\documentclass[11pt,twoside,a4paper,notitlepage]{article}

\usepackage{times}
\usepackage[a4paper,margin=1.15in]{geometry}

\usepackage[english]{babel}

\usepackage{fancyhdr} 

\usepackage[sf,SF]{subfigure}

\usepackage{paralist}

\usepackage{ifthen}

\newboolean{maybeSTOC}
\newboolean{maybePoster}
\newboolean{maybeSIAM}

\newboolean{maybeArticle}
\newboolean{maybeReport}     
\newboolean{maybeThesis}     

\ifthenelse
{\isundefined{\affaddr}}
{\setboolean{maybeSTOC}{false}}
{\setboolean{maybeSTOC}{true}}

\ifthenelse
{\isundefined{\conference}}
{\setboolean{maybePoster}{false}}
{\setboolean{maybePoster}{true}}

\ifthenelse
{\isundefined{\sicomp}}
{\setboolean{maybeSIAM}{false}}
{\setboolean{maybeSIAM}{true}}

\ifthenelse
{\isundefined{\chapter}}
{\setboolean{maybeArticle}{true}}
{\setboolean{maybeArticle}{false}}

\usepackage[latin1]{inputenc}
\usepackage{amsmath}
\usepackage{amssymb}
\usepackage{amsfonts}
\usepackage{mathtools}

\usepackage{varioref}

\usepackage{xspace}

\ifthenelse
{\boolean{maybeSTOC} \or \boolean{maybeSIAM}}
{}
{\usepackage{local-amsthm}}

\usepackage{array}

\ifthenelse
{\boolean{maybeSTOC} 
  \or \boolean{maybePoster}
  \or \boolean{maybeSIAM}}
{}
{\usepackage{nada-typography}}

\DeclareMathAlphabet{\mathsfsl}{OT1}{cmss}{m}{sl}

\newcommand{\eqperiod}{\enspace .}
\newcommand{\eqcomma}{\enspace ,}

\newcommand{\italicitem}[1][]{\item[\textit{#1}]}

\newcommand{\introduceterm}[1]{{\emph{#1}}}

\newcommand{\wrt}{with respect to\xspace}

\newcommand{\eg}{for instance\xspace} 

\newcommand{\ie}{i.e.,\ }

\newcommand{\st}{such that\xspace}

\newcommand{\etal}{et al.\@\xspace}

\newcommand{\ifaoif}{if and only if\xspace}
\newcommand{\wolog}{without loss of generality\xspace}
\newcommand{\Wolog}{Without loss of generality\xspace}

\newcommand{\DAG}{directed acyclic graph\xspace}

\newcommand{\polysize}{polynomial-size\xspace}

\newcommand{\superpoly}{superpolynomial\xspace}

\newcommand{\Ordocompact}[1]{{\ensuremath{\mathrm{O} \bigl( #1 \bigr)}}}
\newcommand{\Ordosmall}[1]{{\ensuremath{\mathrm{O} ( #1 )}}}

\newcommand{\Tightcompact}[1]{{\ensuremath{\Theta \bigl( #1 \bigr)}}}
\newcommand{\Tightsmall}[1]{{\ensuremath{\Theta ( #1 )}}}

\newcommand{\Lowerboundcompact}[1]{{\ensuremath{\Omega \bigl( #1 \bigr)}}}
\newcommand{\Lowerboundsmall}[1]{{\ensuremath{\Omega ( #1 )}}}

\newcommand{\lowerboundsmall}[1]{\ensuremath{\omega ( #1 )}}
\newcommand{\polyboundsmall}[1]{\ensuremath{\mathrm{poly} ( #1 )}}

\newcommand{\Bigoh}[1]{{\ensuremath{\mathrm{O} \bigl( #1 \bigr)}}}
\newcommand{\bigoh}[1]{{\ensuremath{\mathrm{O} ( #1 )}}}

\newcommand{\Bigtheta}[1]{{\ensuremath{\Theta \bigl( #1 \bigr)}}}
\newcommand{\bigtheta}[1]{{\ensuremath{\Theta ( #1 )}}}

\newcommand{\Bigomega}[1]{{\ensuremath{\Omega \bigl( #1 \bigr)}}}
\newcommand{\bigomega}[1]{{\ensuremath{\Omega ( #1 )}}}

\newcommand{\Littleomega}[1]{\ensuremath{\omega \bigl( #1 \bigr)}}

\newcommand{\problemlanguageformat}[1]{\textsc{#1}\xspace}
\newcommand{\langstd}{\ensuremath{L}}

\newcommand{\TAUTOLOGY}{\problemlanguageformat{tautology}}

\newcommand{\SATISFIABILITY}{\problemlanguageformat{satisfiability}}
\newcommand{\MINIMALUNSATISFIABILITY}%
  {\problemlanguageformat{minimal unsatisfiability}}

\newcommand{\complclassformat}[1]{\textsf{#1}\xspace}
\newcommand{\cocomplclass}[1]%
        {\mbox{\complclassformat{co}-\complclassformat{#1}}\xspace}

\newcommand{\Pclass}{\complclassformat{P}}
\newcommand{\NP}{\complclassformat{NP}}
\newcommand{\NPclass}{\NP}
\newcommand{\coNP}{\cocomplclass{NP}}

\newcommand{\CoNP}{\coNP}

\newcommand{\DP}{\complclassformat{DP}}

\newcommand{\refsec}[1]{Section~\ref{#1}}

\newcommand{\Refsec}[1]{Section~\ref{#1}}

\newcommand{\reftwosecs}[2]{Sections~\ref{#1} and~\ref{#2}}
\newcommand{\refthreesecs}[3]{Sections~\ref{#1}, \ref{#2}, and~\ref{#3}}

\newcommand{\reffig}[1]{Figure~\ref{#1}}
\newcommand{\reffigP}[1]{Figure~\vref{#1}}
\newcommand{\Reffig}[1]{Figure~\ref{#1}}

\newcommand{\reftwofigs}[2]{Figures~\ref{#1} and~\ref{#2}}

\newcommand{\refth}[1]{Theorem~\ref{#1}}
\newcommand{\reftwoths}[2]{Theorems~\ref{#1} and~\ref{#2}}
\newcommand{\refthreeths}[3]{Theorems~\ref{#1}, \ref{#2}, and~\ref{#3}}

\newcommand{\reflem}[1]{Lemma~\ref{#1}}
\newcommand{\reftwolems}[2]{Lemmas~\ref{#1} and~\ref{#2}}

\newcommand{\refpr}[1]{Proposition~\ref{#1}}
\newcommand{\reftwoprs}[2]{Propositions~\ref{#1} and~\ref{#2}}

\newcommand{\refcor}[1]{Corollary~\ref{#1}}

\newcommand{\refdef}[1]{Definition~\ref{#1}}
\newcommand{\reftwodefs}[2]{Definitions~\ref{#1} and~\ref{#2}}

\newcommand{\refobs}[1]{Observation~\ref{#1}}
\newcommand{\reftwoobs}[2]{Observations~\ref{#1} and~\ref{#2}}

\newcommand{\refex}[1]{Example~\ref{#1}}
\newcommand{\reftwoexs}[2]{Examples~\ref{#1} and~\ref{#2}}

\newcommand{\refproperty}[1]{Property~\ref{#1}}

\newcommand{\Refth}[1]{Theorem~\ref{#1}}
\newcommand{\Reflem}[1]{Lemma~\ref{#1}}
\newcommand{\Refpr}[1]{Proposition~\ref{#1}}
\newcommand{\Refcor}[1]{Corollary~\ref{#1}}

\newcommand{\refthP}[1]{Theorem~\vref{#1}}
\newcommand{\reflemP}[1]{Lemma~\vref{#1}}
\newcommand{\refprP}[1]{Proposition~\vref{#1}}

\newcommand{\refdefP}[1]{Definition~\vref{#1}}

\newcommand{\refpropertyP}[1]{Property~\vref{#1}}

\newcommand{\refrule}[1]{rule~\ref{#1}}
\newcommand{\reftworules}[2]{rules~\ref{#1} and~\ref{#2}}

\newcommand{\refpart}[1]{part~\ref{#1}}
\newcommand{\Refpart}[1]{Part~\ref{#1}}

\newcommand{\refcase}[1]{case~\ref{#1}}

\ifthenelse
{\isundefined{\refeq}}
{\newcommand{\refeq}[1]{\eqref{#1}}}
{\renewcommand{\refeq}[1]{\eqref{#1}}}

\newcommand{\assigneq}{:=}

\newcommand{\ceilingcompact}[1]{\ensuremath{\bigl \lceil #1 \bigr \rceil}}
\newcommand{\ceilingsmall}[1]{\ensuremath{\lceil #1 \rceil}}

\newcommand{\floorcompact}[1]{\ensuremath{\bigl \lfloor #1 \bigr \rfloor}}
\newcommand{\floorsmall}[1]{\ensuremath{\lfloor #1 \rfloor}}

\newcommand{\floor}[1]{\lfloor #1 \rfloor}
\newcommand{\Floor}[1]{\bigl \lfloor #1 \bigr \rfloor}

\newcommand{\Nplus}     {\mathbb{N}^{+}}

\newcommand{\minofsetcompact}[3][:]{\min \bigl\{ #2 #1 #3 \bigr\}}

\newcommand{\Maxofexpr}[2][]{\max_{#1} \bigl\{ #2 \bigr\}}
\newcommand{\Minofexpr}[2][]{\min_{#1} \bigl\{ #2 \bigr\}}
\newcommand{\maxofexpr}[2][]{\max_{#1} \{ #2 \}}
\newcommand{\minofexpr}[2][]{\min_{#1} \{ #2 \}}

\newcommand{\maxofset}[3][:]{\max \{ #2 #1 #3 \}}
\newcommand{\minofset}[3][:]{\min \{ #2 #1 #3 \}}

\DeclareMathOperator{\Expop}{E}

\newcommand{\twincommandJN}[6]%
    {#1#2#3\vphantom{#2#5}\mspace{-2.25mu}#4.#5#6}

\newcommand{\CondExp}[2]%
    {\Expop\twincommandJN{\bigl[}{#1}{\bigl|}{\bigr}{\,#2}{\bigr]}}
\newcommand{\CONDEXP}[2]%
     {\Expop\twincommandJN{\left[}{#1}{\left|}{\right}{\,#2}{\right]}}

\newcommand{\CondProb}[3][]%
    {\Pr_{#1}\twincommandJN{\bigl[}{#2}{\bigl|}{\bigr}{\,#3}{\bigr]}}
\newcommand{\CONDPROB}[3][]%
    {\Pr_{#1}\twincommandJN{\left[}{#2}{\left|}{\right}{\,#3}{\right]}}

\newcommand{\isdistras}[2]{\ensuremath{#1} \sim \ensuremath{#2}}

\newcommand{\funcdescr}[3]{\ensuremath{ #1 : #2 \mapsto #3}}

\newcommand{\domainof}[1]{\ensuremath{\mathrm{Dom} ( #1 )}}

\newcommand{\edges}[1]{{\ensuremath{E( #1 )}}}

\newcommand{\vertices}[1]{{\ensuremath{V( #1 )}}}

\newcommand{\vneighbour}[1]{\ensuremath{N(#1)}}
\newcommand{\vneighbourcompact}[1]{\ensuremath{N\bigl(#1\bigr)}}

\newcommand{\pathstd}{\ensuremath{P}}

\newcommand{\pathfromto}[3]{#1 : #2 \rightsquigarrow #3}

\newcommand{\succnode}[1]{\formatfunctiontoset{succ}(#1)}
\newcommand{\prednode}[1]{\formatfunctiontoset{pred}(#1)}

\newcommand{\setcompact}[1]{{\ensuremath{\bigl\{ #1 \bigr\}}}}
\newcommand{\setsmall}[1]{{\ensuremath{\{ #1 \}}}}
\newcommand{\setdescrcompact}[3][\mid]{{\setcompact{ #2 #1 #3 }}}
\newcommand{\setdescrsmall}[3][\mid]{{\setsmall{ #2 #1 #3 }}}

\newcommand{\setsizecompact}[1]{\ensuremath{\bigl\lvert#1\bigr\rvert}}
\newcommand{\setsizesmall}[1]{\ensuremath{\lvert#1\rvert}}
\newcommand{\setsizelarge}[1]{\ensuremath{\Biggl \lvert#1\Biggr\rvert}}

\newcommand{\set}[1]{\{ #1 \}}
\newcommand{\Set}[1]{\bigl\{ #1 \bigr\}}

\newcommand{\setdescr}[3][\mid]{\setsmall{ #2 #1 #3 }}
\newcommand{\Setdescr}[3][|]%
     {\twincommandJN{\bigl\{}{#2}{\bigl#1}{\bigr}{\,#3}{\bigr\}}}
\newcommand{\SETDESCR}[3][|]%
     {\twincommandJN{\left\{}{#2}{\left#1}{\right}{\,#3}{\right\}}}

\newcommand{\Setdescrbrackets}[3][|]%
     {\twincommandJN{\bigl[}{#2}{\bigl#1}{\bigr}{\,#3}{\bigr]}}
\newcommand{\SETDESCRBRACKETS}[3][|]%
     {\twincommandJN{\left[}{#2}{\left#1}{\right}{\,#3}{\right]}}

\newcommand{\Setsize}[1]{\ensuremath{\bigl\lvert#1\bigr\rvert}}
\newcommand{\setsize}[1]{\ensuremath{\lvert#1\rvert}}

\newcommand{\intersection}{\cap}

\newcommand{\union}{\cup}
\newcommand{\Union}{\bigcup}
\newcommand{\Unionnodisplay}{\textstyle \bigcup}

\newcommand{\unionSP}{\, \union \, }
\newcommand{\intersectionSP}{\, \intersection \, }

\newcommand{\disjointunion}{\ensuremath{\overset{.}{\cup}}}
\newcommand{\disjointunionSP}{\disjointunion}

\newcommand{\DisjointunionInText}%
    {{\smash{\overset{\mbox{\boldmath{.}}}{\bigcup}}}\vphantom{\bigcup}}

\newcommand{\intclcl}[2]{{\ensuremath{[ #1 , #2 ] }}}

\newcommand{\intnfirst}[1]{\ensuremath{[#1]}}

\newcommand{\Lor}{\ensuremath{\bigvee}}

\newcommand{\limpl}{\ensuremath{\rightarrow}}

\newcommand{\olnot}[1]{\overline{#1}}
\newcommand{\stdnot}[1]{\olnot{#1}}
\newcommand{\falsenum}{0}

\newcommand{\cnfform}{\cnfshort for\-mu\-la\xspace}

\newcommand{\cnfshort}{CNF\xspace}

\newcommand{\xcnfform}[1]{\mbox{\ensuremath{#1}-}\cnfform}
\newcommand{\kcnfform}{\xcnfform{\clwidth}}
\newcommand{\xclause}[1]{\mbox{\ensuremath{#1}-clause}\xspace}
\newcommand{\kclause}{\xclause{\clwidth}}

\newcommand{\nvar}{{\ensuremath{n}}}
\newcommand{\nclause}{{\ensuremath{m}}}
\newcommand{\clwidth}{{\ensuremath{k}}}

\newcommand{\randkcnfnclwrepl}[3][\clwidth]%
        {\ensuremath{\mathcal{F}^{#2, #3}_{#1}}}
\newcommand{\randkcnfnclwreplstd}%
        {\randkcnfnclwrepl{\clwidth}{\nvar}{\nclause}}

\newcommand{\israndkcnfnclwrepl}[4]%
  {\isdistras{#1}{\randkcnfnclwrepl[#2]{#3}{#4}}}

\newcommand{\randkcnfprobcl}[3]%
        {\ensuremath{\mathcal{F}^{#2}_{#1} \bigl(#3 \bigr)}}

\newcommand{\pcfor}[4][to]{for #2 := #3 #1 #4 do}
\newcommand{\pcformath}[4][to]%
    {\pcfor[#1]{\ensuremath{#2}}{\ensuremath{#3}}{\ensuremath{#4}}}

\newcommand{\pcassigncompact}[2]{#1 := #2}
\newcommand{\pcassignmathcompact}[2]%
        {\pcassigncompact{\ensuremath{#1}}{\ensuremath{#2}}}

\newcommand{\inductionformat}[1]{\textit{#1}}

\newcommand{\BASE}[1][]
        {\inductionformat
                {%
                        \ifthenelse{\equal{#1}{}}%
                                {Base case: }%
                                {Base case (#1):}%
                }%
        }

\newcommand{\isin}[2]{{\ensuremath{#1 \in #2}}}

\ifthenelse
{\boolean{maybeSTOC}}
{\input{definitions-STOC.tex}}
{\ifthenelse
  {\boolean{maybeSIAM}}
  {\input{definitions-SIAM.tex}}
  {\ifthenelse
    {\boolean{maybeArticle}}
    {
\usepackage[hyperindex]{hyperref}

\newtheorem{standardlocalcounter}{Dummy}[section]
\newtheorem{standardglobalcounter}{Dummy}

\theoremstyle{plain}    

\newtheorem{theorem}[standardlocalcounter]{Theorem}
\newtheorem{lemma}[standardlocalcounter]{Lemma}
\newtheorem{proposition}[standardlocalcounter]{Proposition}
\newtheorem{corollary}[standardlocalcounter]{Corollary}
\newtheorem{observation}[standardlocalcounter]{Observation}

\newtheorem{openproblem}[standardglobalcounter]{Open Problem}

\theoremstyle{definition}

\newtheorem{property}[standardlocalcounter]{Property}
\newtheorem{definition}[standardlocalcounter]{Definition}

\theoremstyle{remark}
\newtheorem{remark}[standardlocalcounter]{Remark}
\newtheorem{example}[standardlocalcounter]{Example}

\newtheoremstyle{meta}
  {3pt}
  {3pt}
  {\scshape \small }
  {}
  {\scshape \small }
  {:}
  { }
  {}

\theoremstyle{meta}

\newtheoremstyle{questions}
  {3pt}
  {3pt}
  {\sffamily \slshape}
  {}
  {\bfseries \sffamily \slshape}
  {:}
  { }
  {}

\theoremstyle{questions}

}
    {\input{definitions-report.tex}}}}

\usepackage{ifpdf}
\usepackage{graphicx}  
\ifpdf         
\fi

\ifthenelse
{\boolean{maybeSTOC}}
{}
{
\setcounter{secnumdepth}{3}
\setcounter{tocdepth}{3}
}

\newcount\hours
\newcount\minutes
\def\SetTime{\hours=\time
\global\divide\hours by 60
\minutes=\hours
\multiply\minutes by 60
\advance\minutes by-\time
\global\multiply\minutes by-1 }
\SetTime
\def\now{\number\hours:\ifnum\minutes<10 0\fi\number\minutes}

\newcommand{\formuladots}{\cdots}

\newcommand{\impl}{\vDash}
\newcommand{\nimpl}{\nvDash}

\newcommand{\singset}[1]{#1}

\newcommand{\proplog}{propo\-sitional logic\xspace}
\newcommand{\tva}{truth value assignment\xspace}

\newcommand{\pps}{propo\-sitional proof system\xspace}

\newcommand{\proofsystemformat}[1]{\ensuremath{\mathfrak{#1}}}

\newcommand{\proofstd}{\ensuremath{\pi}}

\newcommand{\treeresnot}{\proofsystemformat{T}}

\newcommand{\treelikeres}{tree-like resolution\xspace}

\newcommand{\resref}{resolution refutation\xspace}

\newcommand{\resderiv}{resolution derivation\xspace}

\newcommand{\resproof}{resolution proof\xspace}

\newcommand{\resstd}{\ensuremath{\pi}}  

\newcommand{\derivof}[4][\derives]
        {{\ensuremath{{#2} : {#3} \, {#1}\, {#4}}}}

\newcommand{\refof}[2]{\derivof{#1}{#2}{\falsenum}}

\newcommand{\resderivspacesequence}[1]{\setsmall{#1}}

\newcommand{\deriveswithall}%
        {\vdash_{\!\!\!{\scriptscriptstyle \forall}}} 
\newcommand{\notderiveswithall}%
        {\nvdash_{\!\!\!{\scriptscriptstyle \forall}}} 

\newcommand{\resrule}[3]{{\ensuremath{\frac{#1 \quad #2}{#3}}}}

\newcommand{\clcfgtransition}[2]{\ensuremath{#1 \rightsquigarrow #2}}
\newcommand{\clcfgtransitioncrammed}[2]%
        {\ensuremath{#1 \!\rightsquigarrow\! #2}}

\newcommand{\truthval}{{\ensuremath{\nu}}}

\newcommand{\tvastd}{{\ensuremath{\alpha}}}
\newcommand{\tvaalt}{{\ensuremath{\beta}}}

\newcommand{\logeval}[2]{{\ensuremath{#2 ( #1 )}}}
\newcommand{\logevalstd}[1]{{\ensuremath{\tvastd ( #1 )}}}

\newcommand{\modtva}[2]{\ensuremath{#1^{#2}}}

\newcommand{\fstd}{{\ensuremath{F}}}
\newcommand{\fvar}{{\ensuremath{G}}}

\newcommand{\formf}{\ensuremath{F}}
\newcommand{\formg}{\ensuremath{G}}
\newcommand{\formh}{\ensuremath{H}}

\newcommand{\emptycl}{0}

\newcommand{\varx}{\ensuremath{x}}
\newcommand{\vary}{\ensuremath{y}}

\newcommand{\lita}{\ensuremath{a}}
\newcommand{\litb}{\ensuremath{b}}
\newcommand{\litc}{\ensuremath{c}}

\newcommand{\cla}{\ensuremath{A}}
\newcommand{\clb}{\ensuremath{B}}
\newcommand{\clc}{\ensuremath{C}}
\newcommand{\cld}{\ensuremath{D}}

\newcommand{\clausesetformat}[1]{\ensuremath{\mathbb{#1}}}
\newcommand{\clsc}{\clausesetformat{C}}
\newcommand{\clsd}{\clausesetformat{D}}

\newcommand{\setsofvarsorlitlarge}[2]%
        {\ensuremath{\mathit{#1}\left({#2}\right)}}
\newcommand{\setsofvarsorlit}[2]%
        {\ensuremath{\mathit{#1}({#2})}}
\newcommand{\setsofvarsorlitcompact}[2]%
        {\ensuremath{\mathit{#1}\bigl({#2}\bigr)}}
\newcommand{\setsofvarsorlitsmall}[2]
        {\ensuremath{\mathit{#1}({#2})}}

\newcommand{\setsofvarsorlitsup}[3]%
        {\ensuremath{\mathit{#1}^{#2}({#3})}}
\newcommand{\setsofvarsorlitsuplarge}[3]%
        {\ensuremath{\mathit{#1}^{#2}\left({#3}\right)}}
\newcommand{\setsofvarsorlitsupcompact}[3]%
        {\ensuremath{\mathit{#1}^{#2}\bigl({#3}\bigr)}}

\newcommand{\varssmall}[1]{\setsofvarsorlitsmall{Vars}{#1}}

\newcommand{\vars}[1]{\setsofvarsorlitsmall{Vars}{#1}}
\newcommand{\lit}[1]{\setsofvarsorlitsmall{Lit}{#1}}

\newcommand{\derivabbrev}[2]{\bigl( #1 \vdash #2 \bigr)}
\newcommand{\derivabbrevsmall}[2]{( #1 \vdash #2 )}
\newcommand{\derivabbrevcompact}[2]{\bigl( #1 \vdash #2 \bigr)}

\newcommand{\refutabbrevsmall}[1]{\derivabbrevsmall{#1}{\falsenum}}
\newcommand{\refutabbrevcompact}[1]{\derivabbrevcompact{#1}{\falsenum}}

\newcommand{\genericmeasure}[2]{\ensuremath{{\mathit{#1}}_{#2}}}

\newcommand{\genericformsmall}[2]{\ensuremath{\mathit{#1}( #2 )}}

\newcommand{\genericrefsmall}[3]%
    {\ensuremath{{\mathit{#1}}_{#2}\refutabbrevsmall{#3}}}
\newcommand{\genericrefcompact}[3]%
    {\ensuremath{{\mathit{#1}}_{#2}\refutabbrevcompact{#3}}}
\newcommand{\genericderiv}[4]%
    {\ensuremath{{\mathit{#1}}_{#2}\derivabbrev{#3}{#4}}}
\newcommand{\genericderivsmall}[4]%
    {\ensuremath{{\mathit{#1}}_{#2}\derivabbrevsmall{#3}{#4}}}
\newcommand{\genericderivcompact}[4]%
    {\ensuremath{{\mathit{#1}}_{#2}\derivabbrevcompact{#3}{#4}}}
\newcommand{\generictaut}[3]%
    {\ensuremath{{\mathit{#1}}_{#2}\derivabbrev{}{#3}}}
\newcommand{\generictautcompact}[3]%
    {\ensuremath{{\mathit{#1}}_{#2}\derivabbrevcompact{}{#3}}}
\newcommand{\generictautsmall}[3]%
    {\ensuremath{{\mathit{#1}}_{#2}\derivabbrevsmall{}{#3}}}

\newcommand{\sizeofsmall}[1]{\genericformsmall{S}{#1}}

\newcommand{\length}[1][]{\genericmeasure{L}{#1}}

\newcommand{\lengthofsmall}[1]{\genericformsmall{L}{#1}}
\newcommand{\lengthrefsmall}[2][]{\genericrefsmall{L}{#1}{#2}}

\newcommand{\lengthderivsmall}[3][]{\genericderivsmall{L}{#1}{#2}{#3}}

\newcommand{\lengthstd}{\length}
\newcommand{\lengthofarg}[1]{\genericformsmall{L}{#1}}

\newcommand{\lengthref}[2][]{\genericrefsmall{L}{#1}{#2}}

\newcommand{\width}[1][]{\genericmeasure{W}{#1}}
\newcommand{\widthofsmall}[2][]{\genericformsmall{W_{#1}}{#2}}

\newcommand{\widthrefsmall}[2][]{\genericrefsmall{W}{#1}{#2}}
\newcommand{\widthderivsmall}[3][]{\genericderivsmall{W}{#1}{#2}{#3}}

\newcommand{\widthofarg}[2][]{\genericformsmall{W_{#1}}{#2}}

\newcommand{\widthref}[2][]{\genericrefsmall{W}{#1}{#2}}

\newcommand{\clspaceofsmall}[1]{\genericformsmall{Sp}{#1}}

\newcommand{\clspacerefsmall}[2][]{\genericrefsmall{Sp}{#1}{#2}}
\newcommand{\clspacerefcompact}[2][]{\genericrefcompact{Sp}{#1}{#2}}
\newcommand{\clspacederivsmall}[3][]{\genericderivsmall{Sp}{#1}{#2}{#3}}

\newcommand{\clspaceof}[1]{\genericformsmall{Sp}{#1}}

\newcommand{\clspaceref}[2][]{\genericrefsmall{Sp}{#1}{#2}}
\newcommand{\Clspaceref}[2][]{\genericrefcompact{Sp}{#1}{#2}}
\newcommand{\clspacederiv}[3][]{\genericderivsmall{Sp}{#1}{#2}{#3}}
\newcommand{\Clspacederiv}[3][]{\genericderivcompact{Sp}{#1}{#2}{#3}}

\newcommand{\varspaceofsmall}[1]{\genericformsmall{VarSp}{#1}}

\newcommand{\varspacerefsmall}[2][]{\genericrefsmall{VarSp}{#1}{#2}}
\newcommand{\varspacerefcompact}[2][]{\genericrefcompact{VarSp}{#1}{#2}}

\newcommand{\varspaceref}[2][]{\genericrefsmall{VarSp}{#1}{#2}}

\newcommand{\formulaformat}[1]{\ensuremath{\mathit{#1}}}
\renewcommand{\formulaformat}[1]{\mathit{#1}}

\newcommand{\formatconfiguration}[1]{\mathbb{#1}}

\newcommand{\formatpebblingstrategy}[1]{\mathcal{#1}}

\newcommand{\transitionarrow}{\rightsquigarrow}
\newcommand{\pebcfgtransition}[2]%
    {\ensuremath{#1 \transitionarrow #2}}
\newcommand{\pebcfgtransitionsqueeze}[2]%
    {#1 \! \transitionarrow \! #2}

\newcommand{\pebbling}[1][P]{\formatpebblingstrategy{#1}}

\newcommand{\pconf}[1][P]{\formatconfiguration{#1}}

\newcommand{\formatpebblingprice}[1]{\text{\textsl{\textsf{#1}}}}

\newcommand{\pebblingprice}[1]{\formatpebblingprice{Peb}(#1)}
\newcommand{\pebblingpricecompact}[1]%
    {\formatpebblingprice{Peb}\bigl(#1\bigr)}
\newcommand{\bwpebblingprice}[1]{\formatpebblingprice{BW-Peb}(#1)}
\newcommand{\bwpebblingpricecompact}[1]%
    {\formatpebblingprice{BW-Peb}\bigl(#1\bigr)}
\newcommand{\bwpebblingpriceempty}[1]%
    {\formatpebblingprice{BW-Peb}^{\emptyset}(#1)}
\newcommand{\bwpebblingpriceemptycompact}[1]%
    {\formatpebblingprice{BW-Peb}^{\emptyset}\bigl(#1\bigr)}

\newcommand{\pebcost}[1]{\formatpebblingprice{cost} ( #1 )}

\newcommand{\pebcostsmall}[1]{\formatpebblingprice{cost} ( #1 )}

\newcommand{\stoptime}{\tau}

\newcommand{\pebcontr}[2][G]{\ensuremath{\formulaformat{Peb}^{#2}_{#1}}}

\newcommand{\pebtreecontr}[2][h]{\pebcontr[{T_{#1}}]{#2}}

\newcommand{\pebdeg}{\ensuremath{d}}

\newcommand{\pebax}[2][\pebdeg]{\ensuremath{\formulaformat{Ax}^{#1} (#2)}}
\newcommand{\pebaxcompact}[2]%
        [\pebdeg]{\ensuremath{\formulaformat{Ax}^{#1} \bigl(#2 \bigr)}}

\newcommand{\varspeb}[1]{\setsofvarsorlitsup{Vars}{\pebdeg}{#1}}

\newcommand{\pqrxvar}[6]%
    {\ensuremath{\stdnot{\varx({#1})}_{#2} \lor \stdnot{\varx({#3})}_{#4} \lor %
    \sourceclausexvar[#6]{#5}}}
\newcommand{\pqrstdxvar}{\pqrxvar{p}{i}{q}{j}{r}{l}}

\newcommand{\pqr}[6]%
    {\ensuremath{\stdnot{#1}_{#2} \lor \stdnot{#3}_{#4} \lor %
    \sourceclausenodisplay[#6]{#5}}}
\newcommand{\pqrstd}{\pqr{p}{i}{q}{j}{r}{l}}
\newcommand{\pqrall}[6]%
        {\setdescrcompact
        {\pqr{#1}{#2}{#3}{#4}{#5}{#6}}{#2,#4 \in \intnfirst{\pebdeg}}}
\newcommand{\pqrallstd}%
        {\setdescrcompact{\pqrstd}{i,j \in \intnfirst{\pebdeg}}}

\newcommand{\sourceclausexvar}[2][n]%
        {\Lor_{#1 = 1}^{\pebdeg} \varx({#2})_{#1}}
\newcommand{\subsourceclausexvar}[3][n]%
        {\Lor_{#1 = {#2}}^{\pebdeg} \varx({#3})_{#1}}
\newcommand{\targetclausexvar}[1][n]{\sourceclausexvar[#1]{z}}

\newcommand{\sourceclausexvarnodisplay}[2][n]%
        {\textstyle \Lor_{#1 = 1}^{\pebdeg} \varx({#2})_{#1}}

\newcommand{\sourceclause}[2][n]{\Lor_{#1 = 1}^{\pebdeg} #2_{#1}}

\newcommand{\targetclause}[1][n]{\sourceclause[#1]{z}}

\newcommand{\sourceclausenodisplay}[2][n]%
        {\textstyle \Lor_{#1 = 1}^{\pebdeg} #2_{#1}}

\newcommand{\sourceclausesuff}[2][n]{\Lor_{#1 \in \intnfirst{\pebdeg}} #2_{#1}}

\newcommand{\relativisation}[1]%
    {\ensuremath{\formulaformat{Rel}\bigl(#1 \bigr)}}

\usepackage{stmaryrd} 

\newcommand{\nclausesof}[1]{\setsize{#1}}

\newcommand{\formatfunctiontoset}[1]{\mathit{#1}}

\newcommand{\formatfunctiontosubconfiguration}[1]{\mathsf{#1}}

\newcommand{\formatfunctiontomulti}[1]{\mathcal{#1}}

\newcommand{\pebconditional}{conditional\xspace}
\newcommand{\pebunconditional}{unconditional\xspace}
\newcommand{\pebcomplete}{complete\xspace}

\DeclareMathOperator{\dummystar}{*}
\newcommand{\pebblingcontrNT}[2][G]%
 {\ensuremath{\dummystar\!\!\formulaformat{Peb}^{#2}_{#1}}}
\newcommand{\pebcontrNT}[2][G]{\pebblingcontrNT[#1]{#2}}

\newcommand{\somenodetrueclause}[1]{\formulaformat{All}^{+}\!({#1})}
\newcommand{\somenodetrueclausedeg}[2]{\formulaformat{All}_{#1}^{+}({#2})}

\newcommand{\cmvleftdelim}{[}
\newcommand{\cmvrightdelim}{]}
\renewcommand{\cmvleftdelim}{\llbracket}
\renewcommand{\cmvrightdelim}{\rrbracket}

\newcommand{\clmentionvert}[2]{{#1} \cmvleftdelim {#2} \cmvrightdelim}

\newcommand{\unrelatedverticesadj}{non-comparable\xspace}
\newcommand{\relatedverticesadj}{comparable\xspace}

\newcommand{\slashedstrickenletter}[1]{{\backslash\mkern-9mu #1}}
            
\newcommand{\strikethroughcommand}[1]{\slashedstrickenletter{#1}}

\newcommand{\abovevertices}[2][G]%
    {{#1}_{#2}^{\hspace{-0.2 pt}\triangledown}}
\newcommand{\aboveverticesNR}[2][G]%
    {{#1}_{\strikethroughcommand{#2}}^{\hspace{-0.3 pt}\triangledown}}
\newcommand{\belowvertices}[2][G]%
    {{#1}^{#2}_{\hspace{-0.6 pt}\vartriangle}}
\newcommand{\belowverticesNR}[2][G]%
    {{#1}^{\strikethroughcommand{#2}}_{\hspace{-0.6 pt}\vartriangle}}

\newcommand{\pebtreecontrNT}[2][]{\pebcontrNT[T_{#1}]{#2}}

\newcommand{\subconftext}{sub\-con\-figu\-ration\xspace}

\newcommand{\blindependent}{independent\xspace}
\newcommand{\bldependent}{dependent\xspace}

\newcommand{\lpebblingpricecompact}[1]%
    {\formatpebblingprice{L-Peb}\bigl(#1\bigr)}

\newcommand{\scnot}[2]{#1 \langle #2 \rangle}
\newcommand{\scnotcompact}[2]{#1 \bigl\langle #2 \bigr\rangle}

\newcommand{\spmerge}[2]
        {\formatfunctiontosubconfiguration{merge}(#1, #2 )}

\newcommand{\spcanonconfcompact}[1]%
        {\formatfunctiontosubconfiguration{canon}\bigl({#1}\bigr)}

\newcommand{\spprojsubsub}[4]%
    {\formatfunctiontosubconfiguration{proj}_{\scnot{#1}{#2}}(\scnot{#3}{#4})}
\newcommand{\spprojsubsubcompact}[4]%
    {\formatfunctiontosubconfiguration{proj}_{\scnot{#1}{#2}}%
    \bigl(\scnot{#3}{#4}\bigr)}
\newcommand{\spprojsubconf}[3]%
    {\formatfunctiontosubconfiguration{proj}_{\scnot{#1}{#2}}({#3})}
\newcommand{\spprojsubconfcompact}[3]%
    {\formatfunctiontosubconfiguration{proj}_{\scnot{#1}{#2}}\bigl({#3}\bigr)}
\newcommand{\spprojconfsub}[3]%
    {\formatfunctiontosubconfiguration{proj}_{#1}(\scnot{#2}{#3})}
\newcommand{\spprojconfsubcompact}[3]%
    {\formatfunctiontosubconfiguration{proj}_{#1}\bigl(\scnot{#2}{#3}\bigr)}
\newcommand{\spprojconfconf}[2]%
    {\formatfunctiontosubconfiguration{proj}_{#1}({#2})}
\newcommand{\spprojconfconfcompact}[2]%
    {\formatfunctiontosubconfiguration{proj}_{#1}\bigl({#2}\bigr)}

\newcommand{\spclossubcompact}[2]%
        {\formatfunctiontoset{cl}\bigl(\scnotcompact{#1}{#2}\bigr)}

\newcommand{\spintersubcompact}[2]%
        {\formatfunctiontoset{int}\bigl(\scnotcompact{#1}{#2}\bigr)}

\newcommand{\spcoversubcompact}[2]%
        {\formatfunctiontoset{cover}\bigl(\scnotcompact{#1}{#2}\bigr)}

\newcommand{\spcoverconfcompact}[1]%
        {\formatfunctiontoset{cover}\bigl({#1}\bigr)}

\newcommand{\spinducedblack}[1]%
    {\formatfunctiontoset{Bl} (#1)}
\newcommand{\spinducedwhite}[1]%
    {\formatfunctiontoset{Wh} (#1)}
\newcommand{\spinducedblackcompact}[1]%
    {\formatfunctiontoset{Bl} \bigl(#1 \bigr)}
\newcommand{\spinducedwhitecompact}[1]%
    {\formatfunctiontoset{Wh} \bigl(#1 \bigr)}

\newcommand{\pathclausedeg}[2][\pebdeg]%
    {\somenodetrueclausedeg[#1]{\vertexpath{#2}}}
\newcommand{\pathclauseNRdeg}[2][\pebdeg]%
    {\somenodetrueclausedeg[#1]{\vertexpathNR{#2}}}

\newcommand{\blacktruth}[1]{\clausesetformat{B} ( {#1} )}

\newcommand{\blacktruthsmall}[1]{\clausesetformat{B} ( {#1} )}

\newcommand{\blacktruthdegexplicit}[4]%
        {\setdescrcompact
        {{\textstyle \Lor_{#2 = 1}^{#3} {#1}_{#2}}}
        {{#1} \in {#4}}}

\newcommand{\binsubtree}[1]{T^{#1}}

\newcommand{\vertexpath}[1]{{P}^{#1}}
\newcommand{\vertexpathNR}[1]{{P}_{*}^{#1}}

\newcommand{\unrelatedNP}[1]%
        {T \setminus \bigl(\binsubtree{#1} \unionSP \vertexpath{#1} \bigr)}
\newcommand{\unrelatedsmallNP}[1]%
        {T \setminus (\binsubtree{#1} \unionSP \vertexpath{#1} )}

\DeclareMathOperator{\levelop}{level}
\DeclareMathOperator{\minlevelop}{minlevel}
\DeclareMathOperator{\maxlevelop}{maxlevel}

\newcommand{\vlevel}[1]{\levelop({#1})}
\newcommand{\vminlevel}[1]{\minlevelop({#1})}
\newcommand{\vmaxlevel}[1]{\maxlevelop({#1})}

\newcommand{\vminlevelcompact}[1]{\minlevelop \bigl( {#1} \bigr)}
\newcommand{\vmaxlevelcompact}[1]{\maxlevelop \bigl( {#1} \bigr)}

\newcommand{\levelstd}{L}
\newcommand{\levelmin}{\levelstd}
\newcommand{\levelmax}{L_{U}}

\DeclareMathOperator{\topvertexop}{top}
\DeclareMathOperator{\bottomvertexop}{bot}

\newcommand{\maxvertex}[1]{\topvertexop(#1)}
\newcommand{\topvertex}[1]{\topvertexop(#1)}

\newcommand{\minvertex}[1]{\bottomvertexop(#1)}
\newcommand{\bottomvertex}[1]{\bottomvertexop(#1)}

\newcommand{\minelement}[1]{#1}

\newcommand{\pyramidgraph}[1][]{\Pi_{#1}}
\newcommand{\pyramidgraphh}{\pyramidgraph[h]}

\newcommand{\aboveverticespyramidNR}[1]{\aboveverticesNR[\pyramidgraph]{#1}}
\newcommand{\belowverticespyramid}[1]{\belowvertices[\pyramidgraph]{#1}}

\newcommand{\pebpyramidcontr}[2][]{\pebcontr[\Pi_{#1}]{#2}}

\newcommand{\sourcepath}{source path\xspace}

\newcommand{\pathsoperator}{\mathfrak{P}}
\newcommand{\setofpathsstd}{\mathfrak{P}}

\newcommand{\pathsviaop}{\pathsoperator_{\textrm{via}}}
\newcommand{\pathsinop}{\pathsoperator_{\textrm{in}}}

\newcommand{\pathsviavertex}[1]{\pathsviaop(#1)}

\newcommand{\pathsinvertex}[1]{\pathsinop(#1)}

\newcommand{\unionpathsviavertex}[1]{\Union \pathsviaop(#1)}

\newcommand{\unionpathsinvertex}[1]{\Union \pathsinop(#1)}

\newcommand{\pathsvertex}[1]{\pathsinvertex{#1}}
\newcommand{\pathsviachain}[1]{\pathsviavertex{#1}}

\newcommand{\tightnessklawe}{tightness\xspace}
\newcommand{\tightklawe}{tight\xspace}
\newcommand{\Tightklawe}{Tight\xspace}
\newcommand{\connectedklawe}{\hidingklawe-connected\xspace}

\newcommand{\hideklawe}{hide\xspace}

\newcommand{\hiddenklawe}{hidden\xspace}

\newcommand{\hidingklawe}{hiding\xspace}
\newcommand{\Hidingklawe}{Hiding\xspace}
\newcommand{\hidingsetklawe}{hiding set\xspace}
\newcommand{\Hidingsetklawe}{Hiding set\xspace}            

\newcommand{\Necessaryhidingsetklawe}{Necessary hiding subset\xspace}

\newcommand{\hidsetgraph}{\mathcal{H}}

\newcommand{\minmeasure}{minimum-measure\xspace}
\newcommand{\minsize}{minimum-size\xspace}
\newcommand{\subsetminimal}{subset-minimal\xspace}
\newcommand{\Minmeasure}{Minimum-measure\xspace}

\DeclareMathOperator{\potentialop}{pot}

\newcommand{\vpotential}[2][]{\potentialop_{#1}({#2})}

\newcommand{\meastopot}[2][]{\vmeasure[#1]{#2}}

\newcommand{\vertabovelevel}[2]{{#1} \{ {\succeq\!#2} \} }
\newcommand{\vertstrictlyabovelevel}[2]{{#1} \{ {\succ\!#2} \} }
\newcommand{\vertbelowlevel}[2]{{#1} \{ {\preceq\!#2} \} }
\newcommand{\vertstrictlybelowlevel}[2]{{#1} \{ {\prec\!#2} \} }
\newcommand{\vertonlevel}[2]{{#1} \{ {\sim\!#2} \} }

\newcommand{\abovelevelblockerminsize}[2]{L_{\succeq{#1}}({#2})}
\newcommand{\abovelevelblockerminsizecompact}%
    [2]{L_{\succeq{#1}}\bigl({#2}\bigr)}

\newcommand{\vjthmeasure}[3][]{m_{#1}^{#2} ( {#3} ) }
\newcommand{\vmeasure}[2][]{m_{#1} ( {#2} ) }
\newcommand{\Vjthmeasure}[3][]{m_{#1}^{#2} \bigl( {#3} \bigr) }
\newcommand{\Vmeasure}[2][]{m_{#1} \bigl( {#2} \bigr) }

\newcommand{\vjthmeasurecompact}[3][]{m_{#1}^{#2} \bigl( {#3} \bigr) }
\newcommand{\vmeasurecompact}[2][]{m_{#1} \bigl( {#2} \bigr) }

\newcommand{\measureleq}{\precsim_{m}}

\newcommand{\hiddenvertices}[1]{\llceil {#1} \rrceil}

\newcommand{\necessaryhidingvert}[2]%
{{#1}{\scriptstyle{\llfloor {#2} \rrfloor}}}

\newcommand{\klawepropertyprefix}{limited hiding-cardinality\xspace}
\newcommand{\Klawepropertyprefix}{Limited hiding-cardinality\xspace}
\newcommand{\KLAWEPROPERTYPREFIX}{Limited Hiding-Cardinality\xspace}

\newcommand{\klawepropacronym}{LHC property\xspace}
\newcommand{\localklawepropacronym}{Local LHC property\xspace}
\newcommand{\genklawepropacronym}{Generalized LHC property\xspace}
\newcommand{\GENKLAWEPROPACRONYM}{Generalized LHC Property\xspace}

\newcommand{\klaweprop}{\Klawepropertyprefix property\xspace}

\newcommand{\KLAWEPROP}{\KLAWEPROPERTYPREFIX Property\xspace}

\newcommand{\localklaweprop}{Local \klawepropertyprefix property\xspace}
\newcommand{\Localklaweprop}{Local \klawepropertyprefix property\xspace}

\newcommand{\genklaweprop}{Generalized \klawepropertyprefix property\xspace}
\newcommand{\Genklaweprop}{Generalized \klawepropertyprefix property\xspace}
\newcommand{\GENKLAWEPROP}{Generalized \KLAWEPROPERTYPREFIX Property\xspace}

\newcommand{\nongenklaweprop}%
{non-generalized \Klawepropertyprefix property\xspace}
\newcommand{\nongenklawepropacronym}%
{non-generalized \klawepropacronym}
\newcommand{\nongenklawepropacronymWithParam}%
{(non-generalized) \klawepropacronym}

\newcommand{\gkpconstant}{C_{K}}

\newcommand{\nepath}{\mbox{NE-path}\xspace}        
\newcommand{\nepathlong}{north-east path\xspace}        
\newcommand{\nwpath}{\mbox{NW-path}\xspace}        
\newcommand{\nwpathlong}{north-west path\xspace}        

\newcommand{\nepaththrough}[1]{P_{\textrm{NE}}({#1})}
\newcommand{\nwpaththrough}[1]{P_{\textrm{NW}}({#1})}

\newcommand{\siblingnonreachabiblitypropertynoref}%
{Sibling non-reachability property\xspace}
\newcommand{\Siblingnonreachabiblitypropertynoref}%
{Sibling non-reachability property\xspace}
\newcommand{\siblingnonreachabiblityproperty}%
{\siblingnonreachabiblitypropertynoref~%
\ref{property:sibling-non-reachability-property}\xspace}
\newcommand{\Siblingnonreachabiblityproperty}%
{\Siblingnonreachabiblitypropertynoref~%
\ref{property:sibling-non-reachability-property}\xspace}

\newcommand{\pebblingdag}{blob-pebblable DAG\xspace}
\newcommand{\Pebblingdag}{Blob-pebblable DAG\xspace}

\newcommand{\mergervertex}{v^*}

\newcommand{\introducetermanmpctext}%
    {a \introduceterm{\mpctext{}}\xspace}

\newcommand{\introducetermamultipebblingtext}%
  {a \introduceterm{\multipebblingtext{}}\xspace}

\newcommand{\blobpebblingtext}{blob-pebbling\xspace}
\newcommand{\Blobpebblingtext}{Blob-pebbling\xspace}
\newcommand{\BLOBPEBBLINGTEXT}{Blob-Pebbling\xspace}

\newcommand{\multipebblingtext}{\blobpebblingtext}
\newcommand{\Multipebblingtext}{\Blobpebblingtext}
\newcommand{\MULTIPEBBLINGTEXT}{\BLOBPEBBLINGTEXT}

\newcommand{\blobpebblegame}{blob-pebble game\xspace}

\newcommand{\multipebblegame}{blob-pebble game\xspace}
\newcommand{\Multipebblegame}{Blob-pebble game\xspace}
\newcommand{\MULTIPEBBLEGAME}{Blob-Pebble Game\xspace}
\newcommand{\multipebble}{blob\xspace}
\newcommand{\Multipebble}{Blob\xspace}

\newcommand{\blobpebble}{blob-pebble\xspace}

\newcommand{\blob}{blob\xspace}

\newcommand{\atomicmultipebbleadj}{atomic\xspace}

\newcommand{\Anatomicmultipebbleadj}{An \atomicmultipebbleadj}

\newcommand{\multipebbling}[1][P]{\pebbling[#1]}
\newcommand{\mpconf}[1][S]{\pconf[#1]}

\newcommand{\mpcost}[1]{\formatpebblingprice{cost}( #1 )}

\newcommand{\mpcostblack}[1]%
        {\formatpebblingprice{cost}_{\mpcblacks}( #1 )}
\newcommand{\mpcostwhite}[1]%
        {\formatpebblingprice{cost}_{\mpcwhites}( #1 )}

\newcommand{\blobpebblingprice}[1]{\formatpebblingprice{Blob-Peb}(#1)}
\newcommand{\blobpebblingpricecompact}[1]%
    {\formatpebblingprice{Blob-Peb}\bigl(#1\bigr)}
\newcommand{\multipebblingprice}[1]{\formatpebblingprice{Blob-Peb}(#1)}
\newcommand{\multipebblingpricecompact}[1]%
    {\formatpebblingprice{Blob-Peb}\bigl(#1\bigr)}

\newcommand{\lpptext}{legal pebble positions\xspace}
\newcommand{\lpptextsing}{legal pebble position\xspace}

\newcommand{\lpp}[1]{\formatfunctiontoset{lpp} ( #1 )}
\newcommand{\lppstd}{\lpp{B}}

\newcommand{\pathsBminusB}{\pathsvertex{B} \setminus{B}}

\newcommand{\mpcblacks}{\formatfunctiontomulti{B}}
\newcommand{\mpcwhites}{\formatfunctiontomulti{W}}

\newcommand{\mpcwhitesof}[1]{\formatfunctiontomulti{W}({#1})}

\newcommand{\mpscblacknot}[1]{[ {#1} ]}
\newcommand{\blackblobnot}[1]{\mpscblacknot{#1}}

\newcommand{\mpscnot}[2]{[ {#1} ] \langle {#2} \rangle}
\newcommand{\mpscnotcompact}[2]%
        {\big[ {#1} \big] \bigl\langle {#2} \bigr\rangle}

\newcommand{\mpscnotstd}[1][]{\mpscnot{B_{#1}}{W_{#1}}}
\newcommand{\mpscnotprime}{\mpscnot{B'}{W'}}

\newcommand{\intrompscnot}[1]{\mpscnot{#1}{\prednode{#1}}}

\newcommand{\unconditionalblackmpscnot}[1]{\mpscnot{#1}{\emptyset}}

\newcommand{\mpctext}{\blobpebblingtext con\-fig\-u\-ra\-tion\xspace}

\newcommand{\anmpctext}{a \mpctext}

\newcommand{\Mpsctext}{Sub\-con\-figu\-ration\xspace}
\newcommand{\mpsctext}{sub\-con\-figu\-ration\xspace}

\newcommand{\anmpsctext}{a sub\-con\-figu\-ration\xspace}

\newcommand{\Mpscfulltext}{Blob sub\-con\-figu\-ration\xspace}
\newcommand{\mpscfulltext}{blob sub\-con\-figu\-ration\xspace}

\newcommand{\predrequalpq}{\prednode{r} = \pqset}
\newcommand{\mpscrpq}{\mpscnot{r}{p,q}}

\newcommand{\mpinducedconf}[1]{\mpconf(#1)}

\newcommand{\mpinducedctminusone}{\mpinducedconf{\clsc_{t-1}}}
\newcommand{\mpinducedct}{\mpinducedconf{\clsc_{t}}}

\newcommand{\sharpimpl}{\vartriangleright}

\newcommand{\sharpimpladv}{precisely\xspace}
\newcommand{\sharpimpladj}{precise\xspace}
\newcommand{\sharpimplsubst}{precise implication\xspace}

\newcommand{\BuW}{B \unionSP W}
\newcommand{\Bup}{B \unionSP \singset{p}}
\newcommand{\Buq}{B \unionSP \singset{q}}
\newcommand{\Bur}{B \unionSP \singset{r}}
\newcommand{\Buv}{B \unionSP \singset{v}}
\newcommand{\Sur}{S \unionSP \singset{r}}
\newcommand{\Wur}{W \unionSP \singset{r}}
\newcommand{\Wup}{W \unionSP \singset{p}}

\newcommand{\Wpuq}{W_p \unionSP \singset{q}}
\newcommand{\Wqup}{W_q \unionSP \singset{p}}
\newcommand{\WpuWq}{W_p \unionSP W_q}
\newcommand{\vrset}{\setsmall{v,r}}
\newcommand{\vrsetNP}{v,r}
\newcommand{\pqset}{\setsmall{p,q}}

\newcommand{\chargeabletext}{chargeable\xspace}        
\newcommand{\chargeableverticestext}{chargeable vertices\xspace}

\newcommand{\chargeablevertices}[1]%
{\formatfunctiontoset{chargeable}({#1}) }
\newcommand{\chargeableverticescompact}[1]%
{\formatfunctiontoset{chargeable}\bigl({#1}\bigr) }

\newcommand{\blackschargedfor}[1][]%
    {\mpcblacks_{#1}}
\newcommand{\whiteschargedfor}[1][]%
    {\mpcwhites_{#1}^{\hspace{-0.3 pt}\vartriangle}}

\newcommand{\reprset}{R}
\newcommand{\reprsetprimed}{R'}
\newcommand{\reprsetith}{R[i]}
\newcommand{\reprvertex}{r}

\DeclareMathOperator{\unblockedop}{unblocked}
\newcommand{\unblocked}[1]{\unblockedop({#1})}

\DeclareMathOperator{\elimop}{elim}

\newcommand{\welim}[2]{\text{$\mpcwhites$-$\elimop$}({#1}, {#2})}
\newcommand{\weliminate}{\mbox{$\mpcwhites$-elim}\-i\-nate\xspace}
\newcommand{\weliminated}{\mbox{$\mpcwhites$-elim}\-i\-nated\xspace}
\newcommand{\welimination}{\mbox{$\mpcwhites$-elim}\-i\-nation\xspace}

\newcommand{\kmatchedindex}{m}

\newcommand{\weightklawe}[1]{w({#1})}
\newcommand{\weightklawecompact}[1]{w\bigl({#1}\bigr)}

\newcommand{\ujustblocking}{U_B}
\newcommand{\uhiding}{U_H}
\newcommand{\uhidingith}[1][i]{U_{H}^{#1}}
\newcommand{\utight}{U_T}
\newcommand{\utightith}{U_{T}^{i}}

\newcommand{\unewblockingset}{U^{*}}

\newcommand{\unewhidingith}[1][i]{U_{*}^{#1}}

\newcommand{\mpconfjustblocked}{\mpconf_B}
\newcommand{\mpconfhidden}{\mpconf_H}

\newcommand{\blacksjustblocked}{\mpcblacks_B}
\newcommand{\blackshidden}{\mpcblacks_H}
\newcommand{\blackshiddenith}{\mpcblacks_{H}^{i}}

\newcommand{\bbvstd}{\mpcblacks}

\newcommand{\whitesbelowjustblocked}%
    {\mpcwhites_{B}^{\hspace{-0.3 pt}\vartriangle}}
\newcommand{\whitesbelowhidden}%
    {\mpcwhites_{H}^{\hspace{-0.3 pt}\vartriangle}}
\newcommand{\whitesbelowhiddenith}{\mpcwhites_{H}^{i}}
\newcommand{\whitestight}%
    {\mpcwhites_{T}^{\hspace{-0.3 pt}\vartriangle}}
\newcommand{\whitestightith}{\mpcwhites_{T}^{i}}

\renewcommand{\theenumi}{\arabic{enumi}}
\renewcommand{\theenumii}{\alph{enumii}}
\renewcommand{\theenumiii}{\roman{enumiii}}
\renewcommand{\theenumiv}{\Alph{enumiv}}     

\renewcommand{\labelenumi}{\theenumi.}

\numberwithin{equation}{section}

\newtheoremstyle{metacommenttheoremstyle}
  {3pt}
  {3pt}
  {\sffamily \itshape \scriptsize
  }
  {}
  {\bfseries \scshape \footnotesize }
  {:}
  { }
  {}

\theoremstyle{metacommenttheoremstyle}

\newcommand{\editcomment}[1]%
{}                                           

\newcommand{\proofcomment}[1]%
{}                                         

\newcommand{\conferenceversioncomment}[1]%
{}                                             

\newcommand{\theauthorJN}{the first author\xspace}

\newcommand{\reftwosecapps}[2]{Sections~\ref{#1} and~\ref{#2}}

\newcommand{\sectionappendixtext}{section\xspace}

\begin{document}

%
%

\title{Towards an Optimal Separation of \\ Space and Length in Resolution%
  \thanks{This is the full-length version of the paper~%
    \cite{NH08TowardsOptimalSeparationSTOC}
    to appear at \emph{STOC~'08}.}}

\author{Jakob Nordström\thanks{%
    Research supported in part by 
    grants from the foundations
    \emph{Johan och Jakob Söderbergs stiftelse}
    and
    \emph{Sven och Dagmar Saléns stiftelse}.}
 \hspace{1.7cm} Johan Håstad \\
                                                    \\
  {\normalsize Royal Institute of Technology (KTH)} \\
  {\normalsize  SE-100 44 Stockholm, Sweden }       \\
  \texttt{\{jakobn,johanh\}@kth.se}}

\date{February 29, 2008}

\maketitle

%
%

\thispagestyle{empty}
%
%
%
%

%
%
%

%
%

\pagestyle{fancy}     
\fancyhead{}         
\fancyfoot{}         
\renewcommand{\headrulewidth}{0pt}
\renewcommand{\footrulewidth}{0pt}
%
%
\fancyhead[CE]{\slshape TOWARDS AN OPTIMAL SEPARATION}
\fancyhead[CO]{\slshape \leftmark}
\fancyfoot[C]{\thepage}
\setlength{\headheight}{13.6pt}

%
%

%

\begin{abstract}
    Most state-of-the-art satisfiability algorithms today are variants
  of the DPLL procedure augmented with clause learning. The main
  bottleneck for such algorithms, other than the obvious one of time,
  is the amount of memory used.  In the field of proof complexity, the
  resources of time and memory correspond to the length and space of
  resolution proofs. There has been a long line of research trying to
  understand these proof complexity measures, as well as relating them
  to the width of proofs, \ie the size of the largest clause in the
  proof, which has been shown to be intimately connected with both
  length and space.
  While strong results have been proven for length and width, our
  understanding of space is still quite poor.  For instance, it
  has remained open whether the fact that a formula is provable in
  short length implies that it is also provable in small space (which
  is the case for length 
  versus
  width), or whether on the contrary
  these measures are completely unrelated in the sense that short
  proofs can be arbitrarily complex with respect to space.
  
  In this paper, we present some evidence that the true answer should
  be that the latter case holds and provide 
  a possible roadmap 
  for how such an optimal separation result could be obtained. 
  We do this by proving a tight bound of $\Tightsmall{\sqrt{n}}$ on
  the space needed for \mbox{so-called} pebbling contradictions over pyramid
  graphs of size~$n$. 
\ifthenelse
{\boolean{maybeSTOC}}
{}
{%
  This yields the first polynomial lower bound on space that is not a
  consequence of a corresponding lower bound on width, as well as an
  improvement of the weak separation of space and width in
  (Nordström~2006) from logarithmic to polynomial.  
}
  
  Also, continuing the line of research initiated by (Ben-Sasson~2002)
  into trade-offs between different proof complexity measures, we
  present a simplified proof of the recent length-space trade-off
  result in (Hertel and Pitassi~2007), and show how our ideas can be
  used to prove a couple of other exponential trade-offs in
  resolution.
\end{abstract}

%
%
\newboolean{STOCsavespace}

\ifthenelse{\boolean{maybeSTOC}}
{\setboolean{STOCsavespace}{true}}
{\setboolean{STOCsavespace}{false}}

\section{Introduction}
\label{sec:introduction}

Ever since the 
\ifthenelse{\boolean{maybeSTOC}}
{ground-breaking}
{fundamental}
\NP-completeness result of 
Cook~\cite{Cook71CooksTheorem}, 
the problem of deciding 
whether a given \proplog formula in conjunctive normal form (CNF)
is satisfiable or not has been on center stage in Theoretical
Computer Science.  In more recent years, \SATISFIABILITY has gone from
a problem of mainly theoretical interest to a practical approach for
solving applied problems.  Although all known Boolean satisfiability
solvers (SAT-solvers) have exponential running time in the worst case,
enormous progress in performance has led to satisfiability algorithms
becoming a standard tool for solving a large number of real-world
problems such as hardware and software verification, experiment
design, circuit diagnosis, and scheduling.

A somewhat surprising aspect 
of this development is that the
most successful SAT-solvers to date are still variants of
the resolution-based 
\ifthenelse{\boolean{maybeSTOC}}
{so-called DPLL procedure}
{Davis-Putnam-Logemann-Loveland (DPLL) procedure}
\mbox{\cite{DLL62MachineProgram,
  DP60ComputingProcedure}}
augmented with 
\introduceterm{clause learning}.
For instance, the great majority of the best algorithms at the 2007
round of the international SAT competitions~\cite{SATcompetition} fit
this description.
DPLL procedures perform a recursive backtrack search in the space of
partial truth value assignments. The idea behind clause learning, or
\introduceterm{conflict-driven learning}, is that at each failure
(backtrack) point in the search tree, the system derives a reason for
the inconsistency in the form of a new clause and then adds this
clause to the original \cnfform (``learning'' the clause).  This can
save a lot of work later on in the 
\ifthenelse{\boolean{maybeSTOC}}
{search,} 
{proof search,} 
when some other
partial truth value assignment fails for similar reasons.
The main bottleneck for this approach, other than the obvious one of
time, is the amount of memory used 
\ifthenelse{\boolean{maybeSTOC}}
{%
  by the algorithms. Thus,%
}
{%
  by the algorithms. 
  Since there is only a finite amount of space, all clauses cannot be
  stored. The difficulty lies in obtaining a highly selective and
  efficient clause caching scheme that nevertheless keeps the
  clauses needed. Thus,%
}
understanding time and memory requirements for clause learning
algorithms, and how these requirements are related to one another, is
a question of great practical importance.
We refer to, e.g.,
\ifthenelse{\boolean{maybeSTOC}}
{\mbox{\cite{BKS03UnderstandingPowerClauseLearning,
KS07StateofSAT}}}
{\mbox{\cite{BKS03UnderstandingPowerClauseLearning,
KS07StateofSAT,
Sabharwal05Thesis}}}
for a more detailed discussion of clause learning (and SAT-solving in
general) with examples of applications.

The study of 
proof complexity originated with the
seminal paper of 
Cook and Reckhow~\cite{CR79Relative}. In its most general form, 
a proof system for a language $\langstd$ is a
predicate $P(x, \proofstd)$, 
computable in time polynomial in
$\setsize{x}$ and $\setsize{\proofstd}$,
such that
for all $x \in \langstd$ there is a string $\proofstd$ 
(a \introduceterm{proof})
for which
$P(x, \proofstd) = 1$, 
whereas for any
$x \not\in \langstd$
it holds for all strings $\proofstd$ that
$P(x, \proofstd) = 0$.
A proof system is said to be polynomially bounded if 
for every $x \in \langstd$
there is a proof $\proofstd_x$ 
of size at most polynomial in~$\setsize{x}$.
%
%
A \introduceterm{\pps{}} is a proof system for
the language of
tautologies in propositional logic. 

From a theoretical point of view, one important motivation for proof
complexity is the intimate connection with the fundamental question of
$\Pclass$ versus $\NPclass$.  Since $\NP$ is exactly the set of
languages with polynomially bounded proof systems, and since
\TAUTOLOGY 
can be seen to be the dual problem of
\SATISFIABILITY,
we have the famous theorem of~%
\cite{CR79Relative}
that \NP = \CoNP \ifaoif there exists a polynomially bounded \pps.
Thus, if it could be shown that there are no
\ifthenelse{\boolean{STOCsavespace}}
{polynomially bounded propositional proof systems,}
{polynomially bounded proof systems for propositional tautologies,}
\Pclass $\neq$ \NP 
would follow as a corollary since \Pclass is closed under complement.
One way of approaching this distant goal is to study stronger and
stronger proof systems and try to prove \superpoly lower bounds on
proof size.  However, although great progress has been made in the
last couple of decades for a variety of proof systems, it seems that
we are still very far from fully understanding the reasoning power of
even quite simple ones.

A second important motivation is that, as was mentioned above, designing
efficient algorithms for proving tautologies (or, equivalently,
testing satisfiability), is a very important problem not only in the
theory of computation but also in applied research and industry.
All automated theorem provers, regardless of whether they actually
produce a written proof, 
explicitly or implicitly define a system in
which proofs are searched for    
and rules which determine what proofs in this system look like.  
Proof complexity analyzes what it takes to
simply write down and verify the proofs that such an automated
theorem-prover might find, ignoring the computational effort needed to
actually find them.  
Thus a lower bound for a proof system tells us that any algorithm, even
an optimal (non-deterministic) one making all the right choices, must
\mbox{necessarily} 
use at least the amount of a certain resource specified by
this bound.  In the other direction, theoretical upper bounds on some
proof complexity measure give us hope of finding good proof search
algorithms \wrt this measure, provided that we can design algorithms
that search for proofs in the system
in an efficient manner.  For DPLL
procedures with clause learning, the time and memory resources used
are measured by the
\introduceterm{length}
and
\introduceterm{space}
of proofs in the resolution proof system.

The field of proof complexity also has rich connections to
cryptography, artificial intelligence and mathematical logic.
\ifthenelse{\boolean{maybeSTOC}}
{
  Two good surveys  
  providing more details are~%
  \mbox{\cite{B00ProofComplexity,
  Segerlind07Complexity}}.%
}
{
  Some good surveys  
  providing more details are~%
  \mbox{\cite{B00ProofComplexity,
      BP98Propositional,
      Segerlind07Complexity}}.%
}

%
%

\subsection{Previous Work}
\label{sec:introduction-previous-work}

Any formula in \proplog can be converted to a \cnfform
that is only linearly larger and is
unsatisfiable \ifaoif the original formula is a tautology.
Therefore, any sound and complete system 
for refuting  \cnfform{}s 
%
%
can be considered as a
general \pps.

\ifthenelse{\boolean{STOCsavespace}}
{%
  Perhaps the single most studied proof system for propositional
  logic, \introduceterm{resolution}, is such a system
  that produces proofs of the unsatisfiability of  \cnfform{}s. 
  \mbox{Resolution} appeared in%
}
{%
  Perhaps the single most studied proof system in propositional
  proof complexity, \introduceterm{resolution}, is such a system
  that produces proofs of the unsatisfiability of  \cnfform{}s. 
  The resolution proof system appeared in%
}
\cite{B37Canonical}
and began to be investigated in connection with automated theorem proving 
in the 1960s
\mbox{\cite{DLL62MachineProgram,
  DP60ComputingProcedure,
  R65Machine-oriented}}.
Because of its simplicity---there is only one derivation rule---and
because all lines in a proof are clauses, 
\ifthenelse{\boolean{STOCsavespace}}
{resolution}
{this proof system}
readily lends itself to proof search algorithms.

Being so simple and fundamental, resolution was also a natural target
to attack when developing methods for proving lower bounds in proof
complexity. 
In this context, it is most straightforward to prove bounds on the
\introduceterm{length} of refutations, \ie the number of clauses,
rather than on the 
\ifthenelse{\boolean{STOCsavespace}}
{total size.}
{total size of refutations.}
The length and size measures are easily seen to be polynomially related.
In 1968, Tseitin~\cite{T68ComplexityTranslated} presented a \superpoly
lower bound on 
\ifthenelse{\boolean{STOCsavespace}}
{length} 
{refutation length} 
for a restricted form of resolution,
called \introduceterm{regular} resolution, but it was not until almost
20~years later that Haken~\cite{H85Intractability} proved the first
\superpoly lower bound for general resolution.
This weakly exponential bound of Haken has later been followed
by many other strong results, among others truly exponential lower
bound on  resolution refutation length for different formula families
in, \eg, 
\ifthenelse{\boolean{maybeSTOC}}
{\cite{BW01ShortProofs,
    CS88ManyHard,
    U87HardExamples}.}
{\cite{BKPS02Efficiency,
    BW01ShortProofs,
    CS88ManyHard,
    U87HardExamples}.}
%
%
%
%
%
%

\ifthenelse{\boolean{maybeSTOC}}
{%
  A second complexity measure for 
  resolution is the \introduceterm{width}, \ie the maximal 
  size of a clause in the refutation.%
}
{%
  A second complexity measure for 
  resolution, first made explicit by Galil~\cite{Galil77Resolution}, is
  the \introduceterm{width}, measured as the maximal size of a clause
  in the refutation.%
}
Ben-Sasson and Wigderson~\cite{BW01ShortProofs}
showed that the minimal width
$\widthref{\fstd}$
of any 
\ifthenelse{\boolean{maybeSTOC}}
{refutation}
{resolution refutation}
of a \kcnfform~$\fstd$ is 
bounded from above by the minimal refutation length 
$\lengthref{\fstd}$ 
by
\begin{equation}
  \label{eq:intro-Ben-Sasson-Wigderson-bound}
  \widthref{\fstd} 
  = 
  \Bigoh{\sqrt{\nvar \log \lengthref{\fstd}}}
  \eqcomma
\end{equation}
where $\nvar$ is the number of variables in~$\fstd$.
Since it is also easy to see that 
\ifthenelse{\boolean{maybeSTOC}}
{refutations} 
{resolution refutations} 
of \polysize
formulas in small width must necessarily be short
(
for the reason that 
$(2 \cdot \# \text{variables})^w$ is an upper bound on
the total number of 
distinct clauses
of width~$w$),
the result in~\cite{BW01ShortProofs} can be interpreted as saying
roughly  that there exists a short refutation of 
\ifthenelse{\boolean{maybeSTOC}}
{$\fstd$}
{the \kcnfform~$\fstd$}
\ifaoif there exists a (reasonably) narrow refutation of~$\fstd$.
This 
gives rise to a natural proof search
heuristic: to find a short refutation, search for
refutations in small width. It was shown in~%
\cite{BIW00Near-optimalSeparation}
that there are formula families for which this heuristic exponentially
outperforms any DPLL procedure regardless of branching function.

The formal study of \introduceterm{space} in resolution was initiated
by Esteban and Torán~%
\cite{ET01SpaceBounds, Toran99LowerBounds}. 
Intuitively, the space 
$\clspaceof{\proofstd}$
\ifthenelse
{\boolean{STOCsavespace}}
{of a refutation} 
{of a resolution refutation} 
$\proofstd$
is the maximal number of
clauses one needs to keep in memory while verifying the refutation,
and the space
$
\clspaceref{\fstd}
$
of refuting~$\fstd$ is defined as the minimal space 
\ifthenelse
{\boolean{STOCsavespace}}
{of any refutation of~$\fstd$.}
{of any 
  refutation of~$\fstd$.}
A number of upper and lower bounds for refutation space in resolution
and other proof systems were subsequently presented in, for example,
\ifthenelse
{\boolean{maybeSTOC}}
{\mbox{\cite{ABRW02SpaceComplexity,
    BG03SpaceComplexity,
    ET03CombinatorialCharacterization}}.}
{\mbox{\cite{ABRW02SpaceComplexity,
    BG03SpaceComplexity,
    EGM04Complexity,
    ET03CombinatorialCharacterization}}.}
Just as for width, the minimum space of refuting a formula can be
upper-bounded by the size of the formula. Somewhat unexpectedly,
however, it also turned out that the lower bounds on resolution
refutation space for 
\ifthenelse
{\boolean{STOCsavespace}}
{several formula families} 
{several different formula families} 
exactly matched previously known lower bounds on refutation 
width. Atserias and Dalmau~\cite{AD02CombinatoricalCharacterization}
showed that this was not a coincidence, but that the inequality
\begin{equation}
  \label{eq:intro-Atserias-Dalmau-bound}
  \widthref{\fstd} 
  \leq
  \clspaceref{\fstd}
  +  
  \bigoh{1}
%
\end{equation}
holds for any \kcnfform~$\fstd$, where the (small) constant term depends
on~$\clwidth$. 
In~\cite{Nordstrom06NarrowProofsMayBeSpaciousSTOCtoappearinSICOMP},
\theauthorJN proved that the inequality~%
\refeq{eq:intro-Atserias-Dalmau-bound}
is asymptotically strict by
exhibiting a \kcnfform family   of size
$\bigoh{n}$
refutable in width
$\widthref{\fstd_n} = \bigoh{1}$
but requiring space
$\clspaceref{\fstd_n} = \bigtheta{\log n}$.
%

The space measure discussed above is known as
\introduceterm{clause space}.
A less well-studied space measure, introduced by
Alekhnovich \etal~%
\cite{ABRW02SpaceComplexity},
is
\introduceterm{variable space},
which counts the maximal number of variable occurrences that must be
kept in memory simultaneously. 
Ben-Sasson~\cite{Ben-Sasson02SizeSpaceTradeoffs}
used this measure to obtain a trade-off result for clause space 
versus width in resolution, proving that there are \kcnfform{}s
$\fstd_n$
that can be refuted in constant clause space and constant width, 
but for which any refutation~$\proofstd_n$ must have
$
\clspaceof{\proofstd_n} 
\cdot
\widthofarg{\proofstd_n}
= \bigomega{n / \log n}
$.
More recently, 
Hertel and Pitassi~\cite{HP07ExponentialTimeSpaceSpeedupFOCS}
showed that there are  \cnfform{}s $\fstd_n$
for which any 
\ifthenelse{\boolean{maybeSTOC}}
{refutation}
{refutation of~$\fstd_n$}
in minimal variable space 
\mbox{$\varspaceref{\fstd_n}$}
must have exponential length, but by adding just $3$ extra units of
storage one can instead get a resolution refutation in linear length.

\subsection{Questions Left Open by Previous Research}
\label{sec:introduction-questions-left-open}

%
%
Despite all the research that has gone into understanding the
resolution proof system, a number of fundamental questions still
remain unsolved.
We touch briefly on two such questions below, and then discuss a third
one, which is the main focus of this paper, in somewhat more detail. 

Equation~
\refeq{eq:intro-Ben-Sasson-Wigderson-bound}
says that short refutation length implies narrow refutation width.
Combining 
Equation~
\refeq{eq:intro-Atserias-Dalmau-bound}
with the observation above that narrow refutations are trivially
short, we get a similar statement that 
small refutation clause space implies short refutation length.
Note, however, that this does \emph{not} 
mean   
that there is a
refutation that is both short and narrow, or that any
small-space refutation must also be short.
The reason is that the resolution refutations on the left- and
right-hand sides of 
\refeq{eq:intro-Ben-Sasson-Wigderson-bound}
and~%
\refeq{eq:intro-Atserias-Dalmau-bound}
need not (and in general will not) be the same one.
%
%
%
%

In view of the minimum-width proof search heuristic mentioned above,
an important question is whether short refutation length of a formula
\ifthenelse{\boolean{maybeSTOC}}
{%
  entails that  there is a refutation that is both short%
} 
{%
  does in fact entail that there is a refutation of it 
  that is both short%
} 
and narrow.  Also, it would be 
\mbox{interesting} 
to know if small space of a
refutation implies that it is short.  It is not known whether there
are such connections or whether on the contrary there exist some kind
of trade-off phenomena here similar to the one for 
space and width    
in~\cite{Ben-Sasson02SizeSpaceTradeoffs}.

A third, even more interesting problem is to clarify the relation
between length and clause space.
For width, 
rewriting  the bound in
\refeq{eq:intro-Ben-Sasson-Wigderson-bound}
in terms of the number of clauses 
$\nclausesof{\fstd_{\nvar}}$
instead of the number of variables we get that
that if the width of refuting $\fstd_{\nvar}$ is 
$\Littleomega{
  \sqrt{
    \nclausesof{\fstd_{\nvar}} \log      \nclausesof{\fstd_{\nvar}}}}$,
then the length of refuting $\fstd_{\nvar}$ must be superpolynomial
in~$\nclausesof{\fstd_{\nvar}}$.
This is known to be almost tight, since
\cite{BG01Optimality}
shows that there is 
a \kcnfform family $\setsmall{\fstd_n}_{n=1}^{\infty}$
with
$\mbox{$\widthref{\fstd_n}$} =
\Bigomega{\sqrt[3]{\nclausesof{\fstd_n}}}$
but
$\lengthref{\fstd_n} = \bigoh{\nclausesof{\fstd_n}}$.
Hence, formula families refutable in polynomial
length can have somewhat wide minimum-width refutations, but not
arbitrarily wide ones.

What does the corresponding relation between 
space and length look
like? 
%
%
The inequality \refeq{eq:intro-Atserias-Dalmau-bound}  
tells us that any correlation between length
and clause space cannot be tighter than the correlation between length and
width, so in particular 
we get from the previous paragraph 
that \kcnfform{}s refutable in polynomial length may have at least
``somewhat spacious'' minimum-space refutations. 
%
At the other end of the spectrum, given any resolution refutation
$\proofstd$ of~$\fstd$ in length~$\lengthstd$ it can be proven using
results from 
\mbox{\cite{ET01SpaceBounds, HPV77TimeVsSpace}} 
that 
$\clspaceof{\proofstd}
=    
\bigoh{\lengthstd / \log \lengthstd}$.
This gives an upper bound on any possible separation of the two
measures.
But is
there a Ben-Sasson--Wigderson kind of upper bound on 
\ifthenelse{\boolean{maybeSTOC}}
{clause space} 
{space} 
in terms of length
similar to \refeq{eq:intro-Ben-Sasson-Wigderson-bound}?
Or are length and  space
on the contrary  unrelated in the sense that there exist
\kcnfform{}s~$\fstd_{\nvar}$ with short refutations but
maximal possible refutation space
$
\mbox{$\clspaceref{\fstd_{\nvar}}$}
=
\Bigomega{\lengthref{\fstd_{\nvar}} / 
  \log \lengthref{\fstd_{\nvar}}}$
in terms of length?

We note that for the restricted case of so-called
tree-like resolution, 
\cite{ET01SpaceBounds}~showed that there is a tight correspondence
between length and 
\ifthenelse{\boolean{maybeSTOC}}
{clause space,} 
{space,} 
exactly as for length versus width. 
The case for general resolution has been discussed in, 
\eg,   
\mbox{\cite{Ben-Sasson02SizeSpaceTradeoffs,
  ET03CombinatorialCharacterization,
  Toran04Space}}, 
but there seems to have been no consensus on what the right answer should
be. However, these papers identify a plausible formula family for
answering the question, namely so-called 
\introduceterm{pebbling contradictions}
defined in terms of pebble games over directed acyclic graphs.

\subsection{Our Contribution}
\label{sec:introduction-our-contribution}

The main result in this paper provides some evidence that the true
answer to the question about the relationship between 
\ifthenelse{\boolean{maybeSTOC}} 
{clause space} 
{space} 
and length
is more likely to be at the latter extreme, \ie that the two measures
can be separated in the strongest sense possible.
\ifthenelse{\boolean{maybeSTOC}} 
{%
  More specifically, as a step towards this goal we prove an
  asymptotically tight bound on the space of refuting pebbling
  contradictions over pyramids.%
} 
{%
  More specifically, as a step towards reaching this goal we prove an
  asymptotically tight bound on the clause space of refuting pebbling
  contradictions over pyramid graphs.%
}

\begin{theorem}
  \label{th:main-theorem}
  The clause space of refuting pebbling contradictions over 
  \ifthenelse{\boolean{maybeSTOC}}
  {pyramid graphs} 
  {pyramids} 
  of height~$h$ in resolution grows as
  $\bigtheta{h}$,
  provided that the number of variables per vertex 
  in the pebbling   contradictions 
  is at least~$2$.
\end{theorem}

\ifthenelse{\boolean{maybeSTOC}}
{%
  This yields the first result separating clause space and length that
  is not a  consequence of a corresponding lower bound on width, as
  well as an exponential improvement of the separation of 
  clause space and width 
  in~\cite{Nordstrom06NarrowProofsMayBeSpaciousSTOCtoappearinSICOMP}.%
}
{%
  This yields the first separation of space and length (in the sense of
  a polynomial lower bound on space for formulas refutable in
  polynomial length)
  that is not a
  consequence of a corresponding lower bound on width, as well as an
  exponential improvement of the separation of space and width 
  in~\cite{Nordstrom06NarrowProofsMayBeSpaciousSTOCtoappearinSICOMP}.%
}

\begin{corollary}
  \label{cor:main-corollary}   
  For all $\clwidth \geq 4$,
  there is a family
  $\setsmall{F_n}_{n=1}^{\infty}$
  of \kcnfform{}s of size 
  $\bigtheta{n}$ 
  that can be refuted
  in resolution
  in length
  $\lengthref{F_n} = \bigoh{n}$
  and width
  $\widthref{F_n} = \bigoh{1}$
  but require clause space
  $\clspaceref{F_n} = \bigtheta{\sqrt{n}}$.
\end{corollary}

%

In addition to our main result, we also make the the observation that
the proof of the recent trade-off result in~%
\cite{HP07ExponentialTimeSpaceSpeedupFOCS}
can be greatly simplified, and the parameters slightly improved.
Using similar ideas, we can also prove exponential trade-offs
for length \wrt clause space and width. Namely, we show that there are
\kcnfform{}s such that if we insist on finding the resolution
refutation in smallest clause space or smallest width, respectively,
then we have to pay with an exponential increase in length.
We state the theorem only for length versus clause space.

\ifthenelse
{\boolean{maybeSTOC}}
{\input{figPebblingContradictionPyramidHeight2.STOC.tex}}
{}

\begin{theorem}
  \label{th:easy-length-clause-space-trade-off}
  \ifthenelse
  {\boolean{maybeSTOC}}
  {%
    There is a family
    $\setsmall{\fstd_n}_{n=1}^{\infty}$
    of \kcnfform{}s
    of size $\bigtheta{n}$ \st:
  }
  {%
    There is a family of \kcnfform{}s
    $\setsmall{\fstd_n}_{n=1}^{\infty}$
    of size $\bigtheta{n}$ \st:
  }
  \begin{compactitem}
  \item
    The minimal  clause space of refuting $\fstd_n$ 
    in resolution
    is
    $\clspaceref{\fstd_n} = \Tightcompact{\sqrt[3]{n}}$.

  \item
    Any resolution refutation
    $\refof{\proofstd}{\fstd_n}$
    in minimal clause space
    must have length
    $\lengthofarg{\proofstd} =  
    \exp \bigl(\Bigomega{\sqrt[3]{n}}\bigr)$.

  \item
    There are  
    \ifthenelse
    {\boolean{maybeSTOC}}
    {refutations $\refof{\proofstd'\!}{\!\fstd_n}$ in}
    {resolution refutations $\refof{\proofstd'}{\fstd_n}$ in} 
    asymptotically minimal 
    clause space
    $\clspaceof{\proofstd'} =
    \Bigoh{\clspaceref{\fstd_n}}$
    and length
    $\lengthofarg{\proofstd'} = \bigoh{n}$, 
    \ie linear in the formula size.
  \end{compactitem}
\end{theorem}

A theorem of exactly the same form can be proven for length versus
width as well.

\ifthenelse{\boolean{maybeSTOC}}    
{\newcommand{\STOCvfill}{\vfill}}
{}

\ifthenelse{\boolean{maybeSTOC}}
{\section{Preliminaries}
\label{sec:sketch-of-preliminaries}}
{
\section{Proof Overview and Paper Organization}
\label{sec:proof-overview-and-paper-organization}

    
Since the proof of our main theorem is fairly involved, we start by
giving an intuitive, high-level description of the proofs of our
results and outlining how this paper is organized.

\subsection{Sketch of Preliminaries}
\label{sec:sketch-of-preliminaries}
}
%
%

A 
\introduceterm{resolution refutation} 
of a \cnfform $\fstd$ 
can be viewed as 
a sequence of
derivation steps on a blackboard.
In each step we may write a clause from~$\fstd$ on the blackboard
(an \introduceterm{axiom} clause),
erase a clause from the blackboard or derive some new
\ifthenelse{\boolean{maybeSTOC}}
{clause implied by the clauses currently written on the blackboard.}
{clause implied by the clauses currently written on the blackboard.%
\footnote{%
  For our proof, it turns out that the exact definition of the
  derivation rule is not essential---our lower bound holds for any
  sound rule. What is important is that we are only allowed to derive  
  new clauses that are implied by the set of clauses currently on the
  blackboard.}} 
The refutation ends when we reach the contradictory empty clause. The 
\introduceterm{length} 
of a  resolution refutation is the number of distinct clauses in the
refutation, the
\introduceterm{width} 
is the size of the largest clause in the refutation, and the 
\introduceterm{clause space} 
is the maximum number of clauses on the blackboard simultaneously.
We write
\mbox{$\lengthrefsmall{\fstd}$},
\mbox{$\widthrefsmall{\fstd}$}
and
\mbox{$\clspacerefsmall{\fstd}$}
to denote the minimum length, width and clause space, respectively, of
any resolution refutation of~$\fstd$. 

The  
\introduceterm{pebble game}  
played on a directed acyclic graph (DAG) $G$  models the 
calculation described by~$G$, where the 
\ifthenelse{\boolean{maybeSTOC}}
{sources}
{source vertices}
contain the input and non-source vertices specify operations on the
values of the predecessors.
Placing a pebble on a vertex~$v$ corresponds to storing in memory the
partial result of the calculation described by the subgraph rooted
at~$v$. Removing a pebble from~$v$ corresponds to deleting the partial
result  of~$v$ from memory. 
\ifthenelse{\boolean{maybeSTOC}}
{%
  A 
  \introduceterm{pebbling}
  of $G$ is a sequence of moves starting with the graph empty
  and ending with all vertices empty except for a pebble on the
  (unique) sink vertex.%
} 
{%
  A 
  \introduceterm{pebbling}
  of a DAG~$G$ is a sequence of moves starting with the empty graph~$G$
  and ending with all vertices in~$G$ empty except for a pebble on the
  (unique) sink vertex.%
} 
The
\introduceterm{cost}
of a pebbling is the maximal number of pebbles used 
simultaneously at any point in time during the pebbling. The 
\introduceterm{pebbling price} 
\ifthenelse{\boolean{maybeSTOC}}
{of $G$} 
{of a DAG~$G$} 
is the minimum cost of any pebbling, 
\ie the minimum number of 
\ifthenelse{\boolean{maybeSTOC}}
{registers} 
{memory registers} 
required to perform the
complete calculation described by$~G$.

\ifthenelse{\boolean{maybeSTOC}}           
{%
  The pebble game on a DAG~$G$ can be encoded as 
  an unsatisfiable   \cnfform
  ${\pebcontr[G]{\pebdeg}}$,
  a so-called
  \introduceterm{pebbling contradiction}
  of degree~$\pebdeg$, 
  as follows
  (see
  \reffig{fig:pebbling-contradiction-for-Pi-2}
  for an 
  example):%
}
{%
  The pebble game on a DAG~$G$ can be encoded as 
  an unsatisfiable   \cnfform
  ${\pebcontr[G]{\pebdeg}}$,
  a so-called
  \introduceterm{pebbling contradiction}
  of degree~$\pebdeg$.
  See
  \reffig{fig:pebbling-contradiction-for-Pi-2}
  for a small example. 
  Very briefly, pebbling contradictions are constructed as follows:%
}
\ifthenelse{\boolean{maybeSTOC}}
{\begin{compactitem}}
{\begin{itemize}}
\ifthenelse{\boolean{maybeSTOC}}{\STOCvfill}{}    
\item Associate $\pebdeg$ variables
  $\varx(v)_1, \ldots, \varx(v)_{\pebdeg}$ with each vertex~$v$
  (in
  \reffig{fig:pebbling-contradiction-for-Pi-2}
  we have $\pebdeg = 2$).
\ifthenelse{\boolean{maybeSTOC}}{\STOCvfill}{}    
\item 
  Specify that all sources have at least one true variable,
  for example, the clause
  $\varx(r)_1 \lor \varx(r)_2$ for the 
  \ifthenelse{\boolean{maybeSTOC}}
  {vertex $r$.}
  {vertex $r$ in \reffig{fig:pebbling-contradiction-for-Pi-2}.}
\ifthenelse{\boolean{maybeSTOC}}{\STOCvfill}{}    
\item 
  \ifthenelse{\boolean{maybeSTOC}}
  {%
    Add clauses propagating the truth from predecessors to
    successors (e.g.\
    clauses 4--7 
    in   \reffig{fig:pebbling-contradiction-for-Pi-2}
    say that
    $
    \mbox{$(\varx(r)_1 \lor  \varx(r)_2 )$}
    \land
    \mbox{$(\varx(s)_1 \lor  \varx(s)_2 )$}
    \limpl
    \mbox{$(\varx(u)_1 \lor  \varx(u)_2 )$}
    $.%
  }
  {%
    Add clauses saying that truth propagates from predecessors to
    successors.
    For instance, 
    for the vertex~$u$ with predecessors $r$ and~$s$,
    clauses 4--7 
    in   \reffig{fig:pebbling-contradiction-for-Pi-2}
    are the CNF encoding of the implication
    $
    \mbox{$(\varx(r)_1 \lor  \varx(r)_2 )$}
    \land
    \mbox{$(\varx(s)_1 \lor  \varx(s)_2 )$}
    \limpl
    \mbox{$(\varx(u)_1 \lor  \varx(u)_2 )$}
    $.%
  }
\ifthenelse{\boolean{maybeSTOC}}{\STOCvfill}{}    
\item 
  \ifthenelse
  {\boolean{maybeSTOC}}
  {To get a contradiction, conclude with} 
  {To get a contradiction, conclude the formula with} 
  $
  \olnot{\varx(z)}_1 \land 
  \formuladots \land 
  \olnot{\varx(z)}_{\pebdeg}
  $ 
  where $z$ is the sink of the DAG.
\ifthenelse{\boolean{maybeSTOC}}{\STOCvfill}{}    
\ifthenelse{\boolean{maybeSTOC}}
{\end{compactitem}}
{\end{itemize}}
We will need the observation from~%
\cite{BIW00Near-optimalSeparation}
that a pebbling contradiction of degree~$\pebdeg$ over a graph with
$n$~vertices can be refuted by resolution
in length
$\Bigoh{\pebdeg^2 \cdot n}$
and width
$\bigoh{\pebdeg}$.

\ifthenelse
{\boolean{maybeSTOC}}
{}
{\begin{figure}[t] %
  \begin{align*}%
    & 
    {( \varx(r)_1 \lor \varx(r)_2 )} 
    &     {\land \ }
    &
    {( \olnot{\varx(u)}_1 \lor \olnot{\varx(v)}_1 \lor \varx(z)_1 \lor \varx(z)_2 ) }
 \\
    {\land \ }
    &
    {( \varx(s)_1 \lor \varx(s)_2 ) }  
    & {\land \ }
    &
    {( \olnot{\varx(u)}_1 \lor \olnot{\varx(v)}_2 \lor \varx(z)_1 \lor \varx(z)_2 ) }
    \\
    {\land \ }
    &
    {( \varx(t)_1 \lor \varx(t)_2 ) }  
    & {\land \ }
    &
    {( \olnot{\varx(u)}_2 \lor \olnot{\varx(v)}_1 \lor \varx(z)_1 \lor \varx(z)_2 ) }
    \\
    {\land \ }
    &
    {( \olnot{\varx(r)}_1 \lor \olnot{\varx(s)}_1 \lor \varx(u)_1 \lor \varx(u)_2 ) } 
    & {\land \ }
    &
    {( \olnot{\varx(u)}_2 \lor \olnot{\varx(v)}_2 \lor \varx(z)_1 \lor \varx(z)_2 ) }
\\
    {\land \ }
    &
    {( \olnot{\varx(r)}_1 \lor \olnot{\varx(s)}_2 \lor \varx(u)_1 \lor \varx(u)_2 ) }
    & {\land \ }
    &  
    {\olnot{\varx(z)}_1}
\\
    {\land \ }
    &
    {( \olnot{\varx(r)}_2 \lor \olnot{\varx(s)}_1 \lor \varx(u)_1 \lor \varx(u)_2 ) }
    & {\land \ }
    &  
    {\olnot{\varx(z)}_2 }
\\
    {\land \ }
    &
    {( \olnot{\varx(r)}_2 \lor \olnot{\varx(s)}_2 \lor \varx(u)_1 \lor \varx(u)_2 ) }
\\ 
    {\land \ }
    &
    {( \olnot{\varx(s)}_1 \lor \olnot{\varx(t)}_1 \lor \varx(v)_1 \lor \varx(v)_2 ) }
\\
    {\land \ }
    &
    {( \olnot{\varx(s)}_1 \lor \olnot{\varx(t)}_2 \lor \varx(v)_1 \lor \varx(v)_2 ) }
\\
    {\land \ }
    &
    {( \olnot{\varx(s)}_2 \lor \olnot{\varx(t)}_1 \lor \varx(v)_1 \lor \varx(v)_2 ) }   
\\
    {\land \ }
    &
    {( \olnot{\varx(s)}_2 \lor \olnot{\varx(t)}_2 \lor \varx(v)_1 \lor \varx(v)_2 ) }
  \end{align*}
  \vspace{-4.6cm}
  \begin{flushright}
    \includegraphics{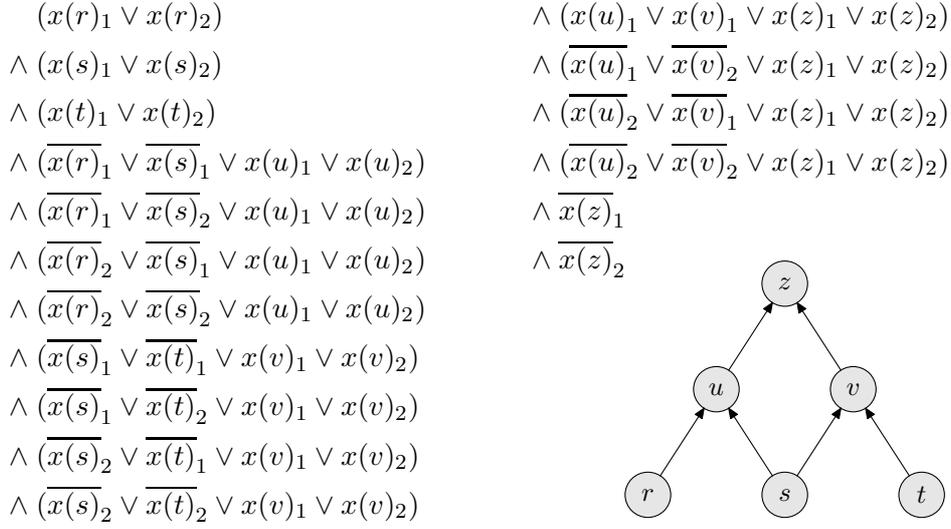} \hspace*{1.35cm}
  \end{flushright}
  \caption{The pebbling contradiction
    $\pebcontr[\Pi_2]{2}$
    for the pyramid graph $\Pi_2$ of height 2.}
  \label{fig:pebbling-contradiction-for-Pi-2}
\end{figure}

}

%
%
\newboolean{STOCshortenSecThreeOne}

\ifthenelse{\boolean{maybeSTOC}}    
{\setboolean{STOCshortenSecThreeOne}{true}}
{\setboolean{STOCshortenSecThreeOne}{false}}

\ifthenelse
{\boolean{maybeSTOC}}
{\newlength{\clauseheight}
\settoheight{\clauseheight}{$\olnot{\varx(s)}_2$}
\setlength{\clauseheight}{1.4\clauseheight}

\ifthenelse
{\boolean{maybeSTOC}}
{\begin{figure*}[t]} 
{\begin{figure}[t]} 
  \subfigure[Clauses on blackboard.]
  {
    \label{fig:intuition-pebbles-pyramid-height-2-a}
    \begin{minipage}[b]{.45\linewidth}
      \centering
      \begin{gather*}
        \left [
          \begin{array}{l}
            { \varx(u)_1 \lor \varx(u)_2 } 
            \rule{0pt}{\clauseheight}
            \\
            { \olnot{\varx(s)}_1 \lor \olnot{\varx(t)}_1 \lor \varx(v)_1 \lor \varx(v)_2 } 
            \rule{0pt}{\clauseheight}
            \\
            { \olnot{\varx(s)}_1 \lor \olnot{\varx(t)}_2 \lor \varx(v)_1 \lor \varx(v)_2 } 
            \rule{0pt}{\clauseheight}
            \\
            { \olnot{\varx(s)}_2 \lor \olnot{\varx(t)}_1 \lor \varx(v)_1 \lor \varx(v)_2 }
            \rule{0pt}{\clauseheight}
            \\
            { \olnot{\varx(s)}_2 \lor \olnot{\varx(t)}_2 \lor \varx(v)_1 \lor \varx(v)_2 }
            \rule{0pt}{\clauseheight}
          \end{array}
        \right ]
      \end{gather*}
    \end{minipage}%
  }
  \hfill
  \subfigure[Corresponding pebbles in the graph.]
  {
    \label{fig:intuition-pebbles-pyramid-height-2-b}
    \begin{minipage}[b]{.45\linewidth}
      \centering
      \includegraphics{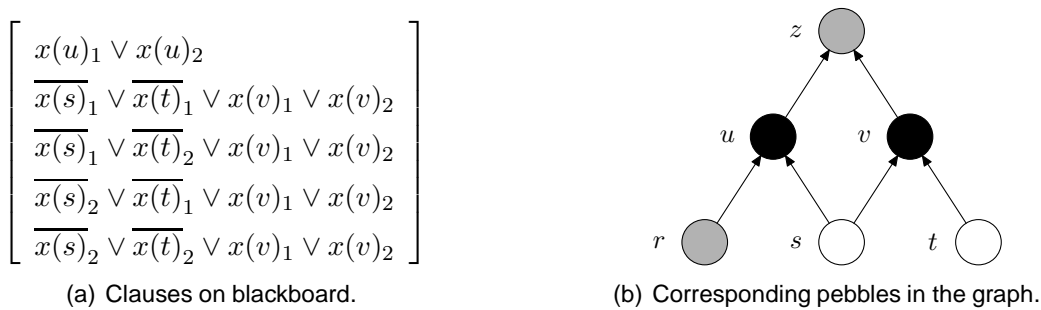}%
    \end{minipage}
  }
%
%
\vspace{-0.1cm}
%
%
  \caption{Example of
    intuitive correspondence between sets of clauses and pebbles.}
  \label{fig:intuition-pebbles-pyramid-height-2}
\ifthenelse
{\boolean{maybeSTOC}}
{\end{figure*}} 
{\end{figure}} 
}
{}

\ifthenelse
{\boolean{maybeSTOC}}
{
\section{Overview of Proofs}
\label{sec:proof-overview}
}
{
\subsection{Proof Idea for  Pebbling Contradictions Space Bound}
\label{sec:proof-overview-idea}
}

Pebble games have been used extensively as a tool to prove time and
space lower bounds and trade-offs for computation. Loosely put, a lower
bound for the pebbling price of a graph says that although the computation
that the graph describes can be performed quickly, it requires large
space.
Our hope is that when we encode pebble games in terms of \cnfform{}s,
these formulas inherit the same properties as the underlying graphs.
That is, if we pick a DAG~$G$ with high pebbling price, since the
corresponding pebbling contradiction encodes a calculation which
\ifthenelse{\boolean{maybeSTOC}}
{needs a lot of memory}
{requires large memory} 
we would like to try to argue that any
resolution refutation of this formula should require large space. Then
a separation result would follow since we already know from~%
\cite{BIW00Near-optimalSeparation}
that the formula can be refuted in short length.

\ifthenelse
{\boolean{maybeSTOC}}
{
\subsection{Proof Idea for  Space Bound}
\label{sec:proof-overview-idea}
}
{}

More specifically, what we would like to do is to establish a
connection between resolution refutations of pebbling contradictions
on the one hand, and the so-called 
\introduceterm{black-white pebble game}~\cite{CS76Storage} 
modelling the non-deterministic computations described by the
underlying graphs on the other. 
\ifthenelse{\boolean{maybeSTOC}}
{Our belief is that} 
{Our intuition is that} 
the resolution
proof system should have to conform to the combinatorics of the pebble
game in the sense that from any 
\ifthenelse{\boolean{maybeSTOC}}
{refutation} 
{resolution refutation} 
of a pebbling
contradiction $\pebcontr[G]{\pebdeg}$ we should be able to extract a
pebbling of the DAG~$G$.

Ideally, we would like to give a proof of a lower bound 
on    
the resolution refutation space of pebbling contradictions 
along the following lines:
\ifthenelse{\boolean{maybeSTOC}}
{\begin{compactenum}}
{\begin{enumerate}}
  
\item
  \label{item:tentative-proof-idea-part-one}
  First, 
  find a natural interpretation of sets of clauses   currently
  ``on the blackboard''   in a refutation of the formula 
  ${\pebcontr[G]{\pebdeg}}$
  in terms of black and white pebbles on the vertices of
  the DAG~$G$.
  
\item
  \label{item:tentative-proof-idea-part-two}
  \ifthenelse
  {\boolean{maybeSTOC}}
  {Then,
  prove that this interpretation
  captures the pebble game in the following sense:
  for any resolution refutation of
  ${\pebcontr[G]{\pebdeg}}$,
  looking at consecutive sets of clauses on the blackboard
  and considering the corresponding sets of pebbles
  we get a black-white pebbling of~$G$.}
  {Then,
  prove that this interpretation of clauses in terms of pebbles
  captures the pebble game in the following sense:
  for any resolution refutation of
  ${\pebcontr[G]{\pebdeg}}$,
  looking at consecutive sets of clauses on the blackboard
  and considering the corresponding sets of pebbles in the graph
  we get a black-white pebbling of~$G$ 
  in accordance with the rules of the pebble game.}

\item 
  \label{item:tentative-proof-idea-part-three}
  Finally, show that the interpretation captures 
  clause
  space in the sense
  that if the content of the blackboard induces $N$ pebbles on the
  graph, then there must be at least $N$ clauses on the blackboard.
%
  
\ifthenelse{\boolean{maybeSTOC}}
{\end{compactenum}}
{\end{enumerate}}
\ifthenelse
{\boolean{STOCshortenSecThreeOne}}
{%
  Combining the above with known lower bounds on the pebbling price
  of~$G$, this would imply a lower bound on the refutation space of
  pebbling contradictions.
  The separation from length and width would then follow from the fact
  that pebbling   contradictions are known to be refutable in linear
  length and   constant width if $\pebdeg$ is fixed.%
} 
{%
  Combining the above with known lower bounds on the pebbling price
  of~$G$, this would imply a lower bound on the refutation space of
  pebbling contradictions and a separation from length and width.  
  For clarity, let us spell out what the formal argument of this would
  look like.
  
  Consider an arbitrary 
  \ifthenelse{\boolean{maybeSTOC}}
  {refutation}
  {resolution refutation}
  of~${\pebcontr[G]{\pebdeg}}$.  From this refutation we extract a
  pebbling of $G$. 
  \ifthenelse{\boolean{maybeSTOC}}           
  {%
    At some time~$t$
    in the obtained pebbling,
    there must be a lot of pebbles in $G$ since this%
  } 
  {%
    At some point in time~$t$ in the obtained pebbling, there must be
    a lot of pebbles on the vertices of $G$ since this%
  } 
  graph was chosen with high pebbling price. But this means that at
  time~$t$, there are a lot of clauses on the blackboard. Since this
  holds for any
  \ifthenelse{\boolean{maybeSTOC}} 
  {refutation,} 
  {resolution refutation,} 
  the refutation space of~${\pebcontr[G]{\pebdeg}}$ must be large.
  The separation result now follows from the fact that pebbling
  contradictions 
  \ifthenelse{\boolean{maybeSTOC}}   
  {are refutable} 
  {are known to be refutable} 
  in linear length and constant width if $\pebdeg$ is fixed.%
}

Unfortunately, this idea does not quite work.  In the next subsection,
we describe the modifications that we are forced to make, and show how
we can make the bits and pieces of our construction fit together to
yield
\refth{th:main-theorem}
and
\refcor{cor:main-corollary}
for the special case of pyramid graphs.

%
%

\subsection{Detailed Overview of Formal Proof  of 
  Space Bound}
\label{sec:proof-overview-detailed}

The black-white pebble game played on a DAG $G$ can be viewed as a way
of proving the end result of the calculation described by~$G$.  Black
pebbles denote proven partial results of the computation.  White
pebbles denote assumptions about partial results which have been used
to derive other partial results (\ie black pebbles), but these
assumptions will have to be verified for the calculation to be
complete. The final goal is a black pebble on the sink~$z$ and no
other pebbles in the graph, corresponding to an unconditional proof of
the end result of the calculation with any assumptions made along the
way having been eliminated.

Translating this to pebbling contradictions, it turns out that a
fruitful way to think of a black pebble on~$v$ is that it should
correspond to truth of the disjunction $\sourceclausexvar[i]{v}$ of
all positive literals over~$v$, or to ``truth of~$v$''.  A white
pebble on a vertex~$w$ can be understood to mean that we need to
\emph{assume} the partial result on $w$ to derive the black pebbles
above $w$ in the graph. Needing to assume the truth of $w$ is the
opposite of knowing the truth of~$w$, so extending the reasoning above
we get that a white-pebbled vertex should correspond to ``falsity
of~$w$'', \ie to all negative literals $\olnot{\varx(w)}_i$, $i \in
\intnfirst{\pebdeg}$, over~$w$.
    
Using this intuitive correspondence, we can translate 
sets of clauses in a resolution refutation of
$\pebcontr[G]{\pebdeg}$
into black and white pebbles in~$G$
as in 
\reffig{fig:intuition-pebbles-pyramid-height-2}.
It is easy to see that 
if we assume
$\varx(s)_1 \lor \varx(s)_2$ 
and
$\varx(t)_1 \lor \varx(t)_2$,
this assumption together with the clauses on the blackboard 
in \reffig{fig:intuition-pebbles-pyramid-height-2-a}
imply 
$\varx(v)_1 \lor \varx(v)_2$,
so $v$ should be black-pebbled and $s$ and~$t$ white-pebbled in
\reffig{fig:intuition-pebbles-pyramid-height-2-b}.
The vertex $u$ is also black since
$\varx(u)_1 \lor \varx(u)_2$ 
certainly is implied by the blackboard.
This translation from clauses to pebbles is arguably quite
straightforward, and seems to yield well-behaved black-white pebblings
for all ``sensible'' resolution refutations of~%
$\pebcontr[G]{\pebdeg}$.

\newlength{\clauseheight}
\settoheight{\clauseheight}{$\olnot{\varx(s)}_2$}
\setlength{\clauseheight}{1.4\clauseheight}

\ifthenelse
{\boolean{maybeSTOC}}
{\begin{figure*}[t]} 
{\begin{figure}[t]} 
  \subfigure[Clauses on blackboard.]
  {
    \label{fig:intuition-pebbles-pyramid-height-2-a}
    \begin{minipage}[b]{.45\linewidth}
      \centering
      \begin{gather*}
        \left [
          \begin{array}{l}
            { \varx(u)_1 \lor \varx(u)_2 } 
            \rule{0pt}{\clauseheight}
            \\
            { \olnot{\varx(s)}_1 \lor \olnot{\varx(t)}_1 \lor \varx(v)_1 \lor \varx(v)_2 } 
            \rule{0pt}{\clauseheight}
            \\
            { \olnot{\varx(s)}_1 \lor \olnot{\varx(t)}_2 \lor \varx(v)_1 \lor \varx(v)_2 } 
            \rule{0pt}{\clauseheight}
            \\
            { \olnot{\varx(s)}_2 \lor \olnot{\varx(t)}_1 \lor \varx(v)_1 \lor \varx(v)_2 }
            \rule{0pt}{\clauseheight}
            \\
            { \olnot{\varx(s)}_2 \lor \olnot{\varx(t)}_2 \lor \varx(v)_1 \lor \varx(v)_2 }
            \rule{0pt}{\clauseheight}
          \end{array}
        \right ]
      \end{gather*}
    \end{minipage}%
  }
  \hfill
  \subfigure[Corresponding pebbles in the graph.]
  {
    \label{fig:intuition-pebbles-pyramid-height-2-b}
    \begin{minipage}[b]{.45\linewidth}
      \centering
      \includegraphics{pyramidHeight2withPebbles.1}%
    \end{minipage}
  }
%
%
\vspace{-0.1cm}
%
%
  \caption{Example of
    intuitive correspondence between sets of clauses and pebbles.}
  \label{fig:intuition-pebbles-pyramid-height-2}
\ifthenelse
{\boolean{maybeSTOC}}
{\end{figure*}} 
{\end{figure}}

The problem is that we have no guarantee that the resolution
refutations will be ``sensible''.  Even though it might seem more or
less clear how an optimal refutation of a pebbling contradiction
should proceed, a particular refutation might contain unintuitive and
seemingly non-optimal derivation steps that do not make much sense
from a pebble game perspective.
In particular, a resolution derivation has no obvious reason always to
derive truth that is restricted to single vertices.  For instance, it
could add the axioms
$\olnot{\varx(u)}_i \lor \olnot{\varx(v)}_2 
\lor \varx(z)_1 \lor \varx(z)_2$,
$i = 1, 2$, 
to the blackboard in
\reffig{fig:intuition-pebbles-pyramid-height-2-a},
derive that the truth of
$s$ and~$t$ implies the truth of either $v$ or~$z$, \ie the clauses
$
\olnot{\varx(s)}_i \lor \olnot{\varx(t)}_j \lor 
\varx(v)_1 \lor \varx(z)_1 \lor \varx(z)_2
$
for $i,j = 1,2$,
and then erase
$\varx(u)_1 \lor \varx(u)_2$ 
from the blackboard.  Although it is hard to see from such a small
example, this turns out to be a serious problem in that there appears
to be no way that we can interpret such derivation steps in terms of
black and white pebbles without making some component in the proof
idea in \refsec{sec:proof-overview-idea} break down.

Instead, what we do is to invent a new pebble game, with white pebbles
just as before, but with black \introduceterm{\multipebble{}s} that
can cover multiple vertices instead of single-vertex black pebbles.  A
\multipebble on a vertex set $V$ can be thought of as truth of some
vertex $v \in V$.  The derivation sketched in the preceding paragraph,
resulting in the set of clauses in
\reffig{fig:intuition-blobs-pyramid-height-2-a}, will then be
translated into white pebbles on $s$ and~$t$ as before and a black
\multipebble covering both $v$ and~$z$ in
\reffig{fig:intuition-blobs-pyramid-height-2-b}.  We define rules in
this \introduceterm{\multipebblegame{}} corresponding roughly to black
and white pebble placement and removal in the usual black-white pebble
game, and add a special \introduceterm{inflation rule} allowing us to
inflate black \multipebble{}s to cover more vertices.

%
%
%

\ifthenelse
{\boolean{maybeSTOC}}
{\begin{figure*}[t]} 
{\begin{figure}[t]} 
  \subfigure[New set of clauses on blackboard.]
  {
    \label{fig:intuition-blobs-pyramid-height-2-a}
    \begin{minipage}[b]{.5\linewidth}
      \centering
      \begin{gather*}
        \left [
          \begin{array}{l}
            { \olnot{\varx(s)}_1 \lor \olnot{\varx(t)}_1 \lor 
              \varx(v)_1 \lor \varx(z)_1 \lor \varx(z)_2 } 
            \rule{0pt}{\clauseheight}
            \\
            { \olnot{\varx(s)}_1 \lor \olnot{\varx(t)}_2 \lor 
              \varx(v)_1 \lor \varx(z)_1 \lor \varx(z)_2 } 
            \rule{0pt}{\clauseheight}
            \\
            { \olnot{\varx(s)}_2 \lor \olnot{\varx(t)}_1 \lor 
              \varx(v)_1 \lor \varx(z)_1 \lor \varx(z)_2 }
            \rule{0pt}{\clauseheight}
            \\
            { \olnot{\varx(s)}_2 \lor \olnot{\varx(t)}_2 \lor 
              \varx(v)_1 \lor \varx(z)_1 \lor \varx(z)_2 }
            \rule{0pt}{\clauseheight}
          \end{array}
        \right ]
      \end{gather*}
    \end{minipage}%
  }
  \hfill
  \subfigure[Corresponding blobs and pebbles.]
  {
    \label{fig:intuition-blobs-pyramid-height-2-b}
    \begin{minipage}[b]{.4\linewidth}
      \centering
      \includegraphics{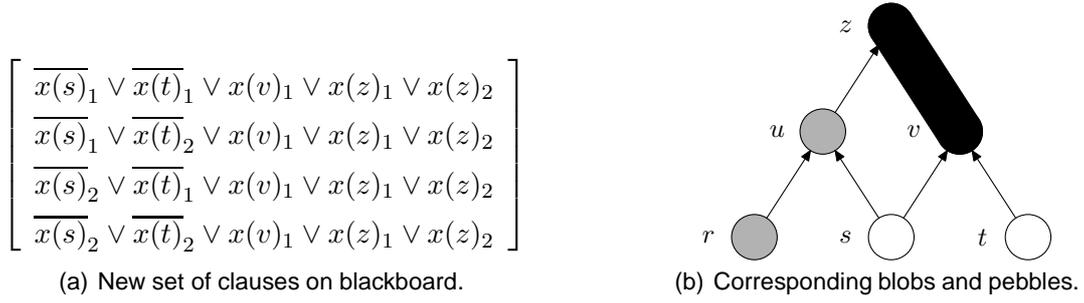}%
      \vspace{-0.1cm}
    \end{minipage}
  }
%
%
\vspace{-0.1cm}
%
%
  \caption{Intepreting sets of clauses as black blobs and white pebbles.}
  \label{fig:intuition-blobs-pyramid-height-2}
\ifthenelse
{\boolean{maybeSTOC}}
{\end{figure*}} 
{\end{figure}}

Once we have this \multipebblegame, we use it to construct a lower
bound proof as outlined in \refsec{sec:proof-overview-idea}.
First, we establish that for a fairly general class of graphs, any
resolution refutation of a pebbling contradiction can be interpreted
as a \multipebblingtext on the DAG in terms of which this pebbling
contradiction is defined.  Intuitively, the reason that this works is
that we can use the inflation rule to analyze apparently non-optimal
steps in the refutation.

\begin{theorem}
  \label{th:overview-pyramid-refutations-and-pebblings}
  Let 
  ${\pebcontr[G]{\pebdeg}}$
  denote the pebbling contradiction of  degree   
  $\pebdeg \geq 1$
  over a layered DAG~$G$.
  Then there is a translation 
  function 
  from
  sets of clauses derived 
  from~${\pebcontr[G]{\pebdeg}}$ 
  into    
  sets of black \multipebble{}s and white pebbles in~$G$ \st
  any resolution   refutation   
  $\proofstd$
  of
  ${\pebcontr[G]{\pebdeg}}$
  corresponds to a  \multipebblingtext~$\multipebbling_{\proofstd}$
  of~$G$
  under this translation.
\end{theorem}

In fact, the only property that we need from the layered graphs in
\refth{th:overview-pyramid-refutations-and-pebblings} is that if $w$
is a vertex with predecessors $u$ and~$v$, then there is no path
between the siblings $u$ and~$v$.
The theorem   
holds for any DAG satisfying this condition.

Next, we carefully design a cost function for black \multipebble{}s
and white pebbles so that the cost of the \multipebblingtext 
$\multipebbling_{\proofstd}$
in 
\refth{th:overview-pyramid-refutations-and-pebblings}
is related to the space of the resolution refutation~$\proofstd$.

\begin{theorem}
  \label{th:overview-pyramid-pebbling-cost-bounded-by-space}
  If 
  $\proofstd$
  is a 
  refutation of a pebbling contradiction
  ${\pebcontr[G]{\pebdeg}}$
  of  degree   
  $\pebdeg > 1$, 
  then the cost of the associated \multipebblingtext
  $\multipebbling_{\proofstd}$
  is 
  bounded by
  the space of $\proofstd$
  by
  $
  \mpcost{\multipebbling_{\proofstd}}
  \leq
  \clspaceof{\proofstd} + \bigoh{1}
  $.
\end{theorem}

Without going into too much detail, in order to make the proof of
\refth{th:overview-pyramid-pebbling-cost-bounded-by-space} work we can
only charge for black \multipebble{}s having distinct lowest vertices
(measured in topological order), so additional \multipebble{}s
with the same bottom vertices are free. Also, we can only charge
for white pebbles below these bottom vertices. 

Finally, we need lower bounds on \multipebblingtext price.
Because of the inflation rule in combination with the peculiar cost
function, the \multipebblegame seems to behave rather differently from
the standard black-white pebble game, and therefore we cannot appeal
directly to known lower bounds on black-white pebbling price.
However, for a more restricted class of graphs than in
\refth{th:overview-pyramid-refutations-and-pebblings}, but still
including binary trees and pyramids, we manage to prove tight bounds
on the \multipebblingtext price by generalizing the lower bound
construction for black-white pebbling in~%
\cite{K80TightBoundPebblesPyramid}.

\begin{theorem}
  \label{th:overview-pyramid-bound-labelled-pebbling-cost}
  Any so-called 
  {layered spreading graph}~$G_h$
  of height~$h$
  has  \multipebblingtext price~$\bigtheta{h}$.
  In particular, this holds for pyramid graphs~%
  $\pyramidgraph[h]$.
\end{theorem}

Putting all of this together, we can prove our main theorem.

\theoremstyle{plain}
\newtheorem*{MainTheorem}
{\Refth{th:main-theorem} (restated)}
\newtheorem*{MainCorollary}
{\Refcor{cor:main-corollary} (restated)}

\begin{MainTheorem}
  Let
  $\pebpyramidcontr[h]{\pebdeg}$
  denote   the pebbling contradiction of  degree   
  $\pebdeg > 1$
  defined over the pyramid graph of height~$h$.
  Then the clause space of refuting  
  $\pebpyramidcontr[h]{\pebdeg}$
  by resolution is 
  $
  \clspaceref{\pebpyramidcontr[h]{\pebdeg}}
  =
  \bigtheta{h}
  $.
\end{MainTheorem}

\begin{proof}
  The upper bound
  $
  \clspaceref{\pebpyramidcontr[h]{\pebdeg}}
  =
  \bigoh{h}
  $
  is easy. A pyramid of height $h$ can be pebbled with
  $h + \bigoh{1}$
  black pebbles,
  and a resolution refutation can mimic such a pebbling in constant
  extra clause space (independent of~$\pebdeg$) to refute
  the corresponding pebbling contradiction.
    
  The interesting part is the lower bound.  Let
  $\proofstd$
  be any resolution refutation of
  ${\pebpyramidcontr[h]{\pebdeg}}$.
  Consider the associated \multipebblingtext 
  $\multipebbling_{\proofstd}$
  provided by  
  \refth{th:overview-pyramid-refutations-and-pebblings}.
  On the one hand, we know that
  \mbox{$
  \mpcost{\multipebbling_{\proofstd}}
  =
  \bigoh{\clspaceof{\proofstd}}
  $}
  by
  \refth{th:overview-pyramid-pebbling-cost-bounded-by-space},
  provided that
  $\pebdeg > 1$.
  On the other hand,
  \refth{th:overview-pyramid-bound-labelled-pebbling-cost}
  tells us that the cost of any \multipebblingtext  of
  $\pyramidgraph[h]$ is $\bigomega{h}$,
  so in particular we must have
  $\mpcost{\multipebbling_{\proofstd}} = \bigomega{h}$.
  Combining these two bounds on
  $\mpcost{\multipebbling_{\proofstd}} $,
  we see that
  $\clspaceof{\proofstd} = \bigomega{h}$.
\end{proof}

The pebbling contradiction
$\pebcontr[G]{\pebdeg}$
is a
\xcnfform{\text{(2+\pebdeg)}}{} 
and for constant $\pebdeg$ the size of the formula
is linear in the number of vertices of~$G$
(compare \reffig{fig:pebbling-contradiction-for-Pi-2}).
Thus, for pyramid graphs $\pyramidgraphh$
the corresponding pebbling contradictions
$\pebpyramidcontr[h]{\pebdeg}$
have size quadratic in the height~$h$.
Also,  when $\pebdeg$ is fixed the upper bounds mentioned at the end of
\refsec{sec:sketch-of-preliminaries}
become
$\lengthref{\pebcontr[G]{\pebdeg}} = \bigoh{n}$
and
$\widthref{\pebcontr[G]{\pebdeg}} = \bigoh{1}$.
\Refcor{cor:main-corollary}
now follows if we set
$\fstd_n = \pebpyramidcontr[h]{\pebdeg}$
for 
$\pebdeg = \clwidth - 2$
and
$h = \floor{\sqrt{n}}$
and use 
\refth{th:main-theorem}.

\begin{MainCorollary}
  For all $\clwidth \geq 4$,
  there is a family
  of \kcnfform{}s 
  $\set{F_n}_{n=1}^{\infty}$
  of size $\bigoh{n}$
  \st
  $\lengthref{F_n} = \bigoh{n}$
  and
  $\widthref{F_n} = \bigoh{1}$
  but
  $\clspaceref{F_n} = \bigtheta{\sqrt{n}}$.
\end{MainCorollary}

\subsection{Overview of Trade-off Results}

Let us also quickly sketch the ideas (or tricks, really) used to prove
our trade-off theorems for resolution.

We show the following version of the 
length-variable space 
trade-off theorem of Hertel and
Pitassi~\cite{HP07ExponentialTimeSpaceSpeedupFOCS}, 
with somewhat improved parameters and a very much simpler proof.

\begin{theorem}
  \label{th:Hertel-Pitassi}
  There is a family of \cnfform{}s $\set{\formf_n}_{n=1}^{\infty}$
  of size $\bigtheta{n}$ \st:
  \begin{compactitem}
  \item 
    The minimal variable space of refuting $\formf_n$
    in resolution
    is
    $\varspacerefsmall{\formf_n} = \bigtheta{n}$.

  \item 
    Any 
    resolution
    refutation $\refof{\proofstd}{\formf_n}$
    in minimal variable space has length
    $\exp (\Lowerboundsmall{\sqrt{n}})$.

  \item 
    \ifthenelse{\boolean{maybeSTOC}}
    {%
      Adding at most $2$ extra units of storage, 
      one can obtain  a refutation in space
      $\mbox{$\varspacerefsmall{\formf_n} + 3$} = \bigtheta{n}$
      and length 
      $\bigoh{n}$,
      \ie linear in the formula size.
    }
    {%
      Adding at most $2$ extra units of storage, it is possible to obtain
      a resolution refutation $\proofstd'$ in variable space
      $ \varspaceofsmall{\proofstd'} =
      \mbox{$\varspacerefsmall{\formf_n} + 3$} = \bigtheta{n}$
      and length 
      $\lengthofsmall{\proofstd'} = \bigoh{n}$,
      \ie linear in the formula size.
    }
  \end{compactitem}
\end{theorem}


The idea behind our proof is as follows. Take formulas 
$\formg_n$
that are really hard for resolution and formulas 
$\formh_m$
which have short refutations but require linear variable space, 
and set 
$
\formf_n =
\formg_n
\land
\formh_m
$
for $m$ chosen so that
$\varspacerefcompact{\formh_{m}}$
is only just larger
\mbox{than $\varspacerefcompact{\formg_{n}}$}.
Then refutations in minimal variable space will have to take care of 
$\formg_n$,
which requires exponential length, but adding one or two literals to the
memory we can attack 
$\formh_{m}$
instead in linear length.

The trade-off result in \refth{th:easy-length-clause-space-trade-off}
for length versus clause space and its twin theorem for length versus
width are shown using similar ideas.
\ifthenelse{\boolean{maybeSTOC}}
{Again, the details can be found in~%
  \cite{NH08TowardsOptimalSeparationECCC}.}
{}

%
%

\subsection{Paper Organization}
    
\Refsec{sec:formal-preliminaries}
provides  
formal definitions of the concepts introduced in
\reftwosecs
{sec:introduction}
{sec:proof-overview-and-paper-organization},
and
\refsec{sec:review-of-some-known-results}
gives precise statements of the results mentioned there, as well as
some other result relevant to this paper. 
The easy proofs of our trade-off theorems
are then immediately presented in
\refsec{sec:simplified-way-of-proving-trade-offs}.

The bulk of the paper is spent proving our main result in
\refth{th:main-theorem}.
In
\refsec{sec:definition-multi-pebble-game},
we define our modified pebble game, the ``\multipebblegame'',
that we will use to analyze resolution refutations of
pebbling contradictions.   
In
\refsec{sec:derivations-induce-chain-multi-pebblings}
we prove that resolution refutations can be translated into
pebblings in this game, which is
\refth{th:overview-pyramid-refutations-and-pebblings}
in
\refsec{sec:proof-overview-detailed}.
In
\refsec{sec:multi-pebble-configuration-cost},
we prove 
\refth{th:overview-pyramid-pebbling-cost-bounded-by-space}
saying that the \multipebblingtext price accurately measures the
clause space of the corresponding resolution refutation.
Finally, 
after giving a detailed description of the
lower bound on black-white pebbling of~%
\cite{K80TightBoundPebblesPyramid}
in
\refsec{sec:pebble-games-pyramids}
(with a somewhat simplified proof that might be of independent interest),
in
\refsec{sec:tight-lower-bound-for-multi-pebbling-pyramid}
we generalize this result in a nontrivial way 
to our \multipebblegame.
This gives us
\refth{th:overview-pyramid-bound-labelled-pebbling-cost}.
Now
\refth{th:main-theorem}
and
\refcor{cor:main-corollary}
follow as in the proofs given at the end of
\refsec{sec:proof-overview-detailed}.

We conclude in
\refsec{sec:open-questions}
by giving suggestions for further research.

%
%

\section{Formal Preliminaries}
\label{sec:formal-preliminaries}

In this \sectionappendixtext, we define resolution, pebble games and
pebbling contradictions. 
%

\subsection{The Resolution Proof System}
\label{sec:resolution-proof-system}

A \introduceterm{literal} is either a
\proplog variable or its negation, denoted
$\varx$ and $\olnot{\varx}$,  respectively.
%
We define $\olnot{\olnot{\varx}} = \varx$.
Two literals $\lita$ and $\litb$
are \introduceterm{strictly distinct}
if $\lita \neq\litb$ and $\lita \neq \olnot{\litb}$,
\ie if they refer to distinct variables.

A
\introduceterm{clause}
$\clc = \lita_1 \lor \formuladots \lor \lita_{\clwidth}$
is a set of literals. 
Throughout this paper, all clauses $\clc$ are assumed to be nontrivial
in the sense that all literals in $\clc$ are pairwise strictly
distinct (otherwise $\clc$ is trivially true).
We say that
$\clc$
is a \introduceterm{subclause} of
$\cld$
if
$\clc \subseteq \cld$. 
A clause containing at most $\clwidth$ literals is called a
\introduceterm{\kclause{}}.

A \introduceterm{\cnfform{}} 
$
\fstd = \clc_1 \land \formuladots \land \clc_m
$
is a set of clauses.
A \introduceterm{\kcnfform{}} is a \cnfform consisting of
\xclause{\clwidth}{}s.
We define the \introduceterm{size} 
$\sizeofsmall{\fstd}$
of the formula $\fstd$  
to be the total number of literals in~$\fstd$
counted with repetitions. More often, we will be interested in the
number of clauses $\nclausesof{\fstd}$ of~$\fstd$.

In this paper, when nothing else is stated it is assumed that
$
\cla, \clb, \clc, \cld 
$ 
denote clauses,
$\clsc, \clsd 
$ 
sets of clauses, 
$\varx, \vary$ propositional variables,
$\lita, \litb, \litc$ literals,
$\tvastd, \tvaalt$ truth value assignments and
$\truthval$ a truth value $0$ or $1$.
We write   
\begin{equation}
\logeval{\vary}{\modtva{\tvastd}{\varx = \truthval}} = 
\begin{cases}
\logeval{\vary}{\tvastd} & \text{if $\vary \neq \varx$,} \\
\truthval & \text{if $\vary = \varx$,}
\end{cases}
\end{equation}
to denote the truth value assignment that agrees with
$\tvastd$ everywhere except possibly at $\varx$, 
to which it assigns the value~$\truthval$.
We let
$\vars{\clc}$
denote the set of variables and
$\lit{\clc}$
the set of literals in  a clause~$\clc$.%
\footnote{%
Although the notation $\lit{\clc}$ is slightly redundant given the
definition of a clause as a set of literals, we include it for
clarity.}
  This notation is extended to sets of clauses by taking
unions.
Also, we employ the standard notation
$\intnfirst{n} = \setsmall{1, 2, \ldots, n}$.

A \introduceterm{\resderiv{}} 
$\derivof{\proofstd}{\fstd}{\cla}$
of a clause~%
$\cla$
from a \cnfform~%
$\fstd$
is a sequence of clauses 
$\proofstd = \setsmall{\cld_1, \dotsc, \cld_{\stoptime}}$
such that 
$\cld_{\stoptime} = \cla$
and each line
$\cld_i$,
$ i \in \intnfirst{\stoptime}$,
either is one of the clauses in~$\fstd$
(\introduceterm{axioms})
or is derived from clauses
$\cld_j,\cld_k$  in~\proofstd{} 
with $j, k < i$
by the 
\introduceterm{resolution rule}
\begin{equation}
\label{eq:resolution-rule}
\resrule{\clb \lor \varx}
        {\clc \lor \stdnot{\varx}}
        {\clb \lor \clc} 
\eqperiod
\end{equation}
We refer to~\eqref{eq:resolution-rule} as
\introduceterm{resolution on the variable}~$\varx$
and to
$\clb \lor \clc$
as the \introduceterm{resolvent} of 
$\clb \lor \varx$
and
$\clc \lor \stdnot{\varx}$
on~$\varx$.
A 
\introduceterm{\resref{}}
of a \cnfform~$\fstd$ is a \resderiv of the empty clause 
$\emptycl$
(the clause with no literals) from~$\fstd$.
Perhaps somewhat confusingly, this is sometimes also referred to as a
\introduceterm{resolution proof} of~$\fstd$.

For a formula $\fstd$ and a set of formulas
$\mathcal{G} = \set{\fvar_1, \ldots, \fvar_n}$,
we say that
$\mathcal{G}$ \introduceterm{implies}~$\fstd$,
denoted 
${\mathcal{G}} \impl {\fstd}$,
if every \tva satisfying all formulas
$\fvar \in \mathcal{G}$ 
satisfies~$\fstd$ as well.
It is well known that resolution is sound and implicationally complete.
That is, 
if there is a \resderiv
$\derivof{\proofstd}{\fstd}{\cla}$, 
then
$\fstd \impl \cla$, 
and if
$\fstd \impl \cla$,
then there is a \resderiv
$\derivof{\proofstd}{\fstd}{\cla'}$
for some $\cla' \subseteq \cla$.
In particular, $\fstd$ is unsatisfiable \ifaoif there is a \resref
of~$\fstd$. 

With every resolution derivation
$\derivof{\proofstd}{\fstd}{\cla}$
\label{page:dag-representation-discussed-here}
we can associate  a DAG~$G_{\proofstd}$,
with the clauses in $\proofstd$ labelling the vertices and
with edges from the assumption clauses to the  
resolvent 
for each application of the resolution rule~\refeq{eq:resolution-rule}.
There might be several different derivations of a clause~$\clc$
in~$\proofstd$, but
if so we can label each occurrence of $\clc$ with a timestamp when it
was derived and keep track of which copy of $\clc$ is used where.
A  \resderiv $\proofstd$ is \introduceterm{tree-like} if any clause in the
derivation is used at most once as a premise in an application of the
resolution rule, \ie if $G_{\proofstd}$ is a tree. (We may make 
different ``time-stamped'' vertex copies of the axiom clauses
in order to make  $G_{\proofstd}$ into a tree).


The
\introduceterm{length}
$\lengthofsmall{\proofstd}$
of a \resderiv~$\proofstd$ is the number of clauses in it.
%
We define the length of deriving a clause $\cla$ from a formula $\fstd$ as
$
\lengthderivsmall{\fstd}{\cla}
=  
\minofexpr[\derivof{\proofstd}{\fstd}{\cla}]{ \lengthofsmall{\proofstd} }
$,
where the minimum is taken over all resolution derivations of~%
$\cla$.
In particular, the length of refuting~$\fstd$ by resolution is denoted
$\lengthrefsmall{\fstd}$.
The length of refuting~$\fstd$ by \treelikeres 
$\lengthrefsmall[\treeresnot]{\fstd}$
is defined by taking the minimum over all \treelikeres refutations~%
$\proofstd_T$ 
of~%
$\fstd$. 

The
\introduceterm{width}
$\widthofsmall{\clc}$
of a clause $\clc$ is $\setsizesmall{\clc}$,
\ie the number of literals appearing in it.
The width of a set of clauses
$\clsc$
is
$\widthofsmall{\clsc} 
=  
        \maxofexpr[\clc \in \clsc]{ \widthofsmall{\clc}}$.
The width of deriving~
$\cla$
from~
$\fstd$ by resolution
is 
$
\widthderivsmall{\fstd}{\cla}
=  \minofexpr[\derivof{\resstd}{\fstd}{\cla}]
{\widthofsmall{\resstd}}
$,
and the width of refuting $\fstd$ is denoted
$
\widthrefsmall{\fstd}
$. 
Note that the minimum width measures in general and \treelikeres
coincide, so it makes no sense 
to make a separate definition for
$\widthrefsmall[\treeresnot]{\fstd}$.

We next define the measure of
\introduceterm{space}.
Following the exposition in~%
\cite{ET01SpaceBounds},
a proof can be seen as a Turing machine computation, with a special
read-only input tape from which the axioms can be downloaded and a
working memory where all derivation steps are made.
The 
\introduceterm{clause space} 
of a \resproof is the maximum number of clauses that
need to be kept in memory simultaneously 
during a verification of the proof.
The 
\introduceterm{variable space} 
is the maximum total space needed, where also
the width of the clauses is taken into account.

For the formal definitions, it is convenient to use 
an 
alternative definition of resolution introduced in~
\cite{ABRW02SpaceComplexity}.

\begin{definition}[Resolution]
\label{def:self-contained-resolution-refutation}
A \introduceterm{clause configuration} 
$\clsc$
is  a set of clauses. A sequence of clause configurations
$\resderivspacesequence{\clsc_0, \ldots, \clsc_{\stoptime}}$ 
is a 
\introduceterm{\resderiv{}} 
from a \cnfform 
$\fstd$ if
$\clsc_0 = \emptyset$
and for all 
$t \in \intnfirst{\stoptime}$,
$\clsc_{t}$ is obtained from $\clsc_{t-1}$ by one%
\footnote{%
  In some previous papers, resolution is defined so as to allow every
  derivation step to \emph{combine} one or zero applications of each of
  the three derivation rules.  Therefore, some of the bounds stated in this
  paper for space as defined next are off by a constant as compared to
  the cited sources.} 
of the following rules:

\begin{description}
  
  \italicitem [Axiom Download]
  $\clsc_{t} = \clsc_{t-1} \union \set{\clc}$ for some
  $\clc \in \fstd$. 
  
  \italicitem [Erasure]
  $\clsc_{t} = \clsc_{t-1} \setminus \set{\clc}$ for some
  $\clc \in \clsc_{t-1}$.
  
  \italicitem [Inference]
  $\clsc_{t} = \clsc_{t-1} \union \set{\cld}$ 
  for some
  $\cld$ 
  inferred by resolution from
  $\clc_1, \clc_2 \in \clsc_{t-1}$.
\end{description}
A resolution derivation
$\derivof{\proofstd}{\fstd}{\cla}$
of a clause~$\cla$ from a formula~$\fstd$ is a derivation 
$\resderivspacesequence{\clsc_0, \ldots, \clsc_{\stoptime}}$ 
\st
$\clsc_{\stoptime} = \setsmall{\cla}$.
A \introduceterm{\resref{}}
of~$\fstd$ is a derivation of the empty clause~%
$\emptycl$ from~$\fstd$.
\end{definition}


\begin{definition}[Clause space
  \cite{ABRW02SpaceComplexity, Ben-Sasson02SizeSpaceTradeoffs}]
\label{def:clause-space}
The 
\introduceterm{clause space}
of a \resderiv 
$
\proofstd 
\!=\!
\resderivspacesequence{\clsc_0, \ldots, 
  \!
  \clsc_{\stoptime}}
$ 
is
$\maxofexpr[t \in \intnfirst{\stoptime}]{\setsizesmall{\clsc_t}}$.
The clause space of deriving 
$\cla$ from 
$\fstd$ is
$
\mbox{$\clspacederivsmall{\fstd}{\cla}$} 
=
\minofexpr[\derivof{\resstd}{\fstd}{\cla}]
{\clspaceofsmall{\resstd}}
$,
and
$\clspacerefsmall{\fstd}$ 
denotes the minimum clause space of any
\resref of~$\fstd$.
\end{definition}

\begin{definition}[Variable space \cite{ABRW02SpaceComplexity}]
The \introduceterm{variable space} 
of a configuration~$\clsc$ is
$\varspaceofsmall{\clsc} = \sum_{\clc \in \clsc} \widthofsmall{\clc}$.
The variable space
of a 
derivation
$\resderivspacesequence{\clsc_0, \ldots, \clsc_{\stoptime}}$ 
is
$\maxofexpr[t \in \intnfirst{\stoptime}] {\varspaceofsmall{\clsc_t}}$, 
and
$\varspacerefsmall{\fstd}$ 
is the minimum variable space of any
\resref of~$\fstd$.
\end{definition}

Restricting the resolution derivations to tree-like resolution, we get the
measures
$\clspacerefsmall[\treeresnot]{\fstd}$ 
and
$\varspacerefsmall[\treeresnot]{\fstd}$ 
in analogy with
$\lengthrefsmall[\treeresnot]{\fstd}$ 
defined above.

Note that if one wanted to be really precise, the size and space
measures should probably measure the number of \emph{bits} needed
rather than the number of literals. However, counting literals makes
matters substantially cleaner, and the difference is at most a
logarithmic factor anyway. Therefore, counting literals seems to be
the established way of measuring formula size and variable space. 

In this paper, 
we will be almost exclusively interested in the clause space of
general resolution refutations. 
When we write simply ``space'' for brevity, we mean clause space.

\subsection{Pebble Games and Pebbling Contradictions}
\label{sec:pebble-games}

Pebble games were devised for studying programming languages and
compiler construction, but have found a variety of
applications in computational complexity theory.
In connection with resolution, pebble games have been employed 
both to analyze resolution derivations \wrt  how much memory they consume
(using the original definition of space in \cite{ET01SpaceBounds}) 
and to construct \cnfform{}s which are hard for different
variants of resolution in various respects
(see for example
\cite{AJPU02ExponentialSeparation,
  BIW00Near-optimalSeparation,
  BEGJ00RelativeComplexity,
  BOP03Complexity}).
An excellent survey of pebbling up to ca 1980 is~\cite{P80Pebbling}.

The black pebbling price of a DAG 
$G$ captures the memory space, 
\ie the number of registers, required to perform the deterministic
computation described by~$G$. The space of a non-deterministic
computation is measured by the black-white pebbling price of~$G$.
We say that vertices of $G$ with indegree~$0$ are 
\introduceterm{sources}
and that vertices with outdegree~$0$ are
\introduceterm{sinks} 
or
\introduceterm{targets}.
In the following, unless otherwise stated we will assume that all DAGs
under discussion have a unique sink and this sink will always be
denoted~$z$. 
The next definition  is adapted from~%
\cite{CS76Storage},
though we use the established pebbling terminology introduced by~%
\cite{HPV77TimeVsSpace}.

\begin{definition}[Pebble game]
\label{def:bw-pebble-game}
Suppose that $G$ is a DAG with sources~$S$ and 
a unique target~$z$. The 
\introduceterm{black-white pebble game}
on~$G$  is the following \mbox{one-player} game. 
At any point in the game, there are black and white pebbles placed on
some vertices of~$G$, at most one pebble per vertex. A
\introduceterm{pebble configuration}
is a pair of subsets $\pconf = (B,W)$ of~$\vertices{G}$,    
comprising the black-pebbled vertices $B$ and white-pebbled vertices~$W$.
The rules of the  game are as follows:
\begin{enumerate}

\item
\label{pebrule:black-placement}
If all immediate predecessors of an empty vertex~$v$ have pebbles on them, a
black pebble may be placed on~$v$.
In particular, a black pebble can always be placed on 
any vertex in~$S$.

\item
\label{pebrule:black-removal}
A black pebble may be removed from any vertex  at any time.

\item
\label{pebrule:white-placement}
A white pebble may be placed on any empty vertex at any time.

\item
\label{pebrule:white-removal}
If all immediate predecessors of a white-pebbled vertex~$v$ have pebbles on
them, the white pebble on $v$ may be removed.
In particular, a white pebble can always be removed from a source vertex.

\end{enumerate}

A \introduceterm{black-white pebbling} 
from $(B_1,W_1)$ to $(B_2,W_2)$
in~$G$ is a sequence of pebble configurations 
$\pebbling = \setsmall{\pconf_0, \ldots, \pconf_{\stoptime}}$
\st
$\pconf_0 =  (B_1, W_1)$,
$\pconf_{\stoptime} = (B_2, W_2)$,
and for all 
$t \in \intnfirst{\stoptime}$,
$\pconf_{t}$ follows from $\pconf_{t-1}$
by one of the rules above.
If
$(B_1, W_1) =  (\emptyset,\emptyset)$,
we say that the pebbling is 
\introduceterm{\pebunconditional{}}, otherwise it is 
\introduceterm{conditional}.

The \introduceterm{cost} of a pebble configuration
$\pconf = (B,W)$
is
$\pebcost{\pconf} = \setsizesmall{B \unionSP W}$
and the cost of a pebbling 
$\pebbling = \setsmall{\pconf_0, \ldots, \pconf_{\stoptime}}$
is 
$\maxofexpr[0 \leq t \leq {\stoptime}] {\pebcostsmall{\pconf_t}}$.
The \introduceterm{black-white pebbling price}
of~$(B,W)$, denoted $\bwpebblingprice{B,W}$,
is the minimum cost of any \pebunconditional{} pebbling reaching~$(B,W)$.

A \introduceterm{\pebcomplete{} pebbling} of $G$,
also called a \introduceterm{pebbling strategy} for~$G$,
is an \pebunconditional{} pebbling reaching $(\setsmall{z}, \emptyset)$. 
The \introduceterm{black-white pebbling price}
of~$G$, denoted $\bwpebblingprice{G}$,
is the minimum cost of any \pebcomplete{} black-white pebbling of~$G$.

A \introduceterm{black pebbling} is a pebbling using black pebbles
only, \ie having $W_t = \emptyset$ for all~$t$.
The \introduceterm{(black) pebbling price}
of~$G$, denoted $\pebblingprice{G}$,
is the minimum cost of any \pebcomplete black pebbling of~$G$.

%
%
%
\end{definition}


We think of the moves in a pebbling as occurring at integral
time intervals $t = 1, 2, \ldots$ and talk about the pebbling move 
``at time $t$''
(which is the move resulting in configuration
$\pconf_t$) 
or the moves 
``during the time interval $\intclcl{t_1}{t_2}$''.

The only pebblings we are really interested in are
\pebcomplete{} pebblings of $G$.
However, when we prove lower bounds for pebbling price it will
sometimes be convenient to be able to reason in terms of partial
pebbling move sequences, \ie \pebconditional pebblings.

%

%
%
%

A 
\introduceterm{pebbling contradiction} 
defined on a DAG~$G$ encodes the pebble game on~$G$ by 
postulating   
the sources to be true and the target to be false, and specifying that
truth propagates through the graph according to the pebbling rules.
The definition below is a generalization of formulas previously studied 
\mbox{in \cite{BEGJ00RelativeComplexity, RM99Separation}.}

\begin{definition}[Pebbling contradiction~%
  \cite{BW01ShortProofs}]
\label{def:pebbling-contradiction}
Suppose that $G$ is
a DAG with sources~$S$, a unique target~$z$
and with all non-source vertices having indegree~$2$, 
and
let $\pebdeg >0$ be an integer.
Associate 
$\pebdeg$ 
distinct variables
$\varx(v)_1, \ldots, \varx(v)_{\pebdeg}$
with every vertex
$\isin{v}{\vertices{G}}$.
The $\pebdeg$th degree 
\introduceterm{pebbling contradiction} over~$G$,
denoted~$\pebcontr{\pebdeg}$,
is the conjunction of the following clauses:
\begin{itemize}

\item
  $\Lor_{i=1}^{\pebdeg} \varx(s)_{i}$ 
  for all 
  $\isin{s}{S}$
  (\introduceterm{source axioms}),
  
\item
  $\stdnot{\varx(z)}_i$
  for all 
  $i \in \intnfirst{\pebdeg}$
  (\introduceterm{target axioms}),
  
\item
  $
  \stdnot{\varx(u)}_{i} \lor \stdnot{\varx(v)}_{j} \lor
  \Lor_{l=1}^{\pebdeg} \varx(w)_{l}
  $
  for all
  $i,j \in \intnfirst{\pebdeg}$
  and all
  $\isin{w}{\vertices{G} \setminus S}$,
  where
  $u, v$ 
  are the two predecessors of~$w$
  (\introduceterm{pebbling axioms}).
\end{itemize}
\end{definition}

The formula 
$\pebcontr{\pebdeg}$
is a   
\xcnfform{\text{(2+\pebdeg)}}{}
with
$\Ordocompact{\pebdeg^2 \cdot \setsizesmall{\vertices{G}}}$
clauses over
\mbox{$\pebdeg \cdot \setsizesmall{\vertices{G}}$}
variables.
An~example pebbling contradiction
is presented in
\reffigP{fig:pebbling-contradiction-for-Pi-2}.

\section{Review of Related Work}
\label{sec:review-of-some-known-results}

This section is an overview of related work, including formal
statements of some previously known results that we will need.
At the end of
\refsec{sec:review-results-pebbling-contradictions}
we also try to provide some of the intuition behind the result proven
in this paper.

\subsection{General Results About Resolution}
\label{sec:review-general-results-resolution}

It is 
not hard to show       
that any \cnfform $\fstd$ over $\nvar$~variables
is refutable in length 
\mbox{$2^{n+1}- 1$}
and width~$n$.
Esteban and Torán~\cite{ET01SpaceBounds}
proved that the clause space of refuting $\fstd$
is upper-bounded by the formula size.
More precisely, the minimal clause space is at most the number of
clauses, or the number of variables, plus a small constant, or in
formal notation
$
\clspacerefsmall{\fstd} \leq
\Minofexpr{\setsizesmall{\fstd}, \setsizesmall{\varssmall{\fstd}}}
+ \bigoh{1}
$.

We will need the fact that there are polynomial-size families of 
 \kcnfform{}s that are very hard \wrt length, width and clause space,
essentially meeting the upper bounds just stated.

\begin{theorem}[\cite{ABRW02SpaceComplexity, 
    BKPS02Efficiency,
    BG03SpaceComplexity, 
    BW01ShortProofs,
    CS88ManyHard,
    Toran99LowerBounds,
    U87HardExamples}]
  \label{th:clause-space-approx-n-clauses}
  There are arbitrarily large unsatisfiable \xcnfform{3}{}s $\fstd_n$
  of size $\bigtheta{n}$   with 
  $\bigtheta{n}$ clauses and
  $\bigtheta{n}$ variables
  for which it holds that 
  $\lengthrefsmall{\fstd_n} = \exp(\bigtheta{n})$,
  $\widthrefsmall{\fstd_n} = \bigtheta{n}$
  and
  $\clspacerefsmall{\fstd_n} = \bigtheta{n}$.
\end{theorem}

Clearly, for such formulas $\fstd_n$ it must also hold that
$
\bigomega{n}
=
\varspacerefsmall{\fstd_n}
=
\Bigoh{n^2}
$.
We note in passing that determining the exact variable space
complexity of a formula family as in 
\refth{th:clause-space-approx-n-clauses}
was mentioned as an open problem in~%
\cite{ABRW02SpaceComplexity}.
To the best of our knowledge this problem is still unsolved.

If a resolution refutation has constant width, it is easy to see that
it must be of size polynomial in the number of variables (just count
the maximum possible number of distinct clauses).
Conversely, if all refutations of a formula are very
wide, it seems reasonable that any refutation of this formula must be
very long as well. This intuition was made precise by Ben-Sasson and
Wigderson~\cite{BW01ShortProofs}.  We state their theorem in the more
explicit form of Segerlind~\cite{Segerlind07Complexity}.

\begin{theorem}[\cite{BW01ShortProofs}]
\label{th:widthboundgeneral}
The width of refuting a \cnfform $\fstd$ is bounded from above by
\begin{equation*}
\widthrefsmall{\fstd}  \leq 
\widthofsmall{\fstd} + 1 + 
3  {\sqrt{\nvar \ln \lengthrefsmall{\fstd}}}
\eqcomma
\end{equation*}
where $\nvar$ is the number of variables in~$\fstd$.
\end{theorem}
%
%

Bonet and Galesi \cite{BG01Optimality} showed that 
this bound on width in terms of length 
is essentially optimal.
For the special case of \treelikeres, however,  it is possible get rid
of the dependence of the number of variables and obtain a tighter bound. 

\begin{theorem}[\cite{BW01ShortProofs}]
\label{th:widthboundtree}
The width of refuting a \cnfform $\fstd$ 
in \treelikeres is bounded from above by
$
\widthref{\fstd} \leq
\widthofarg{\fstd} + \log \lengthref[\treeresnot]{\fstd}
$.
\end{theorem}

For reference, we collect the result in \cite{BG01Optimality}    
together with some other bounds showing that
there are formulas that are easy \wrt length
but moderately hard \wrt width and clause space and state them 
as a theorem.%
\footnote{%
  Note that 
  \cite{BG01Optimality},     
  where an explicit resolution refutation
  upper-bounding the proof complexity measures is presented, does not
  talk about clause space, but it is straightforward to verify that the
  refutation there can be carried out in 
  length~$\Bigoh{n^3}$
  and clause space~$\bigoh{n}$.}

\begin{theorem}[\cite{ABRW02SpaceComplexity,
      BG01Optimality,    
      S96}]
  \label{th:moderately-hard-formulas}
  There are arbitrarily large unsatisfiable \xcnfform{3}{}s $\fstd_n$
  of size   $\Bigtheta{n^3}$ 
  with 
  $\Bigtheta{n^3}$   clauses and
  $\Bigtheta{n^2}$ variables
  \st 
  $\widthrefsmall{\fstd_n} = \bigtheta{n}$
  and
  $\clspacerefsmall{\fstd_n} = \bigtheta{n}$,
  but for which
  there are resolution refutations 
  $\refof{\proofstd_n}{\fstd_n}$
  in length
  $\lengthofsmall{\proofstd_n} = \Bigoh{n^3}$,
  width
  $\widthofsmall{\proofstd_n} = \bigoh{n}$
  and clause space
  $\clspaceofsmall{\proofstd_n} = \bigoh{n}$.
\end{theorem}

As was mentioned above, the fact that all known lower bounds on
refutation clause space coincided with lower bounds on width lead to
the conjecture that the width measure is a lower bound for the clause
space measure.  This conjecture was proven true by Atserias and
Dalmau~\cite{AD02CombinatoricalCharacterization}.

\begin{theorem}[\cite{AD02CombinatoricalCharacterization}]
  \label{th:small-clause-space-implies-small-width}
  For any \cnfform $\fstd$, it holds that
  $\clspacerefsmall{\fstd} - 3 
  \geq 
  \widthrefsmall{\fstd} - \widthofsmall{\fstd}$.%
\end{theorem}

In other words, 
the extra clause space exceeding the minimum 3 needed for any resolution 
derivation is bounded from below by the extra width exceeding 
the width of the formula.  
This inequality was later shown 
by \theauthorJN
to be asymptotically strict in the following sense.

\begin{theorem}[\cite{Nordstrom06NarrowProofsMayBeSpaciousSTOCtoappearinSICOMP}]
  \label{th:narrow-proofs-may-be-spacious}
  For all $\clwidth \geq 4$,
  there is a family
  $\setsmall{F_n}_{n=1}^{\infty}$
  of \kcnfform{}s of size 
  $\bigtheta{n}$ 
  \st
  $\lengthrefsmall{F_n} = \bigoh{n}$
  and
  $\widthrefsmall{F_n} = \bigoh{1}$
  but
  $\clspacerefsmall{F_n} = \bigtheta{\log n}$.
\end{theorem}

An immediate corollary of 
\refth{th:small-clause-space-implies-small-width}
is that for \polysize \kcnfform{}s constant clause space implies
polynomial proof length. 
We are interested in finding out what holds in the other direction,
\ie if upper bounds on length imply upper bounds on space.

For the special case of \treelikeres, 
it is known    
that there is an upper bound on clause space in terms of length
exactly analogous to the one on width in terms of length in
\refth{th:widthboundtree}.

\begin{theorem}[\cite{ET01SpaceBounds}]
  \label{th:short-tree-length-implies-small-clause-space}
  For any \treelikeres refutation
  $\proofstd$ of a \cnfform $\fstd$ it holds that
  $\clspaceofsmall{\proofstd} \leq 
  \ceilingsmall{\log \lengthofsmall{\proofstd}} + 2$.
  In particular,
  $\clspacerefsmall{\fstd} \leq 
  \ceilingsmall{\log \lengthrefsmall[\treeresnot]{\fstd}} + 2$.
\end{theorem}
%
%
%
%

For general resolution, 
since clause space is lower-bounded by width according to
\refth{th:small-clause-space-implies-small-width}, 
the separation of width and length of~%
\cite{BG01Optimality} in
\refth{th:moderately-hard-formulas} 
tells us that \kcnfform{}s
refutable in polynomial length can still have ``somewhat spacious''
minimum-space refutations. But exactly how spacious can they be?
Does space behave as width \wrt length also in general resolution, or
can one get stronger lower bounds on space for formulas refutable in
polynomial length?

All polynomial lower bounds on clause space 
known prior to this paper
can be explained as immediate consequences of
\refth{th:small-clause-space-implies-small-width}
applied on lower bounds on width.
Clearly, any space lower bounds derived in this way cannot get us
beyond the
``Ben-Sasson--Wigderson barrier''   
implied by \refth{th:widthboundgeneral}
saying that  if the width of refuting~$\fstd$ is 
$\Littleomega{\sqrt{\nclausesof{\fstd} \log \nclausesof{\fstd}}}$,
then the length of refuting $\fstd$ must be superpolynomial
in~$\nclausesof{\fstd}$.
Also, since matching upper bounds on clause space have been known for
all of these formula families, they have not been candidates for
showing stronger separations of space and length.
Thus, the best known separation of clause space and length
has been the formulas in
\refth{th:moderately-hard-formulas} 
refutable in linear length
$\lengthrefsmall{\fstd_n} = \bigoh{\nclausesof{\fstd_n}}$
but requiring space
$\clspacerefsmall{\fstd_n} =
\Bigtheta{\sqrt[3]{\nclausesof{\fstd_n}}}$,
as implied by the same bound on width.

Let us also discuss upper bounds on what kind of separations are a
priori possible.
Given any resolution refutation~$\refof{\proofstd}{\fstd}$,
we can write down its DAG representation $G_{\proofstd}$ 
(described on page~%
\pageref{page:dag-representation-discussed-here})
with $\lengthofsmall{\proofstd}$ vertices corresponding to the
clauses, and with all non-source vertices having
\mbox{fan-in~$2$}.
We can then transform $\proofstd$ into as space-efficient a refutation
as possible by considering an optimal black pebbling
of~$G_{\proofstd}$ as follows:
when a pebble is placed on a vertex we derive the
corresponding clause, and when the pebble is removed again we erase
the clause from memory.
This yields a refutation
$\proofstd'$ in clause space
$\pebblingprice{G_{\proofstd}}$
(incidentally, this is the original definition in~%
\cite{ET01SpaceBounds}
of the clause space of a resolution refution~$\proofstd$).
Since it is known 
that any constant indegree DAG on $n$~vertices can be black-pebbled
in cost 
$\bigoh{n / \log n}$
(see
\refth{th:upper-bound-n-log-n-pebbling}),
this shows that
$
\mbox{$\clspacerefsmall{\fstd}$}
=
\Bigoh{\lengthrefsmall{\fstd} / \log \lengthrefsmall{\fstd}}
$
is a trivial upper bound on space in terms of length.

Now we can rephrase the question above about space and length 
in the following way:
Is there a Ben-Sasson--Wigderson kind of lower bound, say
$
\lengthrefsmall{\fstd}
=
\exp 
\bigl(
\Bigomega{\clspacerefsmall{\fstd}^2 / \setsizesmall{\fstd}}
\bigr)
$
or so,
on length in terms of space?
Or do there exist \kcnfform{}s~$\fstd$ with short refutations but
maximum possible refutation space
$
\mbox{$\clspacerefsmall{\fstd}$}
=
\Bigomega{\lengthrefsmall{\fstd} / \log \lengthrefsmall{\fstd}}$
in terms of length?
Note that the refutation length
$\lengthrefsmall{\fstd}$
must indeed be short in this case---essentially linear, since any
formula~$\fstd$ can be refuted in space~%
$\bigoh{\setsizesmall{\fstd}}$
as was noted above.
Or is the relation between refutation space and refutation length
somewhere in between these extremes?

This is the main question addressed in this paper.
We believe that clause space and length can be strongly separated in
the sense that there are formula families with maximum possible
refutation space in terms of length.
As a step towards proving this we improve the lower bound in
\refth{th:narrow-proofs-may-be-spacious}
from $\bigtheta{\log n}$ to $\bigtheta{\sqrt{n}}$,
thus providing the first polynomial lower bound on space that is not
the consequence of a corresponding bound on width.
%
We next review
some results about the tools that we use to do this.

\subsection{Results About Pebble Games}
\label{sec:review-results-pebble-games}

There is an extensive literature on pebbling, mostly from the 70s
and~80s. We just quickly mention four results relevant to this paper.

Perhaps the simplest graphs to pebble are complete binary trees $T_h$ of
height~$h$.  The black pebbling price of $T_h$ can be established by
an easy induction over the tree height.  For black-white pebbling,
general bounds for the pebbling price of trees of any arity were
presented in
\cite{LT80SpaceComplexityPebbleGamesTrees}.
For the case of binary trees, 
this result can be simplified to an exact equality 
(a proof of which can be found in Section 4 of~%
\cite{Nordstrom05NarrowProofsMayBeSpacious}).

\begin{theorem}
  \label{th:bounds-pebbling-price-trees}
  For a complete binary tree $T_h$ of height $h \geq 1$ it holds that
  $\pebblingprice{T_h} = h + 2$
  and
  $\bwpebblingprice{T_h} = \Floor{\frac{h}{2}} + 3$.
\end{theorem}

In this paper, we will focus on pyramid graphs,
an example of which can be found in
\reffig{fig:pebbling-contradiction-for-Pi-2}.

\begin{theorem}[\cite{C74ObservationTimeStorageTradeOff,
    K80TightBoundPebblesPyramid}]
  \label{th:bounds-pebbling-price-pyramids}
  For a pyramid graph 
  $\pyramidgraph[h]$
  of height $h \geq 1$ it holds that
  $\pebblingprice{\pyramidgraph[h]} = h + 2$
  and
  $\bwpebblingprice{\pyramidgraph[h]} = h/2 + \bigoh{1}$.
\end{theorem}

As we wrote in
\refsec{sec:proof-overview-and-paper-organization},
we are interested in DAGs with as 
high a pebbling price as possible measured in terms of the
number of vertices. For a DAG $G$ with $n$ vertices and constant
in-degree, the best we can hope for is $\bigoh{n / \log n}$.

\begin{theorem}[\cite{HPV77TimeVsSpace}]
  \label{th:upper-bound-n-log-n-pebbling}
  For \DAG{}s $G$ with 
  $n$~vertices
  and constant maximum indegree,
  it holds that
  $\pebblingprice{G} = \Bigoh{n / \log n}$.
\end{theorem}

This bound is asymptotically tight both for black and black-white
pebbling. 

\begin{theorem}[\cite{GT78VariationsPebbleGame,PTC76SpaceBounds}]
  \label{th:maximal-pebbling-price}
  There is 
  a family  of explicitly constructible%
  \footnote{%
This was not known at the time of the original theorems in
\cite{GT78VariationsPebbleGame,PTC76SpaceBounds}.
What is needed is an explicit construction of
superconcentrators of linear density, and
it has since been shown how to do this
(with
\cite{AC03SmallerExplicitSuperconcentrators}
apparently being the currently best construction).%
} 
  DAGs~$G_n$
  with
  $\bigtheta{n}$~vertices
  and vertex indegrees $0$ or~$2$
  such that
  $\pebblingprice{G} = \bigtheta{n / \log n}$
  and
  $\bwpebblingprice{G} = \bigtheta{n / \log n}$.
\end{theorem}

It should be pointed out that although the black and black-white
pebbling prices coincide asymptotically in all of the theorems above,
this is not the case in general.
In
\cite{KS88OnThePowerOfWhitePebbles},
a family of DAGs with a quadratic difference in the number of
pebbles between the black and the black-white pebble game was
presented. We note that this is the best separation possible, 
since by~%
\cite{MadH81ComparisonOfTwoVariationsOfPebbleGame}
the difference in black and black-white pebbling price can be at most
quadratic. 

%
%

\subsection{Results About Pebbling Contradictions Plus Some Intuition}
\label{sec:review-results-pebbling-contradictions}

Although any constant indegree will be fine for the results covered in
this subsection, we restrict our attention to DAGs with vertex
indegrees $0$ or~$2$ 
since these are the graphs that will be studied in the rest of this paper.

It was observed in~%
\cite{BIW00Near-optimalSeparation}
that
$\pebcontr{\pebdeg}$
can be refuted in resolution by deriving 
$\Lor_{i=1}^{\pebdeg} \varx(v)_{i}$ 
for all
$v \in \vertices{G}$
inductively in topological order and then resolving with the target axioms
$\stdnot{\varx(z)}_i$, \mbox{$i \in \intnfirst{\pebdeg}$}.
Writing down this resolution proof, one gets the following proposition
(which is proven together with
\refpr{pr:clause-space-upper-bounded-by-pebbling-price}
below).

\begin{proposition}[\cite{BIW00Near-optimalSeparation}]
  \label{pr:pebbling-contradiction-has-short-resolution-refutation}
For any DAG $G$ with all vertices having indegree $0$ or~$2$,
there is a resolution refutation
$\derivof{\proofstd}{\pebcontr{\pebdeg}}{\emptycl}$
in length 
$\lengthofsmall{\proofstd} 
= \Bigoh{\pebdeg^2 \cdot \setsizesmall{\vertices{G}}}$
and width
\mbox{$\widthofsmall{\proofstd} = \bigoh{\pebdeg}$}.
\end{proposition}

Tree-like resolution is good at refuting 
first-degree
pebbling contradictions
$\pebcontr[G]{1}$ but is bad at refuting $\pebcontr[G]{\pebdeg}$
for $\pebdeg \geq 2$.

\begin{theorem}[\cite{Ben-Sasson02SizeSpaceTradeoffs}]
  \label{th:constant-refutation-space-degree-one-pebbling-contradiction}
  For any DAG $G$ with all vertices having indegree $0$ or~$2$,
  there is a tree-like \resref $\proofstd$ of
  $\pebcontr[G]{1}$ \st
  $\lengthofsmall{\proofstd} = \bigoh{\setsizesmall{\vertices{G}}}$
  and
  $\clspaceofsmall{\proofstd} = \bigoh{1}$.
\end{theorem}

\begin{theorem}[\cite{BIW00Near-optimalSeparation}]
  \label{th:connection-refutation-length-pebbling-price}
For any DAG $G$  with all vertices having indegree $0$ or~$2$,
$
\lengthref[\treeresnot]{\pebcontr{2}} = 2 ^
{\bigomega{\pebblingprice{G}}}
$.
\end{theorem}

As to space, it is not too difficult to see that the  
black pebbling price of $G$ provides an upper bound for the refutation
clause space of
${\pebcontr[G]{\pebdeg}}$.

%
%

\begin{proposition}
  \label{pr:clause-space-upper-bounded-by-pebbling-price}
  For any DAG $G$  with 
  vertex indegrees $0$ or~$2$,
  $\clspaceref{\pebcontr[G]{\pebdeg}}
  \leq
  \mbox{$\pebblingprice{G} + \bigoh{1}$}$.
\end{proposition}

Essentially, this is just a matter of combining an optimal
black pebbling of~$G$  with the resolution refutation idea from~%
\cite{BIW00Near-optimalSeparation}
sketched above.
Since we need 
the upper bounds on width and space  in 
\reftwoprs
{pr:pebbling-contradiction-has-short-resolution-refutation}
{pr:clause-space-upper-bounded-by-pebbling-price}
in the proof of our main theorem, we write down the details for
completeness.  

\ifthenelse
{\boolean{maybeSIAM}}
{\manualproofSIAM{Proof of
  \reftwoprs
  {pr:pebbling-contradiction-has-short-resolution-refutation}
  {pr:clause-space-upper-bounded-by-pebbling-price}}}
{\begin{proof}[Proof of
  \reftwoprs
  {pr:pebbling-contradiction-has-short-resolution-refutation}
  {pr:clause-space-upper-bounded-by-pebbling-price}]}
  Consider first the bound on space.

  Given a black pebbling of~$G$,  
  we construct a resolution refutation of
  ${\pebcontr[G]{\pebdeg}}$
  such that if at some point in time there are black pebbles on a set of
  vertices~$V$, then 
  we have the clauses
  $\setdescrcompact{\sourceclausexvar[i]{v}}{v \in V}$  
  in memory.
  When some new vertex $v$ is pebbled, we derive
  $\sourceclausexvar[i]{v}$  
  from the clauses already in memory.
  We claim that with a little care, this can be done in constant extra space
  independent of~$\pebdeg$.
  When a black pebble is removed from~$v$, we erase the clause
  $\sourceclausexvar[i]{v}$.  
  We conclude the resolution proof  by resolving
  $\sourceclausexvar[i]{z}$
  for the target~$z$
  with all target axioms
  $\olnot{\varx(z)}_i$, $i \in \intnfirst{\pebdeg}$,
  in space~$3$.

  It is clear that given our claim
  about the constant extra space needed when a vertex is black-pebbled, 
  this yields a resolution refutation in
  space equal to the pebbling cost plus some constant.
  In particular, 
  given an optimal black pebbling of~$G$,
  we get a refutation in space
  $\mbox{$\pebblingprice{G} + \bigoh{1}$}$.

  To prove the claim,
  note first that it trivially holds for source vertices~$v$, since 
  $\sourceclausexvar[i]{v}$
  is an axiom of the formula.
  Suppose for
  a non-source vertex $r$ with predecessors $p$ and $q$ 
  that at some point in time a black pebble is placed on~$r$.
  Then $p$ and $q$ must be black-pebbled, so 
  by induction
  we have the clauses
  $\sourceclausexvar[i]{p}$
  and
  $\sourceclausexvar[j]{q}$
  in memory.
  We will use that
  the clause
  $\stdnot{\varx(p)}_i \lor \sourceclausexvar[l]{r}$
  for any~$i$
  can be derived 
  in additional space~$3$ by resolving
  $\sourceclausexvar[j]{q}$
  with
  $\pqrstdxvar$ for $j \in \intnfirst{\pebdeg}$,
  leaving the easy verification of this fact to the reader.
  To derive
  $\sourceclausexvar[l]{r}$,
  first
  resolve
  $\sourceclausexvar[i]{p}$
  with
  $\stdnot{\varx(p)}_1 \lor \sourceclausexvar[l]{r}$
  to get
  $\subsourceclausexvar[i]{2}{p} \lor \sourceclausexvar[l]{r}$,
  and then resolve this clause with the clauses
  $\stdnot{\varx(p)}_i \lor \sourceclausexvar[l]{r}$
  for $i = 2, \ldots, \pebdeg$
  one by one
  to get
  $\sourceclausexvar[l]{r}$
  in total extra space~$4$.

  It is easy to see that this proof has width
  $\bigoh{\pebdeg}$,
  which proves the claim about width in
  \refpr{pr:pebbling-contradiction-has-short-resolution-refutation}.
  To get the claim about length, 
  we observe that the subderivation 
  needed when a vertex is black-pebbled has length
  $\Bigoh{\pebdeg^2}$.
  If we use a pebbling that
  black-pebbles all vertices once in topological order without ever
  removing a pebble, we
  get a refutation in length
  $\lengthofsmall{\proofstd} 
  = \Bigoh{\pebdeg^2 \cdot \setsizesmall{\vertices{G}}}$.%
\ifthenelse
{\boolean{maybeSIAM}}
{\explicitqedhere}
{\end{proof}}

Thus, the refutation clause space of a pebbling contradiction is 
upper-bounded by the black pebbling price of the underlying DAG.
\Refpr{pr:clause-space-upper-bounded-by-pebbling-price}
is not quite an optimal strategy  \wrt clause space, though.
For binary trees
\cite{ET03CombinatorialCharacterization}
improved this bound   somewhat  to
$
\mbox{$\clspaceref{\pebcontr[T_h]{2}}$} 
\leq 
\frac{2}{3}h + \bigoh{1}
$
by constructing resolution proofs that try to mimic not black
pebblings but instead optimal \emph{black-white} pebblings of~$T_h$ 
as presented in~%
\cite{LT80SpaceComplexityPebbleGamesTrees}.
%
%
And
for one variable per vertex,
we know 
from~%
\refth{th:constant-refutation-space-degree-one-pebbling-contradiction}
that
$
\mbox{$\clspaceref{\pebcontr[G]{1}}$} 
= \bigoh{1}$.

Proving lower bounds on space for pebbling contradictions 
of degree $\pebdeg \geq 2$
has turned out to be much harder. For quite some time there was no
lower bound on 
$\clspaceref{\pebcontr[G]{\pebdeg}}$
for any DAG~$G$ in general resolution (in terms of pebbling price or
otherwise). 
In~\cite{ET03CombinatorialCharacterization},
a lower bound 
$\mbox{$\clspaceref[\treeresnot]{\pebtreecontr[h]{\pebdeg}}$} 
= h + \bigoh{1}$
was obtained for the special case of tree-like resolution.
Unfortunately, this does not tell us anything about  general resolution.
For tree-like resolution,  if the only way of deriving 
a clause
$\cld$ is from
clauses 
$\clc_1, \clc_2$
\st
$\clspacederiv[\treeresnot]{\fstd}{\clc_i} \geq s$,
then it holds that
$\clspacederiv[\treeresnot]{\fstd}{\cld} \geq s + 1$
since one of the clauses $\clc_i$ must be kept in memory while deriving
the other clause. This seems to be very different from how 
general resolution works \wrt space.
In
\cite{Nordstrom06NarrowProofsMayBeSpaciousSTOCtoappearinSICOMP},
\theauthorJN showed a lower bound
$
\clspaceref{\pebcontr[T_h]{\pebdeg}} 
=
\bigomega{h}
$
for binary trees and $\pebdeg \geq 2$, which matches the upper
bound up to a constant factor.
As the techniques in 
\cite{Nordstrom06NarrowProofsMayBeSpaciousSTOCtoappearinSICOMP}
do not yield anything for more general graphs,
this is all that was known prior to this paper.

%
%
We now try to present our own intuition for what the correct lower bound on
the refutation clause space of pebbling contradictions 
\emph{should be}. 
Although the reasoning is quite informal and non-rigorous, our hope is that
it will help the reader to navigate the formal proofs that will follow.

As we noted above, the resolution refutation of 
$\pebtreecontr[h]{2}$
in~\cite{ET03CombinatorialCharacterization}
used to prove the
\mbox{$\frac{2}{3}h + \bigoh{1}$} upper bound for binary tree pebbling
contradictions  is structurally quite similar to the optimal
black-white pebbling of  $T_h$ presented in 
\cite{LT80SpaceComplexityPebbleGamesTrees}, 
and it somehow feels implausible that  
any resolution refutation 
would be able to do significantly better.  
Also, the lower bound in
\cite{Nordstrom06NarrowProofsMayBeSpaciousSTOCtoappearinSICOMP}
is proven by relating resolution refutations to black-white pebblings
and deriving a lower bound on clause space in terms of pebbling price.
This raises the suspicion that the black-white pebbling price
$\bwpebblingprice{G}$
might be a lower bound for
$\clspaceref{\pebcontr[G]{\pebdeg}}$
also for more general graphs as long as
$\pebdeg \geq 2$.
    
This suspicion is somewhat strengthened by the fact that
for variable space, we do have such a lower bound 
in terms of black-white pebbling price.%
\footnote{%
  To be precise, the result in \cite{Ben-Sasson02SizeSpaceTradeoffs}
  is for $\pebdeg=1$,  but the proof generalizes easily to 
  any $\pebdeg \in \Nplus$.} 

\begin{theorem}[\cite{Ben-Sasson02SizeSpaceTradeoffs}]
  \label{th:variable-space-geq-bw-pebbling-price}
  For any $\pebdeg \in \Nplus$,
  $\varspaceref{\pebcontr[G]{\pebdeg}} \geq \bwpebblingprice{G}$.
\end{theorem}

If the refutation clause space of pebbling contradictions for general
DAGs would be constant or very slowly growing, 
\refth{th:variable-space-geq-bw-pebbling-price} would imply that
as $\bwpebblingprice{G}$ grows larger,  
the clauses in memory get wider, and thus weaker. Still it would
somehow be possible to derive a contradiction from a 
very small number
of these clauses of unbounded width.  This appears counterintuitive.

On the other hand, for one variable per vertex, \ie  $\pebdeg = 1$,
refutations of
$\pebcontr[G]{1}$ in constant space have exactly these  
``counterintuitive'' properties. The resolution refutation of
$\pebcontr[G]{1}$
in~%
\refth{th:constant-refutation-space-degree-one-pebbling-contradiction}
is constructed by first downloading the pebbling axiom for the target $z$
and then moving the false literals downwards by resolving with pebbling axioms 
for vertices $v \in \vertices{G} \setminus S$ in reverse topological order.
This finally yields a clause
$\Lor_{v \in S} \olnot{\varx(v)}_1 \lor \varx(z)_1$
of width $\setsizesmall{S} + 1$,
which can be eliminated by resolving 
with the source axioms $\varx(v)_1$ one by one  for all $v \in S$ 
and then with the target axiom~%
$\olnot{\varx(z)}_1$
to yield the empty clause~$\emptycl$.
    
If we want to establish 
a non-constant lower bound on
$\clspaceref{\pebcontr[G]{\pebdeg}}$  for $\pebdeg \geq 2$,
we have to pin down why this case is different.
Intuitively, the difference is that with only one variable per vertex,
a single clause
$\olnot{\varx(v_1)}_1 \lor \ldots \lor \olnot{\varx(v_m)}_1$
can express the disjunction of the falsity of an arbitrary number of
vertices 
$v_1, \ldots, v_m$,
but for 
$\pebdeg = 2$,
the straightforward way of expressing that both variables
$\varx(v_i)_1$ and $\varx(v_i)_2$
are false for at least one out of $m$~vertices requires
$2^{m}$~clauses.

As was argued in 
\refsec{sec:proof-overview-and-paper-organization},
to prove a lower bound on the refutation clause space 
of pebbling contradictions it seems natural to try to interpret
resolution refutations of 
$\pebcontr[G]{\pebdeg}$ 
in terms of pebblings of the underlying graph~$G$.
Let us say that a vertex $v$ is 
``true'' if
$\Lor_{i=1}^{\pebdeg} \varx(v)_{i}$
has been derived 
and ``false'' if
$\olnot{\varx(v)}_i$ has been derived for all $i \in \intnfirst{\pebdeg}$.
Any resolution proof refutes a pebbling contradiction by
deriving that some vertex $v$ is both true and false
and then resolving to get~$\emptycl$.
Let $w$ be any vertex with predecessors $u,v$.
Then we can see that
if we have derived that $u$ and $v$ are true,  by downloading
$
\stdnot{\varx(u)}_{i} \lor \stdnot{\varx(v)}_{j} \lor
\Lor_{l=1}^{\pebdeg} \varx(w)_{l}
$
for all
$i,j \in \intnfirst{\pebdeg}$
we can derive
$\Lor_{l=1}^{\pebdeg} \varx(w)_{l}$.
This appears analogous  to the rule that if
$u$ and $v$ are black-pebbled we can place a black pebble on~$w$.
In the opposite direction, if we know 
$\olnot{\varx(w)}_l$ 
for all 
$l \in \intnfirst{\pebdeg}$,
using the axioms 
$ \stdnot{\varx(u)}_{i} \lor \stdnot{\varx(v)}_{j} \lor 
 \Lor_{l=1}^{\pebdeg} \varx(w)_{l} $ 
we can derive that either $u$ or $v$ is false.  This looks similar to
eliminating a white pebble on $w$ by placing white pebbles on the
predecessors $u$ and~$v$, and then removing the pebble from~$w$.
Generalizing this loose, intuitive reasoning, we argue 
that a set of
black-pebbled vertices $V$ should correspond to the derived
conjunction of truth of all $v \in V$, and that a set of white-pebbled
vertices $W$ should correspond to the derived disjunction of falsity
of some $w \in W$.

Suppose that we could show that as the resolution derivation proceeds,
the black and white pebbles corresponding to different clause
configurations as outlined above move about on 
the vertices of
$G$ in accordance with
the rules of the pebble game. If so, we would get 
that there is some clause configuration $\clsc$ corresponding to a lot of
pebbles. This could in turn hopefully yield a 
lower bound for the refutation clause space. For if $\clsc$ corresponds to
$N$ black pebbles, \ie implies $N$ disjoint clauses, it seems likely
that  $\setsizesmall{\clsc}$ should be linear in~$N$.
And if $\clsc$ corresponds to $N$ white pebbles, 
$\setsizesmall{\clsc}$~should grow with~$N$ if $\pebdeg \geq 2$, since 
$\clsc$ has to force $\pebdeg$ literals false simultaneously for one
out of $N$~vertices.

This is the guiding intuition that served as a starting point for
proving the results in this paper. 
And although quite a few complications arise along the way,
%
%
we believe that it is important when reading the paper not to let all 
technical details obscure the rather simple intuitive correspondence
sketched above.

\theoremstyle{plain}
\newtheorem*{HertelPitassiTheorem}
{\Refth{th:Hertel-Pitassi} (restated)}
\newtheorem*{LengthSpaceTheorem}
{\Refth{th:easy-length-clause-space-trade-off} (restated)}
\newtheorem*{LengthWidthTheorem}
{\Refth{th:easy-length-width-trade-off} (restated)}

\section{A Simplified Way of Proving Trade-off Results}
\label{sec:simplified-way-of-proving-trade-offs}

Before we launch into the proof of the main result of this paper,
however,
we quickly present our simplification of the length-space
trade-off result in~%
\cite{HP07ExponentialTimeSpaceSpeedupFOCS},
and show how the same ideas can be used to prove other related theorems.
We also point out two key ingredients needed for our proofs to work
and discuss possible conclusions to be drawn regarding proving
trade-off results for resolution.
We remark that this section is a somewhat polished write-up of the
results previously announced in
\cite{Nordstrom07SimplifiedWay}.

%
%

We will need the following easy observation.

\begin{observation}
  \label{obs:disjoint-variable-sets}
  Suppose that
  $
  \formf = \formg \land \formh
  $
  where
  $\formg$ and  $\formh$
  are unsatisfiable \cnfform{}s over disjoint sets of variables. Then
  any resolution refutation
  $\refof{\proofstd}{\formf}$
  must contain a refutation of either
  $\formg$ or~$\formh$.
\end{observation}

\begin{proof}
  By induction, 
  we can never resolve a clause derived from $\formg$ with a clause
  derived from~$\formh$, since the sets of variables of the two
  clauses are disjoint.
\end{proof}

\providecommand{\formf}{\ensuremath{F}}
\providecommand{\formg}{\ensuremath{G}}
\providecommand{\formh}{\ensuremath{H}}

\providecommand{\funcf}{\ensuremath{f}}
\providecommand{\funcg}{\ensuremath{g}}
\providecommand{\funch}{\ensuremath{h}}

\subsection{A Proof of Hertel and Pitassi's Trade-off Result}
\label{sec:simplified-proof}

Using the notation in
\refsec{sec:formal-preliminaries},
and improving the parameters somewhat, the
length-variable space trade-off theorem of 
Hertel and Pitassi~\cite{HP07ExponentialTimeSpaceSpeedupFOCS}
can be stated as follows.


\begin{HertelPitassiTheorem}
  There is a family of \cnfform{}s $\setsmall{\formf_n}_{n=1}^{\infty}$
  of size $\Tightsmall{n}$ \st:
  \begin{compactitem}
  \item 
    The minimal variable space of refuting $\formf_n$
    in resolution
    is
    $\varspacerefsmall{\formf_n} = \Tightsmall{n}$.

  \item 
    Any 
    resolution
    refutation $\refof{\proofstd}{\formf_n}$
    in minimal variable space has length
    $\exp (\Lowerboundsmall{\sqrt{n}})$.

  \item 
    Adding at most $2$ extra units of storage, 
    one can 
    obtain  a 
    refutation $\proofstd'$ in 
    space
    $\varspaceofsmall{\proofstd'} =
    \mbox{$\varspacerefsmall{\formf_n} + 3$} = \Tightsmall{n}$
    and length 
    $\lengthofsmall{\proofstd'} = \Ordosmall{n}$,
    \ie linear in the formula size.
  \end{compactitem}
\end{HertelPitassiTheorem}

We note that the \cnfform{}s used by Hertel and Pitassi, as well as
those in our proof, have clauses of width~$\Tightsmall{n}$.


\begin{proof}[Proof of \refth{th:Hertel-Pitassi}]
  Let
  $\formg_n$
  be \cnfform{}s as in
  \refth{th:clause-space-approx-n-clauses}
  having size $\Tightsmall{n}$,
  refutation length
  $\lengthrefsmall{\formg_n} = \exp(\Lowerboundsmall{n})$
  and refutation clause space
  $\clspacerefsmall{\formg_n} = \Tightsmall{n}$.
  Let us define
  $\funcg(n) = \varspacerefsmall{\formg_n}$
  to be the refutation variable space of the formulas.
  Then it holds that   
  $
  \Lowerboundsmall{n}
  =
  \funcg(n)
  =
  \Ordocompact{n^2}
  $.

  Let $\formh_m$ be the formulas
  \begin{equation}
    \formh_m = 
    \vary_1 \land \formuladots \land \vary_m \land 
    (\olnot{\vary}_1 \lor \formuladots \lor \olnot{\vary}_m)
    \eqperiod
  \end{equation}
  It is not hard to see that
  there are resolution refutations 
  $\refof{\proofstd}{\formh_m}$
  in length
  $\lengthofsmall{\proofstd} = 2m + 1$
  and variable space 
  $\varspaceofsmall{\proofstd} = 2m$,
  and that 
  $\lengthrefsmall{\formh_m} = 2m + 1$
  and
  $\varspacerefsmall{\formh_m} = 2m$
  are also the lower bounds
  (all clauses must be used in any refutation, and
  the minimum space refutation must start by 
  downloading the
  wide clause and some unit clause, and then resolve).

  Now define
  \begin{equation}
    \formf_n = \formg_n \land \formh_{\floorsmall{\funcg(n)/2} + 1}
  \end{equation}
  where
  $\formg_n$ and 
  $\formh_{\floorsmall{\funcg(n)/2} + 1}$ 
  have disjoint sets of variables.
  By
  \refobs{obs:disjoint-variable-sets},
  any resolution refutation of
  $\formf_n$ 
   refutes either
  $\formg_n$ or~$\formh_{\floorsmall{\funcg(n)/2} + 1}$.
  We have
  \begin{equation}
    \varspacerefcompact{\formh_{\floorsmall{\funcg(n)/2} + 1}}
    = 2 \cdot ( \floorsmall{\funcg(n)/2} + 1)
    > \funcg(n)
    = \varspacerefsmall{\formg_n}
    \eqcomma
  \end{equation}
  so a resolution refutation in minimal variable space must refute
  $\formg_n$ in length
  $\exp(\Lowerboundsmall{n})$.
  However, allowing at most two more literals in memory, 
  the resolution refutation can disprove the formula
  $\formh_{\floorsmall{\funcg(n)/2} + 1}$ 
  instead
  in length linear in the (total) formula size.

  Thus, we have a formula family
  $\setsmall{\formf_n}_{n=1}^{\infty}$
  of size
  $
  \Lowerboundsmall{n} =
  \sizeofsmall{\formf_n} = 
  \Ordocompact{n^2}
  $
  refutable in length and variable space both linear in the formula size, 
  but where any minimum variable space refutation must have length
  $\exp(\Lowerboundsmall{n})$.
  Adjusting the indices as needed, we get a formula family with a
  trade-off of the form stated in 
  \refth{th:Hertel-Pitassi}.
\end{proof}

\subsection{Some Other Trade-off Results for Resolution}
\label{sec:some-other-trade-offs}

Using a similar trick as in the previous 
subsection,     
we can prove the
following length-clause  space trade-off.

\begin{LengthSpaceTheorem}
    There is a family of \kcnfform{}s
  $\setsmall{\fstd_n}_{n=1}^{\infty}$
  of size $\Tightsmall{n}$ \st:
  \begin{compactitem}
  \item
    The minimal  clause space of refuting $\fstd_n$ 
    in resolution
    is
    $\clspacerefsmall{\fstd_n} = \Tightcompact{\sqrt[3]{n}}$.

  \item
    Any resolution refutation
    $\refof{\proofstd}{\fstd_n}$
    in minimal clause space
    must have length
    $\lengthofsmall{\proofstd} =  
    \exp \bigl(\Lowerboundcompact{\sqrt[3]{n}}\bigr)$.

  \item
    There are  
    resolution 
    refutations
    $\refof{\proofstd'}{\fstd_n}$
    in 
    asymptotically minimal 
    clause space
    $\clspaceofsmall{\proofstd'} =
    \Ordocompact{\clspacerefsmall{\fstd_n}}$
    and length
    $\lengthofsmall{\proofstd'} = \Ordosmall{n}$, 
    \ie linear in the formula size.
  \end{compactitem}
\end{LengthSpaceTheorem}

The same game can be played with refutation width as well.

\begin{theorem}
  \label{th:easy-length-width-trade-off}
    There is a family of \kcnfform{}s
  $\setsmall{\fstd_n}_{n=1}^{\infty}$
  of size $\Tightsmall{n}$ \st:
  \begin{compactitem}
  \item
    The minimal width of refuting $\fstd_n$ is
    $\widthrefsmall{\fstd_n} = \Tightcompact{\sqrt[3]{n}}$.

  \item
    Any          
    refutation
    $\refof{\proofstd}{\fstd_n}$
    in minimal width
    must have length
    $\lengthofsmall{\proofstd} 
    =          
    \exp \bigl(\Lowerboundcompact{\sqrt[3]{n}}\bigr)$.

  \item
    There are                             
    refutations
    $\refof{\proofstd'}{\fstd_n}$
    with    
    $\widthofsmall{\proofstd'} =
    \Ordocompact{\widthrefsmall{\fstd_n}}$
    and         
    $\lengthofsmall{\proofstd'} = \Ordosmall{n}$. 
  \end{compactitem}
\end{theorem}

We only present the proof of
\refth{th:easy-length-clause-space-trade-off},
as
\refth{th:easy-length-width-trade-off}
is proved in exactly the same manner.

\begin{proof}[Proof of    \refth{th:easy-length-clause-space-trade-off}]
  Let
  $\formg_n$
  be a \xcnfform{3}{} family as in
  \refth{th:clause-space-approx-n-clauses}
  having size $\Tightsmall{n}$,
  refutation length
  $\lengthrefsmall{\formg_n} = \exp(\Tightsmall{n})$,
  and refutation clause space
  $\clspacerefsmall{\formg_n} = \Tightsmall{n}$.
  Let
  $\formh_m$
  be a \xcnfform{3}{} family as in
  \refth{th:moderately-hard-formulas}
  of size
  $\Tightcompact{m^3}$ 
  \st 
  $\lengthrefsmall{\formh_m} = \Ordocompact{m^3}$
  and
  $\clspacerefsmall{\formh_m} = \Tightsmall{m}$.
  Define
  \begin{equation}
    \funcg(n) = \minofsetcompact[\,|\,]{m}
          {\clspacerefsmall{\formh_m} >
          \clspacerefsmall{\formg_n}}
          \eqperiod
  \end{equation}
  Note that since
  ${\clspacerefsmall{\formh_m}} = \Lowerboundsmall{m}$
  and
  $\clspacerefsmall{\formg_n} = \Ordosmall{n}$,
  we know that
  $\funcg(n) = \Ordosmall{n}$.

  Now as before let
  $
  \formf_n = \formg_n \land \formh_{\funcg(n)}
  $,
  where 
  $\formg_n$ and 
  $\formh_{\funcg(n)}$ 
  have disjoint sets of variables.
  By
  \refobs{obs:disjoint-variable-sets},
  any resolution refutation of
  $\formf_n$ 
  is a refutation of either
  $\formg_n$ or~$\formh_{\funcg(n)}$.
  Since 
  $\funcg(n)$
  has been chosen so that
  $\mbox{$\clspacerefcompact{\formh_{\funcg(n)}}$} > \clspacerefsmall{\formg_n}$,
  a refutation in minimal clause space has to refute~$\formg_n$, which
  requires exponential length.
  However, since
  $\funcg(n) = \Ordosmall{n}$,  
  \refth{th:moderately-hard-formulas}
  tells us that there are refutations of
  ${\formh_{\funcg(n)}}$
  in length
  $\Ordocompact{n^3}$
  and clause space~%
  $\Ordosmall{n}$.
%
\end{proof}

\subsection{Making the Main Trick Explicit}
\label{sec:main-trick}

The proofs of the theorems in 
\reftwosecs 
{sec:simplified-proof}
{sec:some-other-trade-offs} 
come very easily; in fact almost \emph{too} easily.  What is it that
makes this possible?  
In this and the next subsection, we want to highlight two key
ingredients in the constructions.

\newcommand{\cplxmeasureone}{M_1}
\newcommand{\cplxmeasuretwo}{M_2}

The common paradigm for the proofs of
\refthreeths
{th:easy-length-clause-space-trade-off}
{th:Hertel-Pitassi}
{th:easy-length-width-trade-off}
%
%
is as follows.
We are given two complexity measures
$\cplxmeasureone$
and 
$\cplxmeasuretwo$
that we want to trade off against one another.
We do this by finding formulas
$\formg_n$
and
$\formh_m$
\st
\begin{itemize}
\item 
  The formulas
  $\formg_n$ are very hard
  with respect to the  first 
  resource measured by
  $\cplxmeasureone$, 
  while 
  $\cplxmeasuretwo\bigl(\formg_n\bigr)$
  is at most some (more or less trivial) upper bound,
  
\item
  The formulas
  $\formh_m$ are very easy 
  with respect to 
  $\cplxmeasureone$,
  but there is some nontrivial 
\emph{lower} bound
  on the usage $\cplxmeasuretwo\bigl(\formh_m\bigr)$
  of the second resource,

\item
  The index
  $m = m(n)$
  is chosen 
  so as to minimize
  $\cplxmeasuretwo\bigl(\formh_{m(n)}\bigr)-
  \cplxmeasuretwo\bigl(\formg_n\bigr) > 0$,
  \ie
  so that
  $\formh_{m(n)}$
  requires \emph{just a little bit} more of the second resource than
  $\formg_n$.
%

\end{itemize}
Then for
$\formf_n = \formg_n \land \formh_{m(n)}$,
if we demand that a resolution refutation
$\proofstd$
must use 
the minimal amount
of the second resource, it will have to use a large amount of the
first resource.  
However, relaxing the requirement on the second resource
by the very small expression
$\cplxmeasuretwo\bigl(\formh_{m(n)}\bigr) -
\cplxmeasuretwo\bigl(\formg_n\bigr)$,
we can get a refutation $\proofstd'$
using small amounts of both resources.

%
%

Clearly, the formula families 
$\setsmall{\fstd_n}_{n=1}^{\infty}$
that we get in this way
are ``redundant'' in the sense that
each formula $\fstd_n$ is the conjunction of two formulas 
$\formg_n$ and~$\formh_m$
which are themselves already unsatisfiable.
Formally, 
we say that a formula $\fstd$ is 
\introduceterm{minimally unsatisfiable}
if $\fstd$ is unsatisfiable, but removing any clause $\clc \in \fstd$, 
the remaining subformula $\fstd \setminus \setsmall{\clc}$
is satisfiable.
We note that if we would add the requirement in
\reftwosecs
{sec:simplified-proof}
{sec:some-other-trade-offs}
that the formulas under consideration should
be minimally unsatisfiable, 
the proof idea outlined above fails completely.
In contrast, the result in
\cite{HP07ExponentialTimeSpaceSpeedupFOCS}
seems to be independent of any such conditions. 
What conclusions can be drawn from this?


On the one hand, trade-off results for minimally unsatisfiable
formulas seem more interesting, since they tell us something about a
property that some natural formula family has, rather than about some
funny phenomena arising because we glue together two 
totally 
unrelated formulas.

On the other hand, one could argue that the main motivation for
studying space is the connection to memory requirements for proof
search algorithms, \eg algorithms using clause learning. And for such
algorithms, a minimality condition might appear somewhat
arbitrary. There are no guarantees that ``real-life'' formulas will be
minimally unsatisfiable, and most probably there is no efficient way of
testing this condition.%
\footnote{%
The problem of deciding minimal unsatisfiability is
\NP-hard{} but not known to be in \NP.
Formally,
a language $L$ is in the complexity class \DP{} if and only if there are
two languages $L_1 \in \NP$ and $L_2 \in \coNP$ such that 
$L = L_1 \intersection L_2$~\cite{P92ComputationalComplexity}.
\MINIMALUNSATISFIABILITY{}
is \DP-complete~\cite{PW88ComplexityOfFacets}, 
and it seems to be commonly believed that
\DP
$\nsubseteq$
\NP $\union$ \coNP.%
}
So in practice, trade-off results for non-minimal
formulas might be just as interesting.


\subsection{An Auxiliary Trick for Variable Space}
\label{sec:auxiliary-trick}


A second important reason why our proof of
\refth{th:Hertel-Pitassi}
gives sharp results 
is that we are allowed to use \cnfform{}s of growing width.
It is precisely because of this that we can easily construct the needed
formulas $\formh_m$ that are hard \wrt variable space but
easy \wrt length.  
If we would have to restrict ourselves to \kcnfform{}s
for $\clwidth$ constant, it would be much more difficult to
find such examples. Although the formulas in
\refth{th:moderately-hard-formulas}
could be plugged in to give 
a slightly weaker trade-off,   
we are not aware of any family of \kcnfform{}s that can provably give
the very sharp result in
\refth{th:Hertel-Pitassi}.
(Note, though, that the formula families used in the proofs of
\reftwoths
    {th:easy-length-clause-space-trade-off}
    {th:easy-length-width-trade-off}
consist of \kcnfform{}s).

This is not the only example of a space measure behaving badly for
formulas of growing width.
We already discussed the lower bound
$ 
\clspacerefsmall{\fstd}  
\geq 
\widthrefsmall{\fstd} - \widthofarg{\fstd}
+ 3
$
on clause space in terms of
length in
\refth{th:small-clause-space-implies-small-width},
and the result in~%
\refth{th:narrow-proofs-may-be-spacious}
%
%
that this inequality is asymptotically strict in the sense
that there are \kcnfform families
$\fstd_n$ 
with
$\widthrefsmall{\fstd_n} = \Ordosmall{1}$
but
$\mbox{$\clspacerefsmall{\fstd_n}$}  = \Tightsmall{\log n}$. 

However, if we are allowed to consider formulas of growing width, the
fact that the inequality in~%
\refth{th:small-clause-space-implies-small-width}
is not tight is entirely trivial.
Namely, let us say that 
a \cnfform $\fstd$ is \introduceterm{$k$-wide} if
all clauses in~$\fstd$ have size at least~$k$.
In~%
\cite{ET01SpaceBounds},
it was proven that
for $\fstd$ a $k$-wide unsatisfiable \cnfform it holds that
$\clspaceref{\fstd} \geq k + 2$.
So in order to get a formula family
$\fstd_n$
\st  
$\widthrefsmall{\fstd_n} - \widthofarg{\fstd_n} = \Ordosmall{1}$
but
$\clspacerefsmall{\fstd_n}  = \lowerboundsmall{1}$,
just pick some suitable formulas 
$\setsmall{\fstd_n}_{n=1}^{\infty}$
of growing width.

In our opinion, these phenomena are clearly artificial.
Since every \cnfform can be rewritten as an equivalent \kcnfform
without increasing the size more than linearly, 
the right approach
when studying space measures in resolution 
seems to be to require that the formulas under study should have
constant width. 

As a final comment before moving on to 
our main result, 
we note that the open trade-off questions mentioned in
\refsec{sec:open-questions}
do not suffer from the technical problems discussed above.

\section{A Game for Analyzing Pebbling Contradictions}
\label{sec:definition-multi-pebble-game}

We now start our construction for the proof of
\refth{th:main-theorem}, 
which will require the rest of this paper.
In this section we present the modified pebble game that we will use
to study the clause space of resolution refutations of pebbling
contradictions. 

\subsection{Some Graph Notation and Definitions}
\label{sec:definition-multi-pebble-game-graph-notation}

We first present some notation and terminology that will be used
in what follows. See
\reffig{fig:vertex-sets-wrt-vertex}
for an illustration of the next definition.

\begin{figure}[tp]
  \centering
  \includegraphics{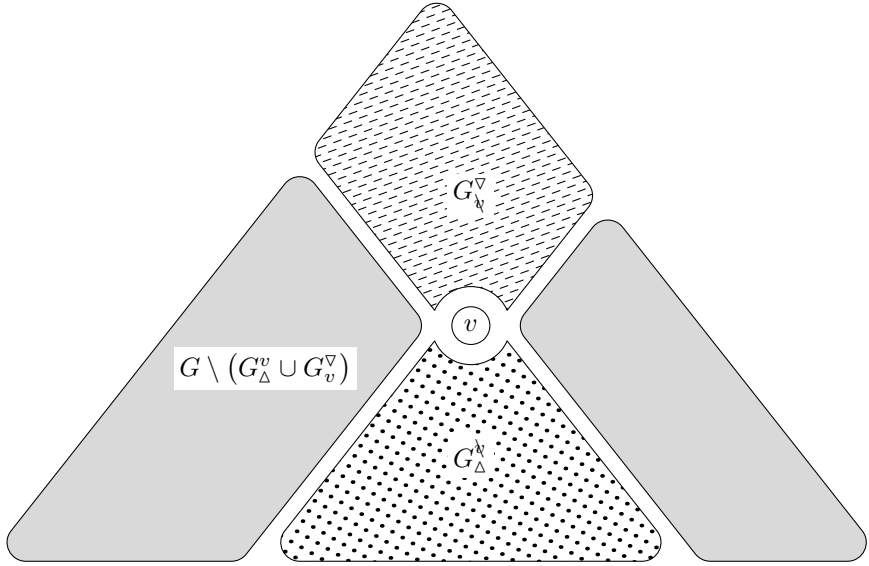}%
  \caption{%
    Notation for sets of vertices in DAG $G$ \wrt a vertex $v$.}
  \label{fig:vertex-sets-wrt-vertex}
\end{figure}

\begin{definition}
We let
$\succnode{v}$
denote the immediate successors and
$\prednode{v}$
denote the immediate predecessors 
of a vertex $v$ in a DAG~$G$. 
Taking the transivite closures of
$\succnode{\cdot}$
and
$\prednode{\cdot}$, we let
$\abovevertices{v}$
denote all vertices reachable from  $v$
(vertices ``above''~$v$)
and
$\belowvertices{v}$
denote all vertices from which $v$ is reachable (vertices
``below''~$v$).  We write
$\belowverticesNR{v}$
and
$\aboveverticesNR{v}$
to denote the corresponding sets with the vertex $v$ itself removed.
If
$\prednode{v} = \setsmall{u,w}$, 
we say that $u$ and $w$ are \introduceterm{siblings}.
If
$u \not \in \belowvertices{v}$
and
$v \not \in \belowvertices{u}$,
we say that $u$ and $v$ are \introduceterm{\unrelatedverticesadj{}}
vertices.
Otherwise they are
\introduceterm{\relatedverticesadj{}}.
\end{definition}

When reasoning about arbitrary vertices
we will often use as a canonical example a vertex $r$ with 
assumed predecessors
$\predrequalpq$.

Note that
for a leaf $v$ we have $\prednode{v} = \emptyset$,
and for the sink $z$ of $G$ we have $\succnode{z} = \emptyset$.
Also note that
$\belowvertices[G]{v}$
and
$\abovevertices[G]{v}$
are sets of vertices, not subgraphs. However, we will allow ourselves
to overload the notation and sometimes use this notation both for the
subgraph and its vertices.  Moreover, as a rule we will overload the
notation for the graph $G$ itself and its vertices, and usually write
only $G$ when we mean~$\vertices{G}$, and when this should be clear
from context. 

For our pebble game to work, we require  of the graphs under study
that they have the following property.

\begin{property}[Sibling non-reachability]
  \label{property:sibling-non-reachability-property}
  We say that a DAG $G$ has
  the \introduceterm{\siblingnonreachabiblitypropertynoref{}} 
  if for all vertices
  $u$ and $v$ that are siblings in $G$, it holds that
  $u \notin \belowvertices{v}$
  and
  $v \notin \belowvertices{u}$,
  \ie the siblings are not reachable from one another.
%
\end{property}

Phrased differently,
\refproperty{property:sibling-non-reachability-property}
asserts that siblings are \unrelatedverticesadj.

A sufficient condition for 
\refproperty{property:sibling-non-reachability-property}
to hold is that if
$v$ is reachable from $u$, then all paths
\mbox{$\pathfromto{P}{u}{v}$}
have the same length. This holds \eg for the class of 
\introduceterm{layered graphs},
and it is also easy to see directly that layered graphs possess
\refproperty{property:sibling-non-reachability-property}.

\begin{definition}[Layered DAG]
  \label{def:layered-graphs}
  A \introduceterm{layered DAG} $G$ is a DAG whose
  vertices are partitioned into
  (nonempty) sets of \introduceterm{layers} 
  $V_0, V_1, \ldots, V_h$
  on
  \introduceterm{levels}
  $0,1, \ldots, h$,
  and whose edges run between consecutive layers.
  That is, 
  if $(u,v)$ is a directed edge, then the level of $u$ is
  $\levelstd - 1$
  and the level of $v$ is
  $\levelstd $
  for some $L \in \intnfirst{h}$.
  We say that $h$ is the \introduceterm{height} of the 
  layered DAG~$G$.
\end{definition}

Throughout this paper, we will assume that all source vertices in a
layered DAG are located on the bottom level~$0$.
Let us next give a formal definitions of the pyramid graphs that are
the focus of this paper.

\begin{definition}[Pyramid graph]
  \label{def:pyramid-graphs}  
  The \introduceterm{pyramid graph} $\pyramidgraphh$
  of height $h$
  is a layered DAG with $h+1$ levels, where there is one vertex on the
  highest level (the sink~$z$), 
  two vertices on the next level et cetera  
  down to $h+1$ vertices at
  the lowest level~$0$. 
  The $i$th vertex at level $\levelstd$ has incoming edges 
  from the $i$th and
  $(i+1)$st vertices at level $\levelstd-1$.
\end{definition}

We also need some notation for contiguous and non-contiguous
topologically ordered sets of vertices in a DAG.

\begin{definition}%
  [Paths and chains]    
\label{def:chain-pebble}
\label{def:chains-and-paths}
  We say that $V$ is a
  \introduceterm{(totally) ordered}
  set of vertices in a DAG~$G$,
  or a \introduceterm{chain},
  if all vertices in $V$ are
  \relatedverticesadj
  (\ie
  if for all $u,v \in V$,
  either
  $u \in \belowvertices{v}$
  or
  $v \in \belowvertices{u}$).
  A
  \introduceterm{path}
  $P$ is a contiguous chain, \ie such that
  $\succnode{v} \intersectionSP P \neq \emptyset$
  for all $v \in P$ except the top vertex.

  We write
  $\pathfromto{P}{v}{w}$ 
  to  denote a path starting in $v$ and ending in $w$.
  A
  \introduceterm{\sourcepath{}}
  is a path that starts at some source vertex of~$G$.
  A \introduceterm{path via $w$} is a path
  \st
  $w \in P$. 
  We will also say that $P$ \introduceterm{visits} $w$. 
  For a chain $V$, we let
  \begin{itemize}
  \item 
    $\minvertex{V}$ denote the bottom vertex of~$V$,
    \ie the unique $v \in V$ \st
    $V \subseteq \abovevertices{v}$,
    
  \item 
    $\maxvertex{V}$ denote the top vertex of~$V$, 
    \ie the unique $v \in V$ \st
    $V \subseteq \belowvertices{v}$,
    
  \item 
    $\pathsinvertex{V}$ denote the set of all paths
    $\pathfromto{P}{\minvertex{V}}{\maxvertex{V}}$
    \introduceterm{via}~$V$ or
    \introduceterm{agreeing with}~$V$,
    \ie \st $V \subseteq P$,
    and

  \item
    $\pathsviavertex{V}$
    denote the set of all \sourcepath{}s 
    \introduceterm{agreeing with}~$V$.
  \end{itemize}                                                                
  We write
  $\unionpathsinvertex{V}$
  to denote the union of the vertices in all paths
  $P \in \pathsinvertex{V}$
  and
  $\unionpathsviavertex{V}$
  for the union of all vertices in paths
  $P \in \pathsviavertex{V}$.
\end{definition}


In the rest of this paper, we will almost exclusively discuss DAGs
with certain structural properties.  The next definition is so that we
will not have to repeat these properties over and over again.

\begin{definition}[\Pebblingdag{}]
  \label{def:blob-pebbling-DAG}
  A
  \introduceterm{\pebblingdag{}}
  is a DAG that has a unique sink, which we will alway denote~$z$,
  that has vertex indegree~$2$ for all non-sources, and that satisfies
  the \siblingnonreachabiblityproperty.
\end{definition}

\subsection{Description of the \MULTIPEBBLEGAME and Formal Definition}
\label{sec:definition-multi-pebble-game-the-definition}

To prove a lower bound on the refutation space of pebbling
contradictions, we want to interpret 
derivation steps  in
terms of pebble placements and removals in the corresponding graph.
In
\refsec{sec:proof-overview-and-paper-organization},
we outlined an intuitive correspondence between clauses and pebbles.
The problem is that if we try to use this correspondence, the pebble
configurations that we get do not obey the rules of the black-white
pebble game. Therefore, we are forced to to change the pebbling rules.
In this section, we present the modified pebble game used for
analyzing resolution derivations.

Our first modification of the pebble game is to alter the rule for
white pebble removal so that a white pebble can be removed from a
vertex when a black pebble is placed on that same vertex.  
This will make the correspondence between pebblings and resolution
derivations much more natural.
Clearly, this is only a minor adjustment, and it is easy to prove
formally that it does not really change anything.

Our second, and far more substantial, modification of the pebble game
is motivated by the fact that in general, a resolution refutation a
priori has no reason to follow our pebble game intuition.  Since
pebbles are induced by clauses, if at some derivation step the
refutation chooses to erase ``the wrong clause'' from the point of
view of the induced pebble configuration, this can lead to pebbles
just disappearing. Whatever our translation from clauses to pebbles
is, a resolution proof that suddenly out of spite erases practically
all clauses must surely lead to practically all pebbles disappearing,
if we want to maintain a correspondence between clause space and
pebbling cost.
This is all in order for black pebbles, but if we allow uncontrolled
removal of white pebbles we cannot hope for any nontrivial lower
bounds on pebbling price   
(just white-pebble the two predecessors of the sink, then black-pebble
the sink itself and finally remove the white pebbles).
    
Our solution to this problem is to keep track of exactly which white
pebbles have been used to get a black pebble on a vertex.  Loosely
put, removing a white pebble from a vertex~$v$ without placing a black
pebble on the same vertex should be in order, provided that all black
pebbles placed on vertices above $v$ in the DAG with the help of the
white pebble on $v$ are removed as well.
We do the necessary bookkeeping by defining
\introduceterm{\subconftext{}s} of pebble configurations,
each \subconftext consisting of black pebble together with all the
white pebbles this black pebble depends on, 
and require that if any pebble in a \subconftext is removed, then
all other pebbles in this \subconftext must be removed as well.

Another problem is that resolution derivation steps can be made that
appear intuitively bad given that we know that the end goal is to
derive the empty clause, but where formally it appears where hard to
nail down wherein this supposed badness lies. To analyze such
apparently non-optimal derivation steps, we introduce an
\introduceterm{inflation}
rule in which a black pebble can be inflated to a 
\introduceterm{\multipebble{}}
covering multiple vertices. The way to think of this is that a black
pebble on a vertex~$v$ corresponds to derived truth ov~$v$, whereas
for a \multipebble pebble on $V$ we only know that some vertex $v\in V$ is
true, but not which one.
For reasons that will perhaps become clearer in
\reftwosecapps
{sec:pebble-games-pyramids}
{sec:tight-lower-bound-for-multi-pebbling-pyramid},
in is natural to consider \multipebble{}s that are chains
(\refdef{def:chains-and-paths}).

We now present the formal definition of the concept
used to ``label'' each black 
\multipebble pebble with
the set of white pebbles (if any) this black pebble is dependent on.
The intended meaning of the notation
$\mpscnot{B}{W}$
is a black \multipebble on $B$ together with the white pebbles $W$
below $v$ with the help of which  we have been able to place the black
\multipebble on~$B$.
These ``associated'' or ``supporting'' white pebbles can be located on
any vertex $w \notin B$ that can be visited by a \sourcepath~$P$
to $\topvertex{B}$ agreeing with~$B$.
Formally, the
\introduceterm{\lpptext{}}
\wrt a chain $B$ with $b = \bottomvertex{B}$
is the set of vertices
\begin{equation}
\label{eq:legal-pebble-positions-definition}
  \lpp{B}
  =
  \belowverticesNR{b} 
  \unionSP
  \left(\unionpathsinvertex{B} \setminus {B} \right)
  =
 \unionpathsviavertex{B} \setminus {B} 
  \eqperiod
\end{equation}
We refer to the structure
$\mpscnot{B}{W}$
grouping together
a black \multipebble $B$ and its associated white pebbles $W$ as a
\introduceterm{\mpscfulltext{}},
or just
\introduceterm{\mpsctext{}}
for short.

\begin{definition}[\Mpscfulltext{}]
  \label{def:C-multi-pebble-configuration}
  \label{def:C-legal-white-pebbles}
  For   sets of vertices 
  $B,W$ in a \pebblingdag~$G$,
  $\mpscnot{B}{W}$
  is a 
  \introduceterm{\mpscfulltext{}}
  if
  $B \neq \emptyset$ is a chain
  and
  $W \subseteq \lpp{B}$.
  We refer to $B$ as a 
  (single) black 
  \introduceterm{\multipebble{}} 
  and to  $W$ as (a number of different) white pebbles
  \introduceterm{supporting}~$B$.
  We also say that $B$ is \introduceterm{\bldependent{}} on~$W$.
  If $W = \emptyset$, 
  $B$ is
  \introduceterm{\blindependent{}}.
  \Multipebble{}s $B$ with
  $\setsizesmall{B} = 1$ 
  are said to be
  \introduceterm{\atomicmultipebbleadj{}}.


  A set of
  \mpscfulltext{}s 
  $
  \mpconf
  = \setdescrcompact{\mpscnotstd[i]}{i = 1, \ldots, m}
  $
  together constitute 
  \introducetermanmpctext.
\end{definition}

Note in particular that it always holds that
$B \intersectionSP W = \emptyset$
for a \mpscfulltext~$\mpscnotstd$.

Since the definition of the game we will play with these
\multipebble{}s and pebbles is somewhat involved, let us first try to
give an intuitive description.

\begin{itemize}
\item 

There is one single rule corresponding to the two
\reftworules
{pebrule:black-placement}
{pebrule:white-placement}
for black and white pebble placement
in the black-white pebble game of
\refdef{def:bw-pebble-game}.
This
\introduceterm{introduction} 
rule says that we can place a black pebble on a vertex~$v$
together with white pebbles on its predecessors (unless $v$ is a
source,
in which case no white pebbles are needed).

\item 
  The analogy for
  \refrule{pebrule:black-removal}
  for black pebble removal in
  \refdef{def:bw-pebble-game}
  is a rule for ``shrinking'' black \multipebble{}s.
%
%
  A vertex~$v$ in a \multipebble can be eliminated by
  \introduceterm{merging}
  two \mpscfulltext{}s,
  provided that there is both a black \multipebble and a white pebble
  on~$v$, and provided that the two black \multipebble{}s involved 
  in this \introduceterm{merger}   do
  not intersect the supporting white pebbles of one another in any
  other vertex than~$v$.
  Removing black pebbles in the black-white pebble game
  corresponds to shrinking \atomicmultipebbleadj
  black \multipebble{}s.
  
\item 
  A black \multipebble can be
  \introduceterm{inflated}
  to cover more vertices, 
  as long as it does not collide with its own supporting 
  white vertices. 
  Also, new supporting white pebbles can be added at
  an inflation move.
  There is no analogy of this move in the usual black-white pebble game.

\item
  The   \refrule{pebrule:white-removal}
  for  white pebble removal 
  also corresponds to merging in the \multipebblegame, since the white
  pebble used in the merger is eliminated as well. In addition,
  however, a white pebble on~$w$ can also 
  disappear if its black \multipebble $B$ changes so that $w$ 
  no longer can be visited on a path via~$B$
%
  (\ie if $w$ is no longer a \lpptextsing \wrt~$B$).

\item
  Other than that, individual white pebbles, and individual black
  vertices covered by \multipebble{}s, can never just disappear. If we
  want to remove a white pebble or parts of a black \multipebble, we
  can do so only by
  \introduceterm{erasing}
  the whole \mpscfulltext.
\end{itemize}
The formal definition follows.
See \reffig{fig:blobpebblingmoves} 
for some examples of \multipebblingtext moves.

\ifthenelse
{\boolean{maybeSTOC}}
{\begin{figure*}[t]}
{\begin{figure}[tp]}
  \centering
  \subfigure[Empty pyramid.]
  {
    \label{fig:blobpebblingmoves-a}
    \begin{minipage}[b]{.47\linewidth}
      \centering
      \includegraphics{blobpebblingmoves.2}%
    \end{minipage}
  }
  \hfill
  \subfigure[Introduction move.]
  {
    \label{fig:blobpebblingmoves-b}
    \begin{minipage}[b]{.47\linewidth}
      \centering
      \includegraphics{blobpebblingmoves.3}%
    \end{minipage}
  }

  \subfigure[Two \mpsctext{}s before merger.]
  {
    \label{fig:blobpebblingmoves-c}
    \begin{minipage}[b]{.47\linewidth}
      \centering
      \includegraphics{blobpebblingmoves.5}%
    \end{minipage}
  }
  \hfill
  \subfigure[The merged \mpsctext.]
  {
    \label{fig:blobpebblingmoves-d}
    \begin{minipage}[b]{.47\linewidth}
      \centering
      \includegraphics{blobpebblingmoves.6}%
    \end{minipage}
  }

  \subfigure[\Mpsctext before inflation.]
  {
    \label{fig:blobpebblingmoves-e}
    \begin{minipage}[b]{.47\linewidth}
      \centering
      \includegraphics{blobpebblingmoves.7}%
    \end{minipage}
  }
  \hfill
  \subfigure[\Mpsctext after inflation.]
  {
    \label{fig:blobpebblingmoves-f}
    \begin{minipage}[b]{.47\linewidth}
      \centering
      \includegraphics{blobpebblingmoves.8}%
    \end{minipage}
  }
\ifthenelse
{\boolean{maybeSTOC}}
{}
{

  \subfigure[Another \mpsctext before inflation.]
  {
    \label{fig:blobpebblingmoves-g}
    \begin{minipage}[b]{.47\linewidth}
      \centering
      \includegraphics{blobpebblingmoves.9}%
    \end{minipage}
  }
  \hfill
 \subfigure[\mbox{After inflation with vanished white pebbles.}]
  {
    \label{fig:blobpebblingmoves-h}
    \begin{minipage}[b]{.47\linewidth}
      \centering
      \includegraphics{blobpebblingmoves.10}%
    \end{minipage}
  }
}
%
%
%
%
  \caption{Examples of moves in the \multipebblegame.}
  \label{fig:blobpebblingmoves}
\ifthenelse
{\boolean{maybeSTOC}}
{\end{figure*}}
{\end{figure}}

\begin{definition}[\Multipebblegame{}]
  \label{def:C-modified-pebble-game}
  \label{def:multi-pebble-game}  
  For a \pebblingdag $G$ and
  \mpctext{}s $\mpconf_0$ and $\mpconf_{\stoptime}$ on~$G$, a
  \introduceterm{\multipebblingtext{}}
  from $\mpconf_0$ to $\mpconf_{\stoptime}$ in $G$ is a sequence  
  $\pebbling[P] =
  \setcompact{\mpconf_0, \ldots, \mpconf_{\stoptime}}$
  of configurations  \st
  for all
  $t \in \intnfirst{\stoptime}$,
  $\mpconf_{t}$~is obtained from $\mpconf_{t-1}$
  by one of the following rules:

  \begin{description}
    
    \italicitem[Introduction]
    $\mpconf_{t} = \mpconf_{t-1} \unionSP 
    \setcompact{\intrompscnot{v}}$.

  \italicitem[Merger]
    $\mpconf_{t}  = 
    \mpconf_{t-1} \unionSP  \setcompact{\mpscnotstd}$
    if there are
    $\mpscnotstd[1], \mpscnotstd[2] \in \mpconf_{t-1}$
    \st
    \begin{enumerate}
    \item 
      $B_1 \unionSP B_2$ is (totally) ordered,
      
    \item 
      $B_1 \intersectionSP W_2 = \emptyset$,

    \item 
      $\setsizesmall{B_2 \intersectionSP W_1} = 1$;
      let $\mergervertex$  denote this unique element
      in
      $
      {B_2 \intersectionSP W_1}$,

    \item
      $B = (B_1 \unionSP B_2) \setminus \setsmall{\mergervertex}$,
      and

    \item
      $W = 
      \bigl(
      (W_1 \unionSP W_2) \setminus \setsmall{\mergervertex}
      \bigr)
      \intersectionSP
      \lppstd$,

    \end{enumerate}
    
    We write
    $\mpscnotstd = \spmerge{\mpscnotstd[1]}{\mpscnotstd[2]}$
    and refer to this as a
    \introduceterm{merger on~$\mergervertex$}.
    
  \italicitem[Inflation]
    $\mpconf_{t} = \mpconf_{t-1} \unionSP
    \setcompact{\mpscnotstd}$
    if there is a $\mpscnotprime \in \mpconf_{t-1}$
    \st
    \begin{enumerate}
    \item 
      $B \supseteq B'$,
      
    \item 
      $B \intersectionSP W' = \emptyset$,
      and
      
    \item 
      $W \supseteq 
      W' \intersectionSP 
      \lppstd
      $.

    \end{enumerate}

    We say that
    $\mpscnotstd$
    is derived from
    $\mpscnotprime$
    by inflation or that
    $\mpscnotprime$
    is \introduceterm{inflated} to  yield
    $\mpscnotstd$.

  \italicitem[Erasure]
    $\mpconf_{t} = \mpconf_{t-1} \setminus \setcompact{\mpscnotstd}$
    for
    $\mpscnotstd \in \mpconf_{t-1}$.

\end{description}
The \multipebblingtext $\multipebbling$ is
\introduceterm{\pebunconditional{}} if
$\mpconf_0 = \emptyset$
and \introduceterm{\pebconditional{}} otherwise.
A
\introduceterm{\pebcomplete{} \multipebblingtext{}} of $G$
is an \pebunconditional{} pebbling $\multipebbling$
ending in
${\mpconf_{\stoptime} = \setcompact{\unconditionalblackmpscnot{z}}}$
for $z$ the unique sink of~$G$.
\end{definition}

\subsection{\MULTIPEBBLINGTEXT Price}

We have not yet defined what the price of a \multipebblingtext is. The
reason is that it is not a priori clear what the ``correct''
definition of \multipebblingtext price should be.

It should be pointed out  that the \multipebblegame has no obvious
intrinsic value---its function is to serve as a tool
to prove lower bounds on the resolution refutation space 
of pebbling contradictions.
The intended structure of our lower bound proof for resolution space
is that we 
want look at resolution refutations of pebbling contradictions, interpret
them in terms of \multipebblingtext{}s on the underlying graphs, and
then translate lower bounds on the price of these
\multipebblingtext{}s into lower bounds on the size of the
corresponding clause configurations.
Therefore, 
we have two requirements for the \multipebblingtext price
$\multipebblingprice{G}$:
\begin{enumerate}
\item 
  It should be sufficiently high to enable us to prove
  good lower bounds on
  $\multipebblingprice{G}$,
  preferrably by relating it to the
  standard black-white pebbling price
  $\bwpebblingprice{G}$.

\item
  It should also be sufficiently low, so that lower bounds on
  $\multipebblingprice{G}$
  translate back to lower bounds on the 
  size of the clause configurations.
\end{enumerate}
So when defining pebbling price in
\refdef{def:multi-pebbling-price}
below, we also have to have in mind the coming
\refdef{def:induced-multi-pebble-subconfiguration}
saying how we will interpret clauses in terms of \multipebble{}s and
pebbles and that these two definitions together should make it
possible for us to lower-bound clause set size in terms of pebbling
cost.

For black pebbles, we could try to charge $1$ for each distinct
\multipebble. But this will not work, since then the second
requirement above fails. 
For the translation of clauses to blobs and pebbles sketched in
\refsec{sec:proof-overview-detailed}
it is possible to construct clause
configurations that correspond to an exponential number of distinct
black \multipebble{}s measured in the clause set size.
The other natural extreme seems to be to charge only for mutually
disjoint black  \multipebble{}s. But this is far too generous, and the first
requirement above fails. To get a trivial example of this, take any 
ordinary
black pebbling of $G$ and translate in into an (atomic)
\multipebblingtext, but then change it so that each black pebble
$\mpscblacknot{v}$
is immediately inflated to 
$\mpscblacknot{\setsmall{v,z}}$
after each introduction move. 
It is straightforward to verify that this would yield a pebbling of
$G$ in constant cost.
For white pebbles, the first idea might be to charge $1$ for every
white-pebbled vertex, just as in the standard pebble game. On
closer inspection, though, this seems to be not quite what we need.

The definition presented below   
turns out to give us both of the desired properties above, and allows
us to prove an optimal bound. Namely,
we define \blobpebblingtext price so as to
charge $1$ for 
\emph{each distinct bottom vertex} among the black \multipebble{}s, 
and so as to charge for the subset of supporting white pebbles
$W \intersectionSP \belowvertices[G]{b}$
in \anmpsctext 
$\mpscnotstd$
that are 
\emph{located below the bottom vertex}
$\bottomvertex{B}$   
of its black \multipebble{}~$B$. 
Multiple distinct \multipebble{}s
with the same bottom vertex come for free, however, and any supporting
white pebbles above the bottom vertex of its own \multipebble are also
free,  although we still have to keep track of them.

\begin{definition}[\Multipebblingtext price]
  \label{def:multi-pebbling-price}
  \label{def:C-multi-pebbling-price}
  For \anmpsctext 
  $\mpscnotstd$, we say that
  $
  \blackschargedfor(\mpscnotstd) = \setsmall{\bottomvertex{B}}
  $
  is the \introduceterm{chargeable black vertex} and that
  $
  \whiteschargedfor(\mpscnotstd)
  =
  W \intersectionSP \belowvertices[G]{\bottomvertex{B}}
  $
  are the \introduceterm{chargeable white vertices}.
  The 
  \introduceterm{chargeable vertices} of
  the \mpsctext
  $\mpscnotstd$
  are all vertices in the union
  $\blackschargedfor(\mpscnotstd) \unionSP \whiteschargedfor(\mpscnotstd)$.
  This definition is extended to \mpctext{}s $\mpconf$ in the natural
  way by letting
  \begin{equation*}
    \blackschargedfor(\mpconf) 
    =
    \Union_{\mpscnotstd \in \mpconf} \blackschargedfor(\mpscnotstd)
    =
    \setdescrcompact
    {\bottomvertex{B}}
    {\mpscnotstd \in \mpconf}
  \end{equation*}
  and
  \begin{equation*}
    \whiteschargedfor(\mpconf) 
    =
    \Union_{\mpscnotstd \in \mpconf} \whiteschargedfor(\mpscnotstd)
    =
    \Union_{\mpscnotstd \in \mpconf}
    \left(
      {W \intersectionSP
        \belowvertices[G]{\bottomvertex{B}}}
    \right) 
    \eqperiod
  \end{equation*}
  The cost of \anmpctext $\mpconf$ is 
  $
  \mpcost{\mpconf} =
  \setsizecompact{
    \blackschargedfor(\mpconf) 
    \unionSP
    \whiteschargedfor(\mpconf)
  }
  $,
  and the cost of a \multipebblingtext
  $\multipebbling = \setcompact{\mpconf_0, \ldots, \mpconf_{\stoptime}}$
  is 
  $
  \mpcost{\multipebbling} =
  \Maxofexpr[t \in \intnfirst{\stoptime}]{\mpcost{\mpconf_t}}
  $.
  
  The \introduceterm{\multipebblingtext price}
  of a \mpscfulltext
  $\mpscnotstd$,
  denoted $\multipebblingprice{\mpscnotstd}$,
  is the minimal cost of any \pebunconditional{} \multipebblingtext 
  $
  \multipebbling
  =
  \setsmall{\mpconf_{0}, \ldots, \mpconf_{\stoptime}}
  $
  \st
  $\mpconf_{\stoptime}
  =
  \setcompact{\mpscnotstd}
  $.
  The
  \multipebblingtext price
  of a DAG~$G$
  is
  $\multipebblingprice{G}
  = 
  \multipebblingprice{\unconditionalblackmpscnot{z}}$,
  \ie the minimal cost of any \pebcomplete{} \multipebblingtext of~$G$.  
\end{definition}

We will also write
$
\mpcwhitesof{\mpconf}
$
to denote the set of all white-pebbled vertices in~$\mpconf$, 
including non-\chargeabletext ones.

\section{Resolution Derivations Induce \MULTIPEBBLINGTEXT{}s}
\label{sec:derivations-induce-chain-multi-pebblings}
\label{sec:induced-pebbling}

For simplicity, 
in this section, as well as in the next one,
we will write
$v_1, \ldots, v_{\pebdeg}$
instead of
$\varx(v)_1, \ldots, \varx(v)_{\pebdeg}$
for the $\pebdeg$ variables associated with
$v$ in a 
$\pebdeg$th degree pebbling contradiction.
That is, in
\reftwosecs
{sec:derivations-induce-chain-multi-pebblings}
{sec:multi-pebble-configuration-cost}
small letters with subscripts will denote only variables in
propositional logic and nothing else.

It turns out that for technical reasons, it is more natural to ignore
the target axioms 
$\olnot{z}_1, \ldots, \olnot{z}_{\pebdeg}$
and focus on resolution derivations of
$\sourceclause[l]{z}$
from the rest of the formula rather than resolution refutations of all
of
$\pebcontr[G]{\pebdeg}$.
Let us write
$
\pebcontrNT{\pebdeg}
= 
\pebcontr{\pebdeg}
\setminus
\setcompact{
  \olnot{z}_1, \ldots, \olnot{z}_{\pebdeg}
}
$
to denote the pebbling formula over~$G$
with the target axioms in the  pebbling contradiction removed.
The next lemma is the formal statement saying that we may just as well
study derivations of
$\targetclause[l]$
from this pebbling formula
$\pebcontrNT{\pebdeg}$
instead of refutations of~%
$\pebcontr{\pebdeg}$.

\begin{lemma}
  \label{lem:C-clause-space-the-same-without-target-axioms}
  \label{lem:clause-space-the-same-without-target-axioms}
  For any DAG $G$ with sink~$z$,
  it holds that
  $
  \clspaceref{\pebcontr{\pebdeg}}
  =
  \clspacederiv
  {\pebcontrNT{\pebdeg}}
  {\targetclause[l]}
  $.
\end{lemma}

\begin{proof}
  For any resolution derivation
  $\derivof{\proofstd^*}{\pebcontrNT{\pebdeg}}{\targetclause[l]}$,
  we can get a resolution  refutation    of
  ${\pebcontr{\pebdeg}}$
  from $\proofstd^*$ in the same space
  by resolving
  $\targetclause[l]$
  with all $\olnot{z}_l$, $l=1, \ldots, \pebdeg$, in space~$3$.

  In the other direction, for
  $\refof{\proofstd}{\pebcontr{\pebdeg}}$
  we can extract a derivation of
  $\targetclause[l]$
  in at most the same space by simply omitting all downloads of and
  resolution steps on $\olnot{z}_l$ in~$\proofstd$,
  leaving the literals $z_l$ in the clauses.
  Instead of the final empty clause $\emptycl$ we get some clause
  $\cld \subseteq \targetclause[l]$,
  and since
  $\pebcontrNT{\pebdeg} \nimpl \cld \subsetneqq \targetclause[l]$
  and resolution is sound, 
  we have $\cld = \targetclause[l]$.\qedhere%
\end{proof}


In view of
\reflem{lem:C-clause-space-the-same-without-target-axioms},
from now on we will only consider resolution derivations from
$\pebcontrNT{\pebdeg}$ and try to convert clause configurations in
such derivations into sets of
\mpscfulltext{}s.

To avoid cluttering the notation with an excessive amount of brackets, we
will sometimes use sloppy notation for sets.
We will allow ourselves to omit curly brackets around singleton sets
when this is clear from context, writing \eg
$V \unionSP \singset{v}$ instead of
$V \unionSP \setsmall{v}$
and
$\mpscnot{B \unionSP \singset{b}}{W \unionSP \singset{w}}$
instead of
$\mpscnot{B \unionSP \setsmall{b}}{W \unionSP \setsmall{w}}$.
Also, we will sometimes omit the curly brackets around sets of
vertices in black \multipebble{}s  and write, \eg, 
$\mpscblacknot{u,v}$
instead of~%
$\mpscblacknot{\setsmall{u,v}}$.

\subsection{Definition of Induced Configurations 
  and Theorem Statement}

If $r$ is a non-source vertex with predecessors
$\prednode{r} = \setsmall{p, q}$,
we say that the
\introduceterm{axioms for $r$}
in $\pebcontrNT{\pebdeg}$ is the set
\begin{equation}
  \pebax{r}  =
  \setdescrcompact
  {\stdnot{p}_i \lor \stdnot{q}_j \lor \sourceclausenodisplay[l]{r}}
  {i,j \in \intnfirst{\pebdeg}}
\end{equation}
and if $r$ is a source, we define
$
\pebax{r}  =
\setcompact
{\sourceclause[i]{r}}
$.
For $V$ a set of vertices in~$G$, we let
$\pebax{V} = \setdescrcompact{\pebax{v}}{v \in V}$.
Note that with this notation, we have
$
\pebcontrNT[G]{\pebdeg}
=
\setdescrcompact
{\pebax{v}}
{v \in \vertices{G}}
$.
For brevity, we introduce the shorthand notation
\begin{equation}
  \label{eq:blacktruth}
  \blacktruth{V} = 
  \setdescrcompact{\sourceclausenodisplay[i]{v}}{v \in V}
\end{equation}
and
\begin{equation}
  \label{eq:somenodetrueclause}
  \somenodetrueclause{V} 
  = {\textstyle \Lor_{v \in V}} \sourceclausenodisplay[i]{v}
  \eqperiod  
\end{equation}
One can think of
$\blacktruth{V}$ 
as ``truth of all vertices in $V$'' and
$\somenodetrueclause{V}$
as ``truth of some vertex in $V$''.

We say that 
a set of clauses $\clsc$ implies a clause $\cld$
\introduceterm{minimally}
if
$\clsc \impl \cld$
but for all
$\clsc' \subsetneqq \clsc$
it holds that
$\clsc' \nimpl \cld$.
If 
$\clsc \impl \emptycl$ minimally, 
$\clsc$ is said to be 
\introduceterm{minimally unsatisfiable}.
We say that
$\clsc$ implies a clause~$\cld$
\introduceterm{maximally}
if
$\clsc \impl \cld$
but for all
$\cld' \subsetneqq \cld$
it holds that
$\clsc' \nimpl \cld'$.
To define our translation of clauses to \mpscfulltext{}s, we use
implications that are in a sense both minimal and maximal.
We remind the reader that the vertex set 
$\lpp{B}$ 
of \lpptext for white pebbles \wrt the chain~$B$
was defined in
Equation~\refeq{eq:legal-pebble-positions-definition}
on page~\pageref{eq:legal-pebble-positions-definition}.

\begin{definition}[Induced \mpscfulltext{}]
  \label{def:induced-multi-pebble-subconfiguration}
  \label{def:C-induced-multi-pebble-subconfiguration}
  Let $G$ be a \pebblingdag and $\clsc$
  a clause configuration derived from
  $\pebcontrNT[G]{\pebdeg}$.
  Then $\clsc$ induces the \mpscfulltext $\mpscnotstd$
  if there is a clause set
  $\clsc_{B} \subseteq \clsc$
  and a vertex set
  $S \subseteq G \setminus B$
  with
  $W = S \intersectionSP \lpp{B}$
  \st
\begin{subequations}
\begin{align}
  \label{eq:C-sharp-condition-one}
  \clsc_{B} \unionSP \blacktruth{S} &\impl \somenodetrueclause{B}
  \\
  \intertext{but for which it holds for all strict subsets
    $\clsc'_{B} \subsetneqq \clsc_{B}$,
    $S' \subsetneqq S$
    and
    $B' \subsetneqq B$
    that}
  \label{eq:C-sharp-condition-two}
  \clsc'_{B} \unionSP \blacktruth{S} &\nimpl \somenodetrueclause{B}
  \eqcomma
  \\
  \label{eq:C-sharp-condition-three}
  \clsc_{B} \unionSP \blacktruth{S'} &\nimpl \somenodetrueclause{B}
  \eqcomma \text{ and }
  \\
  \label{eq:C-sharp-condition-four}
  \clsc_{B} \unionSP \blacktruth{S} &\nimpl \somenodetrueclause{B'}
  \eqperiod
\end{align}
\end{subequations}
We write
$\mpinducedconf{\clsc}$
to denote the set of  all \mpscfulltext{}s
induced by $\clsc$.

To save space, 
when all conditions
\mbox{\refeq{eq:C-sharp-condition-one}--\refeq{eq:C-sharp-condition-four}}
hold, we write
\begin{equation}
  \clsc_{B} \unionSP \blacktruth{S} \sharpimpl \somenodetrueclause{B}
\end{equation}
and refer to this as
\introduceterm{\sharpimplsubst{}}
or say that
the clause set
$\clsc_{B} \unionSP \blacktruth{S}$ 
implies 
the clause
$\somenodetrueclause{B}$
\introduceterm{\sharpimpladv{}}.
Also, we say that the \sharpimplsubst
$
\clsc_{B} \unionSP \blacktruth{S} \sharpimpl \somenodetrueclause{B}
$
\introduceterm{witnesses}
the induced \mpscfulltext
$\mpscnotstd$.
\end{definition}

In the following, we will  use the definition of 
\sharpimplsubst~$\sharpimpl$
also for clauses $\somenodetrueclause{V}$
where the vertex set $V$ is not a chain.

Let us see that this definition agrees with the intuition 
presented in \refsec{sec:proof-overview-detailed}. 
\Anatomicmultipebbleadj black pebble on a single vertex~$v$
corresponds, as promised, to the fact that 
$\sourceclause[i]{v}$
is implied by the current set of clauses.  A black \multipebble on $V$
without supporting white pebbles is induced precisely when the
disjunction 
$
\somenodetrueclause{V} = 
\Lor_{v \in V} \sourceclause[i]{v}
$
of the corresponding clauses follow from the clauses in memory,
but no disjunction over a strict subset of vertices
$V' \subsetneqq V$
is implied. Finally, the supporting white pebbles just indicate that
if we indeed had the information corresponding to black pebbles on
these vertices, the clause corresponding to the  supported black
\multipebble could be derived.
Remember that our cost measure does not take into account
the size of \multipebble{}s.  This is natural
since we are interested
in clause space, and since large \multipebble{}s, in an intuitive sense,
corresponds to large (\ie wide) clauses rather than many clauses.


The main result of this section is as follows.

\begin{theorem}
  \label{th:C-translation-of-resolution-to-pebbling}
  \label{th:translation-of-resolution-to-multi-pebbling}
  Let
  $\proofstd = \setcompact{\clsc_0, \ldots, \clsc_{\stoptime}}$
  be a resolution derivation of
  $\targetclause[i]$
  from    
  $\pebcontrNT{\pebdeg}$
  for a \pebblingdag~$G$.
  Then the induced \mpctext{}s
  $
  \setcompact{
    \mpinducedconf{\clsc_0}, \ldots, \mpinducedconf{\clsc_{\stoptime}}
  }
  $
  form the ``backbone'' of 
  a \pebcomplete{}
  \multipebblingtext $\multipebbling$ of $G$
  in the sense that 
  \begin{itemize}
  \item
    $\mpinducedconf{\clsc_0} = \emptyset$,

  \item
    $\mpinducedconf{\clsc_{\stoptime}} 
    = \setsmall{\unconditionalblackmpscnot{z}}$,
    and

  \item
    for every $t \in \intnfirst{\stoptime}$,
    the transition
    $
    \clcfgtransition
    {\mpinducedconf{\clsc_{t-1}}}
    {\mpinducedconf{\clsc_{t}}}
    $
    can be accomplished 
    in accordance with the \multipebblingtext rules
    in cost
    $
    \Maxofexpr{
      \mpcost{\mpinducedconf{\clsc_{t-1}}},
      \mpcost{\mpinducedconf{\clsc_{t}}}}
    + \Ordosmall{1}
    $.   
  \end{itemize}
  In particular, to any   resolution derivation
  $\derivof
  {\proofstd}
  {\pebcontrNT{\pebdeg}}
  {\targetclause[i]}
  $
  we can associate a \pebcomplete{} \multipebblingtext
  $\multipebbling_{\proofstd}$  
  of~$G$ \st
  $
  \mpcost{\multipebbling_{\proofstd}}
  \leq
  \Maxofexpr[\clsc \in \proofstd]{\mpcost{\mpinducedconf{\clsc}}}
  + \Ordosmall{1}
  $.
\end{theorem}

We prove the theorem by forward induction over the
derivation~$\proofstd$. 
By the pebbling rules in 
\refdef{def:C-modified-pebble-game},
any \mpsctext $\mpscnotstd$
may be erased    freely at any time.
Consequently, we need not worry about \mpsctext{}s
disappearing during the transition from
$\clsc_{t-1}$ to~$\clsc_{t}$.
What we do need to check, though,  is that no
\mpsctext $\mpscnotstd$
appears inexplicably in
$\mpinducedconf{\clsc_{t}}$
as a result of a 
derivation step
$\clcfgtransition{\clsc_{t-1}}{\clsc_{t}}$,
but that we can always derive any
$
\mpscnotstd
\in
\mpinducedconf{\clsc_{t}}
\setminus
\mpinducedconf{\clsc_{t-1}}
$
from
$\mpinducedconf{\clsc_{t-1}}$
by the \multipebblingtext rules.
Also, when several pebbling moves are needed to get from
$\mpinducedconf{\clsc_{t}}$
to~%
$\mpinducedconf{\clsc_{t-1}}$,
we need to check that these intermediate moves do not affect the
pebbling cost by more than an additive constant.

The proof boils down to a case analysis of the different possibilities
for the derivation step
$\clcfgtransition{\clsc_{t-1}}{\clsc_{t}}$.
Since the analysis is quite lengthy, we divide it into subsections.
But first of all we need some technical lemmas.

\subsection{Some Technical Lemmas}

%

The next three lemmas are not hard, but will prove quite useful. 
We present the proofs for completeness.

\begin{lemma}
  \label{lem:C-pure-literals-stay}
  \label{lem:pure-literals-stay}
  Let
  $\clsc$ be a set of clauses
  and $\cld$ a clause
  such that
  $\clsc \impl \cld$
  minimally and 
  $\lita \in \lit{\clsc}$
  but
  $\olnot{\lita} \not\in \lit{\clsc}$.
  Then $\lita \in \lit{\cld}$.
\end{lemma}

\begin{proof}
  Suppose not. Let 
  $\clsc_1
  = \setdescrsmall{\clc \in \clsc}{ \lita \in \lit{\clc}}$
  and
  $\clsc_2 = \clsc \setminus \clsc_1$.
  Since
  $\clsc_2 \nimpl \cld$ there is a truth value assignment
  $\tvastd$ \st
  $\logevalstd{\clsc_2} = 1$
  and
  $\logevalstd{\cld} = 0$.
  Note that 
  $\logevalstd{\lita} = 0$,
  since otherwise
  $\logevalstd{\clsc_1} = 1$
  which would contradict
  $\clsc_1 \unionSP \clsc_2 = \clsc \impl \cld$.
  It follows that 
  $\olnot{\lita} \notin \lit{\cld}$.
  Flip $\lita$ to true and denote the resulting truth value assignment
  by~%
  ${\modtva{\tvastd}{\lita = 1}}$.
  By construction
  $\logeval{\clsc_1}{\modtva{\tvastd}{\lita = 1}} = 1$
  and $\clsc_2$ and $\cld$ are not affected since
  $
  \setsmall{\lita, \olnot{\lita}}
  \intersectionSP
  \bigl(
  \lit{\clsc_2} \unionSP \lit{\cld}
  \bigr)
  = \emptyset
  $,
  so
  $\logeval{\clsc}{\modtva{\tvastd}{\lita = 1}} = 1$
  and
  $\logeval{\cld}{\modtva{\tvastd}{\lita = 1}} = 0$.
  Contradiction.\qedhere%
\end{proof}


\begin{lemma}
  \label{lem:C-move-literal-right}
  Suppose that
  $\clc, \cld$ are clauses and
  $\clsc$ is a set of clauses.
  Then
  $\clsc \unionSP \setcompact{\clc} \impl \cld$
  \ifaoif
  $\clsc \impl \stdnot{\lita} \lor \cld$
  for all
  ${\lita \in \lit{\clc}}$.
\end{lemma}

\begin{proof}
  Assume that
  $\clsc \unionSP \setcompact{\clc} \impl \cld$
  and consider  any assignment $\tvastd$ \st
  $\logevalstd{\clsc} = 1$
  and
  \mbox{$\logevalstd{\cld} = 0$}
  (if there is no such $\tvastd$, then
  $\clsc \impl \cld \subseteq \olnot{\lita} \lor \cld$).
  Such an $\tvastd$ must set  $\clc$ to false, \ie   all
  $\olnot{\lita}$
  to true.
  Conversely, if 
  $\clsc \impl \stdnot{\lita} \lor \cld$
  for all
  ${\lita \in \lit{\clc}}$
  and  
  $\tvastd$ is such that 
  $\logevalstd{\clsc}
  =
  \logevalstd{\clc}
  =
  1
  $,
  it must hold that
  $\logevalstd{\cld} = 1$,
  since otherwise
  $\logevalstd{\olnot{\lita} \lor \cld} = 0$
  for some literal
  $\lita \in \lit{\clc}$
  satisfied by~$\tvastd$.
\end{proof}

\begin{lemma}
  \label{lem:C-no-complements-in-minimal-implication}
  Suppose that 
  $\clsc \impl \cld$ minimally.
  Then no literal from $\cld$ can occur negated in $\clsc$, 
  \ie
  it holds that
  $
  \setdescrsmall{\olnot{\lita}}{\lita \in \lit{\cld}}
  \intersectionSP
  \lit{\clsc}
  = \emptyset
  $.
\end{lemma}

\begin{proof}
  Suppose not. Let 
  $\clsc_1
  = \setdescrsmall
  {\clc \in \clsc}
  {\exists \lita \text{ such that }
    \olnot{\lita} \in \lit{\clc} 
    \text{ and }
    \lita \in \lit{\cld}} 
  $
  and
  $\clsc_2 = \clsc \setminus \clsc_1$.
  Since
  $\clsc_2 \nimpl \cld$ there is an $\tvastd$ \st
  $\logevalstd{\clsc_2} = 1$
  and
  $\logevalstd{\cld} = 0$.
  But then
  $\logevalstd{\clsc_1} = 1$,
  since every
  $\clc \in \clsc_1$
  contains a negated literal 
  $\olnot{\lita}$
  from~$\cld$, and these literals are all
  set to true by~$\tvastd$. Contradiction.
\end{proof}

We also need the following key technical lemma
connecting implication with inflation moves.

\begin{lemma}
  \label{lem:C-implication-implies-derivability-by-inflation}
  Let
  $\clsc$ be a clause set derived from
  $\pebcontrNT[G]{\pebdeg}$. Suppose that
  $B$ is  a chain and that
  $S \subseteq G \setminus B$
  is a vertex set \st
  $\clsc \unionSP \blacktruth{S} \impl \somenodetrueclause{B}$
  and let 
  $W = S \intersectionSP \lpp{B}$.
  Then the \mpscfulltext   $\mpscnotstd$   is derivable by inflation
  from some   $\mpscnotprime \in \mpinducedconf{\clsc}$.
%
%
\end{lemma}

\begin{proof}
  Pick
  $\clsc' \subseteq \clsc$,
  $S' \subseteq S$
  and
  $B' \subseteq B$
  minimal \st
  $\clsc' \unionSP \blacktruth{S'} \impl \somenodetrueclause{B'}$.
  Then
  $\clsc' \unionSP \blacktruth{S'} \sharpimpl
  \somenodetrueclause{B'}$
  by definition.
  Note, furthermore, that 
  $B' \neq \emptyset$ 
  since the clause set on the left-hand side must be
  non-contradictory.   Also,
  $\clsc' \neq \emptyset$
  since
  $
  B' \intersectionSP S'
  \subseteq
  B \intersectionSP S
  = \emptyset
  $,
  so by
  \reflem{lem:C-pure-literals-stay}
  it cannot be that
  $
  \blacktruth{S'} \impl \somenodetrueclause{B'}
  $.
  This means that
  $\clsc$
  induces
  $\mpscnotprime$
  for
  $W' = S' \intersectionSP \lpp{B'}$.
  We claim that 
  $\mpscnotprime$
  can be inflated to
  $\mpscnotstd$, from which the lemma follows.

  To verify this claim, note that   first two conditions
  $B' \subseteq B$
  and
  $B \intersectionSP W' \subseteq B \intersectionSP S = \emptyset$
  for inflation moves in \refdef{def:multi-pebble-game} clearly hold
  by construction. As to the third condition, we get
  \begin{equation*}
    W' \intersectionSP \lpp{B}
    =
    \bigl( S' \intersectionSP \lpp{B'} \bigr) \intersectionSP \lpp{B}
    \subseteq
    S \intersectionSP \lpp{B}
    =
    W
  \end{equation*}
  which proves the claim.
\end{proof}

We now start the case analysis 
in the proof of  \refth{th:C-translation-of-resolution-to-pebbling}
for the different possible derivation
steps in a resolution derivation.

\subsection{Erasure}
\label{sec:derivations-induce-case-erasure}

Suppose that
$\clsc_{t} = \clsc_{t-1} \setminus \setsmall{\clc}$
for
$\clc \in \clsc_{t-1}$.
%
%
It is easy to see that the only possible outcome of erasing
clauses is that  \mpscfulltext{}s disappear.
We note for future reference that this implies that the
\blobpebblingtext cost decreases monotonically when going from
$\mpinducedconf{\clsc_{t-1}}$
to~%
$\mpinducedconf{\clsc_{t}}$.

\subsection{Inference}
\label{sec:derivations-induce-case-inference}


Suppose that
$\clsc_{t} = \clsc_{t-1} \unionSP \setsmall{\clc}$
for some clause $\clc$ derived from $\clsc_{t-1}$.
No \mpscfulltext{}s can disappear at an inference move since
$\clsc_{t-1} \subseteq \clsc_{t}$.
Suppose that
$\mpscnotstd$
is a new \mpsctext at time~$t$ arising from
$\clsc_{B} \subseteq \clsc_{t-1}$
and
$S \subseteq G \setminus B$
\st
$W = S \intersectionSP \lppstd$
and
$
\clsc_{B} \unionSP \setsmall{\clc} \unionSP \blacktruth{S} 
\sharpimpl
\somenodetrueclause{B}
$.
Since $\clc$ is derived from $\clsc_{t-1}$, we have
$\clsc_{t-1} \impl \clc$. 
Thus it holds that
$\clsc_{t-1}  \unionSP \blacktruth{S} \impl \somenodetrueclause{B}$
and
\reflem{lem:C-implication-implies-derivability-by-inflation}
tells us that
$\mpscnotstd$
is derivable by inflation from~$\mpinducedconf{\clsc_{t-1}}$.

Since no \mpsctext disappears, the pebbling cost
increases monotonically when going from
$\mpinducedconf{\clsc_{t-1}}$
to
$\mpinducedconf{\clsc_{t}}$
for an inference step, which is again noted
for future reference.

\subsection{Axiom Download}
\label{sec:derivations-induce-case-axiom-download}


This is the interesting case.
Assume that a new \mpscfulltext
$\mpscnotstd$
is induced at time $t$ as the result of a download of an axiom
$\clc \in \pebax{r}$.
Then $\clc$ must be one of the clauses inducing the \mpsctext, and we
get that there are
$\clsc_B \subseteq \clsc_{t-1}$
and
$S \subseteq G \setminus B$
with
$W = S \intersectionSP \lpp{B}$
\st
\begin{equation}
  \label{eq:C-time-t-plus-one-sharp-implication}
  \clsc_{B} \unionSP \setsmall{\clc} \unionSP \blacktruth{S} 
  \sharpimpl
  \somenodetrueclause{B}
  \eqperiod
\end{equation}
Our intuition is that download of an axiom clause
$\clc \in \pebax{r}$
in the resolution derivation should correspond to an introduction of 
$\intrompscnot{r}$
in the induced \multipebblingtext. We want to prove that any other
\mpscfulltext $\mpscnotstd$ in $\mpinducedconf{\clsc_{t}}$
is derivable by the pebbling rules from
$\mpinducedconf{\clsc_{t-1}} \unionSP \singset{\intrompscnot{r}}$.
Also, we need to prove that the pebbling moves needed to go from
$\mpinducedconf{\clsc_{t-1}}$
to
$\mpinducedconf{\clsc_{t}}$
do not increase the \blobpebblingtext cost by more than an additive
constant compared to 
$
\Maxofexpr{
\mpcost{\mpinducedconf{\clsc_{t-1}}},
\mpcost{\mpinducedconf{\clsc_{t}}}
}
=
\mpcost{\mpinducedconf{\clsc_{t}}}
$.

We do the proof  
by a case analysis over $r$ depending on where in the
graph this vertex is located in relation to $B$. 
To simplify the proofs for the different cases, we first show a general
technical lemma about  pebble induction at axiom download.

\begin{lemma}
  \label{lem:C-lemma-axiom-download-pebble-induction}
  Suppose that
  $\clsc_{t} = \clsc_{t-1} \unionSP \clc$
  for an axiom
  $\clc \in \pebax{r}$
  and that
  $\mpscnotstd$
  is a new \mpscfulltext
  induced at time $t$
  as witnessed by~%
  \refeq{eq:C-time-t-plus-one-sharp-implication}.
  %
  Then it holds that:
  \begin{enumerate}
  \item
  \label{part:part-one-lemma-axiom-download-pebble-induction}
    $r \notin S$.
    
  \item
  \label{part:part-two-lemma-axiom-download-pebble-induction}
    $\prednode{r} \intersectionSP B = \emptyset$.

    \item
      \label{part:part-three-lemma-axiom-download-pebble-induction}
      If $r \notin B$,
      then
      $\clsc_{t-1}$ induces
      $\mpscnot{B}{W \unionSP (\setsmall{r} \intersectionSP \lpp{B})}$
      if $r$ is a source,  
      and       otherwise this \mpsctext
      can be derived 
      from $\mpinducedconf{\clsc_{t-1}}$
      by inflation.

    \item
      \label{part:part-four-lemma-axiom-download-pebble-induction}
      If $r$ is a non-source vertex and
      $v \in \prednode{r}$ 
      is such that
      $v \in \lpp{B} \setminus S$,
      then
      we can derive
      $\mpscnot{B \unionSP v}{S \intersectionSP \lpp{B \unionSP v}}$
      from $\mpinducedconf{\clsc_{t-1}}$
      by inflation.
  \end{enumerate}
\end{lemma}

\begin{proof}
  Suppose that
  $
  \mpscnotstd \in 
  \mpinducedconf{\clsc_{t}} \setminus \mpinducedconf{\clsc_{t-1}} 
  $.
  For \refpart{part:part-one-lemma-axiom-download-pebble-induction},
  noting that 
  $\blacktruth{r} \impl \clc$
  for
  $\clc \in \pebax{r}$
  we see that
  $r \notin S$,
  as   otherwise the implication
  \refeq{eq:C-time-t-plus-one-sharp-implication}
  cannot be \sharpimpladj   since $\clc$ can be omitted.

  If $r$ is a source
  \refpart{part:part-two-lemma-axiom-download-pebble-induction}
  is trivial, so suppose 
  $\prednode{r} = \setsmall{p,q}$
  and
  $\clc = \pqrstd$.
  Then it follows from
  \reflem{lem:C-no-complements-in-minimal-implication}
  that
  $\setsmall{p,q} \intersectionSP B = \emptyset$.
    
  For 
  \refpart{part:part-three-lemma-axiom-download-pebble-induction},
  if $r$ is a source, we have
  $\clc = \sourceclause[i]{r}$
  and
  \refeq{eq:C-time-t-plus-one-sharp-implication}
  becomes
  \begin{equation}
    \label{eq:C-time-t-source-sharp-implication}
    \clsc_{B} \unionSP \blacktruth{S \unionSP r} 
    \sharpimpl
    \somenodetrueclause{B}
  \end{equation}
  for 
  $S \unionSP r \subseteq G \setminus B$,
  which shows that
  $\clsc_{t-1}$
  induces
  \begin{equation}
    \begin{split}
      \mpscnot{B}{(S \unionSP r) \intersectionSP \lpp{B}}
    &=
    \mpscnot{B}{(S \intersectionSP \lpp{B}) \unionSP (r \intersectionSP \lpp{B})}
    \\
    &=
    \mpscnot{B}{(W \unionSP (r \intersectionSP \lpp{B})}
    \eqperiod
    \end{split}
  \end{equation}
  If $r$ is a non-source we do not get a \sharpimplsubst but still have 
  \begin{equation}
    \label{eq:C-time-t-non-source-implication}
    \clsc_{B} \unionSP \blacktruth{S \unionSP r} 
    \impl
    \somenodetrueclause{B}
  \end{equation}
  and
  \reflem{lem:C-implication-implies-derivability-by-inflation}
  yields that
  $
  \mpscnot{B}{(S \unionSP r) \intersectionSP \lpp{B}}
  =
  \mpscnot{B}{W \unionSP (r \intersectionSP \lpp{B})}
  $
  is derivable by inflation from~$\mpinducedconf{\clsc_{t-1}}$.
    
If $v \in \prednode{r}$ in
\refpart{part:part-four-lemma-axiom-download-pebble-induction},
the downloaded axiom can be written on the form
$\clc = \clc' \lor \olnot{v}_i$.
Applying
\reflem{lem:C-move-literal-right}
on
\refeq{eq:C-time-t-plus-one-sharp-implication}
we get
\begin{equation}
  \label{eq:C-time-t-B-union-v-implication}
  \clsc_{B} \unionSP \blacktruth{S} 
  \impl
  \somenodetrueclause{B} \lor v_i
  \subseteq
  \somenodetrueclause{B \unionSP v}
  \eqperiod
\end{equation}
By assumption,  we have that
$B \unionSP v$ is a chain and
that
$S \subseteq G \setminus (B \unionSP v)$, 
so
\reflem{lem:C-implication-implies-derivability-by-inflation}
says that
$\mpscnot{B \unionSP v}{S \intersectionSP \lpp{B \unionSP v}}$
is derivable from $\mpinducedconf{\clsc_{t-1}}$ 
by inflation.
\end{proof}

What we get from
\reflem{lem:C-lemma-axiom-download-pebble-induction}
is not in itself sufficient to derive the new \mpscfulltext
$\mpscnotstd$ 
in the \multipebblegame, but the lemma provides 
\mpsctext{}s that will be used as building blocks in the derivations of
$\mpscnotstd$
below.

Now we are ready for the case analysis over the vertex $r$
for the downloaded axiom clause 
$\clc \in \pebax{r}$.
Recall that the assumption is that 
there exists a \mpscfulltext  
$
\mpscnotstd \in
\mpinducedconf{\clsc_{t}} \setminus \mpinducedconf{\clsc_{t-1}}
$
induced through~%
\refeq{eq:C-time-t-plus-one-sharp-implication}
for
$\clsc_B \subseteq \clsc_{t-1}$
and
$S \subseteq G \setminus B$
with
$W = S \intersectionSP \lpp{B}$.
Remember also that we want to explain all new \mpsctext{}s in
$\mpinducedconf{\clsc_{t}} \setminus \mpinducedconf{\clsc_{t-1}}$
in terms of pebbling moves from
$\mpinducedconf{\clsc_{t}} \unionSP \setsmall{\intrompscnot{r}}$.
As illustrated in
\reffig{fig:casesforvertex},
the cases for $r$ are:
\begin{enumerate}
\item 
  $r \in G \setminus 
  \bigl( 
  \belowvertices{\minelement{b}}
  \unionSP
  \unionpathsinvertex{B}
  \bigr)
  $
  for
  $\minelement{b} = \minvertex{B}$,
  
\item 
  $r \in \unionpathsinvertex{B} \setminus B$,
  
\item 
  $r \in  B \setminus \setsmall{\minelement{b}}$
  for
  $\minelement{b} = \minvertex{B}$,
  
\item
  $r = {\minvertex{B}}$, 
  and

\item 
$r \in \belowverticesNR{\minelement{b}}$
  for
  $\minelement{b} = \minvertex{B}$.
\end{enumerate}
%

\begin{figure}[tp]
  \centering
  \includegraphics{casesforvertex.1}%
  \caption{%
    Cases for vertex $r$ 
    \wrt new black \multipebble $B$
    at download of axiom
    $\clc \in\pebax{r}$.}
  \label{fig:casesforvertex}
\end{figure}

%
%
\renewcommand{\theenumi}{\Roman{enumi}}

\subsubsection{Case 1: 
$r \in G \setminus 
  \bigl( 
  \belowvertices{\minelement{b}} \unionSP \unionpathsinvertex{B}
  \bigr)
  $
  for
  $\minelement{b} = \minvertex{B}$
}
\label{sec:derviations-induce-case-one}

If
$r \in G \setminus 
\bigl( 
\belowvertices{\minelement{b}} \unionSP \unionpathsinvertex{B}
\bigr)
$,
this means that the vertex $r$ is outside the set of vertices covered
by \sourcepath{}s via $B$ to $\topvertex{B}$. In other words,
$r \notin \lpp{B} \unionSP B$
and
\refpart{part:part-three-lemma-axiom-download-pebble-induction}
of
\reflem{lem:C-lemma-axiom-download-pebble-induction}
yields that
$
\mpscnotcompact{B}{W \unionSP (r \intersectionSP \lpp{B})}
=
\mpscnotstd
$
is derivable from
$\mpinducedctminusone$
by inflation.
Note that we need no intermediate \mpsctext{}s in this case.

\subsubsection{Case 2: 
$r \in \unionpathsinvertex{B} \setminus B$}
\label{sec:derviations-induce-case-two}

This is the first more challenging case, and we do it in some detail
to show how the reasoning goes. The proofs for the rest of the cases
are analogous and will be presented in slightly more condensed form.

The condition
$r \in \unionpathsinvertex{B} \setminus B$
says that the vertex $r$ is located on some path from
$\bottomvertex{B}$ via $B$ to $\topvertex{B}$
strictly above the bottom vertex
$b = \bottomvertex{B}$.
In particular, this means that $r$ cannot be a source vertex.
Let
$\predrequalpq$
and denote the downloaded axiom clause
$\clc = \pqrstd$.

\Refpart{part:part-three-lemma-axiom-download-pebble-induction}
of
\reflem{lem:C-lemma-axiom-download-pebble-induction}
says that we can derive the \mpscfulltext
\begin{equation}
  \label{eq:C-axdow-a}
  \mpscnot{B}{W \unionSP (r \intersectionSP \lpp{B})}
  =
  \mpscnot{B}{\Wur}
\end{equation}
by inflation from
$\mpinducedctminusone$,
where the equality holds since
$r \in \unionpathsinvertex{B} \setminus B \subseteq \lpp{B}$
by
\refdef{def:C-multi-pebble-configuration}.
Also, since $r$ is on some path above $b$, at least one of the
predecessors of $r$ must be located on some path from $b$ as well.
That is, translating what was just said into our notation we have that 
the fact that
$r \in \unionpathsinvertex{B}
\intersectionSP \aboveverticesNR{b}$
implies that either
$p \in \unionpathsinvertex{B}$
or
$q \in \unionpathsinvertex{B}$
or both.
By symmetry, we get two cases:
$p \in \unionpathsinvertex{B}, \ q \notin \unionpathsinvertex{B}$
and
$\setsmall{p,q} \subseteq \unionpathsinvertex{B}$.
Let us look at them in order.

\begin{enumerate}
\item
  $p \in \unionpathsinvertex{B}, \ q \notin \unionpathsinvertex{B}$:
  We make a subcase analysis depending on whether
  $p \in \BuW$
  or not.
  Recall from  
  \refpart{part:part-two-lemma-axiom-download-pebble-induction}
  of
  \reflem{lem:C-lemma-axiom-download-pebble-induction}
  that
  $p \notin B$. The two remaining cases are
  $p \in W$
  and
  $p \notin \BuW$.

  \begin{enumerate}
  \item
    $p \in W$:
    Let $v$ be the uppermost vertex in $B$ below $p$, 
    or in formal notation
    \begin{equation}
      \label{eq:C-top-G-p-intersection-B}
      v = \maxvertex{\belowvertices{p} \intersectionSP B}
      \eqperiod
    \end{equation}
    Such a vertex $v$ must exist since
    $p \in \unionpathsinvertex{B} \setminus B$.
    Since $p$ is above $v$ and is a predecessor of~$r$, it lies on
    some path from $v$ to~$r$, \ie
    $p \in \unionpathsinvertex{\vrset} \setminus \vrset$.
    For the sibling $q$ we have
    $q \notin \unionpathsinvertex{\vrset}$. 
    This is so     since
    $q \notin \unionpathsinvertex{B}$
    and for any path
    $P \in \pathsinvertex{\setsmall{v,r}}$ 
    it holds that
    $P \subseteq \unionpathsinvertex{B}$
    since there is nothing inbetween $v$ and $r$ in $B$,
    \ie 
    $\bigl( \unionpathsinvertex{\vrset} \setminus \vrset \bigr)
    \intersectionSP B = \emptyset$.
    Also, 
    $q \notin \belowverticesNR{p} \supseteq \belowverticesNR{v}$
    because of
    the \siblingnonreachabiblityproperty.
    Hence, it must hold that 
    $q \notin \lpp{\vrset}$. 

    We can use this information to make \blobpebblingtext moves
    resulting in $\mpscnotstd$ as follows.
    First introduce $\mpscrpq$ and inflate this \mpsctext to 
    \begin{equation}
      \label{eq:C-axdow-b}
      \mpscnot{\vrsetNP}{\pqset \intersectionSP \lpp{\vrset}}
      =
      \mpscnot{\vrsetNP}{p}
      \eqperiod
    \end{equation}
    Then derive the \mpsctext
    $\mpscnot{B}{\Wur}$
    in \refeq{eq:C-axdow-a} 
    by inflation from
    $\mpinducedctminusone$.
    Finally, merge the two \mpsctext{}s
    \refeq{eq:C-axdow-a}
    and
    \refeq{eq:C-axdow-b}.
    The result of this  merger move is
    $\mpscnot{\Buv}{\Wup} = \mpscnotstd$.

  \item
    $p \notin \BuW$:
    Note that
    $p \in \pathsBminusB$
    by assumption. 
    Also, it must hold that
    $p \notin S$
    since otherwise we would get the contradiction
    $
    p \in S \intersectionSP (\pathsBminusB)
    \subseteq S \intersectionSP \lpp{B} = W
    $.
    Thus,
    $p \in \lpp{B} \setminus S$
    and 
    \refpart{part:part-four-lemma-axiom-download-pebble-induction}
    of
    \reflem{lem:C-lemma-axiom-download-pebble-induction}
    yields that we can derive the \mpscfulltext
    \begin{equation}
      \label{eq:C-axdow-c}
      \mpscnot{\Bup}{W_p} \text{\ \ for \ }
      W_p \subseteq W
    \end{equation}
    by inflation from $\mpinducedctminusone$,    where
    $
    W_p = S \intersectionSP \lpp{\Bup}
    \subseteq
    S \intersectionSP \lpp{B} = W
    $
    since
    $\lpp{\Bup} \subseteq \lpp{B}$
    if $p \in \unionpathsinvertex{B}$.
    (This last claim is easily verified directly from
    \refdef{def:C-multi-pebble-configuration}.)

    With
    $v = \maxvertex{\belowvertices{p} \intersectionSP B}$
    as in
    \refeq{eq:C-top-G-p-intersection-B},
    introduce
    $\mpscrpq$
    and inflate to
    $\mpscnot{\vrsetNP}{p}$
    as in
    \refeq{eq:C-axdow-b}.
    Merging the \mpsctext{}s 
    \refeq{eq:C-axdow-b}
    and
    \refeq{eq:C-axdow-c}
    yields
    \begin{equation}
      \label{eq:C-axdow-d}
      \mpscnot{B \unionSP \vrset}{W_p}
      =
      \mpscnot{\Bur}{W_p}
    \end{equation}
    and a second merger of the resulting \mpsctext
    \refeq{eq:C-axdow-d}
    with the \mpsctext in
    \refeq{eq:C-axdow-a}
    produces
    $
    \mpscnot{B}{W \unionSP W_p} = \mpscnotstd
    $.
  \end{enumerate}

  This finishes the case
  $p \in \unionpathsinvertex{B}$,
  $q \notin \unionpathsinvertex{B}$.

\item
  \label{case:two-two} 
  $\setsmall{p,q} \subseteq \unionpathsinvertex{B}$:
  By
  \refpart{part:part-two-lemma-axiom-download-pebble-induction}
  of
  \reflem{lem:C-lemma-axiom-download-pebble-induction}
  $\setsmall{p,q} \intersectionSP B = \emptyset$,
  so
  $\setsmall{p,q} \subseteq \pathsBminusB$.
  By symmetry, we have the following subcases for
  $p$ and $q$ \wrt membership in $B$ and $W$.
  \begin{enumerate}
    \item
      $\pqset \subseteq W$,
    \item
      $p \in W, \  q \notin W$,
    \item
      $\pqset \intersectionSP (\BuW) = \emptyset$.      
  \end{enumerate}
  We analyze these subcases one by one.
  
  \begin{enumerate}
    \item
      $\pqset \subseteq W$:
      This is easy.
      Just introduce
      $\mpscrpq$
      and merge 
      this \mpsctext
      with the \mpsctext
      \refeq{eq:C-axdow-a}
      to get
      $
      \mpscnot{B}{W \unionSP \pqset} = 
      \mpscnotstd
      $.

    \item
      \label{case:case-two-two-b}
      $p \in W, \  q \notin W$:
      In this case it must hold that
      $q \notin S$
      since otherwise we would have
      $
      q \in 
      S \intersectionSP (\pathsBminusB)
      \subseteq
      S \intersectionSP \lpp{B} = W
      $
      contradicting the assumption. Thus
      $q \in (\pathsBminusB) \setminus S
      \subseteq \lpp{B} \setminus S$
      and
      \refpart{part:part-four-lemma-axiom-download-pebble-induction}
      of
      \reflem{lem:C-lemma-axiom-download-pebble-induction}
      allows us to derive
      \begin{equation}
        \label{eq:C-axdow-e}
        \mpscnot{\Buq}{W_q} 
        \text{\ \ for \ }
        W_q \subseteq W
      \end{equation}
      by inflation from $\mpinducedctminusone$.
      Here we have
      $
      W_q = S \intersectionSP \lpp{\Buq}
      \subseteq
      S \intersectionSP \lpp{B} = W
      $
      since
      $\lpp{\Buq} \subseteq \lpp{B}$
      when $q \in \unionpathsinvertex{B}$.
      
      Introduce
      $\mpscrpq$ and merge with the \mpsctext
      \refeq{eq:C-axdow-e}
      to get
      \begin{equation}
        \label{eq:C-axdow-f}
        \mpscnot{\Bur}{\Wqup} 
      \end{equation}
      and then merge
      \refeq{eq:C-axdow-f}
      with
      $\mpscnot{B}{\Wur}$
      from
      \refeq{eq:C-axdow-a}
      to get
      $
      \mpscnot{B}{W \unionSP \Wqup}
      = \mpscnotstd
      $.

    \item
      $\pqset \intersectionSP \BuW = \emptyset$:
      Just as for the vertex $q$ in case 
      \refcase{case:case-two-two-b},
      here it holds for both $p$ and $q$ that
      $\pqset \subseteq \lpp{B} \setminus S$.
      \Refpart{part:part-four-lemma-axiom-download-pebble-induction}
      of
      \reflem{lem:C-lemma-axiom-download-pebble-induction}
      yields \mpsctext{}s
      $\mpscnot{\Bup}{W_p}$ for 
      $W_p\subseteq W$
      as in
      \refeq{eq:C-axdow-c}
      and
      $\mpscnot{\Buq}{W_q}$ for
      $W_q \subseteq W$
      as in
      \refeq{eq:C-axdow-e}
      derived by inflation from $\mpinducedctminusone$.
      
      Introduce
      $\mpscrpq$
      and merge with
      \refeq{eq:C-axdow-c}
      on $p$ to get
      \begin{equation}
        \label{eq:C-axdow-g}
        \mpscnot{\Bur}{\Wpuq}
      \end{equation}
      and then merge
      \refeq{eq:C-axdow-g}
      with
      \refeq{eq:C-axdow-e}
      on $q$ resulting in
      \begin{equation}
        \label{eq:C-axdow-h}
        \mpscnot{\Bur}{\WpuWq}
        \eqperiod
      \end{equation}
      Finally, merge
      \refeq{eq:C-axdow-h}
      with
      \refeq{eq:C-axdow-a}
      on $r$ to get
      $
      \mpscnot{B}{W \unionSP \WpuWq} = \mpscnotstd
      $.

  \end{enumerate}
\end{enumerate}

This concludes the case 
$r \in \unionpathsinvertex{B} \setminus B$.
We can see that in all subcases, the new \mpscfulltext
$\mpscnotstd$
is derivable from
$\mpinducedctminusone \unionSP \intrompscnot{r}$
by inflation moves followed by mergers on some subset of
$\setsmall{p,q,r}$.

Let us analyze the cost of deriving 
$\mpscnotstd$.
We want to bound the cost of the intermediate \mpsctext{}s 
that are used in the
transition from
$\mpinducedctminusone$
to
$\mpinducedct$
but are not present in~%
$\mpinducedct$.
We first note that for the \mpsctext{}s
$\mpscnot{B}{\Wur}$,
$\mpscnot{\Bup}{W_p}$,
$\mpscnot{\Buq}{W_q}$
and
$\mpscnot{\Bur}{W'}$ for various~$W' \subseteq W$,
the \chargeableverticestext are all subsets of the
\chargeableverticestext of the final \mpsctext~%
$\mpscnotstd$.
This is so since
$b = \bottomvertex{B}$
is the bottom vertex in all these black \multipebble{}s,
and all \chargeabletext white vertices are contained in
$
W \intersectionSP \belowvertices[G]{b}
$.
The \mpsctext{}s 
$\mpscrpq$
and
$\mpscnot{\vrsetNP}{p}$
for
$v = \maxvertex{\belowvertices{p} \intersectionSP B}$
can incur an extra cost, however, but this cost is clearly bounded by
$\setsizesmall{\setsmall{p,q,r,v}} = 4$.

\subsubsection{Case 3: 
  $r \in  B \setminus \setsmall{\minelement{b}}$
  for
  $\minelement{b} = \minvertex{B}$%
}
\label{sec:derviations-induce-case-three}

First we note that in this case, we can no longer use
\refpart{part:part-three-lemma-axiom-download-pebble-induction}
of
\reflem{lem:C-lemma-axiom-download-pebble-induction}
to derive the \mpscfulltext
$\mpscnot{B}{\Wur}$
of~\refeq{eq:C-axdow-a}.
The vertex $r$ cannot be added to the support $S$ since it is
contained in~$B$.
Also, 
we note that $r$ cannot be a source since
it is above the bottom vertex~$b$.
As usual, let us write
$\predrequalpq$.

Observe that just as in case~2 
(\refsec{sec:derviations-induce-case-two})
we must have either
$p \in \unionpathsinvertex{B}$
or
$q \in \unionpathsinvertex{B}$
or both. By symmetry we get the same two cases
for membership of $p$ and $q$ in
$\unionpathsinvertex{B}$, namely
$p \in \unionpathsinvertex{B}, \ q \notin \unionpathsinvertex{B}$
and
$\setsmall{p,q} \subseteq \unionpathsinvertex{B}$.

\begin{enumerate}
\item
  $p \in \unionpathsinvertex{B}, \ q \notin \unionpathsinvertex{B}$:
  As before,
  $p \notin B$
  by
  \refpart{part:part-two-lemma-axiom-download-pebble-induction}
  of
  \reflem{lem:C-lemma-axiom-download-pebble-induction}.
  We make a subcase analysis depending on whether
  $p \in W$
  or
  $p \notin \BuW$.

  As in
  \refeq{eq:C-top-G-p-intersection-B}
  we let
  $v = \maxvertex{\belowvertices{p} \intersectionSP B}$
  and note that
  $p \in \unionpathsinvertex{\vrset} \setminus \vrset$.
  For $q$ we have
  $q \notin \unionpathsinvertex{\vrset}$
  since
  $q \notin \unionpathsinvertex{B}$
  but
  $\vrset \subseteq \unionpathsinvertex{B}$
  and there is nothing inbetween $v$ and $r$ in $B$.
  Also,
  $q \notin \belowverticesNR{p} \supseteq \belowverticesNR{v}$
  because of the \siblingnonreachabiblityproperty.
  Hence, it holds that 
  $q \notin \lpp{\vrset}$.
  
  \begin{enumerate}
  \item
    $p \in W$:
    Introduce
    $\mpscrpq$,
    inflate
    $\mpscrpq$
    to
    $
    \mpscnot{\vrsetNP}{\pqset \intersectionSP \lpp{\vrset}}
    =
    \mpscnot{\vrsetNP}{p}
    $
    as in~\refeq{eq:C-axdow-b}
    and continue the inflation to
    $
    \mpscnot{B \unionSP \vrset}{\Wup}
    =
    \mpscnotstd
    $.
    
  \item
    $p \notin \BuW$:
    Just as in case~2, 
    $p \notin W$
    implies
    $p \notin S$,
    so
    $p \in \lpp{B} \setminus S$
    and we can use
    \refpart{part:part-four-lemma-axiom-download-pebble-induction}
    of
    \reflem{lem:C-lemma-axiom-download-pebble-induction}
    to derive
    $\mpscnot{\Bup}{W_p}$
    for
    $W_p \subseteq W$
    as in 
    \refeq{eq:C-axdow-c}.
    Introduce
    $\mpscrpq$,
    inflate to
    $
    \mpscnot{\vrsetNP}{p}
    $
    as in 
    \refeq{eq:C-axdow-b}
    and merge
    \refeq{eq:C-axdow-b}
    and
    \refeq{eq:C-axdow-c}
    on $p$ resulting in
    $
    \mpscnot{B \unionSP \vrset}{W_p}
    = \mpscnot{B}{W_p}
    $,
    which can be inflated to
    $\mpscnotstd$.

  \end{enumerate}

\item
  $\pqset \subseteq \unionpathsinvertex{B}$:
  We have the same possibilities to consider for containment of
  $p$ and $q$ in $\BuW$ as in
  case 2(\ref{case:two-two}) 
  on page~\pageref{case:two-two}.

  \begin{enumerate}
    \item
      $\pqset \subseteq W$:
      This is immediate.
      Introduce 
      the \mpsctext
      $\mpscrpq$ and inflate to
      $
      \mpscnot{\Bur}{W \unionSP \pqset} = \mpscnotstd
      $.

    \item
      \label{case:case-three-two-b}
      $p \in W, \ q \notin \BuW$:
      Apply
      \refpart{part:part-four-lemma-axiom-download-pebble-induction}
      of
      \reflem{lem:C-lemma-axiom-download-pebble-induction}
      to derive
      $\mpscnot{\Buq}{W_q}$
      for
      $W_q \subseteq W$
      by inflation from $\mpinducedctminusone$.
      Then introduce      
      $\mpscrpq$
      and merge on $q$ to get
      the \mpsctext
      $
      \mpscnot{\Bur}{\Wqup} =       \mpscnot{B}{\Wqup}
      $,
      which can be inflated further to
      $
      \mpscnot{B}{\Wqup \unionSP W}
      = \mpscnotstd
      $.

    \item
      $\pqset \intersectionSP(\BuW) = \emptyset$:
      In the same way as in 
      \refcase{case:case-three-two-b},
      derive
      the \mpsctext{}s
      $\mpscnot{\Bup}{W_p}$
      and
      $\mpscnot{\Buq}{W_q}$
      with
      $\WpuWq \subseteq W$
      from $\mpinducedctminusone$ by inflation.
      Introduce
      $\mpscrpq$
      and merge twice, first on $p$ and then on $q$, to get
      $\mpscnot{B}{\WpuWq}$,
      which can be inflated to 
      $\mpscnotstd$.
  \end{enumerate}
\end{enumerate}

This concludes the case 
$r \in  B \setminus \setsmall{\minelement{b}}$.
We see that in all subcases the new \mpscfulltext
$\mpscnotstd$
is derivable from
$\mpinducedctminusone \unionSP \intrompscnot{r}$
by inflation moves followed by mergers on some subset of
$\setsmall{p,q}$, 
possibly followed by one more inflation move.

As in the previous case, the bottom vertex in all of the black
\multipebble{}s
$\mpscblacknot{\Bup}$,
$\mpscblacknot{\Buq}$
and
$\mpscblacknot{\Bur}$
is
$b = \bottomvertex{B}$,
and the corresponding \chargeabletext white pebbles are subsets of
those of~$W$.
The extra cost caused by
the \mpsctext{}s 
$\mpscrpq$
and
$\mpscnot{\vrsetNP}{p}$
is at most~$4$.

%
%
\renewcommand{\theenumi}{\alph{enumi}}
\renewcommand{\labelenumi}{(\theenumi)}

\subsubsection{Case 4: 
$r = {\minvertex{B}}$}
\label{sec:derviations-induce-case-four}

If $r$ is a source,
any 
$\mpscnotstd$
with
$r \in B$
can be derived by introducing
$\intrompscnot{r} = \mpscnot{r}{\emptyset}$
and inflating.
Suppose therefore that
$r = \minvertex{B}$
is not a source and let
$\predrequalpq$.
Then it holds that
$\pqset \subseteq \belowverticesNR{r} \subseteq \lpp{B}$,
\ie the vertex sets
$\Bup$
and
$\Buq$
are both chains.

By symmetry, we have three cases for $p$ and $q$ \wrt membership in $W$.
(It is still true that
$\pqset \intersectionSP B = \emptyset$
by
\refpart{part:part-two-lemma-axiom-download-pebble-induction}
of
\reflem{lem:C-lemma-axiom-download-pebble-induction}.)

\begin{enumerate}
  \item
    $\pqset \subseteq W$:
    Immediate.
    Introduce
    $\mpscrpq$
    and inflate to
    $
    \mpscnot{\Bur}{W \unionSP \pqset} = \mpscnotstd
    $.

  \item
    \label{case:case-four-b}
    $p \in W, \ q \notin W$:
    Enlist the help 
    of our old friend
    \reflem{lem:C-lemma-axiom-download-pebble-induction},
    \refpart{part:part-four-lemma-axiom-download-pebble-induction},
    to derive
    $\mpscnot{\Buq}{W_q}$
    for
    $W_q \subseteq W$
    by inflation from $\mpinducedctminusone$
    (where
    $W_q \subseteq W$
    holds since
    $ \lpp{\Buv} \subseteq \lpp{B} $
    if
    $v \in \belowverticesNR{\minelement{b}}$).
    Introduce
    $\mpscrpq$
    and merge with
    $\mpscnot{\Buq}{W_q}$
    to get
    $
    \mpscnot{\Bur}{\Wqup}
    =    \mpscnot{B}{\Wqup}
    $.
    Then inflate
    $\mpscnot{B}{\Wqup}$ 
    to
    $\mpscnot{B}{\Wqup \unionSP W} = \mpscnotstd$.

  \item
    $\pqset \intersectionSP W = \emptyset$:
    Following an established tradition, mimic
    \refcase{case:case-four-b}
    and derive
    $\mpscnot{\Bup}{W_p}$
    and
    $\mpscnot{\Buq}{W_q}$
    with
    $\WpuWq \subseteq W$
    by inflation from
    $\mpinducedctminusone$.
    Introduce
    $\mpscrpq$,
    do two mergers to get
    $\mpscnot{B}{\WpuWq}$
    and inflate to
    $\mpscnotstd$.
    
\end{enumerate}

This takes care of the case
$r = \minelement{b}$.
Again, in all subcases
our new \mpsctext
$\mpscnotstd$
is derivable from
$\mpinducedctminusone \unionSP \intrompscnot{r}$
by inflation moves followed by mergers on some subset of
$\setsmall{p,q}$, 
possibly followed by one more inflation move.

This time the \multipebble{}s
$\mpscblacknot{\Bup}$
and
$\mpscblacknot{\Buq}$
can cause an extra intermediate cost of~$1$ each for
the bottom vertices $p$ and~$q$, and
$\mpscrpq$
potentially adds an extra cost~$1$ for~$r$, giving that the
intermediate extra cost is bounded by~$3$.

\subsubsection{Case 5: 
  $r \in \belowverticesNR{\minelement{b}}$
  for
  $\minelement{b} = \minvertex{B}$}
\label{sec:derviations-induce-case-five}

This final case is very similar to the previous case
$r = \bottomvertex{B}$.
Note first that
$r \in \belowverticesNR{\minelement{b}}
\subseteq \lpp{B}$.
If $r$ is a source, then
$\clc = \sourceclause[i]{r}$
and we have
\begin{equation}
\clsc_B \unionSP \setsmall{C} \unionSP \blacktruth{S} =
\clsc_B \unionSP \blacktruth{\Sur}
\sharpimpl \somenodetrueclause{B}
\end{equation}
at time $t-1$,
which shows that
$\mpscnot{B}{\Wur} \in \mpinducedctminusone$.
Hence, we can introduce
$\intrompscnot{r} = \mpscnot{r}{\emptyset}$
and merge on $r$ to get
$\mpscnotstd$.

As usual, the more interesting case is when
$r$ is a non-source with $\predrequalpq$.
The case analysis is just as in case~4
(\refsec{sec:derviations-induce-case-four}).
However, note that now we can again use
\refpart{part:part-three-lemma-axiom-download-pebble-induction}
of
\reflem{lem:C-lemma-axiom-download-pebble-induction}
to derive
$\mpscnot{B}{\Wur}$
from
$\mpinducedctminusone$ by inflation since it holds that $r \notin B$.

\begin{enumerate}
  \item
    $\pqset \subseteq W$:
    Introducing
    $\mpscrpq$
    and merging with
    $\mpscnot{B}{\Wur}$
    yields
    $\mpscnotstd$.

  \item
    \label{case:case-five-b}
    $p \in W, \ q \notin W$:
    Appeal to
    \refpart{part:part-four-lemma-axiom-download-pebble-induction}
    of
    \reflem{lem:C-lemma-axiom-download-pebble-induction}
    to get
    $\mpscnot{\Buq}{W_q}$
    for
    $W_q \subseteq W$
    by inflation from $\mpinducedctminusone$.
    Introduce
    $\mpscrpq$
    and merge to get
    $\mpscnot{\Bur}{\Wqup}$,
    and merge again with
    $\mpscnot{B}{\Wur}$
    to get
    $\mpscnotstd$.    

  \item
    $\pqset \intersectionSP W = \emptyset$:
    As in 
    \refcase{case:case-four-b}
    above for $q$, 
    derive
    $\mpscnot{\Bup}{W_p}$
    and
    $\mpscnot{\Buq}{W_q}$
    with
    $\WpuWq \subseteq W$
    by inflation from
    $\mpinducedctminusone$.
    Introduce
    $\mpscrpq$
    and do two mergers to get
    $\mpscnot{\Bur}{\WpuWq}$.
    Finally merge
    $\mpscnot{\Bur}{\WpuWq}$
    with
    $\mpscnot{B}{\Wur}$
    to get
    $\mpscnotstd$.    
    
\end{enumerate}

This takes care of the case
$r = \belowverticesNR{\minelement{b}}$.
We note that in all subcases of this case,
$\mpscnotstd$
is derivable from
$\mpinducedctminusone \unionSP \intrompscnot{r}$
by inflation moves followed by mergers on some subset of
$\setsmall{p,q, r}$.
Again,
the extra intermediate pebbling cost is bounded by
$
\setsizesmall{\setsmall{p,q,r}} = 3
$.

%
%
\renewcommand{\theenumi}{\arabic{enumi}}
\renewcommand{\labelenumi}{\theenumi.}

\subsection{Wrapping up the Proof}

If
$
\proofstd = \setcompact{\clsc_0, \ldots, \clsc_{\stoptime}}
$
is a derivation of
$\targetclause[i]$
from~%
$\pebcontrNT[G]{\pebdeg}$,
it is easily verified from
\refdef{def:induced-multi-pebble-subconfiguration}
that
$\mpinducedconf{\clsc_0} = \mpinducedconf{\emptyset} = \emptyset$
and
$\mpinducedconf{\clsc_{\stoptime}} 
= \mpinducedconf{\setsmall{\targetclause[i]}} 
= \setsmall{\unconditionalblackmpscnot{z}}$.

In
\refthreesecs
{sec:derivations-induce-case-erasure}
{sec:derivations-induce-case-inference}
{sec:derivations-induce-case-axiom-download},
we have shown how to do the intermediate \multipebblingtext moves to
get from
$\mpinducedctminusone$
to~%
$\mpinducedct$
in the case of erasure, inference and axiom download, respectively.
For erasure and inference, the \blobpebblingtext cost changes
monotonically 
during the transition
$
\clcfgtransition
    {\mpinducedctminusone}
    {\mpinducedct}
$.
In the case of axiom download, there can be an extra cost of~$4$
incurred for deriving each
$
\mpscnotstd
\in \mpinducedct \setminus \mpinducedctminusone
$.
We have no a priori upper bound on
$
\setsizecompact{\mpinducedct \setminus \mpinducedctminusone}
$,
but if we just derive the new \mpsctext{}s one by one and erase all
intermediate \mpsctext{}s inbetween these derivations, we will keep
the total extra cost below~$4$.

This shows that the \pebcomplete{} \multipebblingtext
$\multipebbling_{\proofstd}$  
of $G$
associated to a resolution derivation
$\derivof
{\proofstd}
{\pebcontrNT{\pebdeg}}
{\targetclause[i]}
$
by the construction in this section 
has \blobpebblingtext cost bounded from above by
$
\mpcost{\multipebbling_{\proofstd}}
\leq
\Maxofexpr[\clsc \in \proofstd]{\mpcost{\mpinducedconf{\clsc}}}
+ 4
$.
\Refth{th:C-translation-of-resolution-to-pebbling}
is thereby proven.

%
%

\section{Induced Blob Configurations Measure Clause Set Size}
\label{sec:multi-pebble-configuration-cost}

In this section we prove that if a set of clauses
$\clsc$
induces \anmpctext
$\mpinducedconf{\clsc}$
according to 
\refdef{def:induced-multi-pebble-subconfiguration},
then the cost of
$\mpinducedconf{\clsc}$
as specified in
\refdef{def:multi-pebbling-price}
is at most 
${\setsize{\clsc}}$.
That is, the cost of an induced \mpctext 
provides a lower bound on the size of the set of clauses
inducing  it.    
This is
\refth{th:linear-cost-pebbles}
below.

Note that we cannot expect a proof of this fact to work regardless of 
the pebbling degree~$\pebdeg$. The induced \multipebblingtext in
\refsec{sec:induced-pebbling}
makes no assumptions about~$\pebdeg$, but for first-degree pebbling
contradictions we know that
$
\mbox{$\clspacederiv{\pebcontrNT[G]{1}}{z_1}$}
=
\mbox{$\clspaceref{\pebcontr[G]{1}}$}
= 
\Ordosmall{1}
$.
Provided  
$\pebdeg \geq 2$, 
though, we show that one has to pay at least
$\setsize{\clsc} \geq {N}$
clauses to get an induced \mpctext of cost~$N$.

We introduce some notation to simply the proofs in what follows.
Let us define
$\varspeb{u} = \setsmall{u_1, \ldots, u_{\pebdeg}}$.
We say that      a vertex $u$  is 
\introduceterm{represented}
in       a clause 
$\clc$
derived from~$\pebcontrNT{\pebdeg}$,
or that 
$\clc$
\introduceterm{mentions}~$u$,
if
$\varspeb{u} \intersectionSP \vars{\clc} \neq \emptyset$.
We write
\begin{equation}
\vertices{\clc} 
=
\Setdescr
{ u \in \vertices{G} }
{ \varspeb{u} \intersection \vars{\clc} \neq \emptyset}  
\end{equation}
to denote all vertices represented in~$\clc$.
We will also refer to
$\vertices{\clc}$
as the set of vertices \introduceterm{mentioned} by~$\clc$.
This notation is extended to sets of clauses by taking unions.  
Furthermore, we write
\begin{equation}
\label{eq:def-clauses-mentioning-vertices}
  \clmentionvert{\clsc}{U}
  =
  \setdescr
  {\clc \in \clsc}
  {\vertices{\clc} \intersectionSP U \neq \emptyset}    
\end{equation}
to denote the subset of all clauses in $\clsc$ mentioning vertices
in a vertex set~$U$.

We now show some technical results about \cnfform{}s that will 
come in handy in the proof of
\refth{th:linear-cost-pebbles}.
Intuitively, we will use
\reflem{lem:literals-represented-in-minimal-implication} 
below together with
\reflemP{lem:pure-literals-stay}
to argue that if a clause set $\clsc$ induces a lot of \mpsctext{}s,
then there must be a lot of variable occurrences in $\clsc$
for variables corresponding to these vertices. Note, however, that
this alone will not be enough, since this will be true also for
pebbling degree \mbox{$\pebdeg=1$}.

\begin{lemma}
  \label{lem:literals-represented-in-minimal-implication}
  Suppose for a set of clauses
  $\clsc$
  and clauses $\cld_1$ and $\cld_2$ 
  with
  $\vars{\cld_1} \intersectionSP \vars{\cld_2} = \emptyset$
  that
  $\clsc \impl \cld_1 \lor \cld_2$
  but
  $\clsc \nimpl \cld_2$.
  Then there is a literal
  $\lita \in \lit{\clsc} \intersectionSP \lit{\cld_1}$.
\end{lemma}

\begin{proof}
  Pick a \tva $\tvastd$ \st
  $\logeval{\clsc}{\tvastd} = 1$
  but
  $\logeval{\cld_2}{\tvastd} = 0$.
  Since
  $\clsc \impl \cld$,
  we must have
  $\logeval{\cld_1}{\tvastd} = 1$.
  Let $\tvastd'$ be the same assignment except that
  all satisfied literals in $\cld_1$ are flipped to false
  (which is possible since they are all strictly distinct by assumption).
  Then 
  $\logeval{\cld_1 \lor \cld_2}{\tvastd'} = 0$
  forces
  $\logeval{\clsc}{\tvastd'} = 0$,  
  so the flip must have falsified   some previously satisfied clause
  in~$\clsc$. 
\end{proof}


The fact that a minimally unsatisfiable \cnfform must have more 
clauses than variables seems to have been proven 
independently a number of times (see, \eg,
\cite{AL86Minimal,
  BET01MinimallyUnsatisfiable,
  CS88ManyHard,  
  Kullmann00Matroid}).
We will need the following formulation of this result, relating
subsets of variables in a minimally implicating \cnfform and the
clauses containing variables from these subsets.

\begin{theorem}
  \label{th:N-clauses-mentioning-S-gt-size-of-S}
  Suppose that $\fstd$ is \cnfform that implies a clause $\cld$
  minimally.
  For any subset of variables $V$ of $\fstd$, let
  $
  \fstd_V = 
  \setdescrsmall
  { \clc \in \fstd }
  { \vars{\clc} \intersectionSP V \neq \emptyset }
  $
  denote the set of clauses containing variables from~$V$.
  Then if
  $V \subseteq \vars{\fstd} \setminus \vars{\cld}$,
  it holds that
  $\setsize{\fstd_V} > \setsize{V}$.
  In particular, if $\fstd$ is a minimally unsatisfiable \cnfform, 
  we have
  $\setsize{\fstd_V} > \setsize{V}$
  for all
  $V \subseteq \vars{\fstd}$.
\end{theorem}

\begin{proof}
  The proof is 
  by induction over $V \subseteq \vars{\fstd} \setminus \vars{\cld}$.

  The base case is easy.
  If $\setsize{V} = 1$, then 
  $\setsize{F_V} \geq 2$, since any 
  $\varx \in V$ must occur 
  both unnegated and negated in 
  $\fstd$ 
  by
  \reflem{lem:pure-literals-stay}.

  The inductive step just generalizes the proof of 
  \reflem{lem:pure-literals-stay}. 
  Suppose that 
  $\setsize{\fstd_{V'}} > \setsize{V'}$ 
  for all strict subsets 
  $V' \subsetneqq V \subseteq \vars{\fstd} \setminus \vars{\cld}$ 
  and consider~$V$.
  Since
  $\fstd_{V'} \subseteq \fstd_V$
  if
  $V' \subseteq V$,
  choosing any $V'$ of size $\setsize{V} - 1$
  we see that
  $
  \setsize{\fstd_V} 
  \geq 
  \setsize{\fstd_{V'}} 
  \geq 
  \setsize{V'} + 1
  = 
  \setsize{V} 
  $.

  If $\setsize{\fstd_V} > \setsize{V}$ there is nothing to prove, 
  so assume that $\setsize{\fstd_V} = \setsize{V}$. 
  Consider the bipartite graph with the variables $V$ and the clauses
  in $\fstd_V$ as   vertices, and edges between variables and clauses
  for all variable   occurrences. Since for all $V' \subseteq V$ the
  set of neighbours  
  $\vneighbour{V'} = \fstd_{V'} \subseteq \fstd_V$ satisfies 
  $\setsize{\vneighbour{V'}} \geq \setsize{V'}$, 
  by Hall's marriage theorem there is a perfect matching between $V$
  and   $\fstd_V$. Use this
  matching to satisfy $\fstd_V$ assigning values to variables in $V$ only.

  The clauses in
  $\fstd' = \fstd \setminus \fstd_V$ 
  are not affected by this partial truth value assignment,
  since they do not contain any occurrences of variables in $V$. 
  Furthermore, by the minimality of $\fstd$ it must hold that $\fstd'$ 
  can be satisfied and $\cld$ falsified simultaneously by assigning
  values to variables in  
  $\vars{\fstd'} \setminus V$.

  The two partial truth value assignments above can be
  combined to an assignment that satisfies all of $\fstd$ but
  falsifies $\cld$, which is a contradiction. 
  Thus $\setsize{\fstd_V} > \setsize{V}$. 
  The theorem follows by induction.
\end{proof}

%
%

Continuing our intuitive argument,
given that
\reftwolems
{lem:pure-literals-stay}
{lem:literals-represented-in-minimal-implication}
tell us that many induced \mpsctext{}s implies the presence of many
variables in~$\clsc$,   we will use
\refth{th:N-clauses-mentioning-S-gt-size-of-S} 
to demonstrate that a lot of different variable occurrences will have to
translate into a lot of different clauses provided that the pebbling
degree $\pebdeg$ is at least~$2$.
Before we prove this formally, let us try to provide some intuition
for why it should be true by studying two special cases.
Recall the notation
$\blacktruth{V} = \Setdescr{\sourceclausesuff[i]{v}}{v \in V}$
and
$\somenodetrueclause{V} = \Lor_{v \in V} \sourceclausesuff[i]{v}$
from 
\refsec{sec:induced-pebbling}.

\begin{example}
  \label{ex:pairwise-disjoint-induced-blobs}
  Suppose that $\clsc$ is a clause set derived from
  $\pebcontrNT[G]{\pebdeg}$
  that induces $N$ \blindependent black \blob{}s
  $B_1, \ldots, B_N$ that are pairwise disjoint, \ie
  $B_i \intersectionSP B_j = \emptyset$ if $i \neq j$.
  Then the implications
  \begin{equation}
    \label{eq:disjoint-induced-blobs}
      \clsc \impl \somenodetrueclause{B_i}
  \end{equation}
  hold for
  $i = 1, \ldots, N$.
  Remember that since
  $\pebcontrNT[G]{\pebdeg}$
  is non-contradictory, so is~$\clsc$.
    
  It is clear that a non-contradictory clause set $\clsc$ satisfying
  \refeq{eq:disjoint-induced-blobs}
  for $i = 1, \ldots, N$  is quite simply the set
  \begin{equation}
    \label{eq:inducing-clause-set-for-disjoint-blobs}
    \clsc =
    \Setdescr{\somenodetrueclause{B_i}}{i = 1, \ldots N}
  \end{equation}
  consisting precisely of the clauses implied.
  Also, it seems plausible that this is the best one can do. 
  Informally, if there
  would be strictly fewer clauses than~$N$, some clause would have to
  mix variables from different \multipebble{}s
  $B_i$ and~$B_j$. But then
  \reflem{lem:pure-literals-stay}
  says that there will be extra clauses needed to ``neutralize'' the
  literals from
  $B_j$ in the implication $\clsc \impl \somenodetrueclause{B_i}$
  and vice versa, so that the total number of clauses would have to be
  strictly greater than~$N$.

  As it turns out, the proof that $\setsize{\clsc} \geq N$ 
  when $\clsc$ induces $N$ pairwise disjoint and \blindependent black
  \blob{}s is very easy. Suppose on the contrary
  that~\refeq{eq:disjoint-induced-blobs} holds for
  $i = 1, \ldots, N$
  but that
  $\setsize{\clsc} < N$.
  Let $\tvastd$ be a satisfying assignment for~$\clsc$.
%
%
  Choose $\tvastd' \subseteq \tvastd$ to be any minimal partial truth value
  assignment fixing
  $\clsc$ to true.
  Then for the size of the domain of $\tvastd'$ we have
  $\setsize{\domainof{\tvastd'}} < N$,
  since at most one distinct literal is needed for every clause
  $\clc \in \clsc$ to fix it to true.
  This means that there is some $B_i$ such that $\tvastd'$ does not
  set any variables in $\varspeb{B_i}$. Consequently $\tvastd'$ can
  be extended to an assignment $\tvastd''$ setting $\clsc$ to true but
  $\somenodetrueclause{B_i}$ to false, which is a contradiction.
  With some more work, and using
  \refth{th:N-clauses-mentioning-S-gt-size-of-S}, one can show that
  $\setsize{\clsc} > N$ 
  if variables from distinct \multipebble{}s are mixed.

  Note that the above argument works for any pebbling degree including
  $\pebdeg = 1$. Intuitively, this means that one can charge for black
  blobs even in the case of first degree pebbling formulas. 
\end{example}

\begin{example}
  \label{ex:many-associated-induced-white-pebbles}
  Suppose that the clause set  $\clsc$ induces an \mpscfulltext
  $\mpscnotstd$  with $W \neq \emptyset$, and let us assume for
  simplicity that $\clsc$ is minimal and $W = S$ so that the implication
  \begin{equation}
    \label{eq:clause-set-inducing-white-pebbles}
    \clsc \unionSP \blacktruth{W} \impl \somenodetrueclause{B}
  \end{equation}
  holds and is minimal.
  We claim that
  $\setsize{\clsc} \geq \setsize{W} + 1$
  provided that $\pebdeg > 1$.

  Since by definition
  $B \intersectionSP W = \emptyset$
  we have
  $
  \vars{\somenodetrueclause{B}}
  \intersectionSP
  \vars{\blacktruth{W}}
  = \emptyset
  $,
  and
  \refth{th:N-clauses-mentioning-S-gt-size-of-S}
  yields that
  $
  \setsize{\clsc \unionSP \blacktruth{W}}
  \geq
  \setsize{\clmentionvert{\clsc}{W} \unionSP \blacktruth{W}}
  >
  \setsize{\vars{\blacktruth{W}}}
  $,
  using the notation from~%
  \refeq{eq:def-clauses-mentioning-vertices}.
  This is not quite what we want---we have a lower bound on
  $\setsize{\clsc \unionSP \blacktruth{W}}$,
  but what we need is a bound on~$\setsize{\clsc}$.
  But if we observe that
  $\setsize{\vars{\blacktruth{W}}}
  = \pebdeg \setsize{W}$
  while
  $\setsize{\blacktruth{W}} = \setsize{W}$,
  we get that
  \begin{equation}
    \label{eq:inequality-set-size-induced-white-pebbles}
    \setsize{\clsc}
    \geq
    \setsize{\vars{\blacktruth{W}}} - \setsize{\blacktruth{W}} + 1
    =
    (\pebdeg - 1) \setsize{W} + 1
    \geq
    \setsize{W} + 1
  \end{equation}
  as claimed.

  We remark that this time we had to use that $\pebdeg > 1$ in order
  to get a lower bound on the clause set size. And indeed, it is not
  hard to see that a single clause on the form
  $\clc = v_1 \lor \Lor_{w \in W} \olnot{w}_1$
  can induce an arbitrary number of white pebbles if
  $\pebdeg = 1$. Intuitively, white pebbles can be had for free in first
  degree pebbling formulas.
\end{example}

In general, matters are more complicated than in
\reftwoexs
{ex:pairwise-disjoint-induced-blobs}
{ex:many-associated-induced-white-pebbles}.
If
$\mpscnotstd[1]$
and
$\mpscnotstd[2]$
are two induced \mpscfulltext{}s, the black \multipebble{}s
$B_1$ and $B_2$ need not be disjoint,
the supporting white pebbles
$W_1$ and $W_2$ might also intersect, and  the black
\multipebble $B_1$ can intersect the supporting white pebbles~$W_2$
of the other \multipebble.
Nevertheless, if we choose with some care which vertices to charge
for, the intuition provided by our examples can still be used to prove
the following theorem.
%

\begin{theorem}
  \label{th:linear-cost-pebbles}
  Suppose that $G$ is a \pebblingdag and let
  $\clsc$ be a set of clauses derived from the pebbling formula
  $\pebcontrNT{\pebdeg}$
  for $\pebdeg \geq 2$.
  Then
  $\setsize{\clsc}  \geq  {\mpcost{\mpinducedconf{\clsc}}}$.
\end{theorem}

\begin{proof}
  Suppose that  the induced set of \mpscfulltext{}s is
  $
  \mpinducedconf{\clsc}
  =
  \Setdescr{\mpscnotstd[i]}{i \in \intnfirst{m}}
  $.
  By
  \refdef{def:multi-pebbling-price}, 
  we have
  $
  \mpcost{\mpinducedconf{\clsc}}
  =
  \Setsize{\blackschargedfor
    \unionSP \whiteschargedfor}
  $
  where
  \begin{equation}
    \blackschargedfor
    =
    \Setdescr
    {\bottomvertex{B_i}}
    {\mpscnotstd[i] \in \mpinducedconf{\clsc}}
  \end{equation}
  and
  \begin{equation}
    \whiteschargedfor
    =
    \Union_{\mpscnotstd[i] \in \mpinducedconf{\clsc}}
    \left(
      {W_i \intersectionSP
        \belowvertices[G]{\bottomvertex{B_i}}}
    \right) 
    \eqperiod
  \end{equation}
  We need to prove that
  $
  \setsize{\clsc}
  \geq
  \Setsize{
    \blackschargedfor
    \unionSP
    \whiteschargedfor
  }
  $.

  We first show that all vertices in
  $\blackschargedfor \unionSP \whiteschargedfor$
  are represented in some clause in~$\clsc$.
  By
  \refdef{def:induced-multi-pebble-subconfiguration},
  for each
  $\mpscnotstd[i] \in \mpinducedconf{\clsc}$
  there is a clause set
  $\clsc_i \subseteq \clsc$
  and a vertex set
  $S_i \subseteq G \setminus B_i$
  with
  $W_i = S_i \intersectionSP \lpp{B_i} \subseteq S_i$
  such that
  \begin{equation}
    \label{eq:C-i-union-blacktruth-W-i-implies-All-B-i}
    \clsc_i \unionSP \blacktruth{S_i} 
    \impl 
    \somenodetrueclause{B_i}
  \end{equation}
  and such that this implication does not hold for any strict subset
  of $\clsc_i$, $S_i$  or~$B_i$.
  Fix (arbitrarily) such $\clsc_i$ and $S_i$
  for every
  $\mpscnotstd[i] \in \mpinducedconf{\clsc}$
  for the rest of this proof.

  For the induced black \multipebble{}s $B_i$ we claim that
  $B_i \subseteq \vertices{\clsc_i}$,
  which certainly implies
  $\bottomvertex{B_i} \in \vertices{\clsc}$.
  To establish this claim, note that for any
  $v \in B_i$
  we can apply
  \reflem{lem:literals-represented-in-minimal-implication}
  with
  $\cld_1 = \sourceclause[j]{v}$ 
  and 
  $\cld_2 = \somenodetrueclause{B_i \setminus \setsmall{v}}$ 
  on the implication
  \refeq{eq:C-i-union-blacktruth-W-i-implies-All-B-i},
  which yields that the vertex $v$ must be represented in 
  $\clsc_i \unionSP \blacktruth{W_i}$ 
  by some positive literal~$v_j$.
  Since
  $B_i \intersectionSP S_i = \emptyset$,
  we have
  $
  \vars{\blacktruthsmall{S_i}}
  \intersectionSP
  \vars{\somenodetrueclause{B_i}}
  = \emptyset
  $
  and thus
  $v_j \in \lit{\clsc_i}$.

  Also, we claim that
  $S_i \subseteq \vertices{\clsc_i}$.
  To see this, note that 
  since \mbox{$B_i \intersectionSP S_i = \emptyset$}
  and the implication~%
  \refeq{eq:C-i-union-blacktruth-W-i-implies-All-B-i}
  is minimal,
  it follows from
  \reflem{lem:pure-literals-stay}
  that for  every 
  $w \in S_i$,
  all literals~%
  $\olnot{w}_j$, $j \in \intnfirst{\pebdeg}$,
  must be present in~$\clsc_i$.   
  Thus, in particular, it holds that
  $W_i \intersectionSP
  \belowvertices[G]{\bottomvertex{B_i}}
  \subseteq \vertices{\clsc_i}$.

  We now prove by induction over subsets  
  $
  \reprset \subseteq 
  \blackschargedfor \unionSP \whiteschargedfor
  $
  that
  $
  \setsize{\clmentionvert{\clsc}{\reprset}}
  \geq
  \setsize{\reprset}
  $.
  The theorem clearly follows from this
  since
  $
  \setsize{\clsc}
  \geq
  \setsize{\clmentionvert{\clsc}{\reprset}}
  $.
  (The reader can think of $\reprset$ as the set of vertices
  \emph{representing}   the \mpctext{}s  
  $\mpscnotstd[i] \in \mpinducedconf{\clsc}$
  in the clause set~$\clsc$.)
%

  The base case $\setsize{\reprset} = 1$ is immediate,
  since we just demonstrated that all vertices 
  $\reprvertex \in  \reprset$ are represented in~$\clsc$.

  For the induction step, suppose that 
  $
  \setsize{\clmentionvert{\clsc}{\reprset'}}
  \geq
  \setsize{\reprset'}
  $
  for all
  $\reprset' \subsetneqq \reprset$.
  Pick a ``topmost'' vertex \mbox{$\reprvertex \in \reprset$},
  \ie \st 
  $\aboveverticesNR[G]{\reprvertex} \intersectionSP \reprset = \emptyset$.
  We associate a \mpscfulltext
  $\mpscnotstd[i] \in \mpinducedconf{\clsc}$
  with $\reprvertex$ as follows.
  If
  $\reprvertex = \bottomvertex{B_i}$
  for some
  $\mpscnotstd[i]$,
  fix
  $\mpscnotstd[i]$
  arbitrarily to such \anmpsctext.
  Otherwise, there must exist some
  $\mpscnotstd[i]$
  \st
  $\reprvertex \in W_i \intersectionSP \belowvertices[G]{\bottomvertex{B_i}}$,
  so fix any such \mpsctext.
  We note that it holds that
  \begin{equation}
    \label{eq:no-r-in-B-i-except-r}
    \reprset 
    \intersectionSP 
    \abovevertices[G]{\bottomvertex{B_i}}
    \subseteq 
    \setsmall{\reprvertex}
  \end{equation}
  for
  $\mpscnotstd[i]$
  chosen in this way.

  Consider the clause set
  $\clsc_i \subseteq \clsc$
  and vertex set
  $S_i \supseteq W_i$
  from
  \refeq{eq:C-i-union-blacktruth-W-i-implies-All-B-i}
  associated with $\mpscnotstd[i]$ above.
  Clearly, by construction
  $\reprvertex 
  \in \vertices{\clsc_i}$
  is one of the vertices of~$R$ mentioned by~$\clsc_i$.
  We claim that the total number of vertices in $\reprset$ mentioned
  by $\clsc_i$ is upper-bounded by the number of clauses in $\clsc_i$
  mentioning these vertices, \ie that
  \begin{equation}
    \label{eq:claimed-bound-N-clauses-C-i-mentioning-U}
    \Setsize{
      \clmentionvert
      {\clsc_i}
      {\reprset}} 
    \geq 
    \Setsize{\reprset \intersectionSP \vertices{\clsc_i}}
    \eqperiod
  \end{equation}                                
  Let us first see that this claim is sufficient to prove the theorem.
  To this end, let
  \begin{equation}
    \label{eq:def-reprsetith}
    \reprsetith 
    = 
    \reprset \intersectionSP \vertices{\clsc_i}
  \end{equation}
  denote the set of all vertices in $\reprset$ mentioned by~$\clsc_i$
  and assume that
  $
  \setsize{\clmentionvert{\clsc_i}{\reprset}} 
  =  \setsize{\clmentionvert{\clsc_i}{\reprsetith}} 
  \geq \setsize{\reprsetith}
  $.
  Observe that
  $\clmentionvert{\clsc_i}{\reprsetith}
  \subseteq 
  \clmentionvert{\clsc}{\reprset}$,
  since
  $\clsc_i \subseteq \clsc$
  and
  $\reprsetith \subseteq \reprset$.
  Or in words:
  the set of clauses in $\clsc_i$ mentioning vertices in
  $\reprsetith$ is certainly a subset of all clauses in $\clsc$
  mentioning any vertex in~$\reprset$.
  Also, by construction 
  $\clsc_i$
  does not mention any vertices in
  $\reprset \setminus \reprsetith$
  since
  $\reprsetith = \reprset \intersectionSP \vertices{\clsc_i}$.
  That is, 
  \begin{equation}
    \clmentionvert{\clsc}{\reprset \setminus \reprsetith}
    \subseteq
    \clmentionvert{\clsc}{\reprset} 
    \setminus 
    \clsc_i
  \end{equation}
  in our notation.                                
  Combining the (yet unproven) claim
  \refeq{eq:claimed-bound-N-clauses-C-i-mentioning-U}
  for 
  $\clmentionvert{\clsc_i}{\reprset} = \clmentionvert{\clsc_i}{\reprsetith}$
  asserting that
  $\Setsize{\clmentionvert{\clsc_i}{\reprsetith}} \geq \setsize{\reprsetith}$
  with the induction hypothesis  for 
  $\reprset \setminus \reprsetith 
  \subseteq \reprset \setminus \setsmall{\reprvertex}
  \subsetneqq \reprset$
  we get
  \begin{align}
    \nonumber
    \Setsize{
      \clmentionvert{\clsc}{\reprset}
    }
    &=
    \Setsize{
      \clmentionvert{\clsc_i}{\reprset}
      \disjointunionSP
      \clmentionvert{(\clsc \setminus \clsc_i)}{\reprset}
    }
    \\ 
    \nonumber
    &\geq
    \Setsize{
      \clmentionvert{\clsc_i}{\reprset \intersection \vertices{\clsc_i}}
      \disjointunionSP
      \clmentionvert{\clsc}{\reprset \setminus \vertices{\clsc_i}}
    }
    \\ 
    &=
    \Setsize{\clmentionvert{\clsc_i}{\reprsetith}}
    +
    \Setsize{\clmentionvert{\clsc}{\reprset \setminus \reprsetith}}
    \\
    \nonumber
    &\geq
    \setsize{\reprsetith}
    +
    \setsize{\reprset \setminus \reprsetith}
    \\
    \nonumber
    &= 
    \setsize{\reprset}
  \end{align}
  and the theorem follows  by induction.

  It remains to verify the claim
  \refeq{eq:claimed-bound-N-clauses-C-i-mentioning-U}
  that
  $
  \setsize{\clmentionvert{\clsc_i}{\reprsetith}} 
  \geq \setsize{\reprsetith}
  $
  for
  $\reprsetith 
  = 
  \reprset \intersectionSP \vertices{\clsc_i}
  \neq
  \emptyset
  $.
  To do so, recall first that
  $r \in \reprsetith$.
  Thus,
  $\reprsetith \neq \emptyset$
  and if
  $\reprsetith = \setsmall{r}$
  we trivially have  
  $
  \setsize{\clmentionvert{\clsc_i}{\reprsetith}}
  \geq
  1
  =
  \setsize{\reprsetith}
  $.
  Suppose therefore that
  $\reprsetith \supsetneqq \setsmall{r}$.

  We want to apply
  \refth{th:N-clauses-mentioning-S-gt-size-of-S} 
  on the formula 
  $
  \fstd
  =
  {\clsc_i} \unionSP \blacktruthsmall{S_i}
  $
  on the left-hand side of the minimal implication
  \refeq{eq:C-i-union-blacktruth-W-i-implies-All-B-i}.
  Let $\reprsetprimed = \reprsetith \setminus \setsmall{r}$,
  write
  $\reprsetprimed = \reprset_1 \disjointunion \reprset_2$
  for 
  $\reprset_1 = \reprsetprimed \intersectionSP S_i$
  and
  $\reprset_2  = \reprsetprimed \setminus \reprset_1$,
  and consider        the subformula
  \begin{equation}
    \begin{split}
      \fstd_{\reprsetprimed}
      &=
      \Setdescr
      { \clc \in \bigl( {\clsc_i} \unionSP \blacktruthsmall{S_i} \bigr) }
      { \vertices{\clc} \intersectionSP \reprsetprimed \neq \emptyset}
      \\
      &=
      {\clmentionvert{\clsc_i}{\reprsetprimed}}
      \unionSP
      \blacktruthsmall{\reprset_1}
    \end{split}      
  \end{equation}
  of
  $\fstd = {\clsc_i} \unionSP \blacktruthsmall{S_i}$.
  A key observation for the concluding part of the argument is that by
  \refeq{eq:no-r-in-B-i-except-r}
  we have
  $
  \varspeb{\reprsetprimed} 
  \intersectionSP 
  \vars{\somenodetrueclause{B_i}} 
  =
  \emptyset
  $.

  For each $w \in \reprset_1$, the clauses in
  $\blacktruthsmall{\reprset_1}$
  contain $\pebdeg$ literals
  $w_1, \ldots, w_{\pebdeg}$
  and these literals must all occur negated in $\clsc_i$ by
  \reflem{lem:pure-literals-stay}.
  For each 
  $u \in \reprset_2$, 
  the clauses in
  ${\clmentionvert{\clsc_i}{\reprsetprimed}}$
  contain at least one variable~$u_i$.
  Appealing to 
  \refth{th:N-clauses-mentioning-S-gt-size-of-S} 
  with
  the subset of variables
  $
  \varspeb{\reprsetprimed} \intersectionSP \vars{\clsc_i}
  \subseteq
  \vars{\fstd} \setminus
  \vars{\somenodetrueclause{B_i}} 
  $, 
  we get    
  \begin{align}
    \nonumber
    \Setsize{\fstd_{\reprsetprimed}}
    &=
    \Setsize{
      {\clmentionvert{\clsc_i}{\reprsetprimed}}
      \unionSP
      \blacktruthsmall{\reprset_1}
    }
    \\
    \label{eq:application-N-clauses-mentioning-S-gt-size-of-S}
    &\geq 
    \Setsize{
      \varspeb{\reprsetprimed} \intersectionSP \vars{\clsc_i} 
    }
    + 1
    \\
    \nonumber
    &\geq
    \pebdeg\Setsize{\reprset_1} + \Setsize{\reprset_2}
    + 1
    \eqcomma
  \end{align}
  and rewriting this as 
  \begin{equation}
    \begin{split}
      \Setsize{\clmentionvert{\clsc_i}{\reprsetith}}
      &\geq
      \Setsize{\clmentionvert{\clsc_i}{\reprsetprimed}}
      \\
      &=
      \Setsize{\fstd_{\reprsetprimed}} 
      - \Setsize{\blacktruthsmall{\reprset_1}} 
      \\
      \label{eq:pebbling-degree-geq-two-used-here}
      &\geq
      (\pebdeg - 1)\Setsize{\reprset_1} + \Setsize{\reprset_2} + 1
      \\
      &\geq
      \Setsize{\reprsetith}  
    \end{split}
  \end{equation}
  establishes the claim.
\end{proof}

We have two concluding remarks.
Firstly, we note that the place where the condition
$\pebdeg \geq 2$
is needed is the very final step~%
\refeq{eq:pebbling-degree-geq-two-used-here}.
This is where an attempted lower bound proof for 
first degree pebbling formulas
$\pebcontrNT[G]{1}$
would fail for the reason that the presence of many white pebbles in 
$\mpinducedconf{\clsc}$
says absolutely nothing about the size of the clause set
$\clsc$
inducing these pebbles.
Secondly, another crucial step in the proof is
that we can choose our representative vertices
$\reprvertex \in \reprset$ 
so that~%
\refeq{eq:no-r-in-B-i-except-r}
holds. 
It is thanks to this fact that the inequalities in~%
\refeq{eq:application-N-clauses-mentioning-S-gt-size-of-S}
go through.
%
%
The way we make sure that~%
\refeq{eq:no-r-in-B-i-except-r}
holds is to charge only for (distinct) bottom vertices in the black
\multipebble{}s, and only for supporting white pebbles below these
bottom vertices.

%
%

\newcommand{\sectionKlawe}{section\xspace}
\newcommand{\subsectionKlawe}{subsection\xspace}

\section{Black-White Pebbling and Layered Graphs}
\label{sec:pebble-games-pyramids}

Having come this far in the paper, we know that resolution derivations
induce \blobpebblingtext{}s. We also know that \blobpebblingtext cost
gives a lower bound on clause set size and hence on the space of the
derivation.
The final component needed to make the proof of
\refth{th:main-theorem}
complete is to show lower bounds on the \blobpebblingtext price~%
$\blobpebblingprice{G_i}$
for some nice family of \pebblingdag{}s $G_i$.

Perhaps the first idea that comes to mind is to try to establish lower
bounds on \blobpebblingtext price  by reducing this  problem 
to the problem of proving lower bounds for the standard black-white
pebble game of \refdef{def:bw-pebble-game}. This is what is done in~%
\cite{Nordstrom06NarrowProofsMayBeSpaciousSTOCtoappearinSICOMP}
for the restricted case of trees.
There, for the pebblings
$\pebbling_{\proofstd}$
that one gets from resolution derivations 
$\derivof{\proofstd}{\pebtreecontrNT{\pebdeg}}{\targetclause[i]}$
in a rather different so-called ``labelled'' pebble game,
an explicit procedure is presented to transform
$\pebbling_{\proofstd}$
into a \pebcomplete black-white pebblings of $T$ in asymptotically
the same cost. The lower bound on pebbling price in the labelled
pebbel game then follows immediately by using the known lower bound
for black-white pebbling of trees in~%
\refth{th:bounds-pebbling-price-trees}.

Unfortunately, the \blobpebblegame seems more difficult than the game in~%
\cite{Nordstrom06NarrowProofsMayBeSpaciousSTOCtoappearinSICOMP}
to analyze in terms of the standard black-white pebble game. 
The problem is the inflation rule (in combination with the cost function).
It is not hard to show that without inflation, the
\blobpebblegame is essentially just a disguised form of black-white
pebbling. Thus, if we could convert any \blobpebblingtext into an
equivalent pebbling not using inflation moves without
increasing the cost  by more than, say, some constant factor, we would
be done. 
But in contrast to the case for the labelled pebble game in~%
\cite{Nordstrom06NarrowProofsMayBeSpaciousSTOCtoappearinSICOMP}
played on binary trees,
we are currently not able to transform  \blobpebblingtext{}s into
black-white pebblings in a cost-preserving way. 

%
%

Instead, what we do is to prove lower bounds directly for the
\blobpebblegame. This is not immediately clear how to do, since the
lower bound proofs for black-white pebbling price in, \eg,
\cite{CS76Storage,
  GT78VariationsPebbleGame,
  K80TightBoundPebblesPyramid,
  LT80SpaceComplexityPebbleGamesTrees}
all break down for the more general \blobpebblegame. We are currently
able to obtain lower bounds only for the limited class of
\introduceterm{layered spreading graphs} 
(to be defined below), a class that includes binary trees and pyramid
graphs.
In our proof, we borrow heavily from the corresponding bound for
black-white pebbling in~%
\cite{K80TightBoundPebblesPyramid},
but we need to go quite deep into the construction in order to
make the changes necessary for the proof go through in the
\blobpebblingtext case. 
In this \sectionKlawe, we therefore give a detailed exposition 
of the lower bound in~%
\cite{K80TightBoundPebblesPyramid}, 
in the process simplifying the proof somewhat.
In the next \sectionKlawe we build on this result to generalize the bound
from the black-white pebble game  to
the \blobpebblegame in \refdef{def:multi-pebble-game}.

\subsection{Some Preliminaries and a Tight Bound for Black Pebbling}

Unless otherwise stated, in the following
$G$ denotes a layered DAG;
$u, v, w, x, y$ denote vertices of~$G$;
$U, V, W, X, Y$ denote sets of vertices;
$P$  denotes a path;
and
$\setofpathsstd$  denotes a set of paths.
We will also  use the following notation.

\begin{definition}[Layered DAG notation]
  \label{def:layered-graphs-notation}
  For a vertex $u$ in a 
  layered DAG $G$ 
  we let
  $\vlevel{u}$ denote the level of $u$.
  For a vertex set~$U$ we let
  $\vminlevel{U} = \minofset{\vlevel{u}}{u \in U}$
  and
  $\vmaxlevel{U} = \maxofset{\vlevel{u}}{u \in U}$
  denote the lowest and highest level, respectively, of any vertex
  in~$U$.  
  Vertices in $U$ on particular levels are denoted as follows:
  \begin{itemize}
  \item
    $\vertabovelevel{U}{j}
    =
    \setdescrsmall{u \in U}{\vlevel{u} \geq j}$
    denotes the subset of all vertices in $U$ on level~$j$ or higher.
  \item
    $\vertstrictlyabovelevel{U}{j}
    =
    \setdescrsmall{u \in U}{\vlevel{u} > j}$
    denotes the vertices in $U$ strictly above level~$j$.
  \item
    $\vertonlevel{U}{j}
    =
    \vertabovelevel{U}{j} \setminus
    \vertstrictlyabovelevel{U}{j}$
    denotes the vertices exactly on level~$j$.
      \end{itemize}
      The vertex sets
      $\vertbelowlevel{U}{j}$
      and
      $\vertstrictlybelowlevel{U}{j}$
      are defined wholly analogously.
\end{definition}
                             
For the layered DAGs $G$ under consideration       
we will assume that all sources are on level~$0$,
that all non-sources have indegree~$2$, and that there is a a unique
sink~$z$. 
Since all layered DAGs also possess the 
\siblingnonreachabiblityproperty, 
this means that we are considering
\pebblingdag{}s
(\refdef{def:blob-pebbling-DAG}), 
and so the \blobpebblegame can be played on them.

Although most of what will be said in what follows holds for
arbitrary layered DAGs, we will focus on pyramids since these are
the graphs that we are most interested in.
\Reffig{fig:pyramid-height-6-a}
presents a pyramid graph with labelled vertices that we will use as a
running example. Pyramid graphs can also be visualized as
triangular fragments of a directed two-dimensional rectilinear 
lattice. Perhaps this can sometimes make it easier for the reader
to see that ``obvious'' statements about properties of pyramids
in some of the proofs below are indeed obvious.
In
\reffig{fig:pyramid-height-6-b},
the pyramid in
\reffig{fig:pyramid-height-6-a}
is redrawn as such a lattice fragment.

\begin{figure}[t]
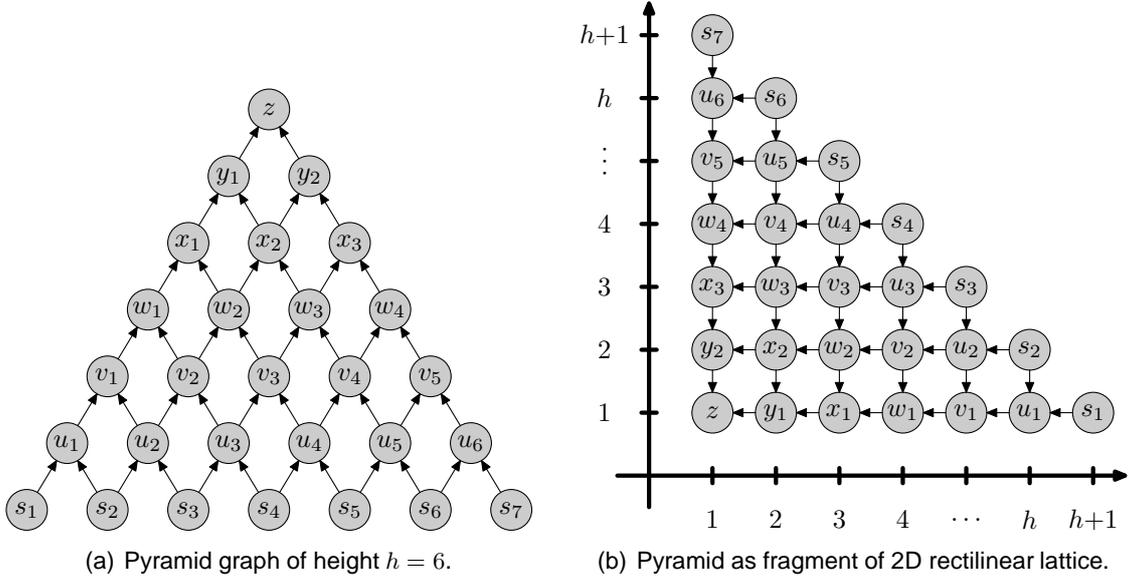
 
  \centering
  \subfigure[Pyramid graph of height $h=6$.]
  {
    \label{fig:pyramid-height-6-a}
    \begin{minipage}[b]{.46\linewidth}
      \centering
      \includegraphics{pyramidHeight6.2}%
    \end{minipage}
  }
  \hfill
  \subfigure[Pyramid as                                 
  fragment of 2D rectilinear lattice.]
  {
    \label{fig:pyramid-height-6-b}
    \begin{minipage}[b]{.50\linewidth}
      \centering
      \includegraphics{pyramidHeight6.4}%
    \end{minipage}
  }
  \caption{The pyramid $\pyramidgraph[6]$ of height $6$ 
    with labelled vertices.}
  \label{fig:pyramid-height-6}
\end{figure}

In the standard black and black-white pebble games, 
we have the following upper bounds on pebbling price
of layered DAGs.

\begin{lemma}
  \label{lem:upper-bound-pebbling-price-layered-DAG}
  For any layered DAG~$G_h$ of height~$h$ 
  with a unique sink~$z$ and
  all non-sources having vertex indegree~$2$,
  it holds that
  $\pebblingprice{G_h} \leq h + \Ordosmall{1}$
  and
  $\bwpebblingprice{G_h} \leq h/2 + \Ordosmall{1}$.
\end{lemma}

\begin{proof}
  The bounds above are true for
  complete binary trees of height~$h$ according to
  \refth{th:bounds-pebbling-price-trees}.  
  It is not hard to see that the corresponding pebbling strategies can
  be used to pebble any layered graph of the same height with at most
  the same amount of pebbles.

  Formally, suppose that the sink $z$ of the DAG $G_h$ has
  predecessors $x$ and~$y$.
  Label the root of $T_h$ by $z_1$ and its predecessors by
  $x_1$ and~$y_1$.
  Recursively, for a vertex in $T_h$ labelled by $w_i$, 
  look at the corresponding vertex $w$ in $G_h$ and suppose that
  $\prednode{w} = \set{u,v}$.
  Then label
  the vertices
  $\prednode{w_i}$ in $T_h$
  by 
  $u_j$ and~$v_k$
  for the   smallest positive indices $j,k$ such that there are
  not already other vertices in  $T_h$ labelled 
  $u_j$ and~$v_k$.
  In
  \reffig{fig:mapDAGtotree}
  there is an illustration of how the vertices in a pyramid
  $\pyramidgraph[3]$ of height~$3$
  are mapped to vertices in the complete binary tree
  $T_3$ in this manner.

  The result is a
  labelling of $T_h$
  where every vertex $v$ in $G_h$ corresponds to one or more distinct
  vertices $v_1, \ldots, v_{k_v}$ in~$T_h$,
  and such that if
  $\prednode{w_i} = \set{u_j, v_k}$ in~$T_h$,
  then
  $\prednode{w} = \set{u, v}$ in~$G_h$.
  Given a pebbling strategy $\pebbling$ for~$T_h$, we can pebble $G_h$
  with at most the same amount of pebbles by mimicking any move on any
  $v_i$ in $T_h$ by performing the same move on $v$ in $G_h$.
  The details are easily verified.
\end{proof}

\begin{figure}[tp]
  \centering
  \subfigure[Pyramid graph $\Pi_3$ of height $3$.]
  {
    \label{fig:mapDAGtotree-a}
    \begin{minipage}[b]{.4\linewidth}
      \centering
      \includegraphics{mapDAGtotree.1}%
    \end{minipage}
  }
  \hfill
  \subfigure[Binary tree $T_3$ with vertex labels 
    from $\Pi_3$.]
  {
    \label{fig:mapDAGtotree-b}
    \begin{minipage}[b]{.55\linewidth}
      \centering
      \includegraphics{mapDAGtotree.2}%
    \end{minipage}
  }
%
%
%
%
  \caption{%
    Binary tree with vertices labelled
    by pyramid graph vertices as in proof of 
    \reflem{lem:upper-bound-pebbling-price-layered-DAG}.%
%
}
  \label{fig:mapDAGtotree}
\end{figure}

In this \sectionKlawe, we will identify some layered graphs $G_h$
for which the bound in
\reflem{lem:upper-bound-pebbling-price-layered-DAG}
is also the asymptotically correct lower bound.
%
%
As a warm-up, 
and also to introduce some important ideas,
let us consider the black pebbling price of 
the pyramid $\pyramidgraphh$ of height~$h$.

\begin{theorem}[\cite{C74ObservationTimeStorageTradeOff}]
  \label{th:bounds-black-pebbling-of-pyramids}
  $\pebblingprice{\pyramidgraphh} = h + 2$ 
  for 
  $h \geq 1$.
\end{theorem}

To prove this lower bound, it turns out that it is sufficient to study
blocked paths in the pyramid.

\begin{definition}
\label{def:path-terminology}
A vertex set
$U$ \introduceterm{blocks}
a path $P$
if 
$U \intersectionSP P \neq \emptyset$.
$U$ blocks a set of paths
$\setofpathsstd$
if $U$ blocks all $P \in \setofpathsstd$.
\end{definition}

\begin{proof}[Proof of
 \refth{th:bounds-black-pebbling-of-pyramids}]
  It is easy to devise (inductively) a black pebbling strategy that uses 
  \mbox{$h+2$} pebbles
  (using, \eg,
  \reflem{lem:upper-bound-pebbling-price-layered-DAG}).
  We show that this is also a lower bound.

  Consider the first time $t$ when all possible paths from sources to
  the sink are blocked by black pebbles. Suppose that $\pathstd$ is (one
  of) the last path(s) blocked. Obviously, $\pathstd$ is blocked by 
  placing a pebble on  some source vertex $u$. The path $\pathstd$ 
  contains $h+1$ vertices, and for each
  vertex $v \in \pathstd \setminus \set{u}$ 
  there is a unique path $\pathstd_v$ that 
  coincides with 
  $\pathstd$ from
  $v$ onwards to the sink but arrives at $v$ in a straight line from a
  source ``in the opposite direction'' of that of $\pathstd$,
  \ie via the immediate predecessor of $v$ not contained in $\pathstd$. At
  time $t-1$ all such paths 
  $\setdescrsmall{\pathstd_v}{v \in P \setminus \set{u}}$ 
  must already be blocked, and since 
  $\pathstd$
  is still open no pebble can block two paths
  $\pathstd_{v} \neq \pathstd_{v'}$
  for
  $v, v' \in \pathstd \setminus \set{u}$,
  $v \neq v'$.
  Thus at time $t$ there
  are at least $h+1$ pebbles on $\pyramidgraphh$. Furthermore, \wolog each
  pebble placement on a source vertex is followed by another
  pebble placement (otherwise
  perform all removals immediately following after time $t$ before
  making the pebble placement at time~$t$). Thus at time $t+1$ there
  are $h+2$~pebbles on~$\pyramidgraphh$.
\end{proof}

We will use the idea in the proof above about a set of paths
converging at different levels to another fixed path repeatedly, so we
write it down as a separate observation.

\begin{observation}
  \label{obs:converging-paths}
  Suppose that $u$ and $w$ are vertices in $\pyramidgraphh$ on levels
  $\levelstd_u < \levelstd_w$ and that
  $\pathfromto{\pathstd}{u}{w}$ is a path from $u$ to $w$.
  Let
  $K = {\levelstd_w-\levelstd_v}$
  and write
  $\pathstd = 
  \set{v_0=u, v_1, \ldots, v_{K}=w}$.
  Then there is a set of $K$ paths
  $
  \setofpathsstd = 
  \set{\pathstd_1, \ldots, \pathstd_K}
  $
  \st
  $\pathstd_i$ 
  coincides with $\pathstd$ from $v_i$ onwards to $w$
  arrives to $v_i$  in a straight line from a source vertex
  via the immediate predecessor of $v_i$ 
  which is not contained in~$\pathstd$, \ie is distinct from~$v_{i-1}$.
  In particular, for any $i,j$ with
  $1 \leq i < j \leq k$ 
  it holds that
  $
  \pathstd_i \intersectionSP \pathstd_j 
  \subseteq
  \pathstd_j \intersectionSP \pathstd 
  \subseteq 
  \pathstd \setminus \set{u}
  $.
\end{observation}

We will refer to the paths
$\pathstd_1, \ldots, \pathstd_K$
as a set of
\introduceterm{converging \sourcepath{}s},
or just converging paths, for 
$\pathfromto{\pathstd}{u}{w}$.
See
\reffig{fig:convergingpaths}
for an example.

\begin{figure}[tp]
  \centering%
  \includegraphics{convergingpaths.1}%
%
%
%
%
  \caption{
    Set of converging \sourcepath{}s (dashed) for the path
    $\pathfromto{P}{u_4}{y_1}$
    (solid).}
  \label{fig:convergingpaths}
\end{figure}

\subsection{A Tight Bound on the Black-White Pebbling Price of Pyramids}
\label{sec:klawe-bw-pebbling-bound}

The rest of this \sectionKlawe contains an exposition of
Klawe~\cite{K80TightBoundPebblesPyramid},
with some simplifications of the proofs.
Much of the notation and terminology has been changed from
\cite{K80TightBoundPebblesPyramid}
to fit better with this paper in general and (in the next \sectionKlawe)
the \blobpebblegame in particular.
Also, it should be noted that we restrict all definitions to 
layered graphs, in contrast to Klawe who deals with a somewhat more
general class of graphs. We concentrate on layered graphs mainly to
avoid unnecessary complications  in the exposition, and since it can
be proven that no graphs in \cite{K80TightBoundPebblesPyramid} can
give a better size/pebbling price trade-off than one gets for layered
graphs anyway. 

Recall from
\refdef{def:chains-and-paths}
%
%
that a 
\introduceterm{path via~$w$} 
is a path $P$  \st $w \in P$. We will also say that $P$ 
\introduceterm{visits}~$w$. 
The notation $\pathsviavertex{w}$ is used to denote all \sourcepath{}s
visiting $w$. Note that a path $P\in\pathsviavertex{w}$ visiting $w$
may continue after~$w$, or may end in~$w$. 

\begin{definition}[\Hidingsetklawe{}]
  \label{def:cover}
  A vertex set $U$ 
  \introduceterm{\hideklawe{}s} 
  a vertex $w$ if $U$   blocks all \sourcepath{}s visiting~$w$, \ie if
  $U$ blocks $\pathsviavertex{w}$.
  $U$ \hideklawe{}s $W$ if $U$ \hideklawe{}s all $w \in W$.  
  If so, we say that
  $U$ is   a
  \introduceterm{\hidingsetklawe{}} for~$W$.
  We write
  $\hiddenvertices{U}$
  to denote the set of all vertices \hiddenklawe{} by $U$.
\end{definition}

%
Our perspective is that we are standing
at the sources of $G$ and looking towards the sink. Then $U$
\introduceterm{\hideklawe{}s} $w$ if we ``cannot see'' $w$ from the sources
since $U$ completely hides $w$.
When
$U$ \introduceterm{blocks} a path $P$ is is possible that we can
``see'' the beginning of the path, but we cannot walk all of the path
since it is blocked somewhere on the way.
The reason why this terminological distinction is convenient 
will become clearer in the next \sectionKlawe.

Note that if $U$ should \hideklawe{} $w$, then in particular it must block
all paths ending in $w$. Therefore, when looking at minimal
\hidingsetklawe{}s we can assume \wolog that no vertex in $U$ is on a
level higher than $w$. 

It is an easy exercise to show that the \hidingklawe{} relation is
transitive, \ie that if $U$ \hideklawe{}s $V$
and $V$ \hideklawe{}s $W$,
then $U$ \hideklawe{}s $W$.

\begin{proposition}
  \label{pr:cover-relation-transitive}
  If
  $V \subseteq \hiddenvertices{U}$
  and
  $W \subseteq \hiddenvertices{V}$
  then
  $W \subseteq \hiddenvertices{U}$.
\end{proposition}

One key concept in Klawe's paper is that of
\introduceterm{potential}.
The potential of
$\pconf = (B,W) $
is intended to measure how
``good'' the configuration $\pconf$ is, or at least how
hard it is to reach in a pebbling.
Note that this is not captured by the cost of the current pebble
configuration. For instance, the final configuration
$\pconf_{\stoptime}
=
(\set{z}, \emptyset)$
is the best configuration conceivable, but only costs $1$.
At the other extreme, the configuration $\pconf$ in a pyramid with, say, 
all vertices on level
$\levelstd$ white-pebbled and all vertices on level $\levelstd+1$ 
black-pebbled is potentially very
expensive (for low levels $\levelstd$), but does not seem very useful.
Since this configuration on the one hand is quite expensive, but on
the other hand is extremely easy to derive
(just white-pebble all vertices on level~$\levelstd$,
and then black-pebble all vertices on level $\levelstd+1$), 
here the cost seems
like a gross overestimation of the ``goodness''
of $\pconf$.

Klawe's potential measure remedies this. 
The potential of a pebble configuration $(B,W)$
is defined as the minimum
measure of any set $U$ that together with $W$ \hideklawe{}s~$B$.
Recall that
$\vertabovelevel{U}{j}
$
denotes the subset of all vertices in $U$ on level $j$ or higher
in a layered graph~$G$.

\begin{definition}[Measure]
  \label{def:measure}
  The
  \introduceterm{$j$th partial measure} of 
  the vertex set $U$ in $G$ is
  \begin{equation*}
    \vjthmeasure[G]{j}{U}
    =    
    \begin{cases}
      j + 2 \setsize{\vertabovelevel{U}{j}} 
      & \text{if $\vertabovelevel{U}{j} \neq \emptyset$,}
      \\
      0 & \text{otherwise,}
    \end{cases}
  \end{equation*}
  and the \introduceterm{measure} of $U$ is
  $\vmeasure[G]{U} = \Maxofexpr[j]{\vjthmeasure[G]{j}{U}}$.
\end{definition}

\begin{definition}[Potential]
  \label{def:potential}
  We say that $U$ is a \hidingsetklawe for
  a black-white pebble configuration
  $\pconf = (B,W)$
  in a layered graph~$G$ if
  $U \unionSP W$ \hideklawe{}s~$B$.
  We define the
  \introduceterm{potential} 
  of the pebble configuration to be
  \begin{equation*}
  \vpotential[G]{\pconf} = \vpotential[G]{B,W}
  =
  \minofset{\meastopot[G]{U}}
           {\text{$U$ is a \hidingsetklawe for $(B,W)$}}
  \eqperiod
  \end{equation*}                                
  If   $U$ is a  \hidingsetklawe for $(B,W)$ 
  with   
  minimal measure $\vmeasure[G]{U}$
  among all vertex sets $U'$ \st
  $U' \unionSP W$ \hideklawe{}s $B$,
  we say that $U$ is a \introduceterm{\minmeasure{}} \hidingsetklawe{}
  for~%
  $\pconf$.
\end{definition}

Since the graph under consideration will almost always be clear from
context, we will tend to omit the subindex~$G$
in measures and potentials.

We remark that although this might not be immediately obvious, there is
quite a lot of nice intuition why
\refdef{def:potential}
is a relevant estimation of how ``good'' a pebble configuration is.
We refer the reader to Section 2 of~%
\cite{K80TightBoundPebblesPyramid}
for a discussion about this.
Let us just note that with this definition, the pebble configuration 
$\pconf_{\stoptime}
=
(\set{z}, \emptyset)$
has high potential, as we shall soon see, 
while the configuration with
all vertices on level $\levelstd$ white-pebbled and all vertices on
level $\levelstd+1$  black-pebbled has potential zero.

\begin{remark}
  Klawe does not use the level of a vertex $u$  in
  \reftwodefs{def:measure}{def:potential},
  but instead the black pebbling price
  $\pebblingprice{\set{u},\emptyset}$
  of the configuration with a black pebble on $u$ and no other pebbles in
  the DAG.  For pyramids, these two concepts are equivalent, and we
  feel that the exposition can be made considerably simpler by using
  levels. 
\end{remark}

Klawe proves two facts about the potentials of the pebble 
configurations in any
%
%
black-white pebbling
$\pebbling =  \set{\pconf_{0}, \ldots, \pconf_{\stoptime}}$
of a pyramid graph $\pyramidgraphh$:
\begin{enumerate}
\item 
  \label{item:potential-property-inductive}
  The potential correctly estimates the goodness of the current
  configuration $\pconf_t$ 
  by taking into account the whole pebbling that has led to
  $\pconf_t$.
  Namely,
  $\vpotential{\pconf_t}
  \leq 
  2 \cdot    
  \maxofexpr[s \leq t]{\pebcost{\pconf_s}}$.

\item 
  \label{item:potential-property-final}
  The final configuration
  $\pconf_{\stoptime} = (\set{z}, \emptyset)$
  has high potential, namely
  $\vpotential{\set{z}, \emptyset}
  = h          
  + \Ordosmall{1}
  $.
\end{enumerate}
Combining these two parts, one clearly gets a 
lower bound on pebbling price.

For pyramids,
\refpart{item:potential-property-final}
is not too hard to show directly.
In fact, it is a useful exercise if one wants to get some feeling for
how the potential works.
\Refpart{item:potential-property-inductive}
is much trickier. It is proven
by induction over the pebbling. 
As it turns out, the whole induction proof hinges on the following key
property.  

\begin{property}[\klaweprop{}]
  \label{property:klawe-property}
  We say that the
  black-white pebble configuration
  $\pconf = (B,W)$ in~$G$ has the
  \introduceterm{\klaweprop{}},
  or just the 
  \introduceterm{\klawepropacronym{}}
  for short,
  if  there is a vertex set $U$
  \st
  \begin{enumerate}
  \item 
    \label{item:KP-cover}
    $U$ is a \hidingsetklawe for $\pconf$,   
  \item 
    \label{item:KP-measure}
    $\vpotential[G]{\pconf} =
    {\meastopot{U}}$,
    
  \item 
    \label{item:KP-size}
    $U = B$ or 
    $
    \setsize{U} < \setsize{B} + \setsize{W}
    = \pebcost{\pconf}
    $.
  \end{enumerate}
  We say that the graph $G$ has the 
  \klaweprop{}
  if all black-white pebble configurations
  $\pconf = (B,W)$ on~$G$
  have the \klaweprop{}.
\end{property}

Note that
requirements~\ref{item:KP-cover} and~\ref{item:KP-measure} 
just say that $U$ is a vertex set that witnesses the potential
of~$\pconf$. 
The important point here is 
requirement~\ref{item:KP-size}, which says (basically) that
if we are given a \hidingsetklawe{} $U$ with minimum measure 
but with size exceeding 
the  cost of the black-white pebble
configuration~$\pconf$, 
then we can pick
\emph{another} \hidingsetklawe{}
$U'$ which keeps the minimum measure but decreases the cardinality to
at most~%
$\pebcost{\pconf}$.

Given
\refproperty{property:klawe-property},
the induction proof for
\refpart{item:potential-property-inductive}
follows quite easily. The main part of the paper~%
\cite{K80TightBoundPebblesPyramid}
is then spent on proving that a class of DAGs including pyramids
have 
\refproperty{property:klawe-property}.
Let us see what the lower bound proof looks like, assuming  that
\refproperty{property:klawe-property}
holds.

\begin{lemma}[Theorem 2.2 in \cite{K80TightBoundPebblesPyramid}]
  \label{lem:potential-property-inductive}
  Let $G$ be a layered graph 
  possessing the \klawepropacronym{}
  and suppose that
  $
  \pebbling = 
  \set{\pconf_0 = \emptyset, \pconf_1, \ldots, \pconf_{\stoptime}}
  $
  is
  any \pebunconditional black-white pebbling on~$G$.
  Then it holds for all $t = 1, \ldots, \stoptime$  that
  $\vpotential[G]{\pconf_t}
  \leq 
  2 \cdot    
  \maxofexpr[s \leq t]{\pebcost{\pconf_s}}$.
\end{lemma}

\begin{proof}
  To simplify the proof, let us assume \wolog that
  no white pebble is ever removed from a source.
  If $\pebbling$ contains such moves, we just substitute for each such
  white pebble placement on $v$ a black pebble placement on $v$
  instead, and when the white pebble is removed we remove the
  corresponding black pebble. It is easy to check that this results in a
  legal pebbling $\pebbling'$ that has exactly the   same cost. 

  The proof is by induction. The base case
  $\pconf_0 = \emptyset$ is trivial.
  For the induction hypothesis,
  suppose that
  $\vpotential{\pconf_{t}}
  \leq 
  2 \cdot    
  \maxofexpr[s \leq t]{\pebcost{\pconf_s}}$
  and let
  $U_{t}$ be a vertex set as in
  \refproperty{property:klawe-property},
  \ie such that 
  $U_t \unionSP W_t$ \hideklawe{}s $B_t$,
  $\vpotential{\pconf_t} = \vmeasure{U_t}$
  and
  $\setsize{U_t}
  \leq
  \pebcost{\pconf_t}
  =
  \setsize{B} +  \setsize{W}
  $.

  Consider $\pconf_{t+1}$.
  We need to show that
  $\vpotential{\pconf_{t+1}}
  \leq 
  2 \cdot    
  \maxofexpr[s \leq t+1]{\pebcost{\pconf_s}}
  $.
  By the induction hypothesis, it is sufficient to show that
  \begin{equation}
    \label{eq:sufficient-inductive-inequality-potential}
    \vpotential{\pconf_{t+1}}
    \leq 
    \maxofexpr{\vpotential{\pconf_{t}}, 2 \cdot \pebcost{\pconf_{t+1}}}
    \eqperiod
  \end{equation}
  We also note that if
  $U_t \unionSP W_{t+1}$ \hideklawe{}s $B_{t+1}$ we are done, since if so
  $
  \vpotential{\pconf_{t+1}}
  \leq
  \vmeasure{U_t}
  =
  \vpotential{\pconf_{t}}
  $.
  We make a case analysis depending on the type of move made to get from
  ${\pconf_{t}}$ to~${\pconf_{t+1}}$.
  \begin{enumerate}
    \item
      Removal of black pebble:
      In this case,
      $
      U_t \unionSP W_{t+1} =
      U_t \unionSP W_{t}
      $
      obviously \hideklawe{}s
      $B_{t+1} \subset B_t$ 
      as well, so
      $
      \vpotential{\pconf_{t+1}}
      \leq
      \vpotential{\pconf_{t}}
      $.

    \item
      Placement of white pebble:
      Again, 
      $
      U_t \unionSP W_{t+1} \supset
      U_t \unionSP W_{t}
      $
      \hideklawe{}s
      $B_{t+1} = B_t$, so
      $
      \vpotential{\pconf_{t+1}}
      \leq
      \vpotential{\pconf_{t}}
      $.

    \item
      Removal of white pebble:
      Suppose that a white pebble is removed from the vertex~$w$, so
      $W_{t+1} = W_t \setminus \set{w}$.      
      As noted above, \wolog $w$ is not a source vertex.
      We claim that
      $U_t \unionSP W_{t+1}$ still \hideklawe{}s $B_{t+1} = B_t$, 
      from which
      $
      \vpotential{\pconf_{t+1}}
      \leq
      \vpotential{\pconf_{t}}
      $
      follows as above.

      To see that the claim is true, note that
      $\prednode{w} \subseteq B_t \unionSP W_t$
      by the pebbling rules, for otherwise we would not be able to
      remove the white pebble on~$w$. If
      $\prednode{w} \subseteq W_t$
      we are done, since then
      $U_t \unionSP W_{t+1}$
      \hideklawe{}s
      $U_t \unionSP W_{t}$
      and we can use the transitivity in
      \refpr{pr:cover-relation-transitive}.
      If instead there is some
      $v \in \prednode{w} \intersectionSP B_t$,
      then
      $U_t \unionSP W_{t} =
      U_t \unionSP W_{t+1} \unionSP \set{w}$
      \hideklawe{}s $v$ by assumption. Since
      $w$ is a successor of $v$, and therefore on a higher level
      than~$v$, we must have
      $U_t \unionSP W_{t} \setminus \set{w}$
      \hidingklawe{}~$v$.
      Thus in any case
      $U_t \unionSP W_{t+1}$
      \hideklawe{}s
      $\prednode{w}$,
      so by transitivity
      $U_t \unionSP W_{t+1}$ \hideklawe{}s~$B_{t+1}$.

    \item
      Placement of black pebble:
      Suppose that a black pebble is placed on~$v$.
      If $v$ is not a source, by the pebbling rules we again have that
      $\prednode{v} \subseteq B_t \unionSP W_t$.
      In particular,
      $B_t \unionSP W_t$
      \hideklawe{}s~$v$ and by transitivity we have that
      $U_t \unionSP W_{t+1} = U_t \unionSP W_{t}$
      \hideklawe{}s $B_t \unionSP \set{v} = B_{t+1}$.

      The case when $v$ is a source turns out to be the only
      interesting one. 
      Now 
      $U_t \unionSP W_t$ 
      does not necessarily \hideklawe{}
      $B_t \unionSP \set{v} = B_{t+1}$
      any longer.
      An obvious fix is to try with
      $U_t \unionSP \set{v} \unionSP W_t$ 
      instead. This set clearly \hideklawe{}s
      $B_{t+1}$,
      but it can be the case that
      $\vmeasure{U_t \unionSP \set{v}}
      >
      \vmeasure{U_t}$.
      This is problematic, since we could have
      $
      \vpotential{\pconf_{t+1}}
      =
      \vmeasure{U_t \unionSP \set{v}}
      >
      \vmeasure{U_t}
      =
      \vpotential{\pconf_{t}}
      $. 
      And we do not know that the inequality
      $\vpotential{\pconf_{t}} \leq 2 \cdot \pebcost{\pconf_t}$ 
      holds, only that
      $
      \vpotential{\pconf_{t}}
      \leq 2 \cdot
      \maxofexpr[s \leq t]
      {\pebcost{\pconf_s}}
      $.
      This means that it can happen that
      $\vpotential{\pconf_{t+1}} > 2 \cdot \pebcost{\pconf_{t+1}}$,
      in which case the induction step fails.
      However, we claim that using the
      \klaweprop{}~\ref{property:klawe-property}
      we can prove for
      $U_{t+1} = U_t \unionSP \set{v}$
      that
      \begin{equation}
      \vmeasure{U_{t+1}}
      =
      \vmeasure{U_t \unionSP \set{v}}
      \leq 
      \maxofexpr{
        \vmeasure{U_t},
        2 \cdot \pebcost{\pconf_{t+1}}}
      \eqcomma
      \end{equation}
      which shows that
      \refeq{eq:sufficient-inductive-inequality-potential}
      holds and the induction steps goes through.

      Namely, suppose that $U_t$ is chosen as in 
      \refproperty{property:klawe-property}
      and consider
      $U_{t+1} = U_t \unionSP \set{v}$.
      Then
      $U_{t+1}$
      is a \hidingsetklawe for
      $\pconf_{t+1} 
      = (B_{t} \unionSP \set{v}, W_{t})$
      and hence
      $\vpotential{\pconf_{t+1}} \leq \vmeasure{U_{t+1}}$.
      For $j > 0$, it holds that
      $\vertabovelevel{U_{t+1}}{j} = \vertabovelevel{U_{t}}{j}$
      and thus
      $\vjthmeasure{j}{U_{t+1}} = \vjthmeasure{j}{U_{t}}$.
      On the bottom level, 
      using that the inequality
      $\setsize{U_t} \leq \pebcost{\pconf_t}$
      holds by the \klawepropacronym,
      we have
      \begin{equation}
        \vjthmeasure{0}{U_{t+1}} 
        = 
        2 \cdot \setsize{U_{t+1}} 
        = 
        2 \cdot (\setsize{U_{t}} + 1)
        \leq
        2 \cdot (\pebcost{\pconf_t} + 1)
        =
        2 \cdot \pebcost{\pconf_{t+1}}
      \end{equation}
      and we get that
      \begin{multline}
        \vmeasure{U_{t+1}}
        = 
        {\textstyle \Maxofexpr[j]{\vjthmeasure{j}{U_{t+1}}}}
        = 
        {\textstyle \Maxofexpr{
            \Maxofexpr[j>0]{\vjthmeasure{j}{U_{t}}},
            \vjthmeasure{0}{U_{t+1}}}}
        \\ 
        \leq
        \maxofexpr{\vmeasure{U_t}, 2 \cdot \pebcost{\pconf_{t+1}}}
        =
        \maxofexpr{\vpotential{\pconf_t}, 2 \cdot \pebcost{\pconf_{t+1}}}
      \end{multline}
      which is exactly what we need.
  \end{enumerate}
  We see that the inequality
  \refeq{eq:sufficient-inductive-inequality-potential}
  holds in all cases in our case analysis, which proves the lemma.
\end{proof}

The lower bound on black-white pebbling price now follows by showing
that the final pebble configuration
$(\set{z}, \emptyset)$
has high potential.

\begin{lemma}
  \label{lem:potential-property-final}
  For $z$ the sink of a pyramid $\pyramidgraphh$
  of height~$h$, 
  the pebble configuration
  $(\set{z}, \emptyset)$
  has potential
  $\vpotential[\pyramidgraphh]{\set{z}, \emptyset}
  = h + 2
  $.
\end{lemma}

\begin{proof}
  This follows easily from the \klaweprop
  (which says that $U$ can be chosen so that either
  $U \subseteq \set{z}$ or $\setsize{U} \leq 0$), 
  but let us show that this
  assumption is not necessary here.
  The set
  $U = \set{z}$
  \hideklawe{}s itself and has measure
  $
  \vmeasure{U} = \vjthmeasure{h}{U} =
  h + 2 \cdot 1 = h + 2$.
  Suppose that $z$ is \hiddenklawe{} by some
  $U' \neq \set{z}$.
  \Wolog $U'$ is minimal, \ie no strict subset of $U'$ \hideklawe{}s~$z$.
  Let $u$ be a vertex in $U'$ on minimal level
  $\vminlevel{U} = \levelstd < h$.
  The fact that $U'$ is minimal implies that there is a path
  $\pathfromto{P}{u}{z}$
  \st
  $(P \setminus \set{u}) \intersectionSP U' = \emptyset$
  (otherwise $U' \setminus \set{u}$ would \hideklawe~$z$).
  By
  \refobs{obs:converging-paths},
  there must exist
  $h - \levelstd$
  converging paths from sources to$~z$ that are
  all blocked by distinct pebbles in 
  $U' \setminus \set{u}$.
  It follows that
  \begin{equation}
  \vmeasure{U'}
  \geq
  \Vjthmeasure{\levelstd}{U'}
  =
  \levelstd + 2 \Setsize{\vertabovelevel{U'}{\levelstd}} 
  =
  \levelstd + 2 \Setsize{U'} 
  \geq
  \levelstd + 2 \cdot (h + 1 - \levelstd)
  >
  h + 2
  \end{equation}
  (where we used that
  ${\vertabovelevel{U'}{\levelstd}} = U'$ 
  since $\levelstd = \vminlevel{U}$).
  Thus
  $U = \set{z}$
  is the unique   \minmeasure \hidingsetklawe{} for
  $(\set{z}, \emptyset)$, 
  and the potential is
  $\vpotential{\set{z}, \emptyset}
  = h + 2
  $.
\end{proof}

Since \cite{K80TightBoundPebblesPyramid} proves that pyramids possess
the \klaweprop, and since there are pebblings that
yield matching upper bounds, we have the following  theorem.

\begin{theorem}[\cite{K80TightBoundPebblesPyramid}]
  \label{th:lower-bound-bwpebbling-assuming-KP}
  $\bwpebblingprice{\pyramidgraphh}
  =
  \frac{h}{2} + \Ordosmall{1}
  $.
\end{theorem}

\begin{proof}
  The upper bound was shown in
  \reflem{lem:upper-bound-pebbling-price-layered-DAG}.
  For the lower bound, 
  \reflem{lem:potential-property-final}
  says that the final pebble configuration
  $(\set{z}, \emptyset)$
  in any \pebcomplete pebbling $\pebbling$ of $\pyramidgraphh$  
  has potential
  $\vpotential{\set{z}, \emptyset} = h + 2$.
  According to
  \reflem{lem:potential-property-inductive},
  $
  \vpotential{\set{z}, \emptyset}
  \leq 
  \mbox{$2 \cdot \pebcost{\pebbling}$}
  $.
  Thus
  $\bwpebblingprice{\pyramidgraphh} \geq h/2 + 1$.
\end{proof}

In the final two \subsectionKlawe{}s of this \sectionKlawe, 
we provide a fairly detailed overview of the proof that pyramids do
indeed possess the \klaweprop{}. As was discussed above,
the reason for giving all the details is that we will need to use and
modify the construction in non-trivial ways in the next \sectionKlawe,
where we will use ideas inspired by Klawe's paper to prove lower
bounds on the pebbling price of pyramids in the \blobpebblegame.

\subsection{Proving the \KLAWEPROP{}}
\label{sec:klawe-bw-pebbling-proving-the-Klawe-property}

We present the proof of that pyramids have the \klaweprop in a
top-down fashion as follows.

\begin{enumerate}
\item 
  First, we study what \hidingsetklawe{}s look like in order to better
  understand their structure. Along the way,
  we make a few definitions and prove some lemmas culminating in
  \refdef{def:covering-set-graph} and
  \reflem{lem:klawe-lemma-three-three}.
 
\item
  We conclude that it seems like a good idea to try to split our
  \hidingsetklawe into disjoint components, prove the
  \klawepropacronym locally, and then add everything together to get a
  proof that works globally. We make an attempt to do this in
  \refth{th:pyramids-have-klawe-property},
  but note that the argument does not quite work.  However, if we
  assume a slightly stronger property locally for our disjoint
  components 
  (\refproperty{property:local-klawe-property}), 
  the proof goes through.

\item
  We then prove this stronger local property
  by assuming that pyramid graphs have a certain
  \introduceterm{spreading} property
  (\refdef{def:spreading-graph}
  and
  \refth{th:pyramids-are-spreading-graphs}),
  and by showing in
  \reftwolems
  {lem:klawe-lemma-three-five}
  {lem:pick-local-good-blocker}
  that the stronger local property holds for such spreading graphs.

\item
  Finally, in
  \refsec{sec:klawe-bw-pebbling-pyramids-are-spreading},
  we give a simplified proof of the theorem in
  \cite{K80TightBoundPebblesPyramid}  
  that pyramids are indeed spreading.
\end{enumerate}
From this, the desired conclusion follows.

For a start, we need two definitions.  The intuition for the first one
is that the vertex set $U$ is \introduceterm{\tightklawe{}} if is does
not contain any ``unnecessary'' vertex $u$ \hiddenklawe by the other
vertices in~$U$.

\begin{definition}[\Tightklawe{} vertex set]
  \label{def:tight}
  The vertex set $U$ is \introduceterm{\tightklawe{}}
  if for all 
  $u \in U$ it holds that
  $u \notin \hiddenvertices{U \setminus \set{u}}$.  
\end{definition}

If $x$ is a vertex hidden by~$U$, 
we can identify a subset of $U$ that is necessary for
\hidingklawe~$x$. 

\begin{definition}[\Necessaryhidingsetklawe{}]
  \label{def:necessary-subcover}
  If
  $x \in \hiddenvertices{U}$,
  we define
  $\necessaryhidingvert{U}{x}$
  to be the subset of $U$ \st for each
  $u \in \necessaryhidingvert{U}{x}$
  there is a \sourcepath{} $P$  ending in $x$ for which
  $P \intersectionSP U = \set{u}$.
\end{definition}

We observe that if $U$ is \tightklawe and $u \in U$, then
$\necessaryhidingvert{U}{u} = \set{u}$. This is not the case for
non-\tightklawe sets.
If we let
$U = \set{u} \unionSP \prednode{u}$
for some non-source $u$,
\refdef{def:necessary-subcover}
yields that
$\necessaryhidingvert{U}{u} = \emptyset$.
The vertices in
$\necessaryhidingvert{U}{x}$
must be contained in every subset of $U$
that \hideklawe{}s~$x$, since for each 
$v \in \necessaryhidingvert{U}{x}$
there is a \sourcepath to $x$ that intersects $U$ only in~$v$.
But if
$U$ is \tightklawe{}, the set
$\necessaryhidingvert{U}{x}$
is also \emph{sufficient} to \hideklawe{}~$x$,
\ie
$x \in \hiddenvertices{\necessaryhidingvert{U}{x}}$.

\begin{lemma}[Lemma 3.1 in
\cite{K80TightBoundPebblesPyramid}]
  \label{lem:klawe-lemma-three-one}
  If
  $U$ is \tightklawe{}
  and
  $x \in \hiddenvertices{U}$,
  then
  $\necessaryhidingvert{U}{x}$
  \hideklawe{}s $x$
  and  this set is also 
  contained in every subset of $U$ that \hideklawe{}s $x$.
\end{lemma}

\begin{proof}
  The necessity was argued above, so
  the interesting part is that
  $x \in \hiddenvertices{\necessaryhidingvert{U}{x}}$.
  Suppose not.
  Let $P_1$ be a \sourcepath{} to $x$ \st
  $P_1 \intersectionSP
  {\necessaryhidingvert{U}{x}}
  = \emptyset$.
  Since $U$ \hideklawe{}s $x$,
  $U$ blocks~$P_1$. Let $v$ be the highest-level element in 
  $P_1 \intersectionSP U$ 
  (\ie, the vertex on this path closest to $x$).
  Since $U$ is \tightklawe{}, 
  $U \setminus \set{v}$ does not \hideklawe{}~$v$. Let
  $P_2$ be a \sourcepath{} to $v$ \st 
  $P_2 \intersectionSP (U \setminus \set{v}) = \emptyset$.
  Then going first along $P_2$ and switching to $P_1 $in $v$ we 
  get a path to $x$ that intersects $U$ only in $v$.
  But if so, we have $v \in {\necessaryhidingvert{U}{x}}$ contrary to
  assumption. 
  Thus, $x \in \hiddenvertices{\necessaryhidingvert{U}{x}}$ must hold.
\end{proof}

Given a vertex set $U$, the \tightklawe subset of $U$ \hidingklawe the
same elements is uniquely determined.

\begin{lemma}
\label{lem:covering-set-is-tight}
For any vertex set $U$ in a layered graph $G$ 
there is a uniquely determined minimal subset
$U^* \subseteq U$
\st
$
\hiddenvertices{U^*}
=
\hiddenvertices{U}
$,
$U^*$ is \tightklawe{},
and for any
$U' \subseteq U$
with
$
\hiddenvertices{U'}
=
\hiddenvertices{U}
$
it holds that
$U^* \subseteq U'$.
\end{lemma}

\newcommand{\utighti}[1][i]{U^*_{#1}}
\newcommand{\uremainingi}[1][i]{U^r_{#1}}
\newcommand{\ufirstlevel}{L}
\newcommand{\umaximallevel}{M}

\begin{proof}
We construct the set $U^*$ bottom-up, layer by layer.
We will let
$\utighti$
be the set of vertices on level~$i$ or lower
in the \tightklawe \hidingsetklawe under construction, and
$\uremainingi$
be the set of vertices in $U$
strictly above level$~i$ remaining to be \hiddenklawe.

Let 
$\ufirstlevel = \vminlevel{U}$.
For 
$i < \ufirstlevel$,
we define
$\utighti = \emptyset$.
Clearly, all vertices on level $\ufirstlevel$ in $U$ must be present also 
in~$U^*$,
since no vertices in
$\vertstrictlyabovelevel{U}{\ufirstlevel}$
can \hideklawe these vertices and vertices on the same level cannot
help \hidingklawe each other. Set
$
\utighti[\ufirstlevel] = 
\vertonlevel{U}{\ufirstlevel}
=
U \setminus \vertstrictlyabovelevel{U}{\ufirstlevel}
$.
Now we can remove from $U$ all vertices \hiddenklawe by
$\utighti[\ufirstlevel]$,
so set
$\uremainingi[\ufirstlevel] = 
U \setminus \hiddenvertices{\utighti[\ufirstlevel]}$.
Note that there are no vertices on or below level~$\ufirstlevel$
left in $\uremainingi[\ufirstlevel]$, \ie
$
\uremainingi[\ufirstlevel] =
\vertstrictlyabovelevel{\uremainingi[\ufirstlevel]}{\ufirstlevel}
$, 
and that
$\utighti[\ufirstlevel]$
\hideklawe{}s the same vertices as does
$\vertbelowlevel{U}{\ufirstlevel}$
(since the two sets are equal).

Inductively, suppose we have constructed the vertex sets
$\utighti[i-1]$
and
$\uremainingi[i-1]$.
Just as above, set
$\utighti[i] = \utighti[i-1] \unionSP \vertonlevel{\uremainingi[i-1]}{i}$
and
$\uremainingi[i] = \uremainingi[i-1] \setminus \hiddenvertices{\utighti[i]}$.
If there are no vertices remaining on level~$i$ to be \hiddenklawe,
\ie if 
$\vertonlevel{\uremainingi[i-1]}{i} = \emptyset$,
nothing happens and we get
$\utighti[i] = \utighti[i-1]$
and
$\uremainingi[i] = \uremainingi[i-1]$.
Otherwise the vertices on level~$i$ in $\uremainingi[i-1]$
are added to $\utighti[i]$
and all of these vertices, as well as any vertices above in  
$\uremainingi[i-1]$
now being \hiddenklawe, are removed  from $\uremainingi[i-1]$ resulting
in a smaller set~$\uremainingi[i]$.

To conclude, we set
$U^* = \utighti[\umaximallevel]$
for
$\umaximallevel = \vmaxlevel{U}$.
By construction, the invariant
\begin{equation}
  \label{eq:invariant-for-lemma-covering-set-is-tight}
  \hiddenvertices{\utighti[i]}
  =
  \hiddenvertices{\vertbelowlevel{U}{i}}
\end{equation}
holds for all levels~$i$.
Thus,
$
\hiddenvertices{U^*}
=
\hiddenvertices{U}
$.
Also, 
$U^*$ must be \tightklawe{}
since if
$v \in U^*$ 
and 
$\vlevel{v} = i$,
by construction
$
{\vertstrictlybelowlevel{U^*}{i}}
$
does not \hideklawe{}~$v$, and (as was argued above) neither does
$\vertabovelevel{U^*}{i} \setminus \set{v}$.
Finally, suppose that
$U' \subseteq U$ is a \hidingsetklawe for $U$
with
$U^* \nsubseteq U'$.
Consider
$v \in U^* \setminus U'$
and suppose
\mbox{$\vlevel{v} = i$}.
On the one hand, we have
$v \notin \hiddenvertices{\utighti[i-1]}$
by construction. On the other hand, by assumption it holds that
$v \in \hiddenvertices{\vertstrictlybelowlevel{U'}{i}}$
and thus
$v \in 
\hiddenvertices{\vertstrictlybelowlevel{U}{i}}
$.
But then by the invariant~%
\refeq{eq:invariant-for-lemma-covering-set-is-tight}
we know that
$v \in \hiddenvertices{\utighti[i-1]}$,
which yields a contradiction. Hence,
$U^* \subseteq U'$
and the lemma follows.
\end{proof}

We remark that $U^*$ can in fact be seen to
contain exactly those elements $u \in U$ such that
$u$ is not \hiddenklawe by $U \setminus \set{u}$.

It follows 
from \reflem{lem:covering-set-is-tight}
that if
$U$
is a \minmeasure \hidingsetklawe{} for
$\pconf = (B,W)$,
we can assume \wolog that
$U \unionSP W$ is 
\tightklawe{}.
More formally, if
$U \unionSP W$
is not
\tightklawe{},
we can consider minimal subsets
$U' \subseteq U$ and $W' \subseteq W$
\st
$U' \unionSP W'$ \hideklawe{}s $B$ and
is \tightklawe{}, and prove the \klawepropacronym for
$B$ and $W'$ \wrt this $U'$ instead.
Then clearly the \klawepropacronym holds also for~$B$ and~$W$. 

Suppose that we have a set $U$ that together with $W$ \hideklawe{}s
$B$. Suppose furthermore that $B$ contains vertices very far apart in
the graph. Then it might very well be the case that
$U \unionSP W$ 
can be split into a number of disjoint subsets
$U_i \unionSP W_i$ 
responsible for \hidingklawe different parts $B_i$ of~$B$,
but which are wholly independent of one another.
Let us give an example of this.

\begin{figure}[tp]
  \centering
  \subfigure[\Hidingsetklawe
  $U$ 
  with large size and measure.
  ]
  {
    \label{fig:hidingsets-a}
    \begin{minipage}[b]{.48\linewidth}
      \centering
      \includegraphics{hidingsets.1}%
    \end{minipage}
  }
  \hfill
  \subfigure[Smaller \hidingsetklawe 
  $U^*$ 
  with smaller measure.]
  {
    \label{fig:hidingsets-b}
    \begin{minipage}[b]{.48\linewidth}
      \centering
      \includegraphics{hidingsets.2}%
    \end{minipage}
  }
%
%
%
%
  \caption{Illustration of \hidingsetklawe{}s
    in
    \refex{ex:why-we-need-connected-components}
    (with vertices in \hidingsetklawe{}s cross-marked).}
  \label{fig:hidingsets}
\end{figure}

\begin{example}
  \label{ex:why-we-need-connected-components}
  Suppose  we have the pebble configuration
  $
  (B,W) 
  =
  (\set{x_1, y_1, v_5}, \set{w_3, s_6, s_7})
  $
  and the \hidingsetklawe
  $U = \set{v_1, u_2, u_3, v_3, s_5}$
  in
  \reffig{fig:hidingsets-a}.
  Then   $U \unionSP W$  \hideklawe{}s $B$, but  $U$
  seems unnecessarily large. 
  To get a better \hidingsetklawe $U^*$, we can leave $s_5$
  responsible for   \hidingklawe{}~$v_5$   but replace
  $\set{v_1, u_2, u_3, v_3}$
  by
  $\set{x_1, y_1}$.
  The resulting set
  $U^* = \set{x_1, y_1, s_5}$
  in
  \reffig{fig:hidingsets-b}
  has both smaller size and smaller measure
  (we leave the straightforward verification of this fact to the reader). 

  Intuitively, it seems that the configuration can be
  split in two components, namely
  $
  (B_1,W_1) =
  (\set{x_1, y_1},  \set{w_3})
  $ 
  with \hidingsetklawe
 $U_1 = \set{v_1, u_2, u_3, v_3}$
  and
  $
  (B_2,W_2) =
  (\set{v_5},\set{s_6, s_7})
  $
  with \hidingsetklawe
  $U_2 = \set{s_5}$,
  and that these two components are independent of one another.
  To improve the \hidingsetklawe~$U$, we need to do something
  locally about the bad
  \hidingsetklawe~$U_1$ in the first component, 
  namely replace it with
  $U^*_1 = \set{x_1, y_1}$,
  but we should keep the locally optimal \hidingsetklawe~$U_2$ in
  the second component.
\end{example}

We want to formalize this understanding of how vertices in
$B$, $W$ and $U$ depend on one another in a
\hidingsetklawe $U \unionSP W$ for~$B$.
The following definition constructs a graph 
that describes the structure of the 
\hidingsetklawe{}s 
that we are studying in terms of these dependencies.

\begin{definition}[\Hidingsetklawe{} graph]
  \label{def:covering-set-graph}
  For a \tightklawe{} (and non-empty) set of vertices $X$ in $G$, the 
  \introduceterm{\hidingsetklawe{} graph}
  $\hidsetgraph = \hidsetgraph(G,X)$
  is
  an undirected graph 
  defined as follows:
  \begin{itemize}
  \item 
    The set of vertices of $\hidsetgraph$ is
    $\vertices{\hidsetgraph}
    = \hiddenvertices{X}$.
    
  \item 
    The set of edges $\edges{\hidsetgraph}$ 
    of $\hidsetgraph$
    consists of all pairs of vertices $(x,y)$
    for  $x,y \in  \hiddenvertices{X}$
    \st 
%
%
    $
    \belowvertices{x} 
    \intersectionSP
    \hiddenvertices{\necessaryhidingvert{X}{x}}
    \intersectionSP
    \belowvertices{y}
    \intersectionSP
    \hiddenvertices{\necessaryhidingvert{X}{y}}
    \!\neq\! 
    \emptyset
    $.
  \end{itemize}
  We say that   
  the vertex set $X$ is 
  \introduceterm{\connectedklawe{}}
  if
  $\hidsetgraph(G,X)$
  is a connected graph.
\end{definition}

When the graph $G$ and vertex set $X$ are clear from context, 
we will sometimes write only $\hidsetgraph(X)$ or even
just~$\hidsetgraph$.   To illustrate
\refdef{def:covering-set-graph}, 
we give an example.

\begin{example}
  \label{ex:hiding-set-graph}
  Consider again the pebble configuration
  $
  (B,W) =
  (\set{x_1, y_1, v_5},\set{w_3, s_6, s_7})
  $
  from
  \refex{ex:why-we-need-connected-components}
  with \hidingsetklawe
  $U = \set{v_1, u_2, u_3, v_3, s_5}$,
  where we have shaded the set of \hiddenklawe vertices in
  \reffig{fig:hidingsetgraph-a}.
  The \hidingsetklawe graph
  $\hidsetgraph(X)$
  for
  $
  X = U \unionSP W  =
  \set{v_1, u_2, u_3, v_3, w_3, s_5, s_6, s_7}
  $
  has been drawn in
  \reffig{fig:hidingsetgraph-b}.
  In accordance with the intuition sketched in
  \refex{ex:why-we-need-connected-components},
  $\hidsetgraph(X)$
  consists of two connected components.

  Note that there are edges from the top vertex $y_1$ in the first
  component to every other vertex in this component and from the top
  vertex $v_5$ to every other vertex in the second 
  component. We will prove presently that this is always the case
  (\reflem{lem:vertex-and-necessary-cover-in-same-component}).
  Perhaps a more interesting edge in $\hidsetgraph(X)$ is, \eg, 
  $(w_1, x_2)$.
  This edge exists since
  $\necessaryhidingvert{X}{w_1} 
  = 
  \set{v_1, u_2, u_3}$
  and
  $\necessaryhidingvert{X}{x_2} = \set{u_2, u_3, v_3, w_3}$
  intersect and since as a consequence of this (which is easily
  verified) we have
  $
  \belowverticespyramid{w_1}
  \intersectionSP
  \hiddenvertices{\necessaryhidingvert{X}{w_1}}
  \intersectionSP
  \belowverticespyramid{x_2}
  \intersectionSP
  \hiddenvertices{\necessaryhidingvert{X}{x_2}}
  \neq \emptyset
  $.
%
%
  For the same reason, there is an edge $(u_5, u_6)$ since
  $\necessaryhidingvert{X}{u_5} = \set{s_5, s_6}$
  and
  $\necessaryhidingvert{X}{u_6} = \set{s_6, s_7}$
  intersect.
\end{example}

\begin{figure}[tp]
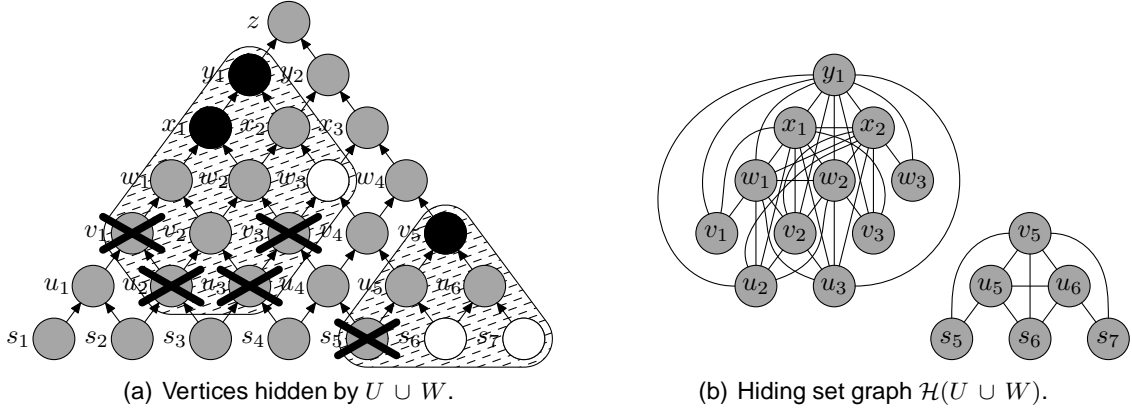

  \centering
  \subfigure[Vertices \hiddenklawe by $U \unionSP W$.]
  {
    \label{fig:hidingsetgraph-a}
    \begin{minipage}[b]{.50\linewidth}
      \centering
      \includegraphics{hidingsets.7}%
    \end{minipage}
  }
  \hfill
  \subfigure[\Hidingsetklawe graph $\hidsetgraph(U \unionSP W)$.]
  {
    \label{fig:hidingsetgraph-b}
    \begin{minipage}[b]{.46\linewidth}
      \centering
      \includegraphics{hidingsets.4}%
    \end{minipage}
  }
  \caption{Pebble configuration with \hidingsetklawe and
    corresponding  \hidingsetklawe graph.}
  \label{fig:hidingsetgraph}
\end{figure}

\begin{lemma}
  \label{lem:vertex-and-necessary-cover-in-same-component}
  Suppose for a \tightklawe vertex set
  $X$
  that
  $x \in \hiddenvertices{X}$
  and
  $y \in \necessaryhidingvert{X}{x}$.
  Then $x$ and $y$ are in the same connected component of
  $\hidsetgraph(X)$.
\end{lemma}

\begin{proof}
  Note first that
  $x,y \in \hiddenvertices{X}$
  by assumption, so $x$ and $y$ are both vertices
  in~$\hidsetgraph(X)$. 
  Since $x$ is above $y$ we have
  $\belowvertices{x} \supseteq \belowvertices{y}$
  and we get
  $
  \belowvertices{x} 
  \intersectionSP
  \hiddenvertices{\necessaryhidingvert{X}{x}}
  \intersectionSP
  \belowvertices{y} 
  \intersectionSP
  \hiddenvertices{\necessaryhidingvert{X}{y}}
  =
  \hiddenvertices{\necessaryhidingvert{X}{x}}
  \intersectionSP
  \belowvertices{y} 
  \intersectionSP
  \set{y}
  =
  \set{y}
  \neq \emptyset
  $.
  Thus,
  $(x,y)$
  is an edge in $\hidsetgraph(X)$, so $x$ and $y$ are certainly in the
  same connected component.
\end{proof}

\begin{corollary}
  \label{cor:vertex-and-necessary-cover-in-same-component}
  If $X$ is \tightklawe and
  $x \in \hiddenvertices{X}$
  then
  $x$ and all of $\necessaryhidingvert{X}{x}$
  are in the same connected component of $\hidsetgraph(X)$.
\end{corollary}

The next lemma says that if $\hidsetgraph(X)$ is a \hidingsetklawe
graph with vertex set 
$V = \hiddenvertices{X}$,
then the connected components
$V_1, \ldots, V_k$
of $\hidsetgraph(X)$ are themselves \hidingsetklawe graphs defined
over the  \connectedklawe subsets 
$X \intersectionSP V_1, \ldots, X \intersectionSP V_k$.

\begin{lemma}[Lemma 3.3 in
\cite{K80TightBoundPebblesPyramid}]
  \label{lem:klawe-lemma-three-three}
  Let $X$ be a \tightklawe set and let $V_i$ be one of the connected
  components in $\hidsetgraph(X)$.
  Then the subgraph of $\hidsetgraph(X)$ induced by $V_i$ is identical
  to the \hidingsetklawe graph
  $\hidsetgraph(X \intersectionSP V_i)$
  defined on the vertex subset
  $X \intersectionSP V_i$.
  In particular, it holds that
  $V_i = \hiddenvertices{X \intersectionSP V_i}$.
\end{lemma}

\begin{proof}
  We need to show that
  $V_i = \hiddenvertices{X \intersectionSP V_i}$
  and that the edges of
  $\hidsetgraph(X)$  in $V_i$ are exactly the edges in
  $\hidsetgraph(X \intersectionSP V_i)$.
  Let us first show that
  $y \in V_i$
  \ifaoif
  $y \in \hiddenvertices{X \intersectionSP V_i}$.

  ($\Rightarrow$)
  Suppose $y \in V_i$.
  Then
  $\necessaryhidingvert{X}{y} \subseteq V_i$
  by
  \refcor{cor:vertex-and-necessary-cover-in-same-component}.
  Also,
  $\necessaryhidingvert{X}{y} \subseteq X$
  by definition,
  so
  $\necessaryhidingvert{X}{y} \subseteq X \intersectionSP V_i$.
  Since
  $y \in \hiddenvertices{\necessaryhidingvert{X}{y}}$
  by
  \reflem{lem:klawe-lemma-three-one},
  clearly
  $y \in \hiddenvertices{X \intersectionSP V_i}$.
  
  ($\Leftarrow$)
  Suppose
  $y \in \hiddenvertices{X \intersectionSP V_i}$.
  Since $X$ is \tightklawe, its subset
  $X \intersectionSP V_i$
  must be \tightklawe as well.
  Applying 
  \reflem{lem:klawe-lemma-three-one}
  twice, we deduce that
  $\necessaryhidingvert{(X \intersectionSP V_i)}{y}$
  \hideklawe{}s $y$ and that
  $
  \necessaryhidingvert{X}{y} 
  \subseteq
  \necessaryhidingvert{(X \intersectionSP V_i)}{y}
  $
  since $\necessaryhidingvert{X}{y}$ is contained in any subset of $X$
  that \hideklawe{}s $y$.
  But then a third appeal to   \reflem{lem:klawe-lemma-three-one}
  yields that 
  $
  \necessaryhidingvert{(X \intersectionSP V_i)}{y}
  \subseteq
  \necessaryhidingvert{X}{y} 
  $
  since
  $
  \necessaryhidingvert{X}{y} 
  \subseteq
  \necessaryhidingvert{(X \intersectionSP V_i)}{y}
  \subseteq
  X \intersectionSP V_i
  $
  and consequently
  \begin{equation}
    \label{eq:necessary-hiding-subset-equalities}
    \necessaryhidingvert{X}{y} 
    =
    \necessaryhidingvert{(X \intersectionSP V_i)}{y}
    \eqperiod 
  \end{equation}
  By
  \refcor{cor:vertex-and-necessary-cover-in-same-component},
  $y$ and all of 
  $
  \necessaryhidingvert{(X \intersectionSP V_i)}{y}
  =
  \necessaryhidingvert{X}{y}
  $ 
  are in the same connected component.
  Since
  $\necessaryhidingvert{X}{y} \subseteq V_i$
  it follows that
  $y \in V_i$.
  
  This shows that
  $V_i = \hiddenvertices{X \intersectionSP V_i}$.
  Plugging
  \refeq{eq:necessary-hiding-subset-equalities}
  into
  \refdef{def:covering-set-graph},
  we see that $(x,y)$ is an edge in
  $\hidsetgraph(X)$  for $x, y \in V_i$ \ifaoif
  $(x,y)$ is an edge in  $\hidsetgraph(X \intersectionSP V_i)$.
\end{proof}

Now we are in a position to describe the structure of the proof that
pyramid graphs have the \klawepropacronym{}.

\begin{theorem}[Analogue of Theorem 3.7 in
  \cite{K80TightBoundPebblesPyramid}]
  \label{th:pyramids-have-klawe-property}
  Let
  $\pconf = (B,W)$ 
  be any black-white pebble configuration on a pyramid~%
  $\pyramidgraph$.
  Then there is a vertex set $U$ \st
  $U \unionSP W$ \hideklawe{}s $B$,
  $\vpotential[\pyramidgraph]{\pconf} =
  {\meastopot{U}}$
  and
  either
  $U = B$ or 
  $\setsize{U} < 
  \setsize{B} + \setsize{W}
  $.
\end{theorem}

The idea is to construct the graph
$\hidsetgraph = \hidsetgraph(\pyramidgraph,U \unionSP W)$,
study the different connected components in~%
$\hidsetgraph$,
find good \hidingsetklawe{}s locally that satisfy the \klawepropacronym 
(which we prove is true for each local
\connectedklawe subset of $U \unionSP W$),
and then add all of these
partial \hidingsetklawe{}s together to get a globally good \hidingsetklawe{}.

Unfortunately, this does not quite work.  Let us nevertheless attempt
to do the proof, note where and why it fails, and then see how Klawe
fixes the broken details.

\begin{proof}[Tentative proof of \refth{th:pyramids-have-klawe-property}]
  Let $U$ be   a set of vertices in $\pyramidgraph$  \st
  $U \unionSP W$ \hideklawe{}s $B$ and 
  $\vpotential{\pconf} = \vmeasure{U}$.
  Suppose that $U$ has minimal size among all such
  sets, and furthermore that among all such
  \minmeasure and
  \minsize sets $U$ has the largest intersection with~$B$.

  Assume \wolog 
  (\reflem{lem:covering-set-is-tight})
  that $U \unionSP W$ is tight, so that we can construct~%
  $\hidsetgraph$.
  Let the connected components of $\hidsetgraph$ be 
  $V_1, \ldots, V_k$.
  For all $i = 1, \ldots, k$, let
  $B_i = B \intersectionSP V_i$,
  $W_i = W \intersectionSP V_i$, and
  $U_i = U \intersectionSP V_i$.
  \Reflem{lem:klawe-lemma-three-three}
  says that
  $U_i \unionSP W_i$ \hideklawe{}s~$B_i$.
  In addition, all $V_i$ are pairwise disjoint, 
  so
  $\setsize{B} = \sum_{i=1}^{k} \setsize{B_i}$,
  $\setsize{W} = \sum_{i=1}^{k} \setsize{W_i}$
  and
  $\setsize{U} = \sum_{i=1}^{k} \setsize{U_i}$.

  Thus, if the \klawepropacronym{}~\ref{property:klawe-property}
  does not hold for $U$ globally, there is some 
  \connectedklawe subset
  $U_i \unionSP W_i$ that \hideklawe{}s $B_i$
  but for which
  $\setsize{U_i} \geq \setsize{B_i} + \setsize{W_i}$
  and
  $U_i \neq B_i$.
  Note that this implies that
  $B_i \nsubseteq U_i$
  since otherwise $U_i$ would not be minimal.

  Suppose that we would know that the \klawepropacronym
  is true for each connected component.
  Then we could find a vertex set
  $U_i^*$
  with
  $U_i^* \subseteq B_i$
  or
  $\Setsize{U_i^*} < \setsize{B_i} + \setsize{W_i}$
  \st
  $U_i^* \unionSP W_i$ \hideklawe{}s $B_i$
  and
  $\Vmeasure{U_i^*} \leq \vmeasure{U_i}$.
  Setting
  $  U^* = (U \setminus U_i) \unionSP U_i^* $,
  we would get a \hidingsetklawe{} with either
  $\setsize{U^*} < \setsize{U}$
  or
  $\setsize{U^* \intersectionSP B} > \setsize{U \intersectionSP B}$.
  The second inequality would hold since if
  $\setsize{U^*} = \setsize{U}$,
  then
  $\Setsize{U^*_i} = \setsize{U_i} \geq
  \setsize{B_i \unionSP W_i}$
  and this would imply
  $U_i^* = B_i$ 
  and thus
  $\Setsize{U_i^* \intersectionSP B_i} > 
  \setsize{U_i \intersectionSP B_i}$. 
  This would contradict how  $U$ was chosen above,
  and we would be home.

  Almost.
  We would also need that $U_i^*$ could be substituted for $U_i$ in
  $U$ without increasing the measure, \ie that
  $  
  \Vmeasure{U_i^*} \leq
  \Vmeasure{U_i}
  $  
  should imply
  $  
  \Vmeasure{(U \setminus U_i) \unionSP U_i^*}
  \leq 
  \Vmeasure{(U \setminus U_i) \unionSP U_i}
  $. 
  And  this turns out not to be true.
\end{proof}

The reason that the proof above does not quite work is that
the measure in 
\refdef{def:measure}
is ill-behaved \wrt unions. Klawe provides the following example of
what can happen.

\begin{example}
  \label{ex:union-of-measures-bad}
  With vertex labels as in
Figures~\ref{fig:pyramid-height-6}
and \mbox{\ref{fig:convergingpaths}--\ref{fig:hidingsetgraph}},
  let
  $X_1 = \set{s_1, s_2}$,
  $X_2 = \set{w_1}$
  and
  $X_3 = \set{s_3}$.
  Then
  $ \vmeasure{X_1} = 4$
  and
  $\vmeasure{X_2} = 5$
  but taking unions with $X_3$ we get that
  $\vmeasure{X_1 \unionSP X_3} = 6$
  and
  $\vmeasure{X_2 \unionSP X_3} = 5$.
  Thus
  $\vmeasure{X_1} < \vmeasure{X_2}$
  but
  $\vmeasure{X_1 \unionSP X_3} > \vmeasure{X_2 \unionSP X_3}$.
\end{example}

So it is not enough to show the \klawepropacronym locally for each connected
component in the graph. We also need that sets
$U_i$ from different components can be combined into a global 
\hidingsetklawe{} while maintaining measure inequalities.
This leads to the following strengthened condition for 
connected components of~$\hidsetgraph$.

\begin{property}[\Localklaweprop{}]
  \label{property:local-klawe-property}
  We say that the pebble configuration
  $\pconf = (B,W)$ 
  has the
  \introduceterm{\localklaweprop{}},
  or just the 
  \introduceterm{\localklawepropacronym{}} 
  for short,
  if  for any vertex set $U$ \st 
  $U \unionSP W$ \hideklawe{}s $B$
  and
  is \connectedklawe,
  we can find a vertex set $U^*$
  \st
  \begin{enumerate}
  \item
    $U^*$ is a \hidingsetklawe for $(B,W)$,
    \item
    for any vertex set
    $Y$
    with
    $Y \intersectionSP U = \emptyset$
    it holds that
    $
    \Vmeasure{Y \unionSP U^*}
    \leq
    \vmeasure{Y \unionSP U}
    $,
  \item
    $U^* \subseteq B$ or
    $\Setsize{U^*} < \setsize{B} + \setsize{W}$.
  \end{enumerate}
  We say that the graph $G$ has the 
  \localklawepropacronym{}
  if all black-white pebble configurations
  $\pconf = (B,W)$ on $G$ do.
\end{property}

Note that if the \localklawepropacronym holds, this in particular
implies that  
$\Vmeasure{U^*} \leq \vmeasure{U}$
(just choose $Y = \emptyset$).
Also, we immediately get that the
\klawepropacronym holds globally.

\begin{lemma}
  \label{lem:lkp-implies-global-klawe-property}
  If $G$ has the 
  \localklaweprop~\ref{property:local-klawe-property}, 
  then $G$ has the 
  \klaweprop~\ref{property:klawe-property}.
\end{lemma}

\begin{proof}
  Consider the tentative proof of
  \refth{th:pyramids-have-klawe-property}
  and look at the point where it breaks down.
  If we instead use the \localklawepropacronym to find
  $U_i^*$, 
  this time we get that
  $  
  \Vmeasure{U_i^*} \leq
  \Vmeasure{U_i}
  $  
  does indeed imply
  $  
  \Vmeasure{(U \setminus U_i) \unionSP U_i^*}
  \leq 
  \Vmeasure{(U \setminus U_i) \unionSP U_i}
  $,
  and the theorem follows.
\end{proof}

An obvious way to get the inequality
$\vmeasure{Y \unionSP U^*} \leq \vmeasure{Y \unionSP U}$
in
\refproperty{property:local-klawe-property}
would be to require that
$\vjthmeasure{j}{U^*} \leq \vjthmeasure{j}{U}$
for all~$j$, but we need to be slightly more general.
The next definition identifies a sufficient condition for 
sets to behave well under unions with respect to the measure in
\refdef{def:measure}.

\begin{definition}
  \label{def:union-respecting-leq}
  We write
  $U \measureleq V$
  if for all
  $j \geq 0$
  there is an $i\leq j$
  \st
  $
  \vjthmeasure{j}{U}
  \leq
  \vjthmeasure{i}{V}
  $.
\end{definition}

Note 
that it is sufficient to verify the condition in
\refdef{def:union-respecting-leq}
for $j = 1, \ldots, \vmaxlevel{U}$.
For $j > \vmaxlevel{U}$
we get
$\vjthmeasure{j}{U} = 0$
and the inequality trivially holds.

It is immediate that
$U \measureleq V$
implies
$\vmeasure{U} \leq \vmeasure{V}$,
but the relation 
$\measureleq$ 
gives us more information than that. 
Usual inequality
$\vmeasure{U} \leq \vmeasure{V}$
holds \ifaoif for every $j$ we can find an  $i$ \st
$\vjthmeasure{j}{U} \leq \vjthmeasure{i}{V}$,
but in the definition of $\measureleq$
we are restricted to finding such an index $i$ that is less than or
equal to~$j$.
So not only is
$\vmeasure{U} \leq \vmeasure{V}$
globally,
but we can also explain locally at each level, by ``looking
downwards'', why
$U$ has smaller measure than $V$.

In
\refex{ex:union-of-measures-bad},
$X_1 \not \measureleq X_2$ since
the relative cheapness of $X_1$ compared to $X_2$
is explained not by a lot of vertices
in $X_2$ on low levels, 
but by one single high-level, and therefore expensive,
vertex in $X_2$ which is far above $X_1$. This is why these sets
behave badly under union. If we have two sets $X_1$ and $X_2$ with
$X_1 \measureleq X_2$,
however, reversals of measure inequalities when taking unions as in 
\refex{ex:union-of-measures-bad}
can no longer occur.

\begin{lemma}[Lemma 3.4 in
\cite{K80TightBoundPebblesPyramid}]
  \label{lem:union-respecting-leq}
  If
  $U \measureleq V$
  and
  $Y \intersectionSP V =   \emptyset$,
  then
  $\vmeasure{Y \unionSP U} \leq \vmeasure{Y \unionSP V}$.
\end{lemma}

\begin{proof}
  To show that
  $\vmeasure{Y \unionSP U} \leq  \vmeasure{Y \unionSP V}$,
  for each level
  $j = 1, \ldots, \vmaxlevel{Y \unionSP U}$  
  we want to find a level $i$ \st
  $
  \vjthmeasure{j}{Y \unionSP U}
  \leq
  \vjthmeasure{i}{Y \unionSP V}
  $.
  We pick the $i \leq j$  provided by the definition of 
  $U \measureleq V$
  such that
  $
  \vjthmeasure{j}{U}
  \leq
  \vjthmeasure{i}{V}
  $.
  Since
  $V \intersectionSP W = \emptyset$
  and
  $i \leq j$ implies
  $\vertabovelevel{Y}{j} \subseteq \vertabovelevel{Y}{i}$,
  we get
  \begin{multline}
    \vjthmeasure{j}{Y \unionSP U}  
    =
    j 
    + 2 \cdot  
    \setsize{\vertabovelevel{(U \unionSP Y)}{j}}
    \leq
    j 
    + 2 \cdot  \setsize{\vertabovelevel{U}{j}}
    + 2 \cdot  \setsize{\vertabovelevel{Y}{j}}
    \leq
    \\
    i 
    + 2 \cdot   \setsize{\vertabovelevel{V}{i}}
    + 2 \cdot   \setsize{\vertabovelevel{Y}{i}}
    =
    \vjthmeasure{i}{Y \unionSP V}  
  \end{multline}
  and the lemma follows.
\end{proof}

So when locally improving a blocking set $U$ that does not satisfy
the \klawepropacronym to some set $U^*$ that does, if we can take
care that $U^* \measureleq U$ in the sense of
\refdef{def:union-respecting-leq}
we get the \localklawepropacronym.
All that remains is to show that this can indeed be done.

When ``improving'' $U$ to $U^*$, we will strive to pick
\hidingsetklawe{}s of minimal size. The next definition makes this
precise.  

\begin{definition}
  \label{def:size-above-level-blocker}
  For any set of vertices $X$, let
  \begin{equation*}
    \abovelevelblockerminsize{j}{X}
    =
    \minofset
    {\setsize{Y}}
    {\text{$\vertabovelevel{X}{j} \subseteq \hiddenvertices{Y}$
        and
        $\vertabovelevel{Y}{j} = Y$}}
  \end{equation*}
  denote the size of a smallest set $Y$ such that all vertices in $Y$
  are on level $j$ or higher and $Y$ \hideklawe{}s all vertices in $X$
  on level   $j$ or higher.
\end{definition}

Note that we only require of $Y$ to \hideklawe
$\vertabovelevel{X}{j}$ and not all of~$X$.
Given the condition that
$Y = \vertabovelevel{Y}{j}$, this set cannot \hideklawe{} any vertices
in $\vertstrictlybelowlevel{X}{j}$.
We make a few easy observations.

\begin{observation}
  \label{obs:about-min-size-covers}
  Suppose that $X$ is a set of vertices in a layered graph $G$. Then:
  \begin{enumerate}
    \item
      \label{item:about-min-size-covers-one}
      $\abovelevelblockerminsize{0}{X}$
      is the minimal size of any \hidingsetklawe{} for~$X$.
    \item
      \label{item:about-min-size-covers-two}
      If
      $X \subseteq Y$,
      then
      $
      \abovelevelblockerminsize{j}{X}
      \leq
      \abovelevelblockerminsize{j}{Y}
      $
      for all~$j$.
    \item
      \label{item:about-min-size-covers-three}
      It always holds that
      $
      \abovelevelblockerminsize{j}{X}
      \leq
      \setsize{\vertabovelevel{X}{j}}
      \leq
      \setsize{X}
      $.
  \end{enumerate}
\end{observation}

\begin{proof}
  \Refpart{item:about-min-size-covers-one}
  follows from the fact that
  $\vertabovelevel{V}{0} = V$
  for any set $V$.
  If
  $X \subseteq Y$,
  then
  $\vertabovelevel{X}{j} \subseteq \vertabovelevel{Y}{j}$
  and any \hidingsetklawe{} for
  $\vertabovelevel{X}{j}$ works also for $\vertabovelevel{Y}{j}$,
  which yields
  \refpart{item:about-min-size-covers-two}.  
  \Refpart{item:about-min-size-covers-three}
  holds since
  $\vertabovelevel{X}{j} \subseteq X$
  is always a possible \hidingsetklawe{} for itself.
\end{proof}

For any vertex set $V$ in any layered graph $G$,
we can always find a set \hidingklawe{} $V$ that
has ``minimal cardinality at each level''
in the sense of
\refdef{def:size-above-level-blocker}.

\begin{lemma}[Lemma 3.5 in
\cite{K80TightBoundPebblesPyramid}]
  \label{lem:klawe-lemma-three-five}
  For any vertex set $V$
  we can find a \hidingsetklawe{} $V^*$
  \st
  $
  \Setsize{\vertabovelevel{V^*}{j}}
  \leq
  \abovelevelblockerminsize{j}{V}
  $
  for all $j$,
  and either
  $V^* = V$
  or
  $\setsize{V^*} <   \setsize{V}$.
\end{lemma}

\begin{proof}
  If
  $
  \setsize{\vertabovelevel{V}{j}}
  \leq
  \abovelevelblockerminsize{j}{V}
  $
  for all $j$,
  we can choose
  $V^* = V$.
  Suppose this is not the case, and let $k$ be minimal \st
  $
  \setsize{\vertabovelevel{V}{k}}
  >
  \abovelevelblockerminsize{k}{V}
  $.
  Let
  $V'$
  be a \minsize \hidingsetklawe{} for
  $\vertabovelevel{V}{k}$
  with
  $V' = \vertabovelevel{V'}{k}$
  and
  $\Setsize{V'} = 
  \setsize{\abovelevelblockerminsize{k}{V}}$
  and set
  $
  V^* = 
  \vertstrictlybelowlevel{V}{k} \disjointunionSP V'$.
  Since
  $\vertstrictlybelowlevel{V}{k}$ 
  \hideklawe{}s itself (any set does),
  we have that 
  $V^*$ \hideklawe{}s 
  $V
  =
  \vertstrictlybelowlevel{V}{k}
  \disjointunionSP
  \vertabovelevel{V}{k}
  $ 
  and that
  \begin{equation}
    \label{eq:size-of-smaller-cover}
    \Setsize{V^*}
    =
    \setsize{\vertstrictlybelowlevel{V}{k}} 
    +
    \setsize{V'}
    <
    \setsize{\vertstrictlybelowlevel{V}{k}} 
    +
    \setsize{\vertabovelevel{V}{k}}
    =
    \setsize{V}
    \eqperiod
  \end{equation}
  Combining
  \refeq{eq:size-of-smaller-cover}
  with
  \refpart{item:about-min-size-covers-one}
  of
  \refobs{obs:about-min-size-covers},
  we see that the minimal index found above must be
  $k = 0$.
  Going through the same argument as above again, we see that
  $
  \Setsize{\vertabovelevel{V^*}{j}}
  \leq
  \abovelevelblockerminsize{j}{V}
  $
  for all $j$,
  since otherwise
  \refeq{eq:size-of-smaller-cover}
  would yield a contradiction to the fact that
  $
  V'  = \vertabovelevel{V'}{0}
  $
  was chosen as a \minsize \hidingsetklawe{} for~$V$.
\end{proof}

We noted above that 
$\abovelevelblockerminsize{0}{X}$
is the cardinality of a \minsize \hidingsetklawe{} of $X$.
For $j > 0$,
the quantity
$\abovelevelblockerminsize{j}{X}$
is large if one needs many vertices
on level $\geq j$
to \hideklawe{} 
$\vertabovelevel{X}{j}$,
\ie if 
$\vertabovelevel{X}{j}$
is ``spread out'' in some sense.
Let us consider a pyramid graph and suppose that
$X$ is a  \tightklawe and \connectedklawe{} set
in which the level-difference
$\vmaxlevel{X} - \vminlevel{X}$
is large.
Then it seems that
$\setsize{X}$ 
should also have to be large, since the pyramid ``fans out'' so
quickly. 
This intuition might be helpful when looking at the next, crucial
definition of Klawe.

\begin{definition}[Spreading graph]
  \label{def:spreading-graph}
  We say that the layered DAG $G$ is a 
  \introduceterm{spreading graph}
  if for every (non-empty) \connectedklawe{} set $X$ in $G$
  and every level
  $j = 1, \ldots, \vmaxlevel{\hiddenvertices{X}}$,
  the \introduceterm{spreading inequality}
  \begin{equation}
    \label{eq:spreading-property}
    \setsize{X}
    \geq
    \abovelevelblockerminsize{j}{\hiddenvertices{X}}
    + j
    - \vminlevel{X}
  \end{equation}
  holds.
\end{definition}

Let us try to give some more intuition for 
\refdef{def:spreading-graph}
by considering two extreme cases in a pyramid graph:
\begin{itemize}
  \item
    For
    $j \leq \vminlevel{X}$,
    we have that     the term
    $j - \vminlevel{X}$
    is non-positive,
    $\vertabovelevel{X}{j} = X$,
    and
    $\hiddenvertices{\vertabovelevel{X}{j}} = \hiddenvertices{X}$.
    In this case,
    \refeq{eq:spreading-property}
    is just the trivial fact that no set 
    that \hideklawe{}s 
    $\hiddenvertices{X}$
    need be larger than $X$ itself.

  \item
    Consider
    $j = \vmaxlevel{\hiddenvertices{X}}$,
    and suppose that
    $\hiddenvertices{\vertabovelevel{X}{j}}$
    is a single vertex $v$ with
    $\necessaryhidingvert{X}{x} = X$.
    Then
    \refeq{eq:spreading-property}
    requires that
    $
    \setsize{X} \geq 1 + \vlevel{x} - \vminlevel{X}
    $,
    and this can be proven to hold by 
    the ``converging paths'' argument of 
    \refth{th:bounds-black-pebbling-of-pyramids}
    and
    \refobs{obs:converging-paths}.
\end{itemize}
Very loosely, 
\refdef{def:spreading-graph}
says that if $X$ contains vertices at low levels that help to
\hideklawe{} other vertices at high levels, then $X$ must be a large
set.
Just as we tried to argue above, the spreading inequality~%
\refeq{eq:spreading-property}
does indeed hold for pyramids. 

\begin{theorem}[\cite{K80TightBoundPebblesPyramid}]
\label{th:pyramids-are-spreading-graphs}
  Pyramids are spreading graphs.
\end{theorem}

Unfortunately, the proof of 
\refth{th:pyramids-are-spreading-graphs}
in~%
\cite{K80TightBoundPebblesPyramid}
is rather involved.
The analysis is divided into two parts, by first showing that a class
of so-called \introduceterm{nice graphs} are spreading, and then
demonstrating that pyramid graphs are nice.
In 
\refsec{sec:klawe-bw-pebbling-pyramids-are-spreading}, 
we give a simplified, direct proof of the fact that pyramids are
spreading that might be of independent interest.

Accepting
\refth{th:pyramids-are-spreading-graphs}
on faith for now, we are ready for the decisive lemma:
If our layered DAG is a spreading graph and if
$U \unionSP W$
is a \connectedklawe set \hidingklawe{} $B$
such that $U$ is too large for the conditions in the
\localklaweprop~\ref{property:local-klawe-property} 
to hold, then replacing $U$ by the \minsize \hidingsetklawe{} in
\reflem{lem:klawe-lemma-three-five}
we get a \hidingsetklawe{} in accordance with
the \localklawepropacronym.

\begin{lemma}[Lemma 3.6 in
\cite{K80TightBoundPebblesPyramid}]
  \label{lem:pick-local-good-blocker}
  Suppose that
  $B,W,U$
  are vertex sets in a layered spreading graph~$G$
  \st
  $U \unionSP W$
  \hideklawe{}s~$B$
  and is
  \tightklawe and \connectedklawe.
  Then there is a vertex set $U^*$
  \st
  $U^* \unionSP W$ \hideklawe{}s $B$,
  $U^* \measureleq U$,
  and either
  $U^* = B$
  or
  $\setsize{U^*}
  < \setsize{B} +  \setsize{W}
  $.
\end{lemma}

Postponing the proof of
\reflem{lem:pick-local-good-blocker}
for a moment, let us note that if we combine this lemma with
\reflem{lem:union-respecting-leq}
and
\refth{th:pyramids-are-spreading-graphs},
the \localklaweprop for pyramids follows.

\begin{corollary}
  \label{cor:pyramids-have-lkp}
  Pyramid graphs have the
  \localklaweprop
  \ref{property:local-klawe-property}.
\end{corollary}

\begin{proof}[Proof of \refcor{cor:pyramids-have-lkp}]
  This is more or less immediate, but we write down the details for
  completeness. Since pyramids are spreading by
  \refth{th:pyramids-are-spreading-graphs},
  \reflem{lem:pick-local-good-blocker} says that $U^*$ is a
  \hidingsetklawe for $(B,W)$ and that
  $U^* \measureleq U$.
  \Reflem{lem:union-respecting-leq}
  then yields that
  $
  \vmeasure{Y \unionSP U^*}
  \leq
  \vmeasure{Y \unionSP U}
  $
  for all $Y$ with $ Y\intersectionSP U = \emptyset$.
  Finally,
  \reflem{lem:pick-local-good-blocker} 
  also tells us that
  $U^* \subseteq B$
  or
  $\setsize{U^*} < \setsize{B} + \setsize{W}$,
  and thus all conditions in
  \refproperty{property:local-klawe-property}
  are satisfied.
\end{proof}

Continuing by plugging 
\refcor{cor:pyramids-have-lkp} 
into
\reflem{lem:lkp-implies-global-klawe-property},
we get the global \klawepropacronym in
\refthP{th:pyramids-have-klawe-property}.
So all that is needed to conclude Klawe's proof of the lower bound for
the black-white pebbling price of pyramids is to prove 
\refth{th:pyramids-are-spreading-graphs}
and
\reflem{lem:pick-local-good-blocker}.
We attend to 
\reflem{lem:pick-local-good-blocker}
right away, deferring a proof of
\refth{th:pyramids-are-spreading-graphs}
to the next \subsectionKlawe.

\begin{proof}[Proof of
    \reflem{lem:pick-local-good-blocker}]
  If
  $\setsize{U} <
  \setsize{B} + \setsize{W}$
  we can pick 
  $U^* = U$
  and be done, so suppose that
  $\setsize{U} \geq
  \setsize{B} + \setsize{W}$.
  Intuitively, this should mean that $U$ is unnecessarily large,
  so it ought to be possible to do better.
  In fact, $U$ is so large that we can just ignore $W$
  and pick a better $U^*$ that \hideklawe{}s $B$ all on its own.

  Namely, let $U^*$ be a \minsize \hidingsetklawe{} for~$B$ as in
  \reflem{lem:klawe-lemma-three-five}.
  Then either
  $U^* = B$
  or
  $\Setsize{U^*}
  <
  \setsize{B}
  \leq
  \setsize{B} + \setsize{W}
  $.
  To prove the lemma, we also  need to show that
  $U^* \measureleq U$,
  which will guarantee that $U^*$ behaves well under union with
  other sets \wrt measure.

  Before we do the the formal calculations, let us try to provide some
  intuition for why it should be the case that   
  $U^* \measureleq U$ holds,
  \ie that for every $j$ we can find an $i \leq j$ \st
  $
  \Vjthmeasure{j}{U^*} \leq
  \vjthmeasure{i}{U}
  $.
  Perhaps it will be helpful at this point for the reader to look at
  \refex{ex:why-we-need-connected-components} again,
  where the replacement of 
  $U_1 = \set{v_1, u_2, u_3, v_3}$ 
  in
  \reffig{fig:hidingsets-a}
  by 
  $U^*_1 = \set{x_1, y_1}$ 
  in
  \reffig{fig:hidingsets-b}
  shows
  \reftwolems
  {lem:klawe-lemma-three-five}
  {lem:pick-local-good-blocker}
  in action.

  Suppose first that
  $j \leq \vminlevel{U \unionSP W} \leq \vminlevel{U}$.
  Then the measure inequality 
  $\vjthmeasure{j}{U^*} \leq   \vjthmeasure{j}{U}$
  is obvious, since
  $\vertabovelevel{U}{j} = U$
  is so large that it can easily pay for all of~$U^*$,
  let alone 
  $\vertabovelevel{U^*}{j} \subseteq U^*$.

  For
  $j > \vminlevel{U \unionSP W}$,
  however, we can worry that although our \hidingsetklawe $U^*$ does
  indeed   have small size, the vertices in $U^*$ might be located on
  high   levels in the graph and be very expensive since they were
  chosen   without regard to measure.
  Just throwing away all white pebbles and picking a new set $U^*$ that
  \hideklawe{}s $B$ on its own is quite a drastic move, and it is not
  hard to construct examples where this is very bad in terms of
  potential   (say, exchanging $s_5$ for $v_5$ in the \hidingsetklawe
  of  \refex{ex:why-we-need-connected-components}).
  The reason that this nevertheless works is that 
  $\setsize{U}$ is so large, that, in addition,
  $U \unionSP W$ is \connectedklawe,
  and that, finally, the graph under consideration is spreading.
  Thanks to this,
  if there are a lot of expensive vertices in 
  $\vertabovelevel{U^*}{j}$
  on or above some high level $j$
  resulting in a large partial measure
  $\Vjthmeasure{j}{U^*}$,
  the number of vertices 
  on or above level
  $\levelstd =  \vminlevel{U \union W}$
  in
  $U = \vertabovelevel{U}{\levelstd}$
  is large enough to yield at least as large a partial measure~%
  $\Vjthmeasure{\levelstd}{U}$.

  Let us do the formal proof, divided into the two cases above.
  \begin{enumerate}
  \item 
    $j \leq \vminlevel{U \unionSP W}$:
    Using the lower bound on the size of $U$ and that level $j$ is no
    higher than the minimal level of~$U$, we get
    \begin{align*}
      \Vjthmeasure{j}{U^*} 
      &= j + 2 \cdot \Setsize{\vertabovelevel{U^*}{j}} 
      && 
      \bigl[\text{ by definition of $\vjthmeasure{j}{\cdot}$ }\bigr] 
      \\
      &\leq j + 2 \cdot \Setsize{U^*} 
      && 
      \bigl[\text{ since $\vertabovelevel{V}{j} \subseteq V$
                 for any $V$ }\bigr] 
      \\
      &\leq j + 2 \cdot \setsize{B}
      && 
      \bigl[\text{ by  construction of $U^*$ in
          \reflem{lem:klawe-lemma-three-five} }\bigr] 
      \\
      &\leq j + 2 \cdot \setsize{U}
      && 
      \bigl[\text{ by  assumption 
          $\setsize{U} \geq 
          \setsize{B} + \setsize{W} \geq
          \setsize{B}$ }\bigr] 
      \\
      &= j + 2 \cdot \Setsize{\vertabovelevel{U}{j}}
      && 
      \bigl[\text{ $\vertabovelevel{U}{j} = U$ 
          since $j \leq \vminlevel{U} $ }\bigr] 
      \\
      &= \vjthmeasure{j}{U}
      && \bigl[\text{ by definition of $\vjthmeasure{j}{\cdot}$ }\bigr]
    \end{align*}
    and we can choose $i=j$ in
    \refdef{def:union-respecting-leq}.

  \item
    $j > \vminlevel{U \unionSP W}$:
    Let
    $\levelmin = \vminlevel{U \unionSP W}$.
    The black pebbles in
    $B$
    are \hiddenklawe by 
    $U \unionSP W$,
    or in formal notation
    $B \subseteq \hiddenvertices{U \unionSP W}$, 
    so
    \begin{equation}
      \label{eq:j-above-minlevel-eq-one}
      \abovelevelblockerminsize{j}{B}
      \leq
      \abovelevelblockerminsizecompact{j}{\hiddenvertices{U \unionSP W}}
    \end{equation}
%
    holds by
    \refpart{item:about-min-size-covers-two}
    of
    \refobs{obs:about-min-size-covers}.
    Moreover,
    $U \unionSP W$
    is a \connectedklawe
    set of vertices in a spreading graph~$G$, so the spreading
    inequality in 
    \refdef{def:spreading-graph}
    says that
    $
    \setsize{U \unionSP W}
    \geq
    \abovelevelblockerminsizecompact{j}{\hiddenvertices{U \unionSP W}}
    + j
    - \levelmin
    $,
    or 
    \begin{equation}
      \label{eq:j-above-minlevel-eq-two}
      j + 
      \abovelevelblockerminsizecompact{j}{\hiddenvertices{U \unionSP W}}
      \leq
      \levelmin
      +
      \setsize{U \unionSP W}
    \end{equation}
    after reordering.
    Combining
    \refeq{eq:j-above-minlevel-eq-one}
    and
    \refeq{eq:j-above-minlevel-eq-two}
    we have that
    \begin{equation}
      \label{eq:j-above-minlevel-eq-three}
      j + 
      \abovelevelblockerminsize{j}{B}
      \leq
      \levelmin
      +
      \setsize{U \unionSP W}
    \end{equation}
    and it follows that
    \begin{align*}
      \vjthmeasure{j}{U^*} 
      &= 
      j + 2 \cdot \Setsize{\vertabovelevel{U^*}{j}} 
      && 
      \bigl[\text{ by definition of $\vjthmeasure{j}{\cdot}$ }\bigr] 
      \\
      &\leq 
      j 
      + \Setsize{\vertabovelevel{U^*}{j}} 
      + \Setsize{U^*} 
      && 
      \bigl[\text{ since $\vertabovelevel{V}{j} \subseteq V$ for any $V$ }\bigr] 
      \\
      &\leq 
      j 
      + \abovelevelblockerminsize{j}{B}
      + \setsize{B} 
      && 
      \bigl[\text{ by  construction of $U^*$ in
          \reflem{lem:klawe-lemma-three-five} }\bigr] 
      \\
      &\leq
      \levelmin
      + \setsize{U \unionSP W}
      + \setsize{B}
      && 
      \bigl[\text{ by the inequality
          \refeq{eq:j-above-minlevel-eq-three} }\bigr]
      \\
      &\leq \levelmin + 2 \cdot \setsize{U}
      && 
      \bigl[\text{ by  assumption 
          $\setsize{U} \geq 
          \setsize{B} + \setsize{W}$ }\bigr] 
      \\
      &= \levelmin + 2 \cdot \setsize{\vertabovelevel{U}{\levelmin}}
      && 
      \bigl[\text{ $\vertabovelevel{U}{\levelmin} = U$ 
          since $\levelmin \leq \vminlevel{U} $ }\bigr] 
      \\
      &= \vjthmeasure{\levelmin}{U}
      && \bigl[\text{ by definition of $\vjthmeasure{\levelmin}{\cdot}$ }\bigr]
    \end{align*}
    Thus, the partial measure of $U$ at the minimum level
    $\levelmin$
    is always larger than the partial measure of $U^*$
    at levels $j$ above this minimum level,
    and we can choose $i=\levelmin$ in
    \refdef{def:union-respecting-leq}.
  \end{enumerate}
  Consequently,
  $
  U^* \measureleq U
  $,
  and the lemma follows.
\end{proof}

Concluding this \subsectionKlawe, we want to make a comment about
\reftwolems
{lem:klawe-lemma-three-five}
{lem:pick-local-good-blocker}
and try to rephrase what they say about \hidingsetklawe{}s.
Given a tight set
$U \unionSP W$ 
\st
$B \subseteq \hiddenvertices{U \unionSP W}$,
we can always pick a $U^*$ as in
\reflem{lem:klawe-lemma-three-five}
with
$U^* = B$
or
$\Setsize{U^*} < \setsize{B}$
and with
$
\Setsize{\vertabovelevel{U^*}{j}}
\leq
\abovelevelblockerminsize{j}{B}
$
for all~$j$.
This will sometimes be a good idea, and sometimes not.
Just as in
\reflem{lem:pick-local-good-blocker},
for
$j > \vminlevel{U \unionSP W}$
we can always prove that
\begin{equation}
  \vjthmeasure{j}{U^*}
  \leq
  \vminlevel{U \unionSP W}
  +
  \setsize{U}
  +
  (\setsize{B} + \setsize{W})
  \eqperiod
\end{equation}
The key message of
\reflem{lem:pick-local-good-blocker}
is that replacing $U$ by $U^*$ is a good idea if $U$
is sufficiently large, namely if 
$\setsize{U} \geq \setsize{B} + \setsize{W}$,
in which case we are guaranteed to get
$\vjthmeasure{j}{U^*} \leq \vjthmeasure{\levelmin}{U}$
for
$\levelmin = \vminlevel{U \unionSP W}$.

\subsection{Pyramids Are Spreading Graphs}
\label{sec:klawe-bw-pebbling-pyramids-are-spreading}

The fact that pyramids are spreading graphs, that is, that they
satisfy the inequality~\refeq{eq:spreading-property}, is
a consequence of the following lemma.

\newcommand{\icecreamconelemma}{Ice-Cream Cone Lemma\xspace}

\begin{lemma}[\icecreamconelemma{}]
  \label{lem:ice-cream-cone-lemma}
  If $X$ is a \tightklawe vertex set in a pyramid $\pyramidgraph$
  such that
  $\hidsetgraph(X)$
  is a connected graph with vertex set 
  $V = \hiddenvertices{X}$,
  then there is a unique vertex $x \in V$
  \st
  $X =   \necessaryhidingvert{X}{x}$
  and
  $
  V = \hiddenvertices{\necessaryhidingvert{X}{x}}
  \subseteq \belowverticespyramid{x}
  $.
\end{lemma}

What the lemma says it that for any \tightklawe vertex set $X$, the
connected components $V_1, \ldots, V_k$ look like ragged ice-cream
cones turned upside down. Moreover, for each ``ice-cream cone'' $V_i$,
all vertices in 
$X \intersectionSP V_i$
are needed to \hideklawe the top vertex.
The two connected components in
\reffig{fig:hidingsetgraph}
are both examples of such ``ice-cream cones.''

Before proving
\reflem{lem:ice-cream-cone-lemma},
we show how this lemma can be used to establish that pyramid graphs
are spreading by a converging-paths argument as in
\refobs{obs:converging-paths}.

\begin{proof}[Proof of \refth{th:pyramids-are-spreading-graphs}]
  Suppose that $X$ is a \tightklawe and \connectedklawe set, \ie such that
  $\hidsetgraph(X)$ is a single connected component with set of
  vertices $V = \hiddenvertices{X}$.
  Let $x \in V$ be the vertex given by
  \reflem{lem:ice-cream-cone-lemma}
  \st
  $X =   \necessaryhidingvert{X}{x}$
  and
  $
  V = \hiddenvertices{\necessaryhidingvert{X}{x}}
  \subseteq \belowverticespyramid{x}
  $,
  and let
  $\umaximallevel = \vlevel{x}$.
  
  For any $j \leq \umaximallevel$ we have
  \begin{equation}
    \label{eq:upper-bound-above-level-blocker-size}
    \abovelevelblockerminsize{j}{\hiddenvertices{X}}
    \leq
    \umaximallevel - j + 1
    \eqperiod
  \end{equation}
  This is so since there are only so many vertices on level $j$ in
  $\belowverticespyramid{x}$ and the set of all these vertices must
  \hideklawe everything in $\hiddenvertices{X}$ above level~$j$
  since
  $\hiddenvertices{X}  \subseteq \belowverticespyramid{x}$.

  By assumption
  $X$
  is \tightklawe and all of $X$ is needed to \hideklawe~$x$,
  \ie
  $X = \necessaryhidingvert{X}{x}$.
  Pick a vertex $v \in X$ on bottom level~$\levelmin = \vminlevel{X}$.
  Since
  $v \in \necessaryhidingvert{X}{x}$
  there is a path
  $\pathfromto{P}{v}{x}$
  such that
  $P \intersectionSP X = \set{v}$.
  Consider the set of converging \sourcepath{}s for $P$ in
  \refobs{obs:converging-paths}.
  All these converging paths
  $P_1, P_2, \ldots, P_{\umaximallevel - \levelmin}$
  must be blocked by distinct vertices in~$X \setminus \set{v}$,
  since
  $P_i \intersectionSP P_j \subseteq P \setminus \set{v}$
  and
  $P \setminus \set{v}$
  does not intersect~$X$. From this the inequality
  \begin{equation}
    \label{eq:lower-bound-X-size-by-converging-paths}
    \setsize{X}
    \geq
    \umaximallevel - \levelmin + 1
  \end{equation}
  follows.
  By combining 
  \refeq{eq:upper-bound-above-level-blocker-size}
  and
  \refeq{eq:lower-bound-X-size-by-converging-paths},
  we get that
  \begin{equation}
    \label{eq:combination-to-get-spreading-inequality}
    \setsize{X} - \abovelevelblockerminsize{j}{\hiddenvertices{X}}
    \geq
    \umaximallevel - \levelmin + 1 - (\umaximallevel - j + 1)
    =
    j - \levelmin
  \end{equation}
  which  is  the required spreading inequality
  \refeq{eq:spreading-property}.
\end{proof}

The rest of this \subsectionKlawe is devoted to proving
the \icecreamconelemma.
We will use that fact that pyramids are planar graphs where we can
talk about left and right. More precisely, the following (immediate)
observation will be central in our proof.

\begin{observation} 
\label{obs:planarity}
Suppose for a planar DAG $G$ that we have a \sourcepath $P$ to a
vertex $w$ and two vertices
$u,v \in \belowverticesNR{w}$ 
on opposite sides of $P$. Then any path
$\pathfromto{Q}{u}{v}$
must intersect~$P$.
\end{observation}

Given a vertex $v$ in a pyramid $\pyramidgraph$, there is a unique
path that passes through $v$ and in every vertex $u$ moves to the
right-hand successor of~$u$. We will refer to this path as the
\introduceterm{\nepathlong{}} 
through~$v$, or just the
\introduceterm{\nepath{}} 
through~$v$ for short, and denote it by~%
$\nepaththrough{v}$.
The path through $v$ always moving to the left is the
\introduceterm{\nwpathlong{}} 
or
\introduceterm{\nwpath{}}
through~$v$, 
and is denoted~%
$\nwpaththrough{v}$.
For instance, for the vertex $v_4$ in our running example pyramid in
\reffig{fig:pyramid-height-6}
we have
$
\nepaththrough{v_4} =
\set{s_4, u_4, v_4, w_4}
$
and
$
\nwpaththrough{v_4} =
\set{s_6, u_5, v_4, w_3, x_2, y_1}
$.
%
To simplify the proofs in what follows, 
we make a couple of observations.

\begin{observation}
  \label{obs:necessary-hiding-set-in-subpyramid}
  Suppose that $X$ is a \tightklawe set of vertices in a
  pyramid~$\pyramidgraph$ and that
  $v \in \hiddenvertices{X}$.
  Then
  $\hiddenvertices{\necessaryhidingvert{X}{v}}
  \subseteq
  \belowverticespyramid{v}$.
\end{observation}

\begin{proof}
  Since all vertices in
  $\necessaryhidingvert{X}{v}$
  have a path to $v$ by definition, it holds that
  $\necessaryhidingvert{X}{v}
  \subseteq
  \belowverticespyramid{v}$.
  Any vertex 
  $u \in \pyramidgraph \setminus \belowverticespyramid{v}$
  must lie either to the left of
  $\nepaththrough{v}$
  or to the right of
  $\nwpaththrough{v}$
  (or both).
  In the first case, 
  $\nepaththrough{u}$ 
  is a path via $u$ that does not intersect~$\necessaryhidingvert{X}{v}$, 
  so $u \notin \hiddenvertices{\necessaryhidingvert{X}{v}}$.
  In the second case,
  we can draw the same conclusion by looking at
  $\nwpaththrough{u}$.
  Thus,
  $\bigl( 
  \pyramidgraph \setminus \belowverticespyramid{v}
  \bigr)
  \intersectionSP
  \hiddenvertices{\necessaryhidingvert{X}{v}}
  = \emptyset
  $.
\end{proof}

\begin{observation}
  \label{obs:hidden-vertex-has-path-with-one-element-intersection}
  Suppose that $X$ is a \tightklawe set of vertices in a
  DAG $G$ and that
  $v \in \hiddenvertices{X}$.
  Then there is a \sourcepath $P$ to $v$ \st
  $\setsize{P \intersectionSP X} = 1$.
\end{observation}

\begin{proof}
  Let $P_1$ be any \sourcepath to~$v$ and note that
  $P_1$ intersects $X$ since $v \in \hiddenvertices{X}$.
  Let
  $y$ be the last vertex on $P_1$ in $P_1 \intersectionSP X$, \ie the
  vertex on the highest level in this intersection. 
  Since $X$ is \tightklawe, there is a \sourcepath $P_2$ to $y$ that
  does not intersect
  $X \setminus \set{y}$.
  Let $P$ be the path that starts like $P_2$ and then switches to
  $P_1$ in~$y$. 
  Then
  $
  \setsize{P \intersectionSP X} 
  =
  \setsize{\set{y}}
  = 1$.
\end{proof}

Using  
\reftwoobs
{obs:necessary-hiding-set-in-subpyramid}
{obs:hidden-vertex-has-path-with-one-element-intersection},
we can simplify the definition of 
the \hidingsetklawe graph.
Note that
\refobs{obs:necessary-hiding-set-in-subpyramid}
is not true for arbitrary layered DAGs, however, or even for arbitrary
layered planar DAGs, so the simplification below does not work in
general.

\begin{proposition}
  \label{pr:simplification-edges}
  Let
  $\hidsetgraph = \hidsetgraph(\pyramidgraph, X)$ 
  be the \hidingsetklawe graph for a \tightklawe set of vertices $X$
  in a pyramid~$\pyramidgraph$, 
  and suppose that
  $u, v \in \hiddenvertices{X}$.
  Then the following conditions are equivalent:
  \begin{enumerate}
  \item 
    \label{item:edge-formal-def}
    $(u,v)$ is an edge in $\hidsetgraph$, \ie
    $
    \belowverticespyramid{u} 
    \intersectionSP
    \hiddenvertices{\necessaryhidingvert{X}{u}}
    \intersectionSP
    \belowverticespyramid{v}
    \intersectionSP
    \hiddenvertices{\necessaryhidingvert{X}{v}}
    \neq \emptyset
    $.
    
  \item
    \label{item:edge-small-simplification}
    $
    \hiddenvertices{\necessaryhidingvert{X}{u}}
    \intersectionSP
    \hiddenvertices{\necessaryhidingvert{X}{v}}
    \neq \emptyset
    $.
    
  \item
    \label{item:edge-bigger-simplification}
    $
    \necessaryhidingvert{X}{u}
    \intersectionSP
    \necessaryhidingvert{X}{v}
    \neq \emptyset
    $.
  \end{enumerate}
\end{proposition}

\begin{proof}
  The directions
  (\ref{item:edge-formal-def})
  $\Rightarrow$
  (\ref{item:edge-small-simplification})
  and
  (\ref{item:edge-bigger-simplification})
  $\Rightarrow$
  (\ref{item:edge-small-simplification})
  are immediate.
  The implication 
  (\ref{item:edge-small-simplification})
  $\Rightarrow$
  (\ref{item:edge-formal-def})
  also follows easily, since
  $
  \hiddenvertices{\necessaryhidingvert{X}{u}}
  \subseteq
  \belowvertices[\pyramidgraph]{u}
  $
  and
  $
  \hiddenvertices{\necessaryhidingvert{X}{v}}
  \subseteq
  \belowvertices[\pyramidgraph]{v}
  $
  by
  \refobs{obs:necessary-hiding-set-in-subpyramid}.
  To prove
  (\ref{item:edge-small-simplification})
  $\Rightarrow$
  (\ref{item:edge-bigger-simplification}),
  fix some vertex
  $
  w \in
  \hiddenvertices{\necessaryhidingvert{X}{u}}
  \intersectionSP
  \hiddenvertices{\necessaryhidingvert{X}{v}}
  $
  and let
  $P$ be a \sourcepath to $w$ as in
  \refobs{obs:hidden-vertex-has-path-with-one-element-intersection}
  with
  $P \intersectionSP X  = \set{y}$ for some vertex~$y$.
  Since
  $
  P \intersectionSP \necessaryhidingvert{X}{u}
  \neq \emptyset \neq
  P \intersectionSP \necessaryhidingvert{X}{u}
  $
  by assumption, we have
  $
  y 
  \in
  \necessaryhidingvert{X}{u} \intersectionSP \necessaryhidingvert{X}{v}
  \neq \emptyset
  $.
\end{proof}

As the first part of the proof of
\reflem{lem:ice-cream-cone-lemma},
we show that all vertices \hiddenklawe by a \connectedklawe
set $X$ are contained in a subpyramid, the top vertex of which is also
\hiddenklawe by~$X$. This gives the ice-cream cone shape alluded
to by the name of
the lemma. 

\begin{lemma}
  \label{lem:hidden-vertices-contained-in-subpyramid}
  Let
  $\hidsetgraph = \hidsetgraph(\pyramidgraph, X)$ 
  be the \hidingsetklawe graph of a \connectedklawe vertex set $X$ in
  a pyramid~$\pyramidgraph$.
  Then there is a unique vertex
  $x \in \hiddenvertices{X}$
  \st
  $\hiddenvertices{X} \subseteq \belowverticespyramid{x}$.
\end{lemma}

\begin{proof}
  It is clear that at most one vertex
  $x \in \hiddenvertices{X}$
  can have the properties stated in the lemma. We show that such
  a vertex exists.
  As a quick preview of the proof, we note that it
  is easy to find a unique vertex $x$ on minimal level \st
  $\hiddenvertices{X} \subseteq \belowverticespyramid{x}$.
  The crucial part of the lemma is that 
  $x$ is \hiddenklawe by~$X$.
  The reason that this holds is that the   
  graph $\hidsetgraph$   is connected. If
  $x \notin \hiddenvertices{X}$, 
  we can find a \sourcepath $P$ to the top vertex $z$ of the pyramid 
  \st $P$ does not intersect $X$ but  
  there are vertices in $\hidsetgraph$
  both to the left and to the right of~$P$.
  But there is no way we can have an edge 
  crossing~$P$ in $\hidsetgraph$,
  so the \hidingsetklawe graph cannot be connected after all.
  Contradiction.

\begin{figure}[tp]
  \centering
  \includegraphics{icecreamfigures.3}%
  \caption{%
    Illustration of proof of
    \reflem{lem:hidden-vertices-contained-in-subpyramid}
    that $\hidsetgraph$ is not connected if
    $x \notin \hiddenvertices{X}$.%
    }
  \label{fig:icecreamtopvertex}
\end{figure}

  The above paragraph really is the whole proof, but let us also provide
  the (somewhat tedious) formal details for completeness.
  To follow the formalization of the argument, the reader might be
  helped by looking at
  \reffig{fig:icecreamtopvertex}.
  Suppose that $\pyramidgraph$ has height $h$ and let
  $s_1, s_2, \ldots, s_{h+1}$ 
  be the sources   
  enumerated from left to right.
  Look at the \nepathlong{}s 
  $\nepaththrough{s_1}, \nepaththrough{s_2}, 
  \ldots$ 
  and let $s_i$ be the first vertex \st
  $\nepaththrough{s_i} \intersectionSP \hiddenvertices{X} 
  \neq
  \emptyset$.
  Similarly, consider
  $\nwpaththrough{s_{h+1}}, 
  \nwpaththrough{s_{h}}, 
  \ldots$
  and let $s_j$ be the first vertex \st
  $\nwpaththrough{s_j} \intersectionSP \hiddenvertices{X} \neq
  \emptyset$.
  It clearly holds that $i \leq j$.

  Let $x$ be the unique vertex where
  $\nepaththrough{s_i}$
  and
  $\nwpaththrough{s_j}$
  intersect.
  By construction, we have
  $\hiddenvertices{X} \subseteq \belowverticespyramid{x}$,
  since no \nepath to the left of 
  $\nepaththrough{s_i} = \nepaththrough{x}$ 
  intersects
  $\hiddenvertices{X}$
  and neither does any \nwpath to the right of
  $\nwpaththrough{s_j} = \nwpaththrough{x}$.
  We need to show that it also holds that
  $x \in \hiddenvertices{X}$.
  
  To derive a contradiction, suppose instead that
  $x \notin \hiddenvertices{X}$.
  By definition, there is a path $P$ from some source $s^*$ to $x$ \st
  $P \intersectionSP \hiddenvertices{X} = \emptyset$.
  $P$ cannot coincide with
  $\nepaththrough{x}$
  or
  $\nwpaththrough{x}$
  since the latter two paths both intersect $\hiddenvertices{X}$
  by construction.
  Since
  $\aboveverticespyramidNR{x} \intersectionSP \hiddenvertices{X}
  = \emptyset$,
  we can extend $P$ to a path
  $\pathfromto{P^*}{s^*}{z}$
  via $x$
  having the property that
  $P^* \intersectionSP \hiddenvertices{X} = \emptyset$
  but there are  vertices in $\hidsetgraph(X)$ 
  both to the left and to the right of~$P^*$,  namely,
  the non-empty sets
  $
  \nepaththrough{x} 
  \intersectionSP 
  \hiddenvertices{X}
  \intersectionSP
  \belowverticespyramid{x}
  $
  and
  $
  \nwpaththrough{x} 
  \intersectionSP 
  \hiddenvertices{X}
  \intersectionSP
  \belowverticespyramid{x}
  $.
  We claim that this implies that $\hidsetgraph$ is not connected. 
  This is a contradiction to the assumptions in the
  statement of the lemma and it follows that
  $x \in \hiddenvertices{X}$ must hold.

  To establish the claim, note that if $\hidsetgraph$ is connected,
  there must exist some edge $(u,v)$  
  between a vertex $u$ to the left of $P^*$ and a vertex $v$ to the
  right of~$P^*$.
  Then
  \refpr{pr:simplification-edges}
  says that
  $
  \hiddenvertices{\necessaryhidingvert{X}{u}}
  \intersectionSP
  \hiddenvertices{\necessaryhidingvert{X}{v}}
  \neq \emptyset
  $.
  Pick any vertex
  $
  w \in
  \hiddenvertices{\necessaryhidingvert{X}{u}}
  \intersectionSP
  \hiddenvertices{\necessaryhidingvert{X}{v}}
  $
  and assume \wolog that $w$ is on the right-hand side of~$P^*$.
  We prove that such a vertex $w$ cannot exist.
  See the example vertices labelled $u$, $v$ and $w$ in
  \reffig{fig:icecreamtopvertex}, 
  which illustrate the fact that
  $w \notin \hiddenvertices{\necessaryhidingvert{X}{u}}$
  if
  $w \in \hiddenvertices{\necessaryhidingvert{X}{v}}$.

  Since $w$ is assumed to be \hiddenklawe by 
  $\hiddenvertices{\necessaryhidingvert{X}{u}}$,
  the \nwpath through~$w$ must intersect
  $\necessaryhidingvert{X}{u}$ somewhere before $w$ or in~$w$.
  Fix any
  $
  y \in
  \nwpaththrough{w}
  \intersectionSP
  \necessaryhidingvert{X}{u}
  \intersectionSP
  \belowverticespyramid{w}
  $
  and note that $y$ must also be located to the right of~$P^*$. By
  \refdef{def:necessary-subcover},
  there is a \sourcepath $P'$ via $y$ to $u$ \st
  $P' \intersectionSP X = \set{y}$. But $P'$ must intersect
  $P^*$ somewhere above~$y$, since $y$ is to the right and $u$ is to
  the left of~$P^*$.
  (Here we use
  \refobs{obs:planarity}.)
  Consider the \sourcepath that starts like $P^*$ and then
  switches to $P'$ at some intersection point in
  $P' \intersectionSP P^* \intersectionSP \aboveverticespyramidNR{y}$.
  This path reaches $u$ but does not intersect~$X$, 
  contradicting the assumption 
  $u \in \hiddenvertices{X}$. 
  It follows that
  $
  \hiddenvertices{\necessaryhidingvert{X}{u}}
  \intersectionSP
  \hiddenvertices{\necessaryhidingvert{X}{v}}
  = \emptyset
  $
  for all $u$ and $v$ on different sides of~$P^*$,
  so there are no edges across $P^*$ in~$\hidsetgraph$.
  This proves the claim.
\end{proof}

The second part needed to prove
\reflem{lem:ice-cream-cone-lemma}
is that all vertices in $X$ are required to
\hideklawe the top vertex 
$x \in \hiddenvertices{X}$ 
found in
\reflem{lem:hidden-vertices-contained-in-subpyramid}.

\begin{lemma}
  \label{lem:tightness-implies-all-hide-top-vertex}
  Let
  $\hidsetgraph = \hidsetgraph(\pyramidgraph, X)$ 
  be the \hidingsetklawe graph of a \connectedklawe vertex set $X$ in
  a pyramid~$\pyramidgraph$ and let  
  $x \in \hiddenvertices{X}$
  be the unique vertex \st
  $\hiddenvertices{X} \subseteq \belowverticespyramid{x}$.
  Then
  $X = \necessaryhidingvert{X}{x}$.
\end{lemma}

\newcommand{\nechidXx}{\necessaryhidingvert{X}{x}}

\begin{proof}
  By definition,
  $\necessaryhidingvert{X}{x} \subseteq X$.
  We want to show that
  $\necessaryhidingvert{X}{x}= X$.
  Again, let us first try to convey some intuition why the lemma is
  true. If
  $ X \setminus \necessaryhidingvert{X}{x} \neq \emptyset$,
  since $X$ is \connectedklawe there must exist some vertex
  \hiddenklawe by all of $X$ but not by just
  $\necessaryhidingvert{X}{x}$
  or
  $X \setminus \necessaryhidingvert{X}{x}$
  (otherwise there can be no edge between the components of
  $\hidsetgraph$ containing
  $\necessaryhidingvert{X}{x}$
  and
  $X \setminus \necessaryhidingvert{X}{x}$,
  respectively).
  But if so, it can be shown that the extra vertices in
  $X  \setminus \necessaryhidingvert{X}{x}$
  help
  $\necessaryhidingvert{X}{x}$ 
  to \hideklawe one of its own vertices. This contradicts the fact
  that $X$ is \tightklawe, so we must have
  $\necessaryhidingvert{X}{x}= X$
  which proves the lemma.

  Let us fill in the formal details in this proof sketch.
  Assume, to derive a contradiction, that
  $\necessaryhidingvert{X}{x} \neq X$. 
  Since $X$ is \tightklawe, it holds that
  $
  (X \setminus \nechidXx)
  \intersectionSP
  \hiddenvertices{\nechidXx}
  = \emptyset
  $,
  so
  $\hidsetgraph$
  contains vertices outside of~$\hiddenvertices{\nechidXx}$.
  Since $\hidsetgraph$ is connected,
  there must exist some edge 
  $\bigl(u,u'\bigr)$ 
  between a pair of vertices
  $u \in \hiddenvertices{X} \setminus \hiddenvertices{\nechidXx}$
  and
  $u' \in \hiddenvertices{\nechidXx}$.
  \Reflem{lem:klawe-lemma-three-one}
  says that
  $\necessaryhidingvert{X}{u'} \subseteq \nechidXx$
  and
  \refpr{pr:simplification-edges}
  then tells us that
  $\necessaryhidingvert{X}{u} \intersectionSP \nechidXx \neq \emptyset$.
  Also, 
  $\necessaryhidingvert{X}{u} \setminus \nechidXx \neq \emptyset$
  since
  $u \notin \nechidXx$.
  For the rest of this proof, fix some arbitrary vertices
  $r \in \necessaryhidingvert{X}{u} \intersectionSP \nechidXx$
  and
  $s \in 
  \necessaryhidingvert{X}{u} \setminus \nechidXx$.
  We refer to
  \reffig{fig:icecreamnontight}
  for an illustration of the proof from here onwards.

\begin{figure}[tp]
  \centering
  \includegraphics{icecreamfigures.4}%
  \caption{Illustration of proof of
    \reflem{lem:tightness-implies-all-hide-top-vertex}
    that all of $X$ is needed to \hideklawe $x$.}
  \label{fig:icecreamnontight}
\end{figure}

  By
  \refdef{def:necessary-subcover},
  there are \sourcepath{}s
  $P_r$ via $r$ to $u$
  and
  $P_s$ via $s$ to $u$
  that intersect $X$ only in $r$ and $s$, respectively.
  Also, there is a \sourcepath $P$ to $x$ \st
  $P \intersectionSP X = \set{r}$
  since
  $r \in \nechidXx$.
  Suppose \wolog that $s$ is to the right of~$P$.
  The paths $P_s$ and $P$ cannot intersect between $s$ and $u$.
  To see this, observe that if $P_s$ crosses $P$ after $s$ but before
  $r$, then by starting with $P$ and switching to $P_s$ at the
  intersection point we get a \sourcepath to $u$ that is not blocked
  by~$X$. And if the crossing is after $r$, we can start with $P_s$
  and then switch to $P$ when the paths intersect, which implies
  that  $s \in \nechidXx$ contrary to assumption.
  Thus $u$ is located to the right of $P$ as well.

  Extend $P_s$ by going north-west from $u$ until hitting $P$, which
  must happen somewhere in between $r$ and~$x$, and then following $P$
  to~$x$. Denote this extended path by $P^{E}_s$ and let $w$ be the
  vertex starting from which $P^{E}_s$ and $P$ coincide.
  The path
  $P^{E}_s$ must intersect $X$ in some more vertex after $s$
  since
  $s \notin \nechidXx$.
  Pick any 
  $v \in P^{E}_s \intersectionSP (X \setminus \set{s})$.
  By construction, $v$ must be located strictly between $u$ and~$w$.
  We claim that 
  $X \setminus \set{v}$ \hideklawe{}s $v$. This contradicts the
  \tightnessklawe of $X$ and the lemma follows.

  To prove the claim, consider any \sourcepath $P_v$ to $v$ and assume
  that $P_v \intersectionSP (X \setminus \set{v}) = \emptyset$. 
  Then, in particular, $r \notin P_v$.
  Suppose that $P_v$ passes to the left of~$r$. By planarity,
  $P_v$ must intersect $P$ somewhere above~$r$. But if so, we can
  construct a \sourcepath $P'$ to $x$ that starts like $P_v$ and
  switches to $P$ at this intersection point. We get
  $P' \intersectionSP X = \emptyset$, which contradicts
  $x \in \nechidXx$.
  If instead $P_v$  passes $r$ on the right, then $P_v$ must cross
  $P_r$ in order to get to~$v$. This implies that there is a
  \sourcepath $P''$ to $u$ \st
  $P'' \intersectionSP X = \emptyset$,
  namely the path obtained by starting to go along $P_v$ and then
  changing to $P_r$ when the two paths intersect above~$r$.
  Thus we get a contradiction in this case as well.
  Hence, $X \setminus \set{v}$ blocks any \sourcepath to $v$ 
  as claimed.
%
\end{proof}

The 
\icecreamconelemma~\ref{lem:ice-cream-cone-lemma} 
now follows.
%
%
Thereby, the proof of the lower bound on
the black-white pebbling price of pyramid graphs in
\refthP{th:lower-bound-bwpebbling-assuming-KP}
is complete.

\newcommand{\sectionBlobBound}{section\xspace}
\newcommand{\subsectionBlobBound}{subsection\xspace}

%
%

\newcommand{\layeredgraphstd}{G}

\newcommand{\BHunionBBunionWH}%
{
  \Setsize{\blackshidden 
    \unionSP \blacksjustblocked
    \unionSP \whitesbelowhidden}
}

\newcommand{\BHunionWH}%
{\Setsize{\blackshidden \unionSP \whitesbelowhidden}}

\newcommand{\BHiunionWHi}%
{\Setsize{\blackshiddenith \unionSP \whitesbelowhiddenith}}

%
%

\section{A Tight Bound for \MULTIPEBBLINGTEXT the Pyramid}
\label{sec:tight-lower-bound-for-multi-pebbling-pyramid}

Inspired by Klawe's ideas in
\refsec{sec:pebble-games-pyramids}, 
we want to do something similar for the \blobpebblegame
in \refdefP{def:multi-pebble-game}.
In this section, we study \pebblingdag{}s
(\refdef{def:blob-pebbling-DAG})
that are also layered.
We show that for all  such DAGs $G_h$ of height~$h$ that are
spreading in the sense of 
\refdef{def:spreading-graph},
it holds that
$\blobpebblingprice{G_h} = \Tightsmall{h}$.
In particular, this bound holds for pyramids~$\pyramidgraphh$
since they are spreading by
\refth{th:pyramids-are-spreading-graphs}.

The constant factor that we get in our lower bound is moderately
small and explicit.  In fact, we believe that it should hold that
$\blobpebblingprice{G_h} \geq h/2 + \Ordosmall{1}$
for layered spreading graphs $G_h$ of height~$h$, 
just as in the standard black-white pebble game.
As we have not made any real attempt to get optimal constants, the
factor in our lower bound can be improved with a minor effort, but
additional ideas seems to be needed to push 
the constant all the way up to~$\frac{1}{2}$.

\subsection{Definitions and Notation for the \MULTIPEBBLINGTEXT Price
Lower Bound}
\label{sec:klawe-M-pebbling-definitions}

Recall that a vertex set $U$ \hideklawe{}s a black pebble on~$b$ if it
blocks all \sourcepath{}s visiting~$v$. 
For a \multipebble $B$, which is a chain by 
\refdef{def:C-multi-pebble-configuration},
it appears natural to extend this definition by requiring that $U$
should block all paths going through all of~$B$.
We recall the terminology and notation from
\refdef{def:chains-and-paths}
that a black \multipebble $B$ and a path $P$ \introduceterm{agree} 
with each other, or that $P$ is a path \introduceterm{via} $B$,
if $B \subseteq P$, and that $\pathsviachain{B}$ 
denotes the set of all \sourcepath{}s agreeing with $B$. 

\begin{definition}[Blocked black \multipebble{}]
  \label{def:blocking-set}
A vertex set $U$ \introduceterm{blocks} a \multipebble $B$ if
$U$ blocks all
$P \in \pathsviachain{B}$.
\end{definition}

A terminological aside:
Recalling the discussion in the beginning of
\refsec{sec:klawe-bw-pebbling-bound},
it seems natural to say that $U$
\introduceterm{blocks} a black \multipebble $B$
rather than \hideklawe{}s it,
since  standing at the sources we might
``see'' the beginning of $B$,
but if we try to
walk any path via $B$ we will fail before reaching the top of~$B$ 
since $U$ blocks the path.
This distinction between \hidingklawe and blocking turns out to be a
very important one in our lower bound proof for \multipebblingtext price.
Of course, if 
$B$ is an \atomicmultipebbleadj black pebble, \ie $\setsize{B} = 1$, the
\hidingklawe{} and blocking relations coincide.

%

Let us next define what it means to block \anmpctext.

\begin{definition}[Unblocked paths]
  \label{def:unblocked-paths}
  For
  $\mpscnot{B}{W}$
  an \mpscfulltext, the set of 
  \introduceterm{unblocked paths}
  for
  $\mpscnot{B}{W}$
  is
  \begin{equation*}
    \unblocked{\mpscnot{B}{W}}
    =
    \setdescr
        {P \in \pathsviachain{B}}
        {\text{$W$ does not block $P$}}
  \end{equation*}
  and
  we say that $U$ blocks $\mpscnot{B}{W}$ if $U$ blocks
  all paths in
  $\unblocked{\mpscnot{B}{W}}$.
  We say that $U$ blocks
  the \mpctext
  $\mpconf$ if $U$ blocks all
  $\mpscnot{B}{W} \in \mpconf$.
%
%
  If so, we say that 
  $U$ is a 
  \introduceterm{blocker} of $\mpscnot{B}{W}$ or~$\mpconf$,
  respectively,
  or a
  \introduceterm{blocking set} for
  $\mpscnot{B}{W}$ or~$\mpconf$.
\end{definition}

Comparing to
\refsec{sec:klawe-bw-pebbling-bound},
note that when blocking a path
$P \in \pathsviachain{B}$,
$U$ can only use the white pebbles $W$ that are associated with $B$
in
$\mpscnotstd$.
Although there might be white pebbles from other \mpsctext{}s
$\mpscnotprime \neq \mpscnotstd$ 
that
would be really helpful, $U$ cannot enlist the help of the 
white pebbles in~$W'$ when blocking~$B$.
The reason for defining the blocking relation in this way is that
these white pebbles can suddenly disappear due to pebbling moves
performed on such 
\mpsctext{}s~$\mpscnotprime$.

Reusing the definition of measure in
\refdefP{def:measure}, we generalize the concept of 
\introduceterm{potential} 
to \mpctext{}s as follows.

\begin{definition}[\Multipebblingtext potential]
  \label{def:potential-M-pebbling}
  The
  \introduceterm{potential} 
  of an \anmpctext $\mpconf$
  is
  \begin{equation*}
  \vpotential{\mpconf} 
  =
  \minofset{\meastopot{U}}
  {\text{$U$ blocks $\mpconf$}}
  \eqperiod
  \end{equation*}                                
  If $U$ is such that $U$ blocks $\mpconf$ and
  $U$ has minimal measure $\vmeasure{U}$ among all blocking sets for
  $\mpconf$,  we say that $U$ is a \introduceterm{\minmeasure{}}
  blocking set   for $\mpconf$.
\end{definition}

To compare \multipebblingtext potential with the black-white pebbling
potential in 
\refdef{def:potential},
consider the following examples with vertex labels as in
Figures~\ref{fig:pyramid-height-6}
and \mbox{\ref{fig:convergingpaths}--\ref{fig:hidingsetgraph}}.
\begin{example}
  \label{ex:first-example-mpebbling-potential}
  For the \mpctext
  $\mpconf =
  \Set{\mpscnot{z}{y_1}, \mpscnot{z}{y_2}}
  $,
  the \minmeasure blocker 
  is 
  $U = \set{z}$.
  In comparison, the standard black-white pebble configuration
  $\pconf
  = (B,W)
  = (\set{z}, \set{y_1,y_2})
  $
  has
  $U = \emptyset$ as \minmeasure \hidingsetklawe{}.
\end{example}

\begin{example}
  \label{ex:second-example-mpebbling-potential}
  For the \mpctext
  $\mpconf =
  \Set{
    \mpscnot{z}{\emptyset},
    \mpscnot{y_1}{x_1,x_2}
  }
  $,
  the \minmeasure blocker  
  is again
  $U = \set{z}$.
  In comparison, for the standard black-white pebble configuration
  $\pconf
  = (B,W)
  = (\set{z,y_1}, \set{x_1,x_2})
  $
  we have the  \minmeasure \hidingsetklawe{}
  $U = \set{x_3}$.
\end{example}

\begin{remark}
  Perhaps it is also worth pointing out that
  \refdef{def:potential-M-pebbling}
  is indeed a strict generalization of
  \refdef{def:potential}.
%
%
  Given a black-white pebble configuration
  $\pconf = (B,W)$
  we can construct an equivalent \mpctext
  $\mpconf(\pconf)$
  \wrt potential
  by setting
  \begin{equation}
    \label{eq:bwpebbling-potential-to-mpebbling-potential}
    \mpconf(\pconf)
    =
    \Setdescr
    {\mpscnotcompact{b}{W \intersectionSP \belowvertices{b}}}
    {b \in B}
  \end{equation}
  but as the examples above show going in the other direction is not
  possible. 
\end{remark}

Since we have accumulated a number of different minimality criteria
for blocking sets, let us pause to clarify the terminology:
\begin{itemize}
\item 
  The vertex set $U$ is a \introduceterm{\subsetminimal{}}, or just
  \introduceterm{minimal}, blocking set for the \mpctext $\mpconf$ if
  no strict subset $U' \subsetneqq U$ is a blocking set for $\mpconf$.

\item
 $U$ is a \introduceterm{\minmeasure{}} blocking set for $\mpconf$ if
 it has minimal measure among all blocking sets for $\mpconf$ (and
 thus yields the potential of $\mpconf$).

\item
  $U$ is a \introduceterm{\minsize{}} blocking set for $\mpconf$ if it
  has minimal   size among all blocking sets for $\mpconf$.
  
\end{itemize}
Note that we can assume \wolog that \minmeasure 
and \minsize blockers are both
\subsetminimal, since throwing away superfluous vertices can only
decrease the measure and size, respectively.
However, \minmeasure blockers need not have minimal size and vice
versa.
For a simple example of this, consider
(with vertex labels as  in
Figures~\ref{fig:pyramid-height-6}
\mbox{and \ref{fig:convergingpaths}--\ref{fig:hidingsetgraph}})
the \mpctext
$\mpconf = \Set{\mpscnot{z}{w_3, w_4}}$
and the two blocking sets
$U_1 = \set{z}$
and
$U_2 = \set{w_1, w_2}$.

\subsection{A Lower Bound Assuming a \GENKLAWEPROPACRONYM{}}
\label{sec:klawe-M-pebbling-lower-bound}

For the \blobpebblegame, a useful generalization of
\refpropertyP{property:klawe-property} 
turns out to be the following.

\begin{property}[\Genklaweprop{}]
  \label{property:generalized-klawe-property}
  We say that \anmpctext $\mpconf$
  on a layered \pebblingdag $G$ has the
  \introduceterm{\genklaweprop{} with parameter
    $\gkpconstant$}
  if there is a vertex set $U$
  \st
  \begin{enumerate}
  \item 
    \label{item:GKP-cover}
    $U$ blocks $\mpconf$,

  \item 
    \label{item:GKP-measure}
    $\vpotential{\mpconf} =
    {\meastopot{U}}$,
    \ie $U$ is a \minmeasure blocker of $\mpconf$,
    
  \item 
    \label{item:GKP-size}
    $\setsize{U} \leq \gkpconstant \cdot \mpcost{\mpconf}$.
  \end{enumerate}
  For brevity, in what follows we will just refer to the
  \introduceterm{\genklawepropacronym{}}.

  We say that the graph $G$ has the 
  \genklawepropacronym{} with parameter $\gkpconstant$
  if all \mpctext{}s $\mpconf$
  on $G$ have the \genklawepropacronym{} with parameter $\gkpconstant$.
\end{property}

When the parameter
$\gkpconstant$
is clear from context, we will just write that
$\mpconf$ or  $G$ has the \genklawepropacronym.

For all layered  \pebblingdag{}s $G_h$ of height~$h$
that have the \genklawepropacronym{} and are spreading,
it holds that
$\blobpebblingprice{G_h} = \Tightsmall{h}$.
The proof of this fact is very much in the spirit of the proofs of 
\reflem{lem:potential-property-inductive}
and
\refth{th:lower-bound-bwpebbling-assuming-KP},
although the details are slightly more complicated. 

\begin{theorem}[Analogue of
  \refth{th:lower-bound-bwpebbling-assuming-KP}]
  \label{th:lower-bound-mpebbling-assuming-GKP}
  Suppose that $G_h$ is a layered \pebblingdag of height~$h$
  possessing the 
  \genklawepropacronym{}~\ref{property:generalized-klawe-property}
  with some fixed parameter $\gkpconstant$.  Then for any  
  \pebunconditional \multipebblingtext
  $
  \multipebbling
  =
  \Set{\mpconf_0 = \emptyset, \mpconf_1, \ldots, \mpconf_{\stoptime}}
  $
  of $G_h$   it holds that
  \begin{equation}
    \label{eq:potential-property-inductive-M}
    \vpotential{\mpconf_t} 
    \leq
    (2 \gkpconstant + 1) 
    \cdot
    \maxofexpr[s \leq t]{\mpcost{\mpconf_s}}
    \eqperiod
  \end{equation}
In particular, for any family of layered \pebblingdag{}s $G_h$ that
are also spreading in the sense of
\refdef{def:spreading-graph}, we have
$\blobpebblingprice{G_h} = \Tightsmall{h}$.
\end{theorem}

We make two separate observations before presenting the proof.

\begin{observation}
  \label{obs:upper-bound-blob-pebbling-price-layered-DAG}
  For any layered DAG 
  $G_h$ of height~$h$ it holds that
  $\blobpebblingprice{G_h} = \Ordosmall{h}$.
\end{observation}

\begin{proof}
  Any layered DAG $G_h$ can be black-pebbled with
  $h + \Ordosmall{1}$ pebbles by
  \refthP{lem:upper-bound-pebbling-price-layered-DAG},
  and it is easy to see that a \multipebblingtext can mimic a black
  pebbling in the same cost.  
\end{proof}

\begin{observation}
  \label{obs:blob-pebbling-potential-of-spreading-graph}
  If $G_h$ is a layered \pebblingdag of height~$h$ that is
  spreading
  in the sense of
  \refdef{def:spreading-graph},
  then
  $\vpotential[G_h]{\unconditionalblackmpscnot{z}} = h + 2$.
\end{observation}

\begin{proof}
  The proof is fairly similar to the corresponding case for pyramids in
  \reflem{lem:potential-property-final}.
  Note, though, that in contrast to
  \reflem{lem:potential-property-final},
  here we cannot get the statement from the
  \genklawepropacronym, but instead have to prove it directly.

  Since
  $\mpscblacknot{z}$
  is an \atomicmultipebbleadj \multipebble, the blocking and
  \hidingklawe relations coincide.
  The set
  $U = \set{z}$
  \hideklawe{}s itself and has measure
  $h + 2$. We show that any other blocking set must have strictly
  larger measure.

  Suppose that $z$ is \hiddenklawe{} by some vertex set
  $U' \neq \set{z}$. This $U'$ is minimal \wolog. In particular, we
  can assume that $U'$ is \tightklawe in the sense of
  \refdef{def:tight} and that
  $U' = \necessaryhidingvert{U'}{z}$.
  Then by
  \refcor{cor:vertex-and-necessary-cover-in-same-component}
  it holds that $U'$  is \connectedklawe.
  Letting
  $L = \vminlevelcompact{U'}$
  and setting
  $j = h$
  in the spreading inequality~\refeq{eq:spreading-property},
  we get that
  $
  \Setsize{U'}
  \geq
  1 + h - L
  $
  and hence 
  $
  \vmeasurecompact{U'}
  \geq
  \vjthmeasurecompact{L}{U'}
  \geq
  \mbox{$L +  2( 1 + h - L)$}
  =
  \mbox{$2 h - L + 2$}
  >
  h + 2
  $
  since
  $L < h$.
\end{proof}

\begin{proof}[Proof of
    \refth{th:lower-bound-mpebbling-assuming-GKP}]
  The statement  in the theorem follows from
  \reftwoobs      
  {obs:upper-bound-blob-pebbling-price-layered-DAG}
  {obs:blob-pebbling-potential-of-spreading-graph}
  combined with the inequality~%
  \refeq{eq:potential-property-inductive-M},
  so just as for
  \refth{th:lower-bound-bwpebbling-assuming-KP}
  the crux of the matter is the 
  induction proof 
  needed to get this inequality.
  
  Suppose that
  $U_{t}$ is such that it blocks $\mpconf_{t}$
  and
  $\vpotential{\mpconf_{t}} =
  {\meastopot{U_{t}}}$.
  By the inductive hypothesis, we have that
  $\vpotential{\mpconf_{t}} 
  \leq
  (2 \gkpconstant + 1) \cdot
  \maxofexpr[s \leq t]{\mpcost{\mpconf_s}}$.
  We want to show
  for $\mpconf_{t+1}$ that 
  $\vpotential{\mpconf_{t+1}} 
  \leq
  (2 \gkpconstant + 1) \cdot
  \maxofexpr[s \leq t+1]{\mpcost{\mpconf_s}}$.
  Clearly, this follows if we can prove that
  \begin{equation}
    \label{eq:sufficient-inductive-inequality-potential-M-pebbling}
    \vpotential{\mpconf_{t+1}} 
    \leq
    \maxofexpr{
      \vpotential{\mpconf_{t}},
      (2 \gkpconstant + 1) \cdot
      \mpcost{\mpconf_t}}
    \eqperiod
  \end{equation}
  We also note that if
  $U_t$ blocks $\mpconf_{t+1}$ we are done, since if so
  $
  \vpotential{\mpconf_{t+1}}
  \leq
  \vmeasure{U_t}
  =
  \vpotential{\mpconf_{t}}
  $.

  We make a case analysis depending on the type of move
  in \refdef{def:multi-pebble-game}
  made to get from
  ${\mpconf_{t}}$ to~${\mpconf_{t+1}}$.
  Analogously with the proof of
  \reflem{lem:potential-property-inductive},
  we want to show that we can use $U_t$ to block $\mpconf_{t+1}$
  as long as the move is not an introduction on a source vertex 
  and then use the \genklawepropacronym to take care of such black
  pebble placements on sources. 
  \begin{description}
    \italicitem[Erasure]
    $\mpconf_{t+1} = \mpconf_{t} \setminus \Set{\mpscnotstd}$
    for
    $\mpscnotstd \in \mpconf_{t}$.
    Obviously, 
    $U_t$ blocks $\mpconf_{t+1} \subseteq \mpconf_{t}$.

  \italicitem[Inflation]
    $\mpconf_{t+1} = \mpconf_{t} \unionSP
    \Set{\mpscnotstd}$
    for
    $\mpscnotstd$ inflated from some
    $\mpscnotprime \in \mpconf_{t}$
    \st
    \begin{subequations}
      \begin{align}
        \label{eq:klawe-inflation-one}
        B' &\subseteq B
        \eqcomma
        \\
        \label{eq:klawe-inflation-two}
        W' 
        \intersectionSP 
        \lppstd
        &\subseteq
        W
        \eqcomma \text{ and }
        \\
        \label{eq:klawe-inflation-three}
        B \intersectionSP W' &= \emptyset
        \eqperiod
      \end{align}
    \end{subequations}
    We claim that $U_t$ blocks $\mpscnotstd$ and thus all of
    $\mpconf_{t+1}$.
    Let us first argue intuitively why.
    Suppose that $P$ is any \sourcepath agreeing with~$B$. 
    This path also agrees with~$B'$, and so must be blocked by
    $U_t \unionSP W'$ by assumption.
    If $U_t$
    blocks $B$ we are done.  We can worry, though, that $U_t$ does not
    block~$P$, but that instead $P$ was blocked by some $w \in W'$ that
    disappeared as a result of the inflation move. But if
    $w \in W'$ is on a path via~$B$, it cannot have disappeared, 
    so this can never happen.

    We now write down the formal details.
    With the notation in
    \refdef{def:unblocked-paths},
    fix any path
    $P \in \unblocked{\mpscnot{B}{W}}$.
    We need to show that
    $P \intersectionSP U_t \neq \emptyset$.
    Let us assume \wolog that $P$ ends in
    $\topvertex{B}$,
    for
    $U_t$ blocks $\mpscnotstd$ precisely if it blocks the paths
    $
    P \intersectionSP \belowvertices{\topvertex{B}}
    $
    for all
    $P \in \unblocked{\mpscnot{B}{W}}$.
    We note that 
    by definition, the fact that $P$ agrees with 
    a chain $V$ 
    and ends in $\topvertex{V}$
    implies that
    \begin{equation}
      \label{eq:P-in-V-union-lpp-V}
      P \subseteq V \disjointunionSP \lpp{V}
      \eqperiod
    \end{equation}                                
    Since $P$ agrees with $B$, or in formal notation~%
    $P \in \pathsviachain{B}$,
    and since
    $B' \subseteq B$
    by~%
    \refeq{eq:klawe-inflation-one},
    we have
    $P \in \pathsviachain{B'}$.
    By assumption, $U_t$ blocks $\mpscnotprime$, which in particular
    means that
    $U_t \unionSP W'$
    intersects the path~$P$ agreeing with~$B'$. We get
    \begin{align*}
      \emptyset
      &\neq
      P \intersectionSP \bigl(U_t \unionSP W'\bigr)
      && 
      \bigl[\text{ 
          by definition of blocking
        }\bigr] 
      \\
      &=
      (P \intersectionSP U_t) 
      \unionSP
      \bigr((P \setminus B) \intersectionSP W' \bigl)
      && 
      \bigl[\text{ 
          since $B \intersectionSP W' = \emptyset$
          by \refeq{eq:klawe-inflation-three}
        }\bigr] 
      \\
      &=
      (P \intersectionSP U_t) 
      \unionSP
      \bigl(P \intersectionSP \lpp{B} \intersectionSP W'\bigr)
      && 
      \bigl[\text{ 
          since $P \subseteq B \disjointunionSP \lpp{B}$
          by \refeq{eq:P-in-V-union-lpp-V}
        }\bigr] 
      \\
      &\subseteq
      (P \intersectionSP U_t) 
      \unionSP
      (P \intersectionSP W) 
      && 
      \bigl[\text{ 
          since $\lppstd \intersectionSP W' \subseteq W$
          by \refeq{eq:klawe-inflation-two}
        }\bigr] 
      \\
      &=
      P \intersectionSP U_t
      && 
      \bigl[\text{ 
          $P \intersectionSP W = \emptyset$
          if $P \in \unblocked{\mpscnotstd}$
        }\bigr] 
    \end{align*}
    so
    $P \intersectionSP U_t \neq \emptyset$
    and the desired conclusion that $U_t$ blocks the path~$P$
    follows. 

  \italicitem[Merger]
    $\mpconf_{t+1}  = 
    \mpconf_{t} \unionSP  \Set{\mpscnotstd}$
    for
    $\mpscnotstd$
    derived by merger of
    $\mpscnotstd[1], \mpscnotstd[2] \in \mpconf_{t}$
    \st
    \begin{subequations}
      \begin{align}
        \label{eq:klawe-merger-two}
        B_1 \intersectionSP W_2 &= \emptyset
        \eqcomma
        \\      
        \label{eq:klawe-merger-one}
        B_2 \intersectionSP W_1 &= \set{\mergervertex}
        \eqcomma
        \\
        \label{eq:klawe-merger-three}
        B &= (B_1 \unionSP B_2) \setminus \set{\mergervertex}
        \eqcomma \text{ and}
        \\
        \label{eq:klawe-merger-four}
        W 
        &= 
        \bigl(
        (W_1 \unionSP W_2) \setminus \set{\mergervertex}
        \bigr)
        \intersectionSP
        \lppstd
        \eqperiod
      \end{align}
    \end{subequations}
    Let us again first argue informally that if 
    a set of vertices  
    $U_t$ blocks two
    \mpsctext{}s
    $\mpscnotstd[1]$
    and~%
    $\mpscnotstd[2]$, 
    it must also block their merger. 
    Let $P$ be any path via~$B$, and suppose in addition that $P$
    visits the merger vertex~$\mergervertex$.
    If so, $P$ agrees with $B_2$ and must be blocked by
    $U_t \unionSP W_2$.
    If on the other hand $P$ agrees with $B$ but does \emph{not}
    visit~$\mergervertex$,  
    it is a path via $B_1$ that in addition does not pass through the
    white pebble in $W_1$ eliminated in the merger. This means that
    $U_t \unionSP W_1 \setminus \set{\mergervertex}$
    must block~$P$. Again, we have to argue that the blocking white
    vertices do not disappear when we apply the intersection with
    $\lpp{B}$ in~\refeq{eq:klawe-merger-four}, but this is
    straightforward to verify.

    So let us show formally that
    $U_t$ blocks $\mpscnotstd$, \ie that for any
    $P \in \unblocked{\mpscnotstd}$ it holds that
    $P \intersectionSP U_t \neq \emptyset$.
    As above, \wolog we consider only paths $P$ ending in
    $\topvertex{B} = \topvertex{B_1 \unionSP B_2}$.
    Recall that 
    \begin{equation}
      \label{eq:B-i-intersection-W-i-empty}
      B_i \intersectionSP W_i = \emptyset
    \end{equation}
    holds for all \mpsctext{}s by definition.
    We divide the analysis into two subcases.
    \begin{enumerate}
      \item
        $P \in 
        \pathsviachain{B_1 \unionSP B_2}
        =
        \pathsviachain{B \unionSP \set{\mergervertex}}
        $.
        If so, in particular it holds that
        $P \in \pathsviachain{B_2}$
        and since $U_t$ blocks $\mpscnotstd[2]$ we have
        \begin{align*}
          \emptyset
          &\neq
          P \intersectionSP \bigl(U_t \unionSP W_2\bigr)
          && 
          \bigl[\text{ 
              by definition of blocking
            }\bigr] 
          \\
          &=
          (P \intersectionSP U_t) 
          \unionSP
          \bigr(
          (P \setminus (B_1 \unionSP B_2)) \intersectionSP W_2 
          \bigl)
          && 
          \bigl[\text{ 
              by \refeq{eq:klawe-merger-two} 
              and \refeq{eq:B-i-intersection-W-i-empty}
            }\bigr] 
          \\
          &=
          (P \intersectionSP U_t) 
          \unionSP
          \bigl(P \intersectionSP \lpp{B_1 \unionSP B_2} 
          \intersectionSP W_2\bigr)
          && 
          \bigl[\text{ 
              by \refeq{eq:P-in-V-union-lpp-V}
            }\bigr] 
          \\
          &=
          (P \intersectionSP U_t) 
          \unionSP
          \bigr(P \intersectionSP
          \lpp{B \unionSP \mergervertex} \intersectionSP W_2 \bigl)
          && 
          \bigl[\text{ 
              just rewriting using \refeq{eq:klawe-merger-three} 
            }\bigr] 
          \\
          &\subseteq
          (P \intersectionSP U_t) 
          \unionSP
          \bigr(
          P 
          \intersectionSP
          (W_2 \setminus \set{\mergervertex})
          \intersectionSP
          \lpp{B}
          && 
          \bigl[\text{ 
              $\lpp{B \union \set{\mergervertex}}
              \subseteq
              \lpp{B} \setminus \set{\mergervertex}$
            }\bigr] 
          \\
          &\subseteq
          (P \intersectionSP U_t) 
          \unionSP
          (P \intersectionSP W) 
          && 
          \bigl[\text{ 
              by \refeq{eq:klawe-merger-four}
            }\bigr] 
          \\
          &=
          P \intersectionSP U_t 
          && 
          \bigl[\text{ 
              since
              $P \in \unblocked{\mpscnotstd}$
            }\bigr] 
        \end{align*}
        so  $U_t$ blocks the path~$P$ in this case.

      \item
        $P \in 
        \pathsviachain{B}
        \setminus
        \pathsviachain{B \unionSP \set{\mergervertex}}
        $.
        This means that 
        $B \subseteq P$
        but
        $B \unionSP \set{\mergervertex}\nsubseteq P$,
        so the path $P$ does not pass through~$\mergervertex$.
        Since $P$ agrees with $B_1$ and 
        $U_t$ blocks $\mpscnotstd[1]$ by assumption, we get that
        \begin{align*}
          \emptyset
          &\neq
          P \intersectionSP \bigl(U_t \unionSP W_1\bigr)
          && 
          \bigl[\text{ 
              by definition of blocking
            }\bigr] 
          \\
          &=
          (P \intersectionSP U_t) 
          \unionSP
          \bigr((P \setminus B) \intersectionSP W_1 \bigl)
          && 
          \bigl[\text{ 
              by \refeq{eq:klawe-merger-one} 
              and \refeq{eq:B-i-intersection-W-i-empty}
            }\bigr] 
      \\
      &=
      (P \intersectionSP U_t) 
      \unionSP
      \bigl(P \intersectionSP \lpp{B} 
      \intersectionSP W_1\bigr)
      && 
      \bigl[\text{ 
          $P \subseteq 
          B \disjointunionSP \lpp{B}$
          by \refeq{eq:P-in-V-union-lpp-V}
        }\bigr] 
      \\
      &=
      (P \intersectionSP U_t) 
      \unionSP
      \bigl(
      P 
      \intersectionSP 
      (W_1 \setminus \set{\mergervertex})
      \intersectionSP 
      \lpp{B} 
      \bigr)
      && 
      \bigl[\text{ 
          since $\mergervertex \notin P$ by assumption
        }\bigr] 
      \\
      &\subseteq
      (P \intersectionSP U_t) 
      \unionSP
      (P \intersectionSP W) 
      && 
      \bigl[\text{ 
          by \refeq{eq:klawe-merger-four}
        }\bigr] 
      \\
      &=
      (P \intersectionSP U_t) 
      && 
      \bigl[\text{ 
          $P \in \unblocked{\mpscnotstd}$
        }\bigr] 
    \end{align*}
    and $U_t$ blocks the path~$P$ in this case as well.

    \end{enumerate}

    \italicitem[Introduction]
    $\mpconf_{t+1} = \mpconf_{t} \unionSP 
    \Set{\intrompscnot{v}}$.
    Clearly,
    $U_t$ blocks $\mpconf_{t+1}$
    if $v$ is a non-source vertex, \ie if 
    \mbox{$\prednode{v} \neq \emptyset$},
    since
    $U_t$ blocks $\mpconf_{t}$
    and
    $\intrompscnot{v}$ blocks itself.

    Suppose however that $v$ is a source vertex, so that the \mpsctext
    introduced is~%
    $\unconditionalblackmpscnot{v}$.
    As in the proof of
    \reflem{lem:potential-property-inductive},
    $U_t$ does not necessarily block $\mpconf_{t+1}$ any longer but
    $U_{t+1} = U_{t} \unionSP \set{v}$
    clearly does.
    For       $j > 0$, it holds that
    $\vertabovelevel{U_{t+1}}{j} = \vertabovelevel{U_{t}}{j}$
    and thus
    $\vjthmeasure{j}{U_{t+1}} = \vjthmeasure{j}{U_{t}}$.
    On the bottom level $j=0$, 
    using that
    $\setsize{U_t} \leq \gkpconstant \cdot \mpcost{\mpconf_t}$
    \genklawepropacronym~\ref{property:generalized-klawe-property}
    we have
    \begin{multline}
      \vjthmeasure{0}{U_{t+1}} 
      = 
      2 \cdot \setsize{U_{t+1}} 
      = 
      2 \cdot (\setsize{U_{t}} + 1)
      \leq
\\
      2 \cdot \bigl( \gkpconstant \cdot \mpcost{\mpconf_t} + 1 \bigr)
      \leq
      2 \cdot \bigl( \gkpconstant \cdot \mpcost{\mpconf_{t+1}} + 1 \bigr)
      \leq
\\
      2 \cdot \bigl( \gkpconstant \cdot \mpcost{\mpconf_{t+1}}  
      + \mpcost{\mpconf_{t+1}} \bigr)
      \leq
      2 (\gkpconstant + 1 ) \cdot \mpcost{\mpconf_{t+1}}
    \end{multline}
    and we get that 
    \begin{multline}
      \vpotential{\mpconf_{t+1}} 
      \leq
      \vmeasure{U_{t+1}}
      \leq
      {\textstyle \Maxofexpr[j]{\vjthmeasure{j}{U_{t+1}}}}
\\
      \leq
      \Maxofexpr{
        \vmeasure{U_t}, 
        (2 \gkpconstant + 1) \cdot \mpcost{\mpconf_{t+1}}}
      =
\\
      \Maxofexpr{
        \vpotential{\mpconf_t}, 
        (2 \gkpconstant + 1) \cdot \mpcost{\mpconf_{t+1}}}
    \end{multline}
    which is what is needed for the induction step to go through.
  \end{description}
  We see that regardless of the pebbling move made in the transition
  $\pebcfgtransition{\mpconf_{t}}{\mpconf_{t+1}}$,
  the inequality
  \refeq{eq:sufficient-inductive-inequality-potential-M-pebbling}
  holds. The theorem follows by the induction principle.
\end{proof}

Hence, in order to prove a lower bound on
$\blobpebblingprice{G_h}$
for layered spreading graphs~$G_h$, 
it is sufficient to find some constant $\gkpconstant$ 
\st these DAGs can be shown to possess the
\genklawepropacronym~\ref{property:generalized-klawe-property}
with parameter~$\gkpconstant$. 
%

\subsection{Some Structural Transformations}
\label{sec:klawe-M-pebbling-structural}

As we tried to indicate by presenting the small toy \mpctext{}s in
\reftwoexs
{ex:first-example-mpebbling-potential}
{ex:second-example-mpebbling-potential},
the potential in the \multipebblegame behaves somewhat differently from
the potential in the standard pebble game.
There are (at least) two important differences:
\begin{itemize}
\item 
  Firstly,
  for the white pebbles we have to keep track of exactly which black
  pebbles they can help to block. This can lead to slightly unexpected
  consequences   such as the blocking set $U$ and the set of white
  pebbles overlapping. 

\item 
  Secondly, for black \multipebble{}s there is a much wider 
  choice where to
  block the \blobpebble{}s than for \atomicmultipebbleadj pebbles.
  It seems that to minimize the potential, blocking black
  \multipebble{}s on (reasonably) low levels should still be a good
  idea. However, we cannot 
  a priori
  exclude the possibility that if a lot of black
  \multipebble{}s intersect in some high-level vertex, adding this
  vertex to a blocking set $U$ might be a better idea.
\end{itemize}
In this \subsectionBlobBound we address the first of these issues. The second
issue, which turns out to be much trickier, is dealt with in the next
\subsectionBlobBound. 

One simplifying observation is that we do not have to prove
\refproperty{property:generalized-klawe-property}
for arbitrary \mpctext{}s. 
Below, we show that one can do some technical preprocessing of the
\mpctext{}s so that it suffices to prove 
the \genklawepropacronym{} for the subclass of 
configurations 
resulting from this preprocessing.%
\footnote{%
Note that we did something similar in
\refsec{sec:klawe-bw-pebbling-proving-the-Klawe-property}
after
\reflem{lem:covering-set-is-tight},
when we argued that if
$U$
is a \minmeasure \hidingsetklawe{} for
$\pconf = (B,W)$,
we can assume \wolog that
$U \unionSP W$ is 
\tightklawe{}.
For if not, we just prove the \klaweprop
for some \tightklawe subset
$
U' \unionSP W'
\subseteq
U \unionSP W
$ 
instead.
This is wholly analogous to the reasoning here, but since matters
become more complex we need to be a bit more careful.
}
Throughout this subsection, we assume that the parameter
$\gkpconstant$
is some fixed constant.

We start slowly by taking care of a pretty obvious redundancy.
Let us say that 
the \mpscfulltext
$\mpscnotstd$
is \introduceterm{self-blocking} if
$W$ blocks $B$.
The \mpctext~%
$\mpconf$ is \introduceterm{self-blocker-free}
if there are no self-blocking \mpsctext{}s in~$\mpconf$.
That is,
if 
$\mpscnotstd$ is self-blocking, $W$ needs no extra help blocking $B$.
Perhaps the simplest example of this is
$\mpscnotstd
=
\intrompscnot{v}$
for a non-source vertex~$v$.
The following proposition is immediate.

\begin{proposition}
  \label{pr:removing-self-blockers}
  For $\mpconf$ any \mpctext, let
  $\mpconf'$ be the \mpctext with all self-blockers in
  $\mpconf$ removed.
  Then
  $
  \mpcost{\mpconf'} \leq  \mpcost{\mpconf}
  $,
  $
  \vpotential{\mpconf'} =  \vpotential{\mpconf}
  $
  and any blocking set $U'$ for $\mpconf'$ is also a blocking set for~$\mpconf$.
\end{proposition}

\begin{corollary}
  \label{cor:sufficient-prove-gkp-for-w-self-blocker-free}
  Suppose that the \genklawepropacronym{} holds for self-blocker-free \mpctext{}s.
  Then the \genklawepropacronym{} holds for all \mpctext{}s.
\end{corollary}

\begin{proof}
  If $\mpconf$ is \emph{not} self-blocker-free,
  take the maximal $\mpconf' \subseteq \mpconf$ 
  that is and the blocking set $U'$ that the \genklawepropacronym{}
  provides for this 
  $\mpconf'$.
  Then
  $U'$ blocks $\mpconf$
  and since the two configurations
  $\mpconf$
  and
  $\mpconf'$
  have the same blocking sets their potentials are equal, so
  $\vpotential{\mpconf} =
  {\meastopot{U'}}$.
  Finally, we have that
  $\setsize{U} 
  \leq
  \mbox{$\gkpconstant \cdot \mpcost{\mpconf'}$}
  \leq 
  \mbox{$\gkpconstant \cdot \mpcost{\mpconf}$}
  $.
  Thus the \genklawepropacronym{} holds for~$\mpconf$.
\end{proof}

We now move on to a more interesting observation.
Looking at
$\mpconf =
  \Set{
    \mpscnot{z}{y_1},
    \mpscnot{z}{y_2}
  }
$
in
\refex{ex:first-example-mpebbling-potential},
it seems that the white pebbles really do not help at all.
One might ask if we could not just throw them away?
Perhaps somewhat surprisingly, the answer is yes, and we can capture
the  intuitive concept of necessary white pebbles and formalize it as
follows. 

\begin{definition}[White sharpening]
  \label{def:white-sharpening}
  Given
  $\mpconf = \Set{\mpscnot{B_i}{W_i}}_{i \in \intnfirst{m}}$,
  we say that
  $\mpconf'$ is a \introduceterm{white sharpening} of
  $\mpconf$
  if
  $\mpconf' = \Set{\mpscnot{B'_i}{W'_i}}_{i \in \intnfirst{m}}$
  for
  $B'_i = B_i$
  and
  $W'_i \subseteq W_i$.
\end{definition}

That is, a white sharpening removes white pebbles and thus makes the
\mpctext stronger or ``sharper'' in the sense that the cost can only
decrease and the potential can only increase.

\begin{proposition}
  \label{pr:white-sharpening}
  If
  $\mpconf'$ is a white sharpening of
  $\mpconf$
  it holds that
  $
  \mpcost{\mpconf'} \leq  \mpcost{\mpconf}
  $
  and
  $
  \vpotential{\mpconf'} \geq  \vpotential{\mpconf}
  $.
  More precisely, any blocking set
  $U'$ for $\mpconf'$
  is also a blocking set for
  $\mpconf$.
\end{proposition}

\begin{proof}
  The statement about cost is immediate from
  \refdef{def:multi-pebbling-price}.
  The statement about potential clearly follows from
  \refdef{def:potential-M-pebbling}
  since it holds that any blocking set
  $U'$ for $\mpconf'$
  is also a blocking set for~$\mpconf$.
\end{proof}

In the next definition,  we suppose that there is some fixed but
arbitrary ordering of the vertices in $G$, and 
that the vertices are considered in this order.

\begin{definition}[White elimination]
  \label{def:white-elimination}
  For
  $\mpscnotstd$
  \anmpsctext
  and $U$ any blocking set for
  $\mpscnotstd$,
  write
  $W = \set{w_1, \ldots, w_s}$,
  set $W^0 \assigneq W $
  and iteratively perform the following for
  $i = 1, \ldots, s$:
  If 
  $U \unionSP (W^{i-1} \setminus \set{w_i})$
  blocks $B$,
  set
  $W^i \assigneq W^{i-1} \setminus \set{w_i}$,
  otherwise set
  $W^i \assigneq W^{i-1}$.
  We define
  the
  \introduceterm{white elimination}
  of $\mpscnotstd$ \wrt $U$ to be
  $\welim{\mpscnotstd}{U}
  = \mpscnot{B}{W^s}$
  for $W^s$ the final set resulting from the procedure above.
  
  For $\mpconf$ \anmpctext
  and $U$ a blocking set for $\mpconf$,
  we define
  \begin{equation}
    \welim{\mpconf}{U}
    =
    \Setdescr
        {\welim{\mpscnotstd}{U}}
        {\mpscnotstd \in \mpconf}
        \eqperiod
  \end{equation}
  We say that the elimination is \introduceterm{strict}
  if
  $\mpconf \neq \welim{\mpconf}{U}$.
  If
  $\mpconf = \welim{\mpconf}{U}$
  we say that
  $\mpconf$ is \introduceterm{white-eliminated}, or
  \introduceterm{\weliminate{}d} for short, \wrt $U$.
\end{definition}

Clearly
$\welim{\mpconf}{U}$
is a white sharpening of $\mpconf$.
And if we pick the right $U$, we simplify the problem of proving the
\genklawepropacronym{} 
a bit more.

\begin{lemma}
  \label{lem:potential-preserving-elimination}
  If
  $U$
  is a \minmeasure blocking set for 
  $\mpconf$,
  then
  $\mpconf' = \welim{\mpconf}{U}$
  is a white sharpening of $\mpconf$
  \st
  $\vpotential{\mpconf'} = \vpotential{\mpconf}$
  and 
  $U$ blocks~$\mpconf'$.
\end{lemma}

\begin{proof}
  Since 
  $\mpconf' = \welim{\mpconf}{U}$
  is a white sharpening of $\mpconf$
  (which is easily verified from
  \reftwodefs
  {def:white-sharpening}
  {def:white-elimination}),
  it holds by
  \refpr{pr:white-sharpening}
  that
  $\vpotential{\mpconf'} \geq \vpotential{\mpconf}$.
  Looking at the construction in
  \refdef{def:white-elimination}, we also see that 
  the white pebbles are ``sharpened away'' with care so that
  $U$ remains a blocking set.
  Thus
  $
  \meastopot{U} \geq
  \vpotential{\mpconf'} = \vpotential{\mpconf}
  =  \meastopot{U}
  $, 
  and the lemma follows.
\end{proof}

\begin{corollary}
  \label{cor:sufficient-prove-gkp-for-w-eliminated}
  Suppose that the \genklawepropacronym{} holds for 
  the set of   all
  \mpctext{}s $\mpconf$ having the property that 
  for all \minmeasure blocking sets $U$ for $\mpconf$   it holds that
  $\mpconf = \welim{\mpconf}{U}$.
  Then the \genklawepropacronym{} holds for all \mpctext{}s.
\end{corollary}

\begin{proof}
  This is essentially the same reasoning as in the proof of
  \refcor{cor:sufficient-prove-gkp-for-w-self-blocker-free}
  plus induction.
  Let $\mpconf$ be any \mpctext. Suppose that there exists a
  \minmeasure blocker  $U$ for $\mpconf$ \st $\mpconf$ is not
  \weliminate{}d \wrt $U$.
  Let
  $\mpconf^{1} = \welim{\mpconf}{U}$.
  Then
  $\mpcost{\mpconf^{1}} \leq \mpcost{\mpconf}$
  by
  \refpr{pr:white-sharpening}
  and
  $\vpotential{\mpconf^{1}} = \vpotential{\mpconf}$
  by
  \reflem{lem:potential-preserving-elimination}.

  If there is a 
  \minmeasure blocker  $U^{1}$ for $\mpconf^{1}$ \st $\mpconf^{1}$ is not
  \weliminate{}d \wrt $U^{1}$,
  set
  $\mpconf^{2} = \welim{\mpconf^{1}}{U^{1}}$.
  Continuing in this manner, we get a chain
  $\mpconf^{1}, \mpconf^{2}, \mpconf^{3}, \ldots$
  of strict \welimination{}s \st
  $\mpcost{\mpconf^{1}} \geq \mpcost{\mpconf^{2}} 
  \geq \mpcost{\mpconf^{3}} \ldots$
  and
  $\vpotential{\mpconf^{1}} = \vpotential{\mpconf^{2}} =
  \vpotential{\mpconf^{3}} = \ldots$
  This chain must terminate at some  configuration
  $\mpconf{^k}$
  since the total number of white pebbles (counted with repetitions)
  decreases in every round.

  Let $U^{k}$ be the blocker that the \genklawepropacronym provides
  for $\mpconf^{k}$. Then
  $U^{k}$ blocks $\mpconf$,
  $\vpotential{\mpconf} = \vpotential{\mpconf^{k}} = \vmeasure{U^{k}}$,
  and
  $\setsize{U^{k}} \leq
  \gkpconstant \cdot \mpcost{\mpconf^{k}} \leq
  \gkpconstant \cdot \mpcost{\mpconf}$.
  Thus the \genklawepropacronym holds for $\mpconf$.
\end{proof}

We note that in particular, it follows from the construction in 
\refdef{def:white-elimination}
combined with
\refcor{cor:sufficient-prove-gkp-for-w-eliminated}
that we can assume \wolog for any blocking
set $U$ and any \mpctext~$\mpconf$
that $U$ does not intersect the set of white-pebbled vertices
in~$\mpconf$. 

\begin{proposition}
  \label{pr:U-and-W-do-not-intersect}
  If  
  $\mpconf = \welim{\mpconf}{U}$,
  then 
  in particular it holds that
  $U \intersectionSP \mpcwhitesof{\mpconf} = \emptyset$.
\end{proposition}

\begin{proof}
  Any 
  $w \in \mpcwhitesof{\mpconf} \intersectionSP U$
  would have been removed in the \welimination.  
\end{proof}

\subsection{A Proof of the \GENKLAWEPROP{}}
\label{sec:klawe-M-pebbling-proof-GKP}

We are now ready  to embark on the proof of the
\genklawepropacronym{} for layered spreading DAGs.

\begin{theorem}
  \label{th:pyramids-possess-GKP}
  All layered \pebblingdag{}s that are spreading 
  possess the
  \genklaweprop~\ref{property:generalized-klawe-property}
  with parameter
  \mbox{$\gkpconstant = 13$}.  
\end{theorem}

Since pyramids are spreading graphs by
\refth{th:pyramids-are-spreading-graphs},
this is all that we need to get the lower bound on 
\multipebblingtext price on pyramids from
\refth{th:lower-bound-mpebbling-assuming-GKP}.
We note that the parameter $\gkpconstant$ in
\refth{th:pyramids-possess-GKP}
can easily be improved. However, our main concern here is not
optimality of  constants but clarity of exposition.

We prove \refth{th:pyramids-possess-GKP} by applying the preprocessing
in the previous \subsectionBlobBound and then (almost) reducing the
problem to the standard black-white pebble game. However, some twists
are added along the way since our potential measure for
\multipebble{}s behave differently from Klawe's potential measure for
black and white pebbles.  Let us first exemplify two problems that
arise if we try to do naive pattern matching on Klawe's proof for the
standard black-white pebble game.

In the standard black-white pebble game,
if $U$ is a \minmeasure \hidingsetklawe{} for
$\pconf = (B,W)$,
\reflem{lem:covering-set-is-tight}
tells us that we can assume \wolog that
$U \unionSP W$ is \tightklawe{}.
This is \emph{not} true
in the \multipebblegame, not even after the transformations in
\refsec{sec:klawe-M-pebbling-structural}.

\begin{figure}[tp]
  \centering
  \subfigure[\Minmeasure but non-\tightklawe blocking set.]
  {
    \label{fig:problematichidingset-a}
    \begin{minipage}[b]{.48\linewidth}
      \centering
      \includegraphics{hidingsets.9}%
    \end{minipage}
  }
  \hfill
  \subfigure[\Tightklawe but
  non-connected blocker for 
  \multipebble.]
  {
    \label{fig:problematichidingset-b}
    \begin{minipage}[b]{.48\linewidth}
      \centering
      \includegraphics{hidingsets.8}%
    \end{minipage}
  }
  \caption{Two \mpctext{}s with problematic blocking sets.}
  \label{fig:problematichidingset}
\end{figure}

\begin{example}
  \label{ex:non-tightness}
  Consider the configuration
  $
  \mpconf 
  =
  \set{
    \mpscnot{w_1}{u_2,u_3},
    \mpscnot{w_4,x_3}{u_4,u_5},
    \mpscnot{x_2,y_2,z}{\emptyset}
  }
  $
  with blocking set
  $U = \set{x_2, u_1, u_6}$
  in
  \reffig{fig:problematichidingset-a}.
  It can be verified that $U$ is a \minmeasure blocking set
  and that the configuration
  $\mpconf$ is \weliminated \wrt~$U$, but
  the set
  $U \unionSP \mpcwhitesof{\mpconf} =
  \set{
    u_1, u_2, u_3, u_4, u_5, u_6, x_2
  }$
  is not \tightklawe{}
  (because of $x_2$).
\end{example}

This can be handled, but a more serious problem is that even if the
set $U \unionSP W$ blocking the chain $B$ is \tightklawe, there is no
guarantee that the vertices in $U \unionSP W$ end up in the same
connected component of the \hidingsetklawe graph 
$\hidsetgraph(U \unionSP W)$ 
in
\refdef{def:covering-set-graph}.

\begin{example}
  \label{ex:non-connectedness}
  Consider the single-\multipebble configuration 
  $\mpconf = \set{\mpscnot{u_5, z}{\emptyset}}$
  in
  \reffig{fig:problematichidingset-b}.
  It is easy to verify that
  $U = \set{v_4, y_2}$
  is a \subsetminimal blocker of $\mpconf$ and also a \tightklawe
  vertex set.  
  This highlights the fact that blocking sets for \mpctext{}s can have
  rather different properties than \hidingsetklawe{}s for standard
  pebbles.  In particular, a minimal blocking set for a single
  \multipebble can have several ``isolated'' vertices at large
  distances from one another.  Among other problems, this leads to
  difficulties in defining connected components of blocking sets for
  \mpsctext{}s.

  The naive attempt to generalize 
  \refdef{def:covering-set-graph} of
  connected components in a  \hidingsetklawe graph
  to blocking sets would place  the vertices
  $v_4$ and $y_2$ 
  in different connected components
  $\set{v_4}$
  and
  $\set{y_2}$,
  none of which blocks 
  $\mpconf = \set{\mpscnot{u_5, z}{\emptyset}}$.
  This is not what we want
  (compare
  \refcor{cor:vertex-and-necessary-cover-in-same-component}
  for \hidingsetklawe{}s for black-white pebble configurations).
  We remark that there really cannot be any other sensible definition
  that places $v_4$ and $y_2$ in the same connected component either,
  at least not if we want to appeal to the spreading properties in
  \refdef{def:spreading-graph}.
  Since the level difference in  $U$ is $3$ but the size of the set is
  only~$2$,   the spreading inequality~\refeq{eq:spreading-property}
  cannot hold for this set.
\end{example}

To get around this problem, we will instead use connected components
defined in terms  of \hidingklawe the singleton black pebbles given by
the bottom vertices of our \multipebble{}s.
For a start, recalling
\reftwodefs{def:cover}{def:blocking-set},
let us make an easy observation relating the \hidingklawe{} and
blocking relations for a \multipebble.

\begin{observation}
  If a vertex set $V$ \hideklawe{}s some vertex $b \in B$,
  then $V$ blocks $B$.
\end{observation}

\begin{proof}
  If $V$ blocks all paths visiting $b$, then in particular it blocks
  the subset of paths that not only visits~$b$ but agree with all
  of~$B$. 
\end{proof}

We will focus on the case when the bottom vertex of a \multipebble is
\hiddenklawe.

\begin{definition}[\Hidingklawe \mpctext{}s]
  \label{def:covering-a-blob}
  We say that the vertex set $U$
  \introduceterm{\hideklawe{}s}
  the \mpsctext
  $\mpscnotstd$
  if $U \unionSP W $ \hideklawe{}s  the vertex~%
  $\bottomvertex{B}$,
  and that $U$ \hideklawe{}s the \mpctext $\mpconf$
  if $U$ \hideklawe{}s all
  $\mpscnotstd \in \mpconf$.  
\end{definition}

If $U$ does not \hideklawe $\mpscnot{B}{W}$,
then $U$ blocks $\mpscnot{B}{W}$
only if
$U \intersectionSP \abovevertices[\layeredgraphstd]{\bottomvertex{B}}$
does.

\begin{proposition}
  \label{pr:not-covering-means-blocking-from-above}
  Suppose that a vertex set $U$ in a layered DAG~$\layeredgraphstd$
  blocks but does not \hideklawe
  the \mpsctext~$\mpscnot{B}{W}$
  and that
  $\mpscnot{B}{W}$
  does not block itself.
  Then 
  $U \intersectionSP 
  \belowvertices[\layeredgraphstd]{\bottomvertex{B}}$
  does not block $\mpscnot{B}{W}$,
  but there is a subset
  $
  U' 
  \subseteq
  U \intersectionSP \abovevertices[\layeredgraphstd]{\bottomvertex{B}}
  $
  that blocks~$\mpscnot{B}{W}$.
\end{proposition}

\begin{proof}
  Suppose that
  $U \unionSP W$
  blocks 
  $B$
  but does not hide
  $b = \bottomvertex{B}$,
  and that
  $W$
  does not block~%
  $B$.
  Then there is a \sourcepath $P_2$ via $B$
  \st
  $P_2 \intersectionSP W = \emptyset$.
  Also,
  there is a \sourcepath $P_1$ to~$b$ \st
  $
  P_1 \intersectionSP (U \unionSP W) = \emptyset 
  $.
  Let
  $P =
  \bigl(
  P_1 \intersectionSP \belowvertices[\layeredgraphstd]{b}  
  \bigr)
  \unionSP 
  \bigl(
  P_2 \intersectionSP \abovevertices[\layeredgraphstd]{b}
  \bigr)
  $
  be the \sourcepath that starts like $P_1$
  and continues like $P_2$ from~$b$ onwards.
  Clearly,
  \begin{equation}
    P
    \intersectionSP
    \bigl(
    \bigl(
    U \intersectionSP \belowvertices[\layeredgraphstd]{b}
    \bigr)
    \unionSP
    W
    \bigr)
    =
    \bigl(
    P_1 \intersectionSP (U \unionSP W)
    \bigr)
    \unionSP
    \bigl(
    P_2 \intersectionSP W
    \bigr)
    = \emptyset
  \end{equation}
  so
  $  U \intersectionSP \belowvertices[\layeredgraphstd]{b}$
  does not block
  $\mpscnot{B}{W}$.

  Suppose that
  $U \intersectionSP \abovevertices[\layeredgraphstd]{b}$
  does not block
  $\mpscnot{B}{W}$.
  Since
  $U \unionSP W$
  does not \hideklawe~$b$,
  there is some \sourcepath $P_1$ to~$b$
  with
  $P_1 \intersectionSP (U \unionSP W) = \emptyset$.
  Also, since
  $U \unionSP W$
  blocks~$B$
  but
  $
  \bigl( U \intersectionSP \abovevertices[\layeredgraphstd]{b} \bigr) 
  \unionSP W
  $
  does not,
  there is a \sourcepath
  $P_2$
  via~$B$
  \st
  $
  P_2
  \intersectionSP
  (U \unionSP W)
  \neq \emptyset
  $
  but
  $
  P_2
  \intersectionSP
  (U  \unionSP W)
  \intersectionSP 
  \abovevertices[\layeredgraphstd]{b}
  = \emptyset
  $.
  But then let
  $P =
  \bigl(
  P_1 \intersectionSP \belowvertices[\layeredgraphstd]{b}  
  \bigr)
  \unionSP 
  \bigl(
  P_2 \intersectionSP \abovevertices[\layeredgraphstd]{b}
  \bigr)
  $
  be the \sourcepath that starts like $P_1$
  and continues like $P_2$ from~$b$ onwards.
  We get that
  $P$
  agrees with $B$
  and that
  $P \intersectionSP (U \unionSP W) = \emptyset$,
  contradicting the assumption that
  $U$ blocks~$\mpscnot{B}{W}$.
\end{proof}

We want to distinguish between \mpsctext{}s that are
\hiddenklawe and \mpsctext{}s that are just blocked, but not
\hiddenklawe.
To this end, let us introduce the notation
\begin{equation}
  \label{eq:definition-of-hidded-subconfigurations}
  \mpconfhidden (\mpconf, U)  =
  \Setdescr
      {\mpscnotstd \in \mpconf}
      {\text{$U$ \hideklawe{}s $\mpscnotstd$}}
\end{equation}
to denote the \mpsctext{}s in $\mpconf$ \hiddenklawe by $U$ and
\begin{equation}
  \label{eq:definition-of-just-blocked-subconfigurations}
  \mpconfjustblocked (\mpconf, U)  =
  \mpconf  \setminus    \mpconfhidden (\mpconf, U) 
\end{equation}
to denote the \mpsctext{}s that are just blocked. We write
\begin{align}
  \label{eq:black-bottom-vertices-hidden}
  \blackshidden  (\mpconf, U) 
  &=
  \setdescr
  {\bottomvertex{B}}
  {\mpscnot{B}{W} \in \mpconfhidden (\mpconf, U) } \\
  \label{eq:black-bottom-vertices-just-blocked}
  \blacksjustblocked (\mpconf, U) 
  &=
  \setdescr
  {\bottomvertex{B}}
  {\mpscnot{B}{W} \in \mpconfjustblocked (\mpconf, U) }
\end{align}
to denote the black bottom vertices in these two subsets of
\mpsctext{}s and note that we can have
$
\blackshidden (\mpconf, U) 
\intersectionSP
\blacksjustblocked (\mpconf, U) 
\neq \emptyset
$.
The white pebbles in these subsets located below the bottom vertices of
the black \multipebble{}s that they are supporting are denoted
\begin{align}
  \label{eq:whites-below-hidden}
  \whitesbelowhidden (\mpconf, U) 
  &=
  \Setdescr
      {W \intersectionSP \belowvertices[\layeredgraphstd]{b}}
      {\mpscnot{B}{W} \in \mpconfhidden (\mpconf, U) , \, 
        b = \bottomvertex{B}}
      \\
\shortintertext{and}
  \label{eq:whites-below-just-blocked}
  \whitesbelowjustblocked (\mpconf, U) 
  &=
  \Setdescr
      {W \intersectionSP \belowvertices[\layeredgraphstd]{b}}
      {\mpscnot{B}{W} \in \mpconfjustblocked (\mpconf, U) , \, 
        b = \bottomvertex{B}}
\eqperiod
\end{align}
This notation will be used heavily in what follows, so we give a
couple of simple but hopefully illuminating examples before we
continue.

\begin{figure}[tp]
  \centering
  \subfigure[%
    $
    \Set{
      \mpscnot{s_4,y_1,z}{v_2},
      \mpscnot{u_3,w_3}{s_3},
      \mpscnot{w_4,x_3}{v_5}}
    $.]	    
  {
    \label{fig:exampleblockingsets-a}
    \begin{minipage}[b]{.48\linewidth}
      \centering
      \includegraphics{hidingsets.10}%
    \end{minipage}
  }
  \hfill
  \subfigure[
    $
    \Set{ 
      \mpscnot{s_4,\! v_4,\! w_3,\! x_3,\! y_2}{\emptyset}, 
      \mpscnot{w_2,\! y_1}{s_3,\! u_3,\! x_1},  
      \mpscnot{w_4}{v_5}
    }$.]
  {
    \label{fig:exampleblockingsets-b}
    \begin{minipage}[b]{.48\linewidth}
      \centering
      \includegraphics{hidingsets.11}%
    \end{minipage}
  }
  \caption{Examples of \mpctext{}s with \hiddenklawe 
    and just blocked \multipebble{}s.}
  \label{fig:exampleblockingsets}
\end{figure}

\begin{example}
  \label{ex:hidden-and-just-blocked-notation}
  Consider the \mpctext{}s and blocking sets in
  \reffig{fig:exampleblockingsets}.
  For the \mpctext
  $
  \mpconf_1 =
  \Set{
    \mpscnot{s_4, y_1, z}{v_2},
    \mpscnot{u_3, w_3}{s_3},
    \mpscnot{w_4, x_3}{v_5}}
  $
  with blocking set
  $U_1 = \set{v_3, v_4}$
  in
  \reffig{fig:exampleblockingsets-a},
  the vertex set
  $\set{v_4, v_5}$
  \hideklawe{}s 
  $w_4 = \bottomvertex{\blackblobnot{w_4, x_3}}$
  but 
  $\blackblobnot{s_4, y_1, z}$
  is blocked but not \hiddenklawe by $\set{v_2, v_3, v_4}$
  and
  $\blackblobnot{u_3,w_3}$
  is blocked but not \hiddenklawe by $\set{v_3}$.
  Thus, we have
  \begin{align*}    
  \mpconfhidden (\mpconf_1, U_1)  
  &=
  \Set{\mpscnot{w_4,x_3}{v_5}}
  \\
  \mpconfjustblocked (\mpconf_1, U_1)  
  &=
  \Set{
    \mpscnot{s_4,y_1,z}{v_2},
    \mpscnot{u_3,w_3}{s_3}
  }
  \\
  \blackshidden  (\mpconf_1, U_1) 
  &=
  \set{w_4}
  \\
  \blacksjustblocked  (\mpconf_1, U_1) 
  &=
  \set{s_4, u_3}
  \\
  \whitesbelowhidden (\mpconf_1, U_1) 
  &=
  \set{v_5}
  \\
  \whitesbelowjustblocked (\mpconf_1, U_1) 
  &=
  \set{s_3}
  \\
  \intertext{
  in this example.  For the configuration
  $
  \mpconf_2 =
  \Set{ 
    \mpscnot{s_4, v_4, w_3, x_3, y_2}{\emptyset}, 
    \mpscnot{w_2, y_1}{s_3, u_3, x_1},  
    \mpscnot{w_4}{v_5}
  }$
  with blocker
  $U_2 = \set{s_2, u_4, u_5}$
  in
  \reffig{fig:exampleblockingsets-b},
  it is straightforward to verify that}
  \mpconfhidden (\mpconf_2, U_2)  
  &=
  \Set{ 
    \mpscnot{w_2, y_1}{s_3, u_3, x_1},  
    \mpscnot{w_4}{v_5}
  }
  \\
  \mpconfjustblocked (\mpconf_2, U_2)  
  &=
  \Set{ 
    \mpscnot{s_4, v_4, w_3, x_3, y_2}{\emptyset}
  }
  \\
  \blackshidden  (\mpconf_2, U_2) 
  &=
  \set{w_2, w_4}
  \\
  \blacksjustblocked  (\mpconf_2, U_2) 
  &=
  \set{s_4}
  \\
  \whitesbelowhidden (\mpconf_2, U_2) 
  &=
  \set{s_3, u_3, v_5}
  \\
  \whitesbelowjustblocked (\mpconf_2, U_2) 
  &=
  \emptyset
  \end{align*}
  are the corresponding sets.

  Let us also use the opportunity to illustrate
  \refdef{def:white-elimination}.
  The \mpctext
  $\mpconf_1$
  is not \weliminated \wrt $U_1$, since $U_1$ also blocks this
  configuration with the white pebble on $s_3$ removed.
  However, a better idea measure-wise is to change the blocking set
  for $\mpconf_1$
  to 
  $U'_1 = \set{s_4, v_4}$,
  which has measure
  $\vmeasure{U'_1} = 4 < 6 = \vmeasure{U_1}$.
  The vertex set $U_2$ can be verified to be a \minmeasure blocker for
  $\mpconf_2$,
  but when
  $\mpconf_2$ is \weliminated \wrt $U_2$ the white pebble on $x_1$
  disappears. 

  As a final remark in this example, we comment that
  although we have not indicated explicitly in
  \reftwofigs
      {fig:exampleblockingsets-a}
      {fig:exampleblockingsets-b}
  which white pebbles $W$ are associated with which black
  \multipebble~$B$
  (as was done in
  \reffig{fig:problematichidingset-a}), 
  this is uniquely determined by the requirement in
  \refdef{def:C-multi-pebble-configuration}
  that
  $W \subseteq \lpp{B}$.
\end{example}

For the rest of this \sectionBlobBound we will assume \wolog
(in view of
\refpr{pr:removing-self-blockers}
and
\refcor{cor:sufficient-prove-gkp-for-w-eliminated})
that we are dealing with \anmpctext
$\mpconf$
and a 
\minmeasure blocker 
$U$
of 
$\mpconf$
\st
$\mpconf$
is free from self-blocking \mpsctext{}s  and  is
\weliminated \wrt~$U$.
As an aside, we note that it is not hard to show
(using
\refdef{def:white-elimination}
and
\refpr{pr:not-covering-means-blocking-from-above})     
that this implies that
$\whitesbelowjustblocked (\mpconf, U) = \emptyset$.
We will tend to drop the arguments
$\mpconf$ and $U$
for
$
\mpconfhidden,
\mpconfjustblocked,
\blackshidden,
\blacksjustblocked,
\whitesbelowhidden$, and
$\whitesbelowjustblocked$,
since  from now on the \mpctext $\mpconf$ and the blocker $U$ will be 
fixed. 
With this notation,
\refth{th:pyramids-possess-GKP}
clearly follows if we can prove the following lemma.

\begin{lemma}
  \label{lem:pyramids-possess-GKP}
  Let
  $\mpconf$
  be any \mpctext on a 
  layered spreading DAG 
  and
  $U$
  be any blocking set for $\mpconf$
  \st
  \begin{enumerate}
  \item 
    $\vpotential{\mpconf} =
    {\meastopot{U}}$,
    \ie $U$ is a \minmeasure blocker of $\mpconf$,
    
  \item 
    $\mpconf$ is 
    free from self-blocking \mpsctext{}s  
    and is \weliminated \wrt~$U$,
    and

  \item 
    $U$ has minimal size among all blocking sets 
    $U'$ for $\mpconf$
    \st
    $\vpotential{\mpconf} =
    {\meastopot{U'}}$.
  \end{enumerate}
  Then
  $\setsize{U} \leq 
  13 \cdot 
  \BHunionBBunionWH
  $.
\end{lemma}

The proof is by contradiction,
although we will have to work harder than for the corresponding 
\refth{th:pyramids-have-klawe-property}
for black-white pebbling and also use (the proof of) the latter
theorem as a subroutine. 
Thus, for the rest of this section, let us assume on the contrary that
$U$ has all the properties stated in 
\reflem{lem:pyramids-possess-GKP}
but that
$
\setsize{U} > 
13 \cdot 
\BHunionBBunionWH
$.
We will show that this leads to a contradiction.

For the \mpsctext{} in
$\mpconfhidden$
that are \hiddenklawe by~$U$, one could argue that 
matters should be reasonably similar to the case for standard
black-white pebbling, and hopefully we could apply similar reasoning
as in 
\refsec{sec:klawe-bw-pebbling-proving-the-Klawe-property}
to prove something useful about the vertex set
\hidingklawe these \mpsctext{}s.
The \mpsctext{}s in
$\mpconfjustblocked$
that are just blocked but not \hiddenklawe, however, seem harder
to get a handle on
(compare  \refex{ex:non-connectedness}).

Let $\uhiding \subseteq U$ be a smallest vertex set
\hidingklawe
$\mpconfhidden$  
and let
$\ujustblocking = U \setminus \uhiding$.
The set
$\ujustblocking$ 
consists of vertices that are not involved in any
\hidingklawe of \mpsctext{}s
in~$\mpconfhidden$, but only in
blocking \mpsctext{}s in $\mpconfjustblocked$ on levels above their
bottom vertices.   
As a first step towards proving
\reflem{lem:pyramids-possess-GKP},
and thus
\refth{th:pyramids-possess-GKP},
we want to argue that $\ujustblocking$ cannot be very large.

Consider the \multipebble{}s in $\mpconfjustblocked$. By definition
they  are not
\hiddenklawe, but are blocked at some level above
$\vlevel{\bottomvertex{B}}$.
Since the vertices in
$\ujustblocking$
are located on high levels, a naive attempt to improve the blocking
set would be to pick some 
$u \in \ujustblocking$
and replace it by the  vertices in
$\blacksjustblocked$ corresponding to the \mpsctext{}s 
in~$\mpconfjustblocked$
that $u$ is
involved in blocking, 
\ie by the set
$
\bbvstd^{u} =
\Setdescr
{\bottomvertex{B}}
{U \setminus \set{u} \text{ does not block } 
  \mpscnotstd \in \mpconfjustblocked}
$.
Note that 
$\bbvstd^{u}$
is lower down in the graph than~$u$,
so 
$(U \setminus \set{u}) \unionSP \bbvstd^{u}$
is obtained from $U$ by moving vertices downwards
and by construction
$(U \setminus \set{u}) \unionSP \bbvstd^{u}$
blocks~$\mpconf$.
But by assumption, $U$~has minimal potential and cardinality, so this
new blocking set cannot be an improvement measure- or
cardinality-wise. 
The same holds if we extend the construction to
subsets $U' \subseteq \ujustblocking$
and the corresponding bottom vertices
$\bbvstd^{U'} \subseteq \blacksjustblocked$.
By assumption we can never find any subset such that
$(U \setminus \set{U'}) \unionSP \bbvstd^{U'}$
is a better blocker than~$U$. It follows that the
cost of the \multipebble{}s that $\ujustblocking$ helps to block must
be larger than the size of~$\ujustblocking$,
and in particular that
$
\setsize{\ujustblocking}
\leq
\setsize{\blacksjustblocked}
$.
Let us write this down as a lemma and prove it properly.

\begin{lemma}
  \label{lem:U-just-blocking-is-small}
  Let
  $\mpconf$
  be any \mpctext on a layered DAG
  and
  $U$
  be any blocking set for~$\mpconf$
  \st
  $\vpotential{\mpconf} =
  {\meastopot{U}}$,
  $U$ has minimal size among all blocking sets 
  $U'$ for~$\mpconf$
  with
  $\vpotential{\mpconf} =
  {\meastopot{U'}}$,
  and
  $\mpconf$ is 
  free from self-blocking \mpsctext{}s  
  and is 
  \weliminated \wrt~$U$.
  Then if
  $\uhiding \subseteq U$ 
  is any smallest set \hidingklawe
  $\mpconfhidden$
  and
  $\ujustblocking = U \setminus \uhiding$,
  it holds that
  $
  \setsize{\ujustblocking}
  \leq
  \setsize{\blacksjustblocked}
  $.
\end{lemma}

Before proving this lemma, we note the immediate corollary that if
the whole blocking set $U$ is
significantly 
larger than~$\mpcost{\mpconf}$,
the lion's share of $U$ by necessity consists not of vertices blocking 
\mpsctext{}s in $\mpconfjustblocked$, but of vertices 
\hidingklawe  \mpsctext{}s in~$\mpconfhidden$.
And recall that we are indeed assuming,
to get a contradiction,
that $U$ is large.

\begin{corollary}
  \label{cor:U-hiding-is-large}
  Assume that
  $\mpconf$ and $U$ 
  are as in
 \reflem{lem:pyramids-possess-GKP}
  but with
  $
  \setsize{U} 
  \!>\!
  13 \cdot 
  \BHunionBBunionWH
  $.
  Let 
  $\uhiding \subseteq U$ 
  be a smallest set \hidingklawe
  $\mpconfhidden$.
  Then
  $
  \setsize{\uhiding}
  >
  12 \cdot
  \BHunionBBunionWH
  $.
\end{corollary}

As was indicated  in the informal discussion preceding
\reflem{lem:U-just-blocking-is-small}, 
the proof of the lemma uses the easy observation that moving vertices
downwards can only decrease the measure.

\begin{observation}
  \label{obs:moving-vertices-downwards-non-increases-measure}
  Suppose that
  $U$, $V_1$ and $V_2$
  are 
  vertex sets in a layered DAG
  \st 
  $U \intersectionSP V_2 = \emptyset$
  and
  there is a one-to-one (but not necessarily onto)
  mapping
  $\funcdescr{f}{V_1}{V_2}$
  with the property that
  $
  \vlevel{v}
  \leq
  \vlevel{f(v)}
  $.
  Then
  $
  \vmeasure{U \unionSP V_1} \leq  \vmeasure{U \unionSP V_2}
  $.
\end{observation}

\begin{proof}
  This follows immediately 
  from
  \refdefP{def:measure}
  since the mapping $f$ tells us that
  \begin{multline*}
  \setsize{\vertabovelevel{(U \unionSP V_1)}{j}}
  \,\leq\,      
  \setsize{\vertabovelevel{U}{j}}
  +
  \setsize{\vertabovelevel{V_1}{j}}
  \,\leq\,      
  \setsize{\vertabovelevel{U}{j}}
  +
  \setsize{f(\vertabovelevel{V_1}{j})}
  \\
  \,\leq\,      
  \setsize{\vertabovelevel{U}{j}}
  +
  \setsize{\vertabovelevel{V_2}{j}}
  \,\leq\,      
  \setsize{\vertabovelevel{(U \unionSP V_2)}{j}}     
  \end{multline*}
  for all~$j$.
\end{proof}

\begin{proof}[Proof of   \reflem{lem:U-just-blocking-is-small}]
  Note first that by
  \refpr{pr:not-covering-means-blocking-from-above},
  for every
  $\mpscnotstd \in \mpconfjustblocked$
  with
  $b = \bottomvertex{B}$
  it holds that
  $
  U \intersectionSP \abovevertices[\layeredgraphstd]{b}
  =
  (\uhiding \disjointunionSP \ujustblocking) 
  \intersectionSP \abovevertices[\layeredgraphstd]{b}
  $
  blocks
  $\mpscnotstd$.
  Therefore, 
  all vertices in $\ujustblocking$
  needed to block $\mpscnotstd$
  can be found in
  $
  \ujustblocking
  \intersectionSP \abovevertices[\layeredgraphstd]{b}
  $.
  Rephrasing this slightly,
  the \mpctext $\mpconf$ is blocked by
  $
  \uhiding
  \disjointunionSP
  \bigl(
  \ujustblocking
  \intersectionSP
  \Union_{b \in \blacksjustblocked} \abovevertices[\layeredgraphstd]{b}
  \bigr) 
  $,
  and since $U$ is \subsetminimal we get that
  \begin{equation}
    \label{eq:U-just-blocking-is-all-above}
    \ujustblocking
    =
    \ujustblocking
    \intersectionSP
    {\textstyle
    \Union_{b \in \blacksjustblocked}
    }
    \abovevertices[\layeredgraphstd]{b}
    \eqperiod
  \end{equation}
  Consider the bipartite graph with
  $\blacksjustblocked$
  and
  $\ujustblocking$
  as the left- and right-hand vertices,
  where the neighbours of each
  $b \in \blacksjustblocked$
  are the vertices
  $
  \vneighbour{b} =
  \ujustblocking \intersectionSP \abovevertices[\layeredgraphstd]{b}
  $
  in $\ujustblocking$ above~$b$.  We have that
  $
  \vneighbour{\blacksjustblocked} 
  =
  \ujustblocking
  \intersectionSP
  \Union_{b \in \blacksjustblocked} \abovevertices[\layeredgraphstd]{b}
  = \ujustblocking
  $
  by~%
  \refeq{eq:U-just-blocking-is-all-above}.
  Let
  $
  \bbvstd' \subseteq \blacksjustblocked
  $
  be a largest set \st
  $  
  \Setsize{\vneighbourcompact{\bbvstd'}}
  <
  \Setsize{\bbvstd'}
  $.
  If
  $\bbvstd' = \blacksjustblocked$
  we are done since this is the inequality
  $\setsize{\ujustblocking} < \setsize{\blacksjustblocked}$.
  Suppose therefore that
  $
  \bbvstd' \subsetneqq \blacksjustblocked
  $
  and
  $
  \setsize{\ujustblocking}
  =
  \setsize{\vneighbour{\blacksjustblocked}}
  >
  \setsize{\blacksjustblocked}
  $.

  For all
  $\bbvstd'' \subseteq \blacksjustblocked \setminus \bbvstd'$
  we must have
  $
  \Setsize{
    \vneighbourcompact{\bbvstd''}
    \setminus
    \vneighbourcompact{\bbvstd'}
  }
  \geq
  \Setsize{\bbvstd''}
  $,
  for otherwise
  $\bbvstd''$ could be added to $\bbvstd'$
  to yield an even larger set 
  $\bbvstd^* = \bbvstd' \unionSP \bbvstd''$
  with
  $  
  \Setsize{\vneighbourcompact{\bbvstd^*}}
  <
  \setsize{\bbvstd^*}
  $
  contrary to the assumption that
  $\bbvstd'$
  has maximal size among all sets with this property.
  It follows by Hall's marriage theorem that there must exist a matching of
  $\blacksjustblocked \setminus \bbvstd'$
  into
  $
  \vneighbourcompact{\blacksjustblocked \setminus \bbvstd'}
  \setminus
  \vneighbourcompact{\bbvstd'}
  =
  \ujustblocking
  \setminus
  \vneighbourcompact{\bbvstd'}
  $.
  Thus,
  $
  \Setsize{
    \blacksjustblocked \setminus \bbvstd'
  }
  \leq
  \Setsize{
    \ujustblocking
    \setminus
    \vneighbourcompact{\bbvstd'}
  }
  $
  and in addition it follows from the way our bipartite graph is
  constructed that   every
  $b \in \blacksjustblocked \setminus \bbvstd'$
  is matched to some
  $u \in \ujustblocking \setminus \vneighbourcompact{\bbvstd'}$
  with
  $\vlevel{u} \geq \vlevel{b}$.

  Clearly,
  all \mpsctext{}s in
  \begin{equation}
    \mpconfjustblocked^{1} =
    \Setdescr
        {\mpscnotstd \in \mpconfjustblocked}
        {\bottomvertex{B} \in 
          \blacksjustblocked \setminus \bbvstd'
        }    
  \end{equation}
  are blocked by
  $
  \blacksjustblocked \setminus \bbvstd'
  $
  (even \hiddenklawe by this set, to be precise).
  Also, as was argued in the beginning of the proof,
  every
  $\mpscnotstd \in \mpconfjustblocked$
  with
  $b = \bottomvertex{B}$
  is blocked by
  $
  \uhiding
  \unionSP
  \bigl(
  \ujustblocking \intersectionSP \abovevertices[\layeredgraphstd]{b}
  \bigr)
  =
  \uhiding
  \unionSP
  \vneighbour{b}
  $,
  so all \mpsctext{}s in
  \begin{equation}
    \mpconfjustblocked^{2} =
    \Setdescr
        {\mpscnotstd \in \mpconfjustblocked}
        {\bottomvertex{B} \in \bbvstd'}
  \end{equation}
  are blocked by
  $\uhiding \unionSP \vneighbourcompact{\bbvstd'}$
  where
  $
  \Setsize{\vneighbourcompact{\bbvstd'}}
  <
  \Setsize{\bbvstd'}
  $.
  And we know that
  $\mpconfhidden$
  is blocked (even \hiddenklawe) by
  $\uhiding$.
  It follows that if we let
  \begin{equation}
    \label{eq:new-U-for-lem-U-just-blocking-is-small}
    U^*
    =
    \uhiding 
    \unionSP 
    \vneighbourcompact{\bbvstd'}
    \unionSP
    \bigr(
    \blacksjustblocked \setminus \bbvstd' 
    \bigl)
  \end{equation}
  we get a vertex set $U^*$ that
  blocks
  $
  \mpconfhidden
  \unionSP
  \mpconfjustblocked^{1}
  \unionSP
  \mpconfjustblocked^{2}
  =
  \mpconf
  $,
  has measure
  $
  \vmeasurecompact{U^*}
  \leq
  \vmeasure{U}
  $
  because of
  \refobs{obs:moving-vertices-downwards-non-increases-measure},
  and has size
  \begin{equation}
    \Setsize{U^*}
    \leq
    \setsize{\uhiding} 
    +
    \Setsize{\vneighbourcompact{\bbvstd'}}
    +
    \Setsize{\blacksjustblocked \setminus \bbvstd'} 
    <
    \setsize{\uhiding} 
    +
    \Setsize{\bbvstd'}
    +
    \Setsize{\blacksjustblocked \setminus \bbvstd'}
    =
    \setsize{U}
  \end{equation}
  strictly less than the size of~$U$.
  But this is a contradiction, since
  $U$ was chosen to be of minimal size.
  The lemma follows.
\end{proof}

The idea in the remaining part  of the proof is as follows:
Fix some smallest subset
$\uhiding \subseteq U$ 
that \hideklawe{}s
$\mpconfhidden$,
and let
$\ujustblocking = U \setminus \uhiding$.
\Refcor{cor:U-hiding-is-large}
says that
$\uhiding$
is the totally dominating part of~$U$
and hence that
$\uhiding$
is very large. But
$\uhiding$
\hideklawe{}s
the \mpscfulltext{}s in
$\mpconfhidden$
very much in a similar way as for \hidingsetklawe{}s in the standard
black-white pebble game. And we know from 
\refsec{sec:klawe-bw-pebbling-proving-the-Klawe-property}
that such sets need not be very large.
Therefore we want to use Klawe-like ideas to derive a contradiction by
transforming $\uhiding$ locally into a
(much) better blocking set for~$\mpconfhidden$.
The problem is that this might leave some \mpsctext{}s in
$\mpconfjustblocked$
not being blocked any longer
(note that in general
$\ujustblocking$
will \emph{not} on its own block~%
$\mpconfjustblocked$).
However, since we have chosen our parameter
\mbox{$\gkpconstant = 13$}  
for the
\genklawepropacronym{}~%
\ref{property:generalized-klawe-property}
so generously and since the transformation in
\refsec{sec:klawe-bw-pebbling-proving-the-Klawe-property}
works for the
\nongenklawepropacronymWithParam with parameter~$1$, 
we expect our locally transformed blocking set to be so much cheaper
that we can afford to take care of any \mpsctext{}s in~%
$\mpconfjustblocked$
that are no longer blocked
simply by adding all bottom vertices for all
black \multipebble{}s in these \mpsctext{}s to the blocking set.

We will not be able to pull this off by just making one local
improvement of the \hidingsetklawe as was done in
\refsec{sec:klawe-bw-pebbling-proving-the-Klawe-property},
though.
The reason is that the local improvement to 
$\uhiding$
could potentially be very small, but lead to very many
\mpsctext{}s in 
$\mpconfjustblocked$
becoming unblocked. If so, we cannot afford adding new vertices
blocking these \mpsctext{}s without risking to increase the size
and/or potential of our new blocking set too much.
To make sure that this does not happen, we instead make
multiple local improvements of
$\uhiding$
simultaneously. 
Our next lemma says that we can do this without losing control of how
the measure behaves.

\begin{lemma}[Generalization of
  \reflem{lem:union-respecting-leq}]
  \label{lem:generalization-union-respecting-leq}
  Suppose that
  $
  U_1, \ldots, U_k,
  V_1, \ldots, V_k,
  Y
  $
  are vertex sets in a layered graph 
  \st  for all
  $i, j \in \intnfirst{k}$, $i \neq j$,
  it holds that
    $U_i \measureleq V_i$,
    $V_i \intersectionSP V_j = \emptyset$,
    $U_i \intersectionSP V_j = \emptyset$ and  
    $Y \intersectionSP V_i = \emptyset$.
  Then 
  $
  \vmeasurecompact{Y \unionSP \Union_{i=1}^{k} U_i} \leq
  \vmeasurecompact{Y \unionSP \Union_{i=1}^{k} V_i}
  $.
\end{lemma}

\begin{proof}
  By induction over~$k$.
  The base case $k=1$ is
  \reflemP{lem:union-respecting-leq}.
  
  For the induction step,
  let
  $
  Y' = Y \unionSP \Union_{i=1}^{k-1} U_i
  $.
  Since
  $U_k \measureleq V_k$
  and
  $Y' \intersectionSP V_k = \emptyset$
  by assumption,
  we get from
  \reflem{lem:union-respecting-leq}
  that
  \begin{equation}
  \vmeasurecompact{Y \unionSP 
    \Unionnodisplay_{i=1}^{k} U_i}
  =
  \vmeasurecompact{Y' \unionSP U_k}
  \leq
  \vmeasurecompact{Y' \unionSP V_k}
  =
  \vmeasurecompact{Y \unionSP 
    \Unionnodisplay_{i=1}^{k-1} U_i \unionSP V_k}
  \eqperiod
  \end{equation}
  Letting
  $
  Y'' = Y \unionSP V_k
  $,
  we see that (again by assumption) it holds for all
  $i, j \in \intnfirst{k-1}$, $i \neq j$,
  that
  $U_i \measureleq V_i$,
  $V_i \intersectionSP V_j = \emptyset$,
  $U_i \intersectionSP V_j = \emptyset$ and 
  $Y'' \intersectionSP V_i = \emptyset$.
  Hence, by the induction hypothesis we have
%
%
  \begin{equation}
    \vmeasurecompact{Y 
      \union  
      \Unionnodisplay_{i=1}^{k-1} U_i \unionSP V_k}
    \!=\!  
    \vmeasurecompact{Y'' 
      \union  
      \Unionnodisplay_{k=1}^{i-1} U_i
    }
    \!\leq\!   
    \vmeasurecompact{Y'' 
      \union  
      \Unionnodisplay_{k=1}^{i-1} V_i
    }
    \!=\!  
    \vmeasurecompact{Y 
      \union  
      \Unionnodisplay_{i=1}^{k} V_i}
  \end{equation}
  and the lemma follows.
\end{proof}

We also need an observation  about the white pebbles in
$\mpconfhidden$.

\begin{observation}
  \label{obs:all-white-pebbles-in-mpconfhidden-are-below}
  For any
  $\mpscnot{B}{W} \in \mpconfhidden$
  with
  $b = \bottomvertex{B}$
  it holds that
  $
  W 
  = 
  W \intersectionSP \belowvertices[\layeredgraphstd]{b}
  $.
\end{observation}

\begin{proof}
  This is so since
  $\mpconf$
  is \weliminated
  \wrt{}~$U$.
  Since
  $U \unionSP W$
  \hideklawe{}s
  $b = \bottomvertex{B}$,
  any vertices in
  $W \intersectionSP \abovevertices[\layeredgraphstd]{b}$
  are superfluous and will be removed by the
  \welimination procedure in
  \refdef{def:white-elimination}.
\end{proof}

Recalling from
\refeq{eq:whites-below-hidden}
that
$
\whitesbelowhidden
=
\Setdescr
  {W \intersectionSP \belowvertices[\layeredgraphstd]{b}}
  {\mpscnot{B}{W} \in \mpconfhidden, \, b = \bottomvertex{B}}
$
this leads to the next, simple but crucial observation.

\begin{observation}
  \label{obs:uhiding-union-whites-below-cover-B}
  The vertex set
  $\uhiding \unionSP \whitesbelowhidden$
  \hideklawe{}s
  the vertices in
  $\blackshidden$
  in the sense of
  \refdef{def:cover}.
\end{observation}

That is, we can consider
$
\bigl(
\blackshidden,
\whitesbelowhidden
\bigr)
$
to be 
(almost)%
\footnote{%
Not quite, since we might have
$\blackshidden
\intersectionSP
\whitesbelowhidden
\neq
\emptyset$.
But at least we know that
$\uhiding \intersectionSP \whitesbelowhidden = \emptyset$
by \welimination
and the roles of $U$ and $W$ in $U \unionSP W$ are fairly
indistinguishable in Klawe's proof anyway, 
so this does not matter.
}
a standard black-white pebble configuration.
This sets the stage for applying the machinery of
\refsec{sec:klawe-bw-pebbling-proving-the-Klawe-property}.

Appealing to
\reflemP{lem:covering-set-is-tight},
let
$X \subseteq
\uhiding \disjointunionSP \whitesbelowhidden$ 
be the unique, minimal \tightklawe set \st
\begin{equation}
  \label{eq:X-has-same-cover}
  \hiddenvertices{X}
  =
  \hiddenvertices{\uhiding \disjointunionSP \whitesbelowhidden}
\end{equation}
and define
\begin{subequations}
\begin{align}
  \label{eq:whitestight}
  \whitestight 
  &= 
  \whitesbelowhidden \intersectionSP X
  \\
  \label{eq:utight}
  \utight 
  &= 
  \uhiding \intersectionSP X
\end{align}
\end{subequations}
to be the vertices in
$\whitesbelowhidden$ and $\uhiding$ that remains in~$X$
after the bottom-up pruning procedure of
\reflem{lem:covering-set-is-tight}.

Let
$\hidsetgraph = \hidsetgraph(\layeredgraphstd, X)$
be the \hidingsetklawe graph of
\refdef{def:covering-set-graph}
for
$
X = \utight \disjointunionSP \whitestight
$.
Suppose that
$V_1, \ldots, V_k$ 
are the connected components of~$\hidsetgraph$,
and define for $i = 1, \ldots, k$
the vertex sets
\begin{subequations}
\begin{align}
  \label{eq:blackshiddenith}
  \blackshiddenith 
  &=
  \blackshidden \intersectionSP V_i
  \\
  \label{eq:whitesbelowhiddenith}
  \whitesbelowhiddenith 
  &=
  \whitesbelowhidden \intersectionSP V_i
  \\
  \label{eq:uhidingith}
  \uhidingith 
  &=
  \uhiding \intersectionSP V_i
  \\
  \intertext{to be the
    black, white and ``\hidingklawe'' vertices within
    component~$V_i$,
    and}
  \label{eq: whitestightith}
  \whitestightith
  &=
  \whitestight \intersectionSP V_i
  \\
  \label{eq:utightith}
  \utightith 
  &=
  \utight \intersectionSP V_i
\end{align}
\end{subequations}
to be the vertices of $\whitesbelowhidden$ and $\uhiding$  in
component~$V_i$ that ``survived'' when moving to the
\tightklawe subset~$X$.
Note that we have the disjoint union equalities
$
\whitesbelowhidden
= 
\DisjointunionInText_{i=1}^{k} \whitesbelowhiddenith
$,
$
\uhiding 
= 
\DisjointunionInText_{i=1}^{k} \uhidingith
$,
et cetera for all of these sets.

Let us also generalize
\refdef{def:measure}
of measure and partial measure to multi-sets of vertices in the
natural way, where we charge separately for each copy of every
vertex. 
This is our way of doing the bookkeeping for the extra vertices
that might be needed later to block $\mpconfjustblocked$ in the final
step of our construction.

This brings us to the key lemma stating how we will locally improve
the blocking sets.

\begin{lemma}[Generalization of
    \reflem{lem:pick-local-good-blocker}]
  \label{lem:generalized-pick-local-good-blocker}
  With the assumptions on 
  the \mpctext 
  $\mpconf$ and 
  the vertex set
  $U$ as in
  \reflem{lem:pyramids-possess-GKP}
  and with notation as above,
  suppose that
  $\uhidingith \unionSP \whitesbelowhiddenith$
  \hideklawe{}s
  $\blackshiddenith$,
  that
  $
  \hidsetgraph
  \bigl(
  \utightith \unionSP \whitestightith
  \bigr)
  $
  is a connected graph, and that
  \begin{equation}
    \label{eq:U-h-ith-geq-six-B-h-ith-plus-W-h-ith}
    \Setsize{\uhidingith}
    \geq
    6 \cdot
    \BHiunionWHi
    \eqperiod
  \end{equation}
  Then we can find a multi-set
  $
  \unewhidingith 
  \subseteq
  \hiddenvertices{\utightith \unionSP \whitestightith}
  $ 
  that  \hideklawe{}s 
  the vertices in 
  $\blackshiddenith$,
  has
  $\floorcompact{\setsize{\uhidingith} / 3}$
  extra copies of some fixed but arbitrary 
  vertex on level
  $\levelmax = \vmaxlevelcompact{\uhidingith}$,
  and satisfies
  $\unewhidingith \measureleq \uhidingith$
  and
  $\Setsize{\unewhidingith} < \Setsize{\uhidingith}$
  (where $\unewhidingith$ is
  measured and counted  as a multi-set with repetitions).
\end{lemma}

\begin{proof}
  Let
  $\unewhidingith$
  be the set found in
  \reflemP{lem:klawe-lemma-three-five},
  which certainly is in
  $\hiddenvertices{\utightith \unionSP \whitestightith}$,
  together with the prescribed extra copies of some
  (fixed but arbitrary)
  vertex that we place on level
  $\vmaxlevelcompact{
    \hiddenvertices{\uhidingith \unionSP
      \whitesbelowhiddenith}}
  \geq
  \levelmax
  $
  to be on the safe side.
  By
  \reflem{lem:klawe-lemma-three-five},
  $\unewhidingith$
  \hideklawe{}s $\blackshiddenith$,
  and the size of 
  $\unewhidingith$
  \emph{counted as a multi-set with repetitions}
  is
  \begin{equation}
    \label{eq:gkp-size-of-new-local-blocker}
    \Setsize{\unewhidingith}
    \leq
    \Setsize{\blackshiddenith}
    + 
    \floorcompact{\setsize{\uhidingith} / 3}
    \leq
    \bigl(
    {\textstyle     \frac{1}{6} + \frac{1}{3} }
    \bigr)   
    \cdot
    \Setsize{\uhidingith}
    <
    \Setsize{\uhidingith}
    \eqperiod
  \end{equation}
  It remains to show that 
  $\unewhidingith \measureleq \uhidingith$.

  The proof of this last measure inequality is very much as in
  \reflem{lem:pick-local-good-blocker},
  but with the distinction that the connected graph that we are
  dealing with is defined over
  $\utightith \disjointunionSP \whitestightith$,
  but we count the vertices in
  $\uhidingith \disjointunionSP \whitesbelowhiddenith$.
  Note, however, that by construction these two unions \hideklawe
  exactly the same set of vertices, \ie
  \begin{equation}
    \label{eq:cover-same-set-of-vertices}
    \hiddenvertices{\utightith \disjointunionSP \whitestightith}
    =
  \hiddenvertices{\uhidingith \disjointunionSP \whitesbelowhiddenith}
  \eqperiod
  \end{equation}
  Recall that by
  \refdefP{def:union-respecting-leq},
  what we need to do
  in order to show that $\unewhidingith \measureleq \uhidingith$
  is to find for each $j$ an $l \leq j$ \st
  $
  \vjthmeasurecompact{j}{\unewhidingith}
  \leq
  \vjthmeasurecompact{l}{\uhidingith}
  $.
  As in
  \reflem{lem:pick-local-good-blocker},
  we divide the proof into two cases.
  \begin{enumerate}
  \item
    If
    $j \leq 
    \vminlevelcompact{\utightith \unionSP \whitestightith}
    =
    \vminlevelcompact{\uhidingith \unionSP \whitesbelowhiddenith}
    $,
    we get
    \begin{align*}
      \vjthmeasurecompact{j}{\unewhidingith} 
      &= j + 2 \cdot \Setsize{\vertabovelevel{\unewhidingith}{j}} 
      && 
      \bigl[\text{ by definition of $\vjthmeasure{j}{\cdot}$ }\bigr] 
      \\
      &\leq j + 2 \cdot \Setsize{\unewhidingith} 
      && 
      \bigl[\text{ since $\vertabovelevel{V}{j} \subseteq V$
                 for any $V$ }\bigr] 
      \\
      &\leq j + 2 \cdot 
      \bigl(
      \setsize{\blackshiddenith} +   
      \floorcompact{\setsize{\uhidingith} / 3} 
      \bigr)
      && 
      \bigl[\text{ by  
          \reflem{lem:klawe-lemma-three-five}
          plus extra vertices }\bigr] 
      \\
      &<  
      j + 2 \cdot \Setsize{\uhidingith}
      && 
      \bigl[\text{ by the assumption in
          \refeq{eq:U-h-ith-geq-six-B-h-ith-plus-W-h-ith} }\bigr]
      \\
      &= j + 2 \cdot \Setsize{\vertabovelevel{\uhidingith}{j}}
      && 
      \bigl[\text{ $\vertabovelevel{\uhidingith}{j} = \uhidingith$ 
          since $j \leq \vminlevel{\uhidingith} $ }\bigr] 
      \\
      &= \vjthmeasure{j}{\uhidingith}
      && 
      \bigl[\text{ by definition of $\vjthmeasure{j}{\cdot}$ }\bigr]
    \end{align*}
    and we can choose $l=j$ in
    \refdef{def:union-respecting-leq}.

  \item
    Consider instead
    $j >
    \vminlevelcompact{\utightith \unionSP \whitestightith}
    $
    and let
    $\levelmin =
    \vminlevelcompact{\utightith \unionSP \whitestightith}$.
    Since the black pebbles in
    $\blackshiddenith$
    are \hiddenklawe by 
    $\utightith \unionSP \whitestightith$,
    \ie
    $
    \blackshiddenith
    \subseteq
    \hiddenvertices{\utightith \unionSP \whitestightith}$
    in formal notation,
    recollecting
    \refdef{def:size-above-level-blocker}
    and
    \refobs{obs:about-min-size-covers},
    \refpart{item:about-min-size-covers-two},
    we see that
    \begin{equation}
      \label{eq:gkp-j-above-minlevel-eq-one}
      \abovelevelblockerminsizecompact{j}{\blackshiddenith}
      \leq
      \abovelevelblockerminsizecompact
          {j}
          {\hiddenvertices{\utightith \unionSP \whitestightith}}
    \end{equation}
    for all $j$.
    Also, since
    $\utightith \unionSP \whitestightith$
    is a \connectedklawe
    vertex set in a spreading graph~$\layeredgraphstd$,
    combining 
    \refdef{def:spreading-graph}
    with the fact that
    $
    \utightith \unionSP \whitestightith
    \subseteq
    \uhidingith \unionSP \whitesbelowhiddenith
    $
    we can derive that 
    \begin{equation}
      \label{eq:gkp-j-above-minlevel-eq-two}
      j + 
      \abovelevelblockerminsizecompact
          {j}
          {\hiddenvertices{\utightith \unionSP \whitestightith}}
      \leq
      \levelmin
      +
      \Setsize{\utightith \unionSP \whitestightith}
      \leq
      \levelmin
      +
      \Setsize{\uhidingith \unionSP \whitesbelowhiddenith}
      \eqperiod
    \end{equation}
    Together, 
    \refeq{eq:gkp-j-above-minlevel-eq-one}
    and
    \refeq{eq:gkp-j-above-minlevel-eq-two}
    say that
    \begin{equation}
      \label{eq:gkp-j-above-minlevel-eq-three}
      j + 
      \abovelevelblockerminsizecompact{j}{\blackshiddenith}
      \leq
      \levelmin
      +
      \Setsize{\uhidingith \unionSP \whitesbelowhiddenith}
    \end{equation}
    and using this inequality we can show that
    \begin{align*}
      \vjthmeasure{j}{\unewhidingith} 
      &= 
      j + 2 \cdot \Setsize{\vertabovelevel{\unewhidingith}{j}} 
      && 
      \bigl[\text{ by definition of $\vjthmeasure{j}{\cdot}$ }\bigr] 
      \\
      &\leq j 
      + 
      \abovelevelblockerminsizecompact{j}{\blackshiddenith}
      + 
      \Setsize{\blackshiddenith} 
      +
      2 \cdot \floorcompact{\setsize{\uhidingith} / 3} 
      && 
      \bigl[\text{ by \reflem{lem:klawe-lemma-three-five}
          + 
          extra vertices }\bigr] 
      \\
      &\leq
      \levelmin
      + \Setsize{\uhidingith \unionSP \whitesbelowhiddenith}
      + \Setsize{\blackshiddenith}
      + 2 \cdot \floorcompact{\setsize{\uhidingith} / 3} 
      && 
      \bigl[\text{ using the inequality 
          \refeq{eq:gkp-j-above-minlevel-eq-three} }\bigr]
      \\
      &\leq
      \levelmin
      + 
      {\textstyle \frac{5}{3}}
      \Setsize{\uhidingith}
      +
      \Setsize{\blackshiddenith}
      + 
      \Setsize{\whitesbelowhiddenith}      
      && 
      \bigl[
      \text{ $\setsize{A \unionSP B} \leq \setsize{A} + \setsize{B}$ }
      \bigr]
      \\
      &\leq
      \levelmin
      + 
      {\textstyle \frac{5}{3}}
      \Setsize{\uhidingith}
      +
      2 \cdot
      \Setsize{\blackshiddenith \unionSP \whitesbelowhiddenith}
      &&
      \bigl[
      \text{ $\setsize{A} + \setsize{B} \leq 
        2 \cdot \setsize{A \unionSP B}$ }
      \bigr]
      \\
      &\leq
      \levelmin
      + 
      2 \cdot
      \Setsize{\uhidingith}
      && 
      \bigl[\text{ by the assumption in
          \refeq{eq:U-h-ith-geq-six-B-h-ith-plus-W-h-ith} }\bigr]
      \\
      &= \levelmin 
      + 2 \cdot 
      \setsize{\vertabovelevel{\uhidingith}{\levelmin}}
      &&
      \bigl[\text{ 
        since $\levelmin \leq \vminlevel{\uhidingith} $ }\bigr] 
      \\
      &= \vjthmeasure{\levelmin}{\uhidingith}
      && 
      \bigl[\text{ 
          by definition of $\vjthmeasure{\levelmin}{\cdot}$ 
        }\bigr]
    \end{align*}
    Thus, the partial measure of $\uhidingith$ at the minimum level
    $\levelmin$
    is always at least as large as
    the partial measure of $\unewhidingith$
    at levels $j$ above this minimum level,
    and we can choose $l=\levelmin$ in
    \refdef{def:union-respecting-leq}.
  \end{enumerate}
  Consequently,
  $\unewhidingith \measureleq \uhidingith$
  and the lemma follows.
\end{proof}

Now we want to determine in which connected components 
of
the \hidingsetklawe graph~$\hidsetgraph$ we should apply
\reflem{lem:generalized-pick-local-good-blocker}.
Loosely put, we want to be sure that changing
$\uhidingith$
to 
$\unewhidingith$
is worthwhile, \ie that we gain enough from this transformation to
compensate for the extra hassle of reblocking \multipebble{}s in
$\mpconfjustblocked$ 
that turn unblocked when we change~%
$\uhidingith$. 
With this in mind,   
let us define the
\introduceterm{weight}
of a component $V_i$ in $\hidsetgraph$ as
\begin{equation}
\label{eq:definition-component-weight}
  \weightklawe{V_i}
  =    
    \begin{cases}
      \ceilingcompact{\setsize{\uhidingith} / 6}
      & 
      \text{if 
        $
        \Setsize{\uhidingith}
        \geq
        6 \cdot
        \BHiunionWHi
        $,}
      \\
      0 & \text{otherwise.}
  \end{cases}
\end{equation}
The idea is that a component $V_i$ has large weight if the
\hidingsetklawe $\uhidingith$ in this component is large compared to
the number of bottom black vertices in
$\blackshiddenith$
\hiddenklawe and the white pebbles
$\whitesbelowhiddenith$
helping 
$\uhidingith$
to \hideklawe~$\blackshiddenith$.
If we concentrate on changing the \hidingsetklawe{}s in components
with non-zero weight, we hope to gain more from the
transformation of $\uhidingith$ into $\unewhidingith$ than we lose
from then having to reblocking~$\mpconfjustblocked$.
And since $\uhiding$ is large, the total weight of the non-zero-weight
components is guaranteed to be reasonably large. 

\begin{proposition}
  \label{pr:weight-of-non-zero-components-is-large}
  With notation as above, the total weight of all connected components
  $V_1, \ldots, \! V_k$
  in the \hidingsetklawe graph
  $
  \hidsetgraph =
  \hidsetgraph\bigl(\layeredgraphstd, \utight \unionSP \whitestight\bigr)
  $
  is
  $
  \sum_{i=1}^{k} \weightklawe{V_i}   
  >
  \BHunionBBunionWH
  $.
\end{proposition}

\begin{proof}
  The total size of the union of all subsets
  $\uhidingith \subseteq \uhiding$
  with sizes
  $
  \Setsize{\uhidingith}
  <
  6 \cdot \BHiunionWHi
  $
  resulting in zero-weight components $V_i$ in $\hidsetgraph$
  is clearly  strictly less than
  \begin{equation}
  6 \cdot \sum_{i=1}^{k} \BHiunionWHi
%
%
  =
  6 \cdot \BHunionWH
  \leq
  6 \cdot \BHunionBBunionWH
  \eqperiod    
  \end{equation}
  Since  according to
  \refcor{cor:U-hiding-is-large}
  we have that
  $
  \Setsize{\uhiding}
  \geq
  12 \cdot \BHunionBBunionWH
  $,
  it follows that the size of the union
  $
  \Union_{\weightklawe{V_i} > 0} 
  \uhidingith
  $
  of all subsets
  $\uhidingith$
  corresponding to non-zero-weight components $V_i$
  must be strictly larger than
  $
  6 \cdot \BHunionBBunionWH
  $.
  But then
  \begin{equation}
    \sum_{\weightklawe{V_i} > 0} \weightklawe{V_i}
    \geq
    \sum_{\weightklawe{V_i} > 0} 
    \ceilingcompact{\setsize{\uhidingith} / 6}
    \geq
    \frac{1}{6}
    \cdot
    \setsizelarge{
      \Union_{\weightklawe{V_i} > 0} 
      \uhidingith
    }
    >
    \BHunionBBunionWH
  \end{equation}
  as claimed in the proposition.
\end{proof}

We have now collected all tools needed to establish the
\genklaweprop for spreading graphs.  Before we wrap up the proof, let
us recapitulate what we have shown so far.

We have divided the blocking set $U$ into a  disjoint union
$\uhiding \disjointunionSP \ujustblocking$
of the vertices
$\uhiding$
not only blocking but actually \emph{\hidingklawe{}}
the \mpsctext{}s in
$\mpconfhidden \subseteq \mpconf$,
and the vertices
$\ujustblocking$
just helping 
$\uhiding$
to block  the remaining \mpsctext{}s in
$\mpconfjustblocked = \mpconf \setminus \mpconfhidden$.
In
\reflem{lem:U-just-blocking-is-small}
and
\refcor{cor:U-hiding-is-large},
we proved that
if $U$ is large (which we are assuming) then
$\ujustblocking$
must be very small compared to
$\uhiding$, so we can basically just ignore~$\ujustblocking$. 
If we want to do
something interesting, it will have to be done with~$\uhiding$.

And indeed,
\reflem{lem:generalized-pick-local-good-blocker}
tells us that we can restructure  
$\uhiding$ 
to get a new vertex set
\hidingklawe
$\mpconfhidden$
and make considerable savings,
but that this can lead to
$\mpconfjustblocked$
no longer being blocked.
By
\refpr{pr:weight-of-non-zero-components-is-large},
there is a large fraction of
$\uhiding$
that resides in the non-zero-weight components of the \hidingsetklawe
graph~$\hidsetgraph$  
(as defined in Equation~%
\refeq{eq:definition-component-weight}).
We would like to show that by judiciously performing the restructuring of
\reflem{lem:generalized-pick-local-good-blocker}
in these components, we can also take care of~%
$\mpconfjustblocked$.

More precisely, 
we claim that 
we can combine the \hidingsetklawe{}s
$\unewhidingith$
from
\reflem{lem:generalized-pick-local-good-blocker}
with some subsets of
$\uhiding \unionSP \ujustblocking$
and
$\blacksjustblocked$
into a new blocking set
$\unewblockingset$
for all of
$
\mpconfhidden \unionSP \mpconfjustblocked
=
\mpconf
$
in such a way that the measure 
$\vmeasurecompact{\unewblockingset}$
does not exceed
$\vmeasure{U} = \vpotential{\mpconf}$
but so that
$
\Setsize{\unewblockingset} 
<
\setsize{U}
$.
But this contradicts the assumptions in
\reflem{lem:pyramids-possess-GKP}.
It follows that the conclusion in
\reflem{lem:pyramids-possess-GKP},
which we assumed to be false in order to derive a contradiction, 
must instead be true.
That is, any set $U$ that is chosen as in
\reflem{lem:pyramids-possess-GKP}
must have size
$\setsize{U} \leq 
13 \cdot 
\BHunionBBunionWH
$.
This in turn implies
\refth{th:pyramids-possess-GKP},
\ie that layered spreading graphs possess the \genklaweprop that we
assumed in order to get a lower bound on \blobpebblingtext price, 
and we are done.

We proceed to establish this final claim.
Our plan is once again to do some bipartite matching with the help of
Hall's theorem.  
Create a weighted bipartite graph with the vertices in
$
\blacksjustblocked
=
\Setdescr
{\bottomvertex{B}}
{\mpscnotstd \in \mpconfjustblocked}
$
on the left-hand side and with the 
non-zero-weight connected components 
among
$V_1, \ldots, V_k$ in~$\hidsetgraph$
in the sense of~%
\refeq{eq:definition-component-weight}
acting as ``supervertices'' on the right-hand side.
Reorder the indices among the connected components
$V_1, \ldots, V_{k}$ 
if needed so that the non-zero-weight components are
$V_1, \ldots, V_{k'}$.
All vertices in the weighted graphs are assigned weights so that each
right-hand side supervertex $V_i$
gets its weight according to~%
\refeq{eq:definition-component-weight},
and each left-hand vertex has weight~$1$.%
\footnote{%
Or, if we like, we can equivalently think of
an unweighted graph, where  each
$V_i$
is a cloud of
$\weightklawe{V_i}$
unique and distinct vertices, 
and where 
$\vneighbour{b}$
in~%
\refeq{eq:neighbouring-supervertices}
always containing either all or none of these vertices.%
}
We define the neighbours of each fixed vertex 
$b \in \blacksjustblocked$
to be
\begin{equation}
  \label{eq:neighbouring-supervertices}
  \vneighbour{b}
  =
  \Setdescr
      {V_i}
      {\weightklawe{V_i} > 0 \text{ and }
        \vmaxlevelcompact{\uhidingith} > \vlevel{b}}
      \eqcomma
\end{equation}
\ie all non-zero-weight components $V_i$ that contain
vertices in the \hidingsetklawe $\uhiding$ that could possibly be
involved in blocking any \mpsctext 
$\mpscnotstd \in \mpconfjustblocked$
having bottom vertex
$\bottomvertex{B} = b$.
This is so since by
\refpr{pr:not-covering-means-blocking-from-above},
any vertex $u \in \uhiding$ helping to block such \anmpsctext
$\mpscnotstd \in \mpconfjustblocked$
must be strictly above~$b$,
so if the highest-level vertices in $\uhidingith$ are on a level
below~$b$, no vertex in $\uhidingith$ can be responsible for
blocking~$\mpscnotstd$.

Let
$\bbvstd' \subseteq \blacksjustblocked$
be a largest set \st
$
\weightklawecompact{\vneighbourcompact{\bbvstd'}}
\leq
\Setsize{\bbvstd'}
$.
We must have
\begin{equation}
  \label{eq:final-push-neighbours-do-not-cover-all-of-H-graph}
  \vneighbourcompact{\bbvstd'}
  \neq
  \Unionnodisplay_{i=1}^{k'} V_i
\end{equation}
since
$
\weightklawecompact{\Union_{i=1}^{k'} V_i}
>
\BHunionBBunionWH
\geq
\Setsize{\blacksjustblocked}
$
by
\refpr{pr:weight-of-non-zero-components-is-large}.
For all
$\bbvstd'' \subseteq \blacksjustblocked \setminus \bbvstd'$
it holds that
\begin{equation}
  \label{eq:final-push-halls-matching-condition-b-bis}
  \weightklawecompact{
    \vneighbourcompact{\bbvstd''}
    \setminus
    \vneighbourcompact{\bbvstd'}
  }
  \geq
  \Setsize{\bbvstd''}
\end{equation}
since otherwise 
$
\bbvstd'
$
would not be of largest size as assumed above.
The inequality~%
\refeq{eq:final-push-halls-matching-condition-b-bis}
plugged into Hall's marriage theorem tells us that there is a matching
of the vertices in 
$
\blacksjustblocked \setminus \bbvstd'
$
to the components in
$
\Union_{i=1}^{k'} V_i
\setminus
\vneighbourcompact{\bbvstd'}
\neq \emptyset
$
with the property that no component~$V_i$ gets matched with more than 
$\weightklawe{V_i}$
vertices from
$
\blacksjustblocked \setminus \bbvstd'
$.

Reorder the components in
the \hidingsetklawe graph~$\hidsetgraph$
so that the matched components in $\hidsetgraph$
are
$V_1, \ldots, V_{\kmatchedindex}$
and the rest of the components are
$V_{\kmatchedindex + 1}, \ldots, V_k$
and so that
$\uhidingith[1], \ldots, \uhidingith[\kmatchedindex]$
and
$\uhidingith[\kmatchedindex + 1], \ldots, \uhidingith[k]$
are the corresponding subsets of
the \hidingsetklawe~$\uhiding$.
Then pick good local blockers
$
\unewhidingith \subseteq V_i
$
as in
\reflem{lem:generalized-pick-local-good-blocker}
for all components
$V_1, \ldots, V_{\kmatchedindex}$.
Now the following holds:
\begin{enumerate}
\item
  By construction and assumption, respectively,
  the vertex set
  $
  \Union_{i=1}^{\kmatchedindex} \unewhidingith
  \unionSP
  \Union_{i=\kmatchedindex + 1}^{k} \uhidingith
  $
  blocks
  (and even \hideklawe{}s)
  $\mpconfhidden$.

\item
  All
  \mpsctext{}s in
  \begin{equation}
    \mpconfjustblocked^{1} =
    \Setdescr
        {\mpscnotstd \in \mpconfjustblocked}
        {\bottomvertex{B} \in \bbvstd'}
  \end{equation}
  are blocked by
  $
  \ujustblocking
  \unionSP
  \vneighbourcompact{\bbvstd'}
  =
  \ujustblocking
  \unionSP
  \Union_{i=\kmatchedindex + 1}^{k} \uhidingith
  $,
  as we have not moved any elements in~$U$ above~${\bbvstd'}$.
  
\item
  With notation as in
  \reflem{lem:generalization-union-respecting-leq},
  let
  $
  Y
  =
  \ujustblocking
  \unionSP
  \Union_{i=\kmatchedindex + 1}^{k} \uhidingith
  $
  and consider
  $\unewhidingith$
  and
  $\uhidingith$
  for $i = 1, \ldots, \kmatchedindex$.
  We have
  $\unewhidingith \measureleq \uhidingith$
  for $i = 1, \ldots, \kmatchedindex$  by
  \reflem{lem:generalized-pick-local-good-blocker}.
  Also, 
  since
  $\uhiding \intersectionSP \ujustblocking = \emptyset$
  and
  $\unewhidingith \subseteq V_i$
  and
  $\uhidingith \subseteq V_i$
  for $V_1, \ldots, V_k$  pairwise disjoint sets of vertices,
  it holds for all
  $i, j \in \intnfirst{\kmatchedindex}$, $i \neq j$,
  that
  $
  \unewhidingith[i] 
  \intersectionSP
  \unewhidingith[j]
  = \emptyset
  $,
  $
  \uhidingith[i]
  \intersectionSP
  \uhidingith[j]
  = \emptyset
  $,
  $
  \unewhidingith[i] 
  \intersectionSP
  \uhidingith[j]
  = \emptyset
  $
  and
  $
  Y
  \intersectionSP
  \uhidingith[j]
  =
  \emptyset
  $.
  Therefore, the conditions in 
  \reflem{lem:generalization-union-respecting-leq}
  are satisfied and we conclude that
  \begin{equation}
    \begin{split}
      \vmeasurecompact{
        \ujustblocking
        \unionSP
        \Unionnodisplay_{i=1}^{\kmatchedindex} \unewhidingith
        \unionSP
        \Unionnodisplay_{i=\kmatchedindex + 1}^{k} \uhidingith
      }
      &=
      \vmeasurecompact{Y \unionSP 
        \Unionnodisplay_{i=1}^{\kmatchedindex} \unewhidingith} 
      \\
      &\leq
      \vmeasurecompact{Y \unionSP 
        \Unionnodisplay_{i=1}^{\kmatchedindex} \uhidingith}
      \\
      &=
      \vmeasurecompact{
        \ujustblocking
        \unionSP
        \Unionnodisplay_{i=1}^{\kmatchedindex} \uhidingith
        \unionSP
        \Unionnodisplay_{i=\kmatchedindex + 1}^{k} \uhidingith
      }
      \\
      &=
      \vmeasure{U}
      \eqcomma
    \end{split}
  \end{equation}
  where we note that
  $
  \ujustblocking
  \unionSP
  \Unionnodisplay_{i=1}^{\kmatchedindex} \unewhidingith
  \unionSP
  \Unionnodisplay_{i=\kmatchedindex + 1}^{k} \uhidingith
  $
  is
  \emph{measured as a multi-set with repetitions}.
  Also,   we have the strict inequality 
  \begin{equation}
    \Setsize{
      \ujustblocking
      \unionSP
      \Unionnodisplay_{i=1}^{\kmatchedindex} \unewhidingith
      \unionSP
      \Unionnodisplay_{i=\kmatchedindex + 1}^{k} \uhidingith
      }
    <
    \setsize{U}
    \eqcomma
  \end{equation}
  where again the multi-set is \emph{counted with repetitions}.

\item
  It remains to take care of the potentially unblocked \mpsctext{}s in
  \begin{equation}
    \mpconfjustblocked^{2} =
    \Setdescr
        {\mpscnotstd \in \mpconfjustblocked}
        {\bottomvertex{B} \in 
          \blacksjustblocked \setminus \bbvstd'
        } 
        \eqperiod
  \end{equation}
  But we derived above that there is a matching of 
  $
  \blacksjustblocked \setminus \bbvstd'
  $
  to 
  $V_1, \ldots, V_{\kmatchedindex}$
  \st
  no~$V_i$ is chosen by more than
  \begin{equation}
    \label{eq:N-matched-leq-N-spare-vertices}
    \weightklawe{V_i}
    =
    \ceilingcompact{\setsize{\uhidingith} / 6}
    \leq
    \floorcompact{\setsize{\uhidingith} / 3}    
  \end{equation}
  vertices from
  $
  \blacksjustblocked \setminus \bbvstd'
  $
  (where we used that
  $\Setsize{\uhidingith} \geq 6$
  if
  $\weightklawe{V_i} > 0$
  to get the last inequality).
  This means that there is a spare blocker vertex in
  $\unewhidingith$
  for each
  $
  b \in
  \blacksjustblocked \setminus \bbvstd'
  $
  that is matched to~$V_i$.
  Also, by the definition of neighbours in our weighted bipartite
  graph, 
  each $b$ is matched to a component with
  $\vmaxlevelcompact{\uhidingith} > \vlevel{b}$.
  By
  \refobs{obs:moving-vertices-downwards-non-increases-measure},
  lowering these spare vertices from
  $\vmaxlevelcompact{\uhidingith}$ to $\vlevel{b}$
  can only decrease the measure.
\end{enumerate}
Finally, throw away any remaining multiple copies in our new
blocking set, and denote the resulting set by~$\unewblockingset$.
We have that
$\unewblockingset$
blocks~$\mpconf$
and that
$
\vmeasurecompact{\unewblockingset}
\leq
\vmeasure{U}
$
but
$
\Setsize{\unewblockingset} 
<
\setsize{U}
$.
This is a contradiction
since $U$ was chosen to be of minimal size, and thus
\reflem{lem:pyramids-possess-GKP}
must hold.
But then
\refth{th:pyramids-possess-GKP}
follows immediately as well, as was noted above.

\subsection{Recapitulation of the Proof of
  \Refth{th:main-theorem}
  and Optimality of Result}

Let us conclude this \sectionBlobBound by recalling why the tight
bound on clause space for refuting pebbling contradictions in
\refth{th:main-theorem}
now follows and by showing that the current construction cannot be
pushed to give a better result.

\begin{theorem}[rephrasing of \refth{th:main-theorem}]
  \label{th:main-theorem-rephrased}
  Suppose that
  $G_h$ is a layered \pebblingdag of height~$h$ that is spreading.
  Then the clause space of refuting  
  the pebbling contradiction 
  $\pebcontr[G_h]{\pebdeg}$
  of  degree   
  $\pebdeg > 1$  
  by resolution is 
  $
  \clspaceref{\pebcontr[G_h]{\pebdeg}}
  =
  \Tightsmall{h}
  $.
\end{theorem}

\begin{proof}
  The   $\Ordosmall{h}$ upper bound on clause space
  follows from the bound
  $\pebblingprice{G_h} \leq h + \Ordosmall{1}$
  on the black pebbling price  in
  \reflemP{lem:upper-bound-pebbling-price-layered-DAG}
  combined with the bound
  $\clspaceref{\pebcontr[G]{\pebdeg}}
  \leq
  \mbox{$\pebblingprice{G} + \Ordosmall{1}$}$
  from
  \refprP{pr:clause-space-upper-bounded-by-pebbling-price}.

  For the lower bound,
  we instead consider the pebbling   formula
  $\pebcontrNT[G_h]{\pebdeg}$  
  without target axioms
  $\olnot{\varx(z)}_{1}, \ldots, \olnot{\varx(z)}_{\pebdeg}$ 
  and use that by
  \reflemP{lem:C-clause-space-the-same-without-target-axioms}
  it holds that
  $
  \Clspaceref{\pebcontr[G_h]{\pebdeg}}
  =
  \Clspacederiv
  {\pebcontrNT[G_h]{\pebdeg}}
    {\targetclausexvar[i]}
  $.
  Fix any resolution derivation
  $\derivof
  {\proofstd}
  {\pebcontrNT[G_h]{\pebdeg}}
  {\targetclausexvar[i]}
  $
  and let
  $\multipebbling_{\proofstd}$  
  be the \pebcomplete{} \multipebblingtext
  of the graph $G$ associated to $\proofstd$ in
  \refthP{th:C-translation-of-resolution-to-pebbling}
  \st
  $
  \mpcost{\multipebbling_{\proofstd}}
  \leq
  \Maxofexpr[\clsc \in \proofstd]{\mpcost{\mpinducedconf{\clsc}}}
  + \Ordosmall{1}
  $.
  On the one hand, 
  \refthP{th:linear-cost-pebbles}
  says that
  $
  {\mpcost{\mpinducedconf{\clsc}}}
  \leq
  \setsize{\clsc}
  $
  provided that
  $\pebdeg > 1$,
  so in particular it must hold that
  $
  \mpcost{\multipebbling_{\proofstd}}
  \leq
  \clspaceofsmall{\proofstd} + \Ordosmall{1}
  $.
  On the other hand,
  $
  \mpcost{\multipebbling_{\proofstd}}
  \geq
  \blobpebblingprice{G_h}
  $
  by definition,
  and by
  \reftwoths    
  {th:lower-bound-mpebbling-assuming-GKP}
  {th:pyramids-possess-GKP}
  it holds that
  $\blobpebblingprice{G_h} = \Lowerboundsmall{h}$.
  Thus
  $\clspaceofsmall{\proofstd} = \Lowerboundsmall{h}$,
  and the theorem follows.
\end{proof}

Plugging in pyramid graphs $\pyramidgraphh$ in
\refth{th:main-theorem-rephrased},
we get \kcnfform{}s $\fstd_n$
of size~$\Tightsmall{n}$
with refutation clause space~$\Tightsmall{\sqrt{n}}$.
This is the best we can get from pebbling formulas over spreading
graphs. 

\begin{theorem}
  \label{th:sqrt-optimality-for-spreading-graphs}
  Let
  $G$
  be any layered spreading graph and suppose that
  $\pebcontr[G]{\pebdeg}$
  has formula size and number of clauses
  $\Tightsmall{n}$.
  Then
  $
  \clspacerefcompact{\pebcontr[G]{\pebdeg}}
  = \Ordosmall{\sqrt{n}}
  $.
\end{theorem}

\begin{proof}
  Suppose that $G$ has height~$h$. Then
  $
  \clspacerefcompact{\pebcontr[G]{\pebdeg}}
  = \Ordosmall{h}
  $
  as was noted above.
  The size of
  $\pebcontr[G]{\pebdeg}$, as well as the number of clauses, is linear
  in the number of vertices
  $\setsize{\vertices{G}}$.
  We claim that the fact that $G$ is spreading implies that
  $\setsize{\vertices{G}} = \Lowerboundcompact{h^2}$, 
  from which the theorem follows.

  To prove the claim, let
  $V_L$ denote the vertices of $G$ on level~$L$.
  Then
  $
  \setsize{\vertices{G}}
  =
  \sum_{L=0}^{h} \setsize{V_L}
  $.
  Obviously, for any~$L$ the set $V_L$ \hideklawe{}s the sink~$z$
  of~$G$.
  Fix for every~$L$ some arbitrary minimal subset
  $V'_L \subseteq V_L$ 
  \hidingklawe~$z$. 
  Then $V'_L$ is
  \tightklawe,
  the graph
  $\hidsetgraph(V'_L)$
  is \connectedklawe by
  \refcor{cor:vertex-and-necessary-cover-in-same-component}, 
  and setting
  $j = h$
  in the spreading inequality
  \refeq{eq:spreading-property}
  we get that
  $
  \Setsize{V'_L}
  \geq
  1 + h - L
  $.
  Hence
  $
  \setsize{\vertices{G}}
  \geq
  \sum_{L=0}^{h}
  \setsize{V'_L}
  =
  \Lowerboundcompact{h^2}
  $.
\end{proof}

The proof  of
\refth{th:sqrt-optimality-for-spreading-graphs}
can also be extended to cover the original definition in~%
\cite{K80TightBoundPebblesPyramid}
of spreading graphs that are not necessarily layered,
but we omit the details.

\section{Conclusion and Open Problems}
\label{sec:open-questions}

We have proven an asymptotically tight bound on the refutation clause
space in resolution of pebbling contradictions over pyramid graphs.
This yields the currently best known separation of length and clause
space in resolution. 
\ifthenelse{\boolean{maybeSTOC}}
{%
  Also, in contrast to previous polynomial lower bounds on clause space,
  our result does not not follow from corresponding lower bounds on
  width for the  same formulas.%
}
{%
  Also, in contrast to previous polynomial lower bounds on clause space,
  our result does not not follow from lower bounds on width for the
  corresponding formulas.%
}
Instead, a corollary of our result is an exponential improvement of
the separation of width and space
in~\cite{Nordstrom06NarrowProofsMayBeSpaciousSTOCtoappearinSICOMP}.
This is a first step towards answering the question of the
relationship between length and space posed in,
\eg,
\mbox{\cite{Ben-Sasson02SizeSpaceTradeoffs,
    ET03CombinatorialCharacterization,
    Toran04Space}.}

\ifthenelse{\boolean{maybeSTOC}}
{More technically,}
{More technically speaking,}
we have established that for all graphs $G$
in the class of ``layered spreading DAGs'' 
\ifthenelse{\boolean{maybeSTOC}}
{(including binary trees and pyramids)} 
{(including complete binary trees and pyramid graphs)} 
the height $h$ of $G$,
which coincides with the black-white pebbling price,
is an asymptotical lower bound for the refutation clause space
$\Clspaceref{\pebcontr[G]{\pebdeg}}$
of pebbling contradictions
$\pebcontr[G]{\pebdeg}$
provided that
$\pebdeg \geq 2$.
Plugging in 
\ifthenelse{\boolean{maybeSTOC}}
{pyramids} 
{pyramid graphs} 
we get an
$\bigomega{\sqrt{n}}$
bound on space, 
which is the best one can get for any spreading graph.

An obvious question is whether this lower bound on clause space in
terms of black-white  pebbling price is true for arbitrary DAGs.
In particular, does it hold for the
family of DAGs $\set{G_n}_{n=1}^{\infty}$ 
in~\cite{GT78VariationsPebbleGame}
of size $\bigoh{n}$
that have maximal black-white pebbling price
$\bwpebblingprice{G_n} = \bigomega{n / \log n}$
in terms of size?
If it could be proven for pebbling contradictions over such graphs
that pebbling price bounds clause space from below, this would immediately
imply that there are \kcnfform{}s refutable in small length that can be
maximally complex \wrt clause space. 

%
%

\begin{openproblem}
  \label{conj:separation-of-space-from-length}
  \ifthenelse{\boolean{maybeSTOC}}
  {Is there a family of}
  {Is there a family of unsatisfiable}
  \kcnfform{}s
  $\setsmall{F_n}_{n=1}^{\infty}$
  of size $\bigoh{n}$  \st
  \mbox{$\lengthrefsmall{F_n} = \bigoh{n}$}
  and
  $\widthrefsmall{F_n} = \bigoh{1}$
  but
  $\clspaceref{F_n} = \bigomega{n / \log n}$?
\end{openproblem}


\ifthenelse
{\boolean{maybeSTOC}}
{}
{%
We are currently working on this problem, but note that these DAGs 
in~\cite{GT78VariationsPebbleGame}
seem to have much more challenging structural properties that makes it
hard to lift the lower bound argument from standard black-white
pebblings to \blobpebblingtext{}s.%
}

%
%
%

A second question, 
\ifthenelse{\boolean{maybeSTOC}}
{%
  more related to 
  \refth{th:easy-length-clause-space-trade-off}
  and our other trade-off results,%
}
{%
  more related to 
  \refth{th:easy-length-clause-space-trade-off}
  and the other trade-off results presented in
  \refsec{sec:simplified-way-of-proving-trade-offs},%
}
is as follows.
\ifthenelse{\boolean{maybeSTOC}}
{%
  We know from~
  \refeq{eq:intro-Ben-Sasson-Wigderson-bound}%
}
{%
  We know from~\cite{BW01ShortProofs}   
  (see \refth{th:widthboundgeneral})%
}
that short resolution refutations imply the existence of narrow
refutations,  and in view of this 
an appealing proof search heuristic 
is to search exhaustively for refutations in minimal width.
One serious drawback of this approach is that
there is no guarantee that the short and narrow refutations are the same
one.
On the contrary, the narrow refutation $\proofstd'$ resulting
from the proof in \cite{BW01ShortProofs}
is potentially exponentially longer than the short proof~$\proofstd$ 
that we start with.
However, we have no examples of formulas where the refutation in
minimum width is actually known to be substantially longer than the
minimum-length refutation. Therefore, 
it would be valuable to know 
whether this increase in length is necessary.
That is, is there a formula family which exhibits a length-width
trade-off in the sense that there are short refutations and narrow
refutations, but all narrow refutations have a length blow-up
(polynomial or superpolynomial)?
Or is the exponential blow-up in~\cite{BW01ShortProofs} just an
artifact of the proof?

\begin{openproblem}
  \label{openproblem:length-vs-width-BW}
    If $\fstd$ is a \kcnfform over $\nvar$ variables
    refutable in length~$\lengthstd$, is it true that there is always
    a refutation $\proofstd$ of~$\fstd$ in 
    width
    {$\widthofarg{\proofstd} = 
      \Bigoh{\sqrt{\nvar \log \lengthstd}}$}
    with length no more than, say,
    {$\lengthofarg{\proofstd} = \bigoh{\lengthstd}$}
    or at most
    {$\polyboundsmall{\lengthstd}$}?
  \end{openproblem}

A similar trade-off question can be posed for clause space.
\ifthenelse{\boolean{maybeSTOC}}
{%
  Given  a refutation in small space, we can prove using
  \refeq{eq:intro-Atserias-Dalmau-bound}%
}
{%
  Given  a refutation in small space, we can prove using
  \cite{AD02CombinatoricalCharacterization}
  (see \refth{th:small-clause-space-implies-small-width})%
}
that there must exist a refutation in short length. 
But again, the short refutation resulting from the proof is not the
same as that with which we started.
For concreteness, let us fix the space to be constant.  If a
polynomial-size \kcnfform has a refutation in 
\ifthenelse{\boolean{maybeSTOC}}
{constant space,}
{constant clause space,}
we know that it must be refutable in polynomial length.  But can we
get a refutation in both short length and small space
simultaneously?

\begin{openproblem}
  \label{openproblem:space-vs-length}
  Suppose that
  $\setsmall{\fstd_n}_{n=1}^{\infty}$
  is a family of polynomial-size \kcnfform{}s 
  with refutation clause space
  $\mbox{$\clspaceref{\fstd_n}$} = \bigoh{1}$. 
  Does this imply that there are refutations
%
%
  $\refof{\proofstd_n}{\fstd_n}$
  simultaneously in length
  {$\lengthofarg{\proofstd_n} = \polyboundsmall{n}$}
  and clause space
  {$\clspaceofsmall{\proofstd_n} = \bigoh{1}$}?
\end{openproblem}
  
Or can it be that restricting the clause space, we sometimes have to end up
with really long refutations?
We would like  to know what holds in this case,
and how it relates to the trade-off results for variable space in~%
\cite{HP07ExponentialTimeSpaceSpeedupFOCS}.

Finally, we note that all bounds on clause space proven so far is in
the regime where the clause space $\clspaceofsmall{\proofstd}$
is less than the number of clauses $\nclausesof{\fstd}$ in~$\fstd$.
This is quite natural, since the size of the formula 
can be shown to be an upper bound on
the minimal clause space needed~\cite{ET01SpaceBounds}.

Such lower bounds on space might not seem too relevant to clause
learning algorithms, since the size of the cache 
in practical applications usually will be very much larger than the  
size of the formula. 
For this reason, it seems to be a highly interesting problem to
determine what can be said if we allow extra clause space.
Assume that we have a \cnfform~$\fstd$ of size roughly~$n$ 
refutable in length
$\lengthrefsmall{\fstd} = \lengthstd$ for $\lengthstd$ suitably large
(say, 
$\lengthstd = \polyboundsmall{n}$
or
$\lengthstd = n^{\log n}$
or so).
Suppose that we allow clause space more than the minimum
$n + \bigoh{1}$,
but less than the trivial upper bound
$\lengthstd / \log \lengthstd$.
Can we then find a 
\ifthenelse{\boolean{maybeSTOC}}
{refutation} 
{resolution refutation} 
using at most that much space
and achieving at most a polynomial increase in length compared to the
minimum?

\begin{openproblem}[\cite{Ben-Sasson07personalcommunication}]
  \label{openproblem:extra-clause-space-vs-length}
  Let $\fstd$ be any \cnfform with
  $\nclausesof{\fstd} = n$ clauses
  (or
  $\setsizesmall{\vars{\fstd}} = n$
  variables).
  Suppose that
  $\lengthrefsmall{\fstd} = \lengthstd$.
  Does this imply that there is  a resolution refutation
  $\refof{\proofstd}{\fstd}$
  in clause space
  $\clspaceofsmall{\proofstd} = \bigoh{n}$
  and length
  $\lengthofarg{\proofstd}=\polyboundsmall{\lengthstd}$?
\end{openproblem}

If so, this could be interpreted as saying that a smart enough clause
learning algorithm can potentially find any short resolution
refutation in reasonable space (and for formulas that cannot be refuted
in short length we cannot hope to find refutations efficiently
anyway).

We conclude with a couple of comments on 
clause space versus clause learning.

Firstly, we note that it is unclear whether one should expect any fast
progress on 
Open Problem~\ref{openproblem:extra-clause-space-vs-length}, at least if
if our experience from the case where
$\clspaceofsmall{\proofstd} \leq \nclausesof{\fstd}$
is anything to go by.
Proving lower bounds on space in this ``low-end regime'' for formulas
easy \wrt length has been (and still is)
very challenging. However, it certainly cannot be excluded that
problems in the range 
$\clspaceofsmall{\proofstd} > \nclausesof{\fstd}$
might be approached with different and more successful techniques.

Secondly, 
we would like to raise the question of whether, in spite of what was
just said before
Open Problem~\ref{openproblem:extra-clause-space-vs-length}, 
lower bounds on 
\ifthenelse{\boolean{maybeSTOC}}
{space}
{clause space} 
can nevertheless give indications as to
which formulas might be hard for clause learning algorithms and why.
Suppose that we know for some \cnfform~$\fstd$ that
$\clspaceref{\fstd}$ 
is large. What this tells us is that any algorithm, even a
non-deterministic one making optimal choices 
concerning 
\ifthenelse{\boolean{maybeSTOC}} 
{which clauses to save or throw away,} 
{which clauses to save or throw away at any given point in time,} 
will have to keep a fairly large number of ``active'' clauses in
memory in order to carry out the refutation.    
Since this is so, a real-life deterministic proof search algorithm,
which has no sure-fire way of knowing 
which clauses are the right ones to concentrate on 
at any given moment, might have to keep working on a
lot of extra clauses in order to be sure that the fairly large
critical set of clauses needed to find a refutation will be among the
``active'' clauses.
%
%

Intriguingly enough, pebbling contradictions over pyramids might in
fact be an example of this.
We know that these formulas are very easy \wrt length and width,
having constant-width refutations that are essentially as short as the
formulas themselves. 
But in~%
\cite{SBK03UsingProblemStructure},
it was shown that state-of-the-art clause learning algorithms can
have 
serious problems with even moderately large 
\ifthenelse{\boolean{maybeSTOC}} 
{pebbling contradictions. (Their
``grid pebbling formulas'' 
are exactly our pebbling
contradictions of degree $\pebdeg=2$ over pyramids.)}
{pebbling contradictions.%
\footnote{%
The ``grid pebbling formulas'' 
in~\cite{SBK03UsingProblemStructure}
are exactly our pebbling
contradictions of degree $\pebdeg=2$ over pyramid graphs.%
}}
Although we are certainly not arguing that this is the whole
story---it was also shown 
in~\cite{SBK03UsingProblemStructure}
that the branching order is a critical factor, and that
given some extra structural information the algorithm can achieve an 
exponential speed-up---we wonder whether the high lower bound on
\ifthenelse{\boolean{maybeSTOC}} 
{space}
{clause space} 
can nevertheless be part of the explanation. 
It should be pointed out that pebbling contradictions are the only
formulas we know of that are really easy \wrt length and width but
hard for clause space. And if there is empirical data showing that for
these very formulas clause learning algorithms can have great
difficulties finding refutations, it might be worth investigating
whether this is just a coincidence or a sign of some deeper
connection.

\ifthenelse{\boolean{maybeSTOC}}
{\section{Acknowledgements}}
{\section*{Acknowledgements}}

We are grateful to  Per Austrin and Mikael Goldmann for generous
feedback during various stages of this work, 
and to Gunnar Kreitz for quickly spotting some bugs in a preliminary
version of the \multipebblegame.
Also, we would like to thank
Paul Beame, 
Maria Klawe,
Philipp Hertel, and \mbox{Toniann Pitassi} for
valuable correspondence concerning their work,
Nathan Segerlind for comments and pointers regarding clause learning,
and Eli Ben-Sasson for stimulating discussions about proof complexity
in general and the problems in \refsec{sec:open-questions} in particular.

%
%


\bibliographystyle{nada-en}   

\bibliography{refArticles,refBooks,refOther}

\begin{thebibliography}{10}

\bibitem{AL86Minimal}
Ron Aharoni and Nathan Linial.
\newblock Minimal non-two-colorable hypergraphs and minimal unsatisfiable
  formulas.
\newblock {\em Journal of Combinatorial Theory}, 43:196\nobreakdash--204, 1986.

\bibitem{ABRW02SpaceComplexity}
Michael Alekhnovich, Eli Ben-Sasson, Alexander~A. Razborov, and Avi Wigderson.
\newblock Space complexity in propositional calculus.
\newblock {\em SIAM Journal on Computing}, 31(4):1184\nobreakdash--1211, 2002.

\bibitem{AJPU02ExponentialSeparation}
Michael Alekhnovich, Jan Johannsen, Toniann Pitassi, and Alasdair Urquhart.
\newblock An exponential separation between regular and general resolution.
\newblock In {\em Proceedings of the 34th Annual ACM Symposium on Theory of
  Computing ({STOC}~'02)}, pages 448\nobreakdash--456, May 2002.

\bibitem{AC03SmallerExplicitSuperconcentrators}
Noga Alon and Michael Capalbo.
\newblock Smaller explicit superconcentrators.
\newblock In {\em Proceedings of the 14th Annual ACM-SIAM Symposium on Discrete
  Algorithms ({SODA}~'03)}, pages 340\nobreakdash--346, 2003.

\bibitem{AD02CombinatoricalCharacterization}
Albert Atserias and Victor Dalmau.
\newblock A combinatorical characterization of resolution width.
\newblock In {\em Proceedings of the 18th {IEEE} Annual Conference on
  Computational Complexity ({CCC}~'03)}, pages 239\nobreakdash--247, July 2003.
\newblock Journal version to appear in \emph{Journal of Computer and System
  Sciences}.

\bibitem{BET01MinimallyUnsatisfiable}
Sven Baumer, Juan~Luis Esteban, and Jacobo Tor\'an.
\newblock Minimally unsatisfiable {CNF} formulas.
\newblock {\em Bulletin of the European Association for Theoretical Computer
  Science}, 74:190\nobreakdash--192, June 2001.

\bibitem{B00ProofComplexity}
Paul Beame.
\newblock Proof complexity.
\newblock In Steven Rudich and Avi Wigderson, editors, {\em Computational
  Complexity Theory}, volume~10 of {\em IAS/Park City Mathematics Series},
  pages 199\nobreakdash--246. American Mathematical Society, 2004.

\bibitem{BKPS02Efficiency}
Paul Beame, Richard Karp, Toniann Pitassi, and Michael Saks.
\newblock The efficiency of resolution and {D}avis-{P}utnam procedures.
\newblock {\em SIAM Journal on Computing}, 31(4):1048\nobreakdash--1075, 2002.

\bibitem{BKS03UnderstandingPowerClauseLearning}
Paul Beame, Henry Kautz, and Ashish Sabharwal.
\newblock Understanding the power of clause learning.
\newblock In {\em Proceedings of the 18th International Joint Conference in
  Artificial Intelligence ({IJCAI}~'03)}, pages 94\nobreakdash--99, 2003.

\bibitem{BP98Propositional}
Paul Beame and Toniann Pitassi.
\newblock Propositional proof complexity: Past, present, and future.
\newblock {\em Bulletin of the European Association for Theoretical Computer
  Science}, 65:66\nobreakdash--89, June 1998.

\bibitem{Ben-Sasson02SizeSpaceTradeoffs}
Eli Ben-Sasson.
\newblock Size space tradeoffs for resolution.
\newblock In {\em Proceedings of the 34th Annual ACM Symposium on Theory of
  Computing ({STOC}~'02)}, pages 457\nobreakdash--464, May 2002.

\bibitem{Ben-Sasson07personalcommunication}
Eli Ben-Sasson.
\newblock Personal communication, 2007.

\bibitem{BG03SpaceComplexity}
Eli Ben-Sasson and Nicola Galesi.
\newblock Space complexity of random formulae in resolution.
\newblock {\em Random Structures and Algorithms}, 23(1):92\nobreakdash--109,
  August 2003.

\bibitem{BIW00Near-optimalSeparation}
Eli Ben-Sasson, Russell Impagliazzo, and Avi Wigderson.
\newblock Near optimal separation of treelike and general resolution.
\newblock {\em Combinatorica}, 24(4):585\nobreakdash--603, September 2004.

\bibitem{BW01ShortProofs}
Eli Ben-Sasson and Avi Wigderson.
\newblock Short proofs are narrow---resolution made simple.
\newblock {\em Journal of the ACM}, 48(2):149\nobreakdash--169, March 2001.

\bibitem{B37Canonical}
Archie Blake.
\newblock {\em Canonical Expressions in {B}oolean Algebra}.
\newblock PhD thesis, University of Chicago, 1937.

\bibitem{BEGJ00RelativeComplexity}
Maria~Luisa Bonet, Juan~Luis Esteban, Nicola Galesi, and Jan Johannsen.
\newblock On the relative complexity of resolution refinements and cutting
  planes proof systems.
\newblock {\em SIAM Journal on Computing}, 30(5):1462\nobreakdash--1484, 2000.

\bibitem{BG01Optimality}
Maria~Luisa Bonet and Nicola Galesi.
\newblock Optimality of size-width tradeoffs for resolution.
\newblock {\em Computational Complexity}, 10(4):261\nobreakdash--276, December
  2001.

\bibitem{BOP03Complexity}
Josh Buresh-Oppenheim and Toniann Pitassi.
\newblock The complexity of resolution refinements.
\newblock In {\em Proceedings of the 18th {IEEE} Symposium on Logic in Computer
  Science ({LICS}~'03)}, pages 138\nobreakdash--147, June 2003.

\bibitem{CS88ManyHard}
Va{\v{s}}ek Chv{\'a}tal and Endre Szemer{\'e}di.
\newblock Many hard examples for resolution.
\newblock {\em Journal of the ACM}, 35(4):759\nobreakdash--768, October 1988.

\bibitem{Cook71CooksTheorem}
Stephen~A. Cook.
\newblock The complexity of theorem-proving procedures.
\newblock In {\em Proceedings of the 3rd Annual ACM Symposium on Theory of
  Computing ({STOC}~'71)}, pages 151\nobreakdash--158, 1971.

\bibitem{C74ObservationTimeStorageTradeOff}
Stephen~A. Cook.
\newblock An observation on time-storage trade off.
\newblock {\em Journal of Computer and System Sciences},
  9:308\nobreakdash--316, 1974.

\bibitem{CR79Relative}
Stephen~A. Cook and Robert Reckhow.
\newblock The relative efficiency of propositional proof systems.
\newblock {\em Journal of Symbolic Logic}, 44(1):36\nobreakdash--50, March
  1979.

\bibitem{CS76Storage}
Stephen~A. Cook and Ravi Sethi.
\newblock Storage requirements for deterministic polynomial time recognizable
  languages.
\newblock {\em Journal of Computer and System Sciences},
  13(1):25\nobreakdash--37, 1976.

\bibitem{DLL62MachineProgram}
Martin Davis, George Logemann, and Donald Loveland.
\newblock A machine program for theorem proving.
\newblock {\em Communications of the ACM}, 5(7):394\nobreakdash--397, July
  1962.

\bibitem{DP60ComputingProcedure}
Martin Davis and Hilary Putnam.
\newblock A computing procedure for quantification theory.
\newblock {\em Journal of the ACM}, 7(3):201\nobreakdash--215, 1960.

\bibitem{EGM04Complexity}
Juan~Luis Esteban, Nicola Galesi, and Jochen Messner.
\newblock On the complexity of resolution with bounded conjunctions.
\newblock {\em Theoretical Computer Science}, 321(2-3):347\nobreakdash--370,
  August 2004.

\bibitem{ET01SpaceBounds}
Juan~Luis Esteban and Jacobo Torán.
\newblock Space bounds for resolution.
\newblock {\em Information and Computation}, 171(1):84\nobreakdash--97, 2001.

\bibitem{ET03CombinatorialCharacterization}
Juan~Luis Esteban and Jacobo Torán.
\newblock A combinatorial characterization of treelike resolution space.
\newblock {\em Information Processing Letters}, 87(6):295\nobreakdash--300,
  2003.

\bibitem{Galil77Resolution}
Zvi Galil.
\newblock On resolution with clauses of bounded size.
\newblock {\em SIAM Journal on Computing}, 6(3):444\nobreakdash--459, 1977.

\bibitem{GT78VariationsPebbleGame}
John~R. Gilbert and Robert~Endre Tarjan.
\newblock {\em Variations of a Pebble Game on Graphs}.
\newblock Technical Report STAN-CS-78-661, Stanford University, 1978.
\newblock Available at the webpage
  \verb+http://infolab.stanford.edu/TR/CS-TR-78-661.html+.

\bibitem{H85Intractability}
Armin Haken.
\newblock The intractability of resolution.
\newblock {\em Theoretical Computer Science}, 39(2-3):297\nobreakdash--308,
  August 1985.

\bibitem{HP07ExponentialTimeSpaceSpeedupFOCS}
Philipp Hertel and Toniann Pitassi.
\newblock Exponential time/space speedups for resolution and the
  {PSPACE}-completeness of black-white pebbling.
\newblock In {\em Proceedings of the 48th Annual IEEE Symposium on Foundations
  of Computer Science ({FOCS}~'07)}, pages 137\nobreakdash--149, October 2007.

\bibitem{HPV77TimeVsSpace}
John Hopcroft, Wolfgang Paul, and Leslie Valiant.
\newblock On time versus space.
\newblock {\em Journal of the ACM}, 24(2):332\nobreakdash--337, April 1977.

\bibitem{KS88OnThePowerOfWhitePebbles}
Balasubramanian Kalyanasundaram and George Schnitger.
\newblock On the power of white pebbles.
\newblock In {\em Proceedings of the 20th Annual ACM Symposium on Theory of
  Computing ({STOC}~'88)}, pages 258\nobreakdash--266, 1988.

\bibitem{KS07StateofSAT}
Henry Kautz and Bart Selman.
\newblock The state of {SAT}.
\newblock {\em Discrete Applied Mathematics}, 155(12):1514\nobreakdash--1524,
  June 2007.

\bibitem{K80TightBoundPebblesPyramid}
Maria~M. Klawe.
\newblock A tight bound for black and white pebbles on the pyramid.
\newblock {\em Journal of the ACM}, 32(1):218\nobreakdash--228, January 1985.

\bibitem{Kullmann00Matroid}
Oliver Kullmann.
\newblock An application of matroid theory to the {SAT} problem.
\newblock In {\em Proceedings of the 15th Annual {IEEE} Conference on
  Computational Complexity ({CCC}~'00)}, pages 116\nobreakdash--124, July 2000.

\bibitem{LT80SpaceComplexityPebbleGamesTrees}
Thomas Lengauer and Robert~Endre Tarjan.
\newblock The space complexity of pebble games on trees.
\newblock {\em Information Processing Letters}, 10(4/5):184\nobreakdash--188,
  July 1980.

\bibitem{MadH81ComparisonOfTwoVariationsOfPebbleGame}
Friedhelm Meyer auf~der Heide.
\newblock A comparison of two variations of a pebble game on graphs.
\newblock {\em Theoretical Computer Science}, 13(3):315\nobreakdash--322, 1981.

\bibitem{Nordstrom05NarrowProofsMayBeSpacious}
Jakob Nordström.
\newblock {\em Narrow Proofs May Be Spacious: Separating Space and Width in
  Resolution}.
\newblock Technical Report TR05-066, Revision 02, Electronic Colloquium on
  Computational Complexity (ECCC), November 2005.

\bibitem{Nordstrom06NarrowProofsMayBeSpaciousSTOCtoappearinSICOMP}
Jakob Nordström.
\newblock Narrow proofs may be spacious: Separating space and width in
  resolution ({E}xtended abstract).
\newblock In {\em Proceedings of the 38th Annual ACM Symposium on Theory of
  Computing ({STOC}~'06)}, pages 507\nobreakdash--516, May 2006.
\newblock Journal version to appear in \emph{SIAM Journal on Computing}.

\bibitem{Nordstrom07SimplifiedWay}
Jakob Nordström.
\newblock {\em A Simplified Way of Proving Trade-off Results for Resolution}.
\newblock Technical Report TR07-114, Electronic Colloquium on Computational
  Complexity (ECCC), September 2007.

\bibitem{NH08TowardsOptimalSeparationSTOC}
Jakob Nordström and Johan Håstad.
\newblock Towards an optimal separation of space and length in resolution
  ({E}xtended abstract).
\newblock In {\em Proceedings of the 40th Annual ACM Symposium on Theory of
  Computing ({STOC}~'08)}, May 2008.
\newblock To appear.

\bibitem{P92ComputationalComplexity}
Christos~H. Papadimitriou.
\newblock {\em Computational Complexity}.
\newblock Addison-Wesley, 1994.

\bibitem{PW88ComplexityOfFacets}
Christos~H. Papadimitriou and David Wolfe.
\newblock The complexity of facets resolved.
\newblock {\em Journal of Computer and System Sciences},
  37(1):2\nobreakdash--13, 1988.

\bibitem{PTC76SpaceBounds}
Wolfgang~J. Paul, Robert~Endre Tarjan, and James~R. Celoni.
\newblock Space bounds for a game on graphs.
\newblock {\em Mathematical Systems Theory}, 10:239\nobreakdash--251, 1977.

\bibitem{P80Pebbling}
Nicholas Pippenger.
\newblock {\em Pebbling}.
\newblock Technical Report RC8258, IBM Watson Research Center, 1980.
\newblock Appeared in Proceedings of the 5th IBM Symposium on Mathematical
  Foundations of Computer Science, Japan.

\bibitem{RM99Separation}
Ran Raz and Pierre McKenzie.
\newblock Separation of the monotone {NC} hierarchy.
\newblock {\em Combinatorica}, 19(3):403\nobreakdash--435, March 1999.

\bibitem{R65Machine-oriented}
John~Alan Robinson.
\newblock A machine-oriented logic based on the resolution principle.
\newblock {\em Journal of the ACM}, 12(1):23\nobreakdash--41, January 1965.

\bibitem{Sabharwal05Thesis}
Ashish Sabharwal.
\newblock {\em Algorithmic Applications of Propositional Proof Complexity}.
\newblock PhD thesis, University of Washington, Seattle, 2005.

\bibitem{SBK03UsingProblemStructure}
Ashish Sabharwal, Paul Beame, and Henry Kautz.
\newblock Using problem structure for efficient clause learning.
\newblock In {\em 6th International Conference on Theory and Applications of
  Satisfiability Testing ({SAT}~'03), Selected Revised Papers}, volume 2919 of
  {\em Lecture Notes in Computer Science}, pages 242\nobreakdash--256.
  Springer, 2004.

\bibitem{SATcompetition}
{\em {T}he international {SAT} {C}ompetitions web page}.
\newblock \verb+http://www.satcompetition.org+.

\bibitem{Segerlind07Complexity}
Nathan Segerlind.
\newblock The complexity of propositional proofs.
\newblock {\em Bulletin of Symbolic Logic}, 13(4):482\nobreakdash--537,
  December 2007.

\bibitem{S96}
Gunnar St{\aa}lmarck.
\newblock Short resolution proofs for a sequence of tricky formulas.
\newblock {\em Acta Informatica}, 33(3):277\nobreakdash--280, May 1996.

\bibitem{Toran99LowerBounds}
Jacobo Torán.
\newblock Lower bounds for space in resolution.
\newblock In {\em Proceedings of the 13th International Workshop on Computer
  Science Logic ({CSL} '99)}, volume 1683 of {\em Lecture Notes in Computer
  Science}, pages 362\nobreakdash--373. Springer, 1999.

\bibitem{Toran04Space}
Jacobo Torán.
\newblock Space and width in propositional resolution.
\newblock {\em Bulletin of the European Association for Theoretical Computer
  Science}, 83:86\nobreakdash--104, June 2004.

\bibitem{T68ComplexityTranslated}
Grigori Tseitin.
\newblock On the complexity of derivation in propositional calculus.
\newblock In A.~O. Silenko, editor, {\em Structures in Constructive Mathematics
  and Mathematical Logic, Part II}, pages 115\nobreakdash--125. Consultants
  Bureau, New~York-London, 1968.

\bibitem{U87HardExamples}
Alasdair Urquhart.
\newblock Hard examples for resolution.
\newblock {\em Journal of the ACM}, 34(1):209\nobreakdash--219, January 1987.

\end{thebibliography}

%
%

\proofcomment{Last ``proof comment''. (To know how many there are.)}
\editcomment{Last ``edit comment''. (To know how many there are.)}

\end{document}